\UseRawInputEncoding

\documentclass[11pt]{article}
\usepackage{}
\usepackage{amssymb}
\usepackage{amsfonts}
\usepackage{mathrsfs}
\usepackage{graphicx}
\usepackage{amsmath}
\usepackage{color}
\usepackage{amsthm}
\usepackage{epstopdf}
\usepackage{subfigure}
\usepackage{pifont,bm}
\usepackage{tikz}
\usepackage{enumerate}
\usepackage{mathrsfs}
\usepackage{cite}
\usepackage{hyperref}
\usepackage{booktabs}
\usepackage{diagbox}
\usepackage{makecell}

\newtheorem{thm}{Theorem}[section]

\newtheorem{prop}[thm]{Proposition}
\theoremstyle{definition}

\newtheorem{rem}[thm]{Remark}

\makeatletter
\def\@biblabel#1{[#1]}
\makeatother

\makeatletter \@addtoreset{equation}{section}

\begin{document}
%\begin{CJK*}{GBK}{song}

\begin{titlepage}
\title{\bf{Double and Triple-Pole Solutions for the Third-Order Flow Equation of the Kaup-Newell System with Zero/Nonzero Boundary Conditions
\footnote{Corresponding authors.\protect\\
\hspace*{3ex} E-mail addresses: ychen@sei.ecnu.edu.cn (Y. Chen)}
}}
\author{Juncai Pu$^{a}$, Yong Chen$^{a,b,*}$\\
%%%%%%%%%%%%%%%%%%%%%%%%%%%%%%%%%%%%%%%%%%%%%%%%%%%%%%%%%%%%%%%%%%%%%%%%%%%%%%%%%%%%%%%%%
%%%%%              以下两行为作者单位
%%%%%%%%%%%%%%%%%%%%%%%%%%%%%%%%%%%%%%%%%%%%%%%%%%%%%%%%%%%%%%%%%%%%%%%%%%%%%%%%%%%%%%%%%
\small \emph{$^{a}$School of Mathematical Sciences, Shanghai Key Laboratory of Pure Mathematics and} \\
\small \emph{Mathematical Practice, East China Normal University, Shanghai, 200241, China} \\
\small \emph{$^{b}$College of Mathematics and Systems Science, Shandong University }\\
\small \emph{of Science and Technology, Qingdao, 266590, China} \\
\date{}}
\thispagestyle{empty}
\end{titlepage}
\maketitle

\vspace{-0.5cm}
\begin{center}
\rule{15cm}{1pt}\vspace{0.3cm}

\parbox{15cm}{\small
{\bf Abstract}\\
\hspace{0.5cm}
In this work, the double and triple-pole solutions for the third-order flow equation of Kaup-Newell system (TOFKN) with zero boundary conditions (ZBCs) and non-zero boundary conditions (NZBCs) are investigated by means of the Riemann-Hilbert (RH) approach stemming from the inverse scattering transformation. Starting from spectral problem of the TOFKN, the analyticity, symmetries, asymptotic behavior of the Jost function and scattering matrix, the matrix RH problem with ZBCs and NZBCs are constructed. Then the obtained RH problem with ZBCs and NZBCs can be solved in the case of scattering coefficients with double or triple zeros, and the reconstruction formula of potential, trace formula as well as theta condition are also derived correspondingly. Specifically, the general formulas of $N$-double and $N$-triple poles solutions with ZBCs and NZBCs are derived systematically by means of determinants. The vivid plots and dynamics analysis for double and triple-pole soliton solutions with the ZBCs as well as double and triple-pole interaction solutions with the NZBCs are exhibited in details. Compared with the most classical second-order flow Kaup-Newell system, we find the third-order dispersion and quintic nonlinear term of the Kaup-Newell system change the trajectory and velocity of solutions. Furthermore, the asymptotic states of the 1-double poles soliton solution and the 1-triple poles soliton solution are analyzed when $t$ tends to infinity.

}

\vspace{0.5cm}
\parbox{15cm}{\small{

\vspace{0.3cm} \emph{Key words: Third-order flow equation of Kaup-Newell system; Riemann-Hilbert approach; Zero/nonzero boundary conditions; Double-pole solutions; Triple-pole solutions}  \\

%\emph{PACS numbers:}  02.30.Ik, 05.45.Yv, 07.05.Mh.
}}
\end{center}
\vspace{0.3cm} \rule{15cm}{1pt} \vspace{0.2cm}

\section{Introduction}\label{Sec1}
For decades, nonlinear integrable systems, which are used to describe complex natural phenomena in the real world, have always been the research hotspots in the field of nonlinear science \cite{Lax(1968),Ablowitz(2003)}. Recently, since abundant solutions for nonlinear integrable systems are one of the most significant characteristics to reveal complex natural phenomena, the study of the integrability and exact solutions for nonlinear integrable systems have been paid more and more attention in optical fiber, fluid dynamics, plasma physics, machine learning and others fields \cite{Hasegawa(1973),Iwao(1997),YuS(1998),Osman(2019),Raissi(2019)}. General speaking, it is difficult to find the localized waves of complex nonlinear integrable systems, with the development of soliton theory in recent decades, some effective methods for solving nonlinear integrable systems have been established, such as inverse scattering transformation (IST) \cite{Gardner(1967)}, Hirota bilinear method \cite{Hirota(2004)}, Darboux transformation \cite{Matveev(1991)}, B\"{a}cklund transformation \cite{Satsuma(1974)}, symmetry theory \cite{Olver(1993)} and deep learning method \cite{Pu(2021)}.

Among these methods mentioned above, the IST is one of the most important and basic theories for solving nonlinear integrable systems. In 1967, Gardner, Greene, Kruskal and Miura (GGKM) discovered the classical IST to solve the initial value problem for the Korteweg-de Vries (KdV) equation with lax pairs \cite{Gardner(1967)}. After that, many researchers try to extend this method to other nonlinear integrable systems which possess so-called lax pairs \cite{Lax(1968)}. Zakharov and Shabat studied the IST of nonlinear Schr\"{o}dinger (NLS) equation in 1972 \cite{Zakharov(1972)}. Later, Ablowitz, Kaup, Newell and Segur (AKNS) proposed a new class of integrable systems, called AKNS systems, and established a general framework for their ISTs in 1973 \cite{Ablowitz(1973),AblowitzM(1973)}. Subsequently, many classical nonlinear integrable systems are proved to be solvable by IST \cite{Wadati(1973),AblowitzMJ(1973),Fokas(1983),Ablowitz(1983)}. However, it is found that the solving process of the classical IST for nonlinear integrable systems with second order spectral problems which was based on the Gel'fand-Levitan-Marchenko integral equations is complicated and tedious, and it is difficult to solve the original IST for the higher order spectral problems of nonlinear integrable systems without the Gel'fand-Levitan-Marchenko theory \cite{Beals(1984)}. Later on, a Riemann-Hilbert (RH) approach was developed which streamlines and simplifies the IST considerably by Zakharov and his collaborators \cite{Zakharov(1984)}. Since 1980s, the RH approach has been applied to nonlinear integrable systems as a more general method than the classical IST. In 1993, Deift and Zhou presented a new and general steepest descent approach to analyzing the asymptotics of oscillatory RH problems, and such problems will arise in the process of evaluating the long-time behavior of nonlinear wave equations solvable by the RH method \cite{Deift(1993)}. Recently, the RH approach has became a powerful tool for constructing IST to obtain abundant solutions of nonlinear integrable systems and dealing with the long time asymptotic behavior of solutions by analysing RH problem. The multisoliton solutions of some important nonlinear integrable systems can be derived by solving a particular RH problem under the reflectionless cases \cite{Yang(2019),Geng(2016),Zhang(2019),Ma(2019),Guo(2012),Peng(2019),Zhang(2020)}. On the other hand, the RH approach has also been utilized to study the initial boundary value problems and the long-time asymptotic behavior for many nonlinear integrable systems \cite{Deift(1993),Lenells(2012),Bilman(2019),Xu(2015),Wang(2019),XuJ(2020)}.

The Kaup-Newell (KN) systems, which are more complex than AKNS system, are of great significance in mathematical physics \cite{Kaup(1978),Qiao(1993)}. The derivative nonlinear Schr\"{o}dinger (DNLS) equation is the most classical second-order flow KN system, which was first derived from Alfv\'{e}n wave propagation in plasma by Mio et al. in 1976, and well described the propagation of small amplitude nonlinear Alfv\'{e}n wave in low plasma \cite{Mio(1976)}. In Ref. \cite{Kaup(1978)}, Kaup and Newell obtained the one-soliton solution and the infinity of conservation laws for the DNLS via using inverse scattering technique. Furthermore, other important KN systems have been widely studied by means of IST and RH approach, such as high order KN equation\cite{ZhuJY(2021)}, the Chen-Lee-Liu equation \cite{Xu(2019)}, the Kundu-type equation \cite{Wen(2020)} and the Gerdjikov-Ivanov equation \cite{ZhangZ(2020)}. However, as far as we know, these aforementioned IST works on these nonlinear integrable systems with ZBCs/NZBCs mainly focus on the case that all discrete spectra are simple. Therefore, a new research topic is to obtain multiple-pole solutions for nonlinear integrable systems with ZBCs/NZBCs by utilizing RH method based on IST, and there are few related works on this aspect at present. For the IST of double-pole solutions, it is worth mentioning that Pichler presented IST for the focusing NLS equation with non-zero boundary conditions at infinity and double zeros of the analytic scattering coefficients \cite{Pichler(2017)}, as well as some important nonlinear integrable systems have been investigated, including defocusing mKdV equations \cite{ZhangG(2020)}, the NLS equation with quartic terms \cite{Wen-AML-2022} and DNLS equation \cite{Zhangg(2020)}. For the IST of triple-pole solutions, there are few nonlinear integrable systems studied, such as the focusing NLS equation \cite{WengWF-PLA-(2021)}, the Gerdjikov-Ivanov equation \cite{Peng(2021)}, the DNLS equation only with ZBCs at infinity \cite{LiuN-AML-(2022)}. In this paper, we consider the double and triple-pole solutions for the third-order flow equation of the KN system (TOFKN) which contains the third-order dispersion and quintic nonlinear terms. In 1999, Imai firstly proposed the coupled TOFKN \cite{Imai(1999)} and pointed out that it is an integrable system, the coupled TOFKN as shown below
\begin{align}\nonumber
\begin{split}
q_{t}=\frac{a_6}{4\varepsilon^2}\bigg(q_{xxx}-\frac{3(qrq_x)_x}{\varepsilon}+\frac{3(q^3r^2)_x}{2\varepsilon^2}\bigg),\\
r_{t}=\frac{a_6}{4\varepsilon^2}\bigg(r_{xxx}+\frac{3(qrr_x)_x}{\varepsilon}+\frac{3(q^2r^3)_x}{2\varepsilon^2}\bigg),
\end{split}
\end{align}
where $q=q(x,t)$, $r=r(x,t)$.

By imposing the condition $r(x,t) = -q(x,t)^*$ (superscript $``*"$ denotes complex conjugation) and taking appropriate parameters $\varepsilon=\mathrm{i},a_6=4$ to the coupled TOFKN, the coupled TOFKN is reduced and obtained the general form of the TOFKN:
\begin{align}\label{T1}
q_{t}+q_{xxx}-3\mathrm{i}(|q|^2q_x)_x-\frac32(|q|^4q)_x=0,
\end{align}
where $|q|^2=q(x,t)q^*(x,t)$. To the best of our knowledge, there are few studies on Eq. \eqref{T1}, and the IST with ZBCs and NZBCs for the \eqref{T1} has not been investigated by utilizing RH method. Lin et al. studied this equation and derived different types of solutions which contain solitons, positons, breathers and rogue waves by using Darboux transformation and generalized Darboux transformation for the KN systems \cite{Lin(2020)}. The TOFKN is completely integrable and associated with the following modified Zakharov-Shabat eigenvalue problem (Lax pairs) \cite{Kaup(1978)}:
\begin{align}\label{T2}
\Psi_x=X\Psi,
\end{align}
\begin{align}\label{T3}
\Psi_t=T\Psi,
\end{align}
where
\begin{align}\nonumber
\begin{split}
&X=X(x,t,\lambda)=\lambda(\mathrm{i}\lambda\sigma_3+Q),\\
&T=T(x,t,\lambda)=\Big(4\lambda^4+\frac32|q|^4-2\lambda^2|q|^2\Big)X+\left(-2\mathrm{i}\lambda^3+3\mathrm{i}\lambda|q|^2\right)\sigma_3Q_x\\
&\qquad\qquad\qquad\quad-\lambda^2(-qq^*_x+q_xq^*)\sigma_3-\lambda Q_{xx},
\end{split}
\end{align}
it is easy to check that Eq. \eqref{T1} as the integrability condition (or zero-curvature condition) $\Psi_{xt}=\Psi_{tx}$ $\big(\text{or } X_t-T_x+[X,T]=0\big)$ of system Eqs. \eqref{T2}-\eqref{T3}, $\Psi=\Psi(\lambda;x,t)$ is a $2\times2$ matrix-valued eigenfunction, $\lambda\in\mathbb{C}$, the potential matrix $Q=Q(x,t)$ is written as
\begin{align}\label{T4}
Q=\bigg(\begin{array}{cc} 0 & q(x,t) \\-q^*(x,t) & 0 \end{array}\bigg),
\end{align}
and the $\sigma_3$ is one of the Pauli's spin matrices given by
\begin{align}\label{T5}
\sigma_1=\bigg(\begin{array}{cc} 0 & 1 \\1 & 0 \end{array}\bigg),\quad\sigma_2=\bigg(\begin{array}{cc} 0 & -\mathrm{i} \\\mathrm{i} & 0 \end{array}\bigg),\quad\sigma_3=\bigg(\begin{array}{cc} 1 & 0 \\0 & -1 \end{array}\bigg).
\end{align}

This paper is organized as follows. In section \ref{Sec2}, the IST for the TOFKN \eqref{T1} with ZBC at infinity is introduced and solved for the double and triple zeros of analytically scattering coefficients by means of the matrix RH problem. The determinant form of explicit $N$-double-pole solutions and $N$-triple-pole solutions have been presented, and the relative plots at fixed $N$ have been given out in detail. Compared with the DNLS in the case of $N=1$, and the correlation analyses are illustrated in detail. Furthermore, we also analyze the asymptotic states of the 1-double poles soliton solution and the 1-triple poles soliton solution when $t$ tends to infinity. In section \ref{Sec3}, the IST for the TOFKN \eqref{T1} with NZBC at infinity is introduced and solved for the double and triple zeros of analytically scattering coefficients. The IST for NZBC at infinity is more complicated than the case of ZBCs since more symmetries and multivalued functions, thus we map the original spectral parameter $\lambda$ into single-valued parameter $z$ by introducing an appropriate two-sheeted Riemann surface. As a result, we derive the general expression for the $N$-double-pole solutions and the $N$-triple-pole solutions for the case of NZBCs by means of determinants, and exhibit vivid plots for the double-pole and triple-pole solutions as taking different $N_1$ and $N_2$ in detail. Moreover, some comparison made between the double-pole dark-bright soliton solution of the TOFKN and the DNLS equation, and the correlation analyses are explained detailly. Summary and conclusion are given out in last section \ref{Sec4}.

\section{The Construction and Solve of RH problem with ZBCs}\label{Sec2}
In this section, we will build the RH problem for the TOFKN \eqref{T1} with ZBCs at infinity as follow
\begin{align}\label{T6}
q(x,t)\sim0,\text{ as }x\rightarrow\pm\infty.
\end{align}

In the following subsection, we will present the IST which contain the direct scattering and the inverse problem for Eq. \eqref{T1} with ZBCs by RH approach respectively.

\subsection{The Direct Scattering with ZBCs}

\subsubsection{Jost Solution, Analyticity and Continuity}
Considering the asymptotic scattering spectrum problem ($x\rightarrow\infty$) of the modified Zakharov-Shabat eigenvalue problem
\begin{align}\label{T7}
\Psi_x=X_0\Psi,
\end{align}
\begin{align}\label{T8}
\Psi_t=T_0\Psi,
\end{align}
where $X_0=\mathrm{i}\lambda^2\sigma_3$ and $T_0=4\lambda^4X_0=4\mathrm{i}\lambda^6\sigma_3$, one can obtain the fundamental matrix solution $\Psi^{\mathrm{bg}}(\lambda;x,t)$ of Eqs. \eqref{T7}-\eqref{T8}
\begin{align}\nonumber
\Psi^{\mathrm{bg}}(\lambda;x,t)=\mathrm{e}^{\mathrm{i}\theta(\lambda;x,t)\sigma_3},\quad\theta(\lambda;x,t)=\lambda^2(x+4\lambda^4t).
\end{align}

Let $\Sigma:=\mathbb{R}\cup \mathrm{i}\mathbb{R}$. Then, the Jost solutions $\psi_{\pm}(\lambda;x,t)$ can be derived as below
\begin{align}\label{T9}
\psi_{\pm}(\lambda;x,t)\thicksim \mathrm{e}^{\mathrm{i}\theta(\lambda;x,t)\sigma_3},\quad\lambda\in\Sigma,\text{ as }x\rightarrow\pm\infty.
\end{align}

In order to obtain the modified Jost solution $\mu_{\pm}(\lambda;x,t)$, we consider a transformation of the form
\begin{align}\label{T10}
\mu_{\pm}(\lambda;x,t)=\psi_{\pm}(\lambda;x,t)\mathrm{e}^{-\mathrm{i}\theta(\lambda;x,t)\sigma_3},
\end{align}
it is evident that
\begin{align}\nonumber
\mu_{\pm}(\lambda;x,t)\thicksim I,\text{ as }x\rightarrow\pm\infty,
\end{align}
where $I$ is $2\times2$ identity matrix. According to the modified Zakharov-Shabat eigenvalue problem Eqs. \eqref{T2}-\eqref{T3} and transformation Eq. \eqref{T10}, one can obtain $\mu_{\pm}(\lambda;x,t)$ satisfy the following equivalent Lax pair:
\begin{align}\label{BT2}
\mu_{\pm,x}(\lambda;x,t)+\mathrm{i}\lambda^2[\mu_{\pm}(\lambda;x,t),\sigma_3]=\lambda Q(x,t)\mu_{\pm}(\lambda;x,t),
\end{align}
\begin{align}\label{BT3}
\mu_{\pm,t}(\lambda;x,t)+4\mathrm{i}\lambda^6[\mu_{\pm}(\lambda;x,t),\sigma_3]=[T(\lambda;x,t)-T_0]\mu_{\pm}(\lambda;x,t),
\end{align}
where the Lie bracket $[L_1,L_2]=L_1L_2-L_2L_1$. Eqs. \eqref{BT2} and \eqref{BT3} can be written in full derivative form
\begin{align}\label{BT4}
d(\mathrm{e}^{-\mathrm{i}\theta(\lambda;x,t)\widehat{\sigma_3}}\mu_{\pm}(\lambda;x,t))=\mathrm{e}^{-\mathrm{i}\theta(\lambda;x,t)\widehat{\sigma_3}}\big\{\big[\lambda Q(x,t)dx+[T(\lambda;x,t)-T_0]dt\big]\mu_{\pm}(\lambda;x,t)\big\},
\end{align}
and $\mu_{\pm}(\lambda;x,t)$ satisfy the Volterra integral equations
\begin{align}\label{T11}
\mu_{\pm}(\lambda;x,t)=I+\int_{\pm\infty}^x\mathrm{e}^{\mathrm{i}\lambda^2(x-y)\widehat{\sigma_3}}(\lambda Q(y,t)\mu_{\pm}(\lambda;y,t))dy,
\end{align}
where $\mathrm{e}^{\chi\widehat{\sigma_3}}E:=\mathrm{e}^{\chi\sigma_3}E\mathrm{e}^{-\chi\sigma_3}$ with $E$ being a $2\times2$ matrix. Let $D^{\pm}:=\{\lambda\in\mathbb{C}|\pm\mathrm{Re}(\lambda)\mathrm{Im}(\lambda)>0\}$, as shown in Fig. \ref{F1}.

\begin{figure}[htbp]
\centerline{\begin{tikzpicture}[scale=1.5]
\path [fill=pink] (-2.5,0) -- (-0.5,0) to
(-0.5,2) -- (-2.5,2);
\path [fill=pink] (-4.5,0) -- (-2.5,0) to
(-2.5,-2) -- (-4.5,-2);
\draw[-][thick](-4.5,0)--(-2.5,0);
\draw[fill] (-2.5,0) circle [radius=0.03];
\draw[->][thick](-2.5,0)--(-0.5,0)node[right]{$\mbox{Re}\lambda$};
\draw[->][thick](-2.5,1)--(-2.5,2)node[above]{$\mbox{Im}\lambda$};
\draw[-][thick](-2.5,1)--(-2.5,0);
\draw[-][thick](-2.5,0)--(-2.5,-1);
\draw[-][thick](-2.5,-1)--(-2.5,-2);
\draw[fill] (-2.5,-0.3) node[right]{$0$};
\draw[fill] (-1.7,0.8) circle [radius=0.03] node[right]{$\lambda_{n}$};
\draw[fill] (-1.7,-0.8) circle [radius=0.03] node[right]{$\lambda^{*}_{n}$};
\draw[fill] (-3.3,0.8) circle [radius=0.03] node[left]{$-\lambda^{*}_{n}$};
\draw[fill] (-3.3,-0.8) circle [radius=0.03] node[left]{$-\lambda_{n}$};
\end{tikzpicture}}
\caption{(Color online) Distribution of the discrete spectrum and jumping curves for the RH problem on complex $\lambda$-plane. Region $D^{+}=\left\{\lambda\in \mathbb{C} \big| \mathrm{Re}\lambda\mathrm{Im}\lambda> 0\right\}$ (pink region), region $D^{-}=\left\{\lambda\in \mathbb{C} \big|\mathrm{Re}\lambda\mathrm{Im}\lambda< 0\right\}$ (white region).}
\label{F1}
\end{figure}
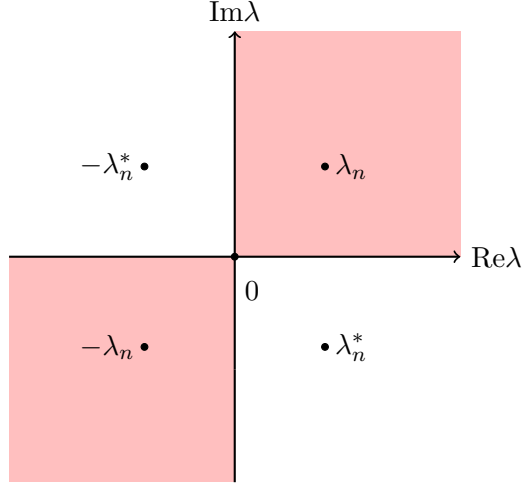

Furthermore, one has the following proposition.
\begin{prop}\label{P1}
Suppose that $q(x,t)\in L^1(\mathbb{R})$ and $\psi_{\pm i}(\lambda;x,t)$ $(\mu_{\pm i}(\lambda;x,t))$ represent the $i$th column of $\psi_{\pm}(\lambda;x,t)$ $\big(\text{modified Jost solutions }\mu_{\pm}(\lambda;x,t)\big)$. Then, the Jost solutions $\psi_{\pm}(\lambda;x,t)$ $\big(\mu_{\pm}(\lambda;x,t)\big)$ possess the properties:\\
$\bullet$ Eq. \eqref{T2} has the unique Jost solutions $\psi_{\pm}(\lambda;x,t)$ $\big(\mu_{\pm}(\lambda;x,t)\big)$ satisfying Eq. \eqref{T9} $($combination of Eqs. \eqref{T9}-\eqref{T10}$)$ on $\Sigma$.\\
$\bullet$ The column vectors $\psi_{+1}(\lambda;x,t)$ $(\mu_{+1}(\lambda;x,t))$ and $\psi_{-2}(\lambda;x,t)$ $(\mu_{-2}(\lambda;x,t))$ can be analytically extended to $D^+$ and continuously extended to $D^+\cup\Sigma$.\\
$\bullet$ The column vectors $\psi_{+2}(\lambda;x,t)$ $(\mu_{+2}(\lambda;x,t))$ and $\psi_{-1}(\lambda;x,t)$ $(\mu_{-1}(\lambda;x,t))$ can be analytically extended to $D^-$ and continuously extended to $D^-\cup\Sigma$.\\
\end{prop}

\begin{proof}
We define modified Jost solutions $\mu_{\pm}=(\mu_{\pm1},\mu_{\pm2})=\bigg(\begin{array}{cc} \mu_{\pm11} & \mu_{\pm21} \\ \mu_{\pm12} & \mu_{\pm22} \end{array}\bigg)$, in which $\mu_{\pm1}$ and $\mu_{\pm2}$ is the first and second columns of $\mu_{\pm}$, respectively. Then, taking $\mu_{-}$ as an example, Eq. \eqref{T11} can be rewritten as
\begin{align}\nonumber
\bigg(\begin{array}{cc} \mu_{-11} & \mu_{-21} \\ \mu_{-12} & \mu_{-22} \end{array}\bigg)=\bigg(\begin{array}{cc} 1 & 0 \\0 & 1 \end{array}\bigg)+\lambda\int_{\pm\infty}^x\Bigg(\begin{array}{cc} q\mu_{-12} & q\mu_{-22}\mathrm{e}^{2\mathrm{i}\lambda^2(x-y)} \\-q^*\mu_{-11}\mathrm{e}^{-2\mathrm{i}\lambda^2(x-y)} & -q^*\mu_{-21} \end{array}\Bigg)dy,
\end{align}
note that $\mu_{-1}=\bigg(\begin{array}{cc} 1 \\ 0 \end{array}\bigg)+\lambda\int_{\pm\infty}^x\bigg(\begin{array}{cc} q\mu_{-12}  \\-q^*\mu_{-11}\mathrm{e}^{-2\mathrm{i}\lambda^2(x-y)} \end{array}\bigg)dy$, in which, we have $\mathrm{e}^{-2\mathrm{i}\lambda^2(x-y)}=\mathrm{e}^{-2\mathrm{i}(x-y)\mathrm{Re}(\lambda^2)}\mathrm{e}^{4(x-y)\mathrm{Re}(\lambda)\mathrm{Im}(\lambda)}$. Since $x-y>0$ and $\mathrm{Re}(\lambda)\mathrm{Im}(\lambda)<0$, it indicate that the first column of $\mu_{-}$ is analytically extended to $D^{-}$. In addition, since $\mathrm{Im}(\lambda^2)=0$ when $\lambda\in\Sigma$, it demonstrate that the first column of $\mu_{-}$ is continuously extended to $D^-\cup\Sigma$. In the same way, we also obtain the analyticity and continuity of $\mu_{-2},\mu_{+1}$ and $\mu_{+2}$. Similarly, the analyticity and continuity for the Jost solutions $\psi_{\pm}(\lambda;x,t)$ can be simply shown from ones of $\mu_{\pm}(\lambda;x,t)$ and the relation \eqref{T10}.

\end{proof}

\begin{prop}\label{P10}
The Jost solutions $\psi_{\pm}(\lambda;x,t)$ satisfy both parts of the modified Zakharov-Shabat eigenvalue problem \eqref{T2}-\eqref{T3} simultaneously.
\end{prop}

\begin{proof}
The Liouville's formula leads to
\begin{align}\nonumber
\mathrm{det}(\psi_{\pm}(\lambda;x,t))=\lim_{x\rightarrow\pm\infty}\mathrm{det}(\psi_{\pm}(\lambda;x,t))=\lim_{x\rightarrow\pm\infty}\mathrm{det}(\mu_{\pm}(\lambda;x,t))=1,
\end{align}
that show $\psi_{\pm}(\lambda;x,t)$ are the fundamental matrix solutions on $\Sigma$. According to the zero-curvature condition $X_t-T_x+[X,T]=0$, one can obtain that $\psi_{\pm,t}(\lambda;x,t)-T\psi_{\pm}(\lambda;x,t)$ also solve the $x$-part \eqref{T2}, that is
\begin{align}\nonumber
\begin{split}
&(\psi_{\pm,t}(\lambda;x,t)-T\psi_{\pm}(\lambda;x,t))_x=X(\psi_{\pm,t}(\lambda;x,t)-T\psi_{\pm}(\lambda;x,t)),\\
&X_t\psi_{\pm}(\lambda;x,t)+X\psi_{\pm,t}(\lambda;x,t)-T_x\psi_{\pm}(\lambda;x,t)-TX\psi_{\pm}(\lambda;x,t)\\
&=X\psi_{\pm,t}(\lambda;x,t)-XT\psi_{\pm}(\lambda;x,t),\\
&(X_t-T_x+XT-TX)\psi_{\pm}(\lambda;x,t)=0.
\end{split}
\end{align}

Thus, there exist the two matrices $U_{\pm}(\lambda;t)$ such that
\begin{align}\nonumber
\psi_{\pm,t}(\lambda;x,t)-T\psi_{\pm}(\lambda;x,t)=\psi_{\pm}(\lambda;x,t)U_{\pm}(\lambda;t), \text{ as } \lambda\in\Sigma,
\end{align}
multiplying both sides by $\mathrm{e}^{-\mathrm{i}\theta(\lambda;x,t)\sigma_3}$, we have
\begin{align}\nonumber
\begin{split}
(\mu_{\pm}(\lambda;x,t))_t+T_0\mu_{\pm}(\lambda;x,t)-T\mu_{\pm}(\lambda;x,t)=\psi_{\pm}(\lambda;x,t)U_{\pm}(\lambda;t)\mathrm{e}^{-\mathrm{i}\theta(\lambda;x,t)\sigma_3},\text{ as } \lambda\in\Sigma,
\end{split}
\end{align}
and letting $x\rightarrow\pm\infty$, one can find $\mu_{\pm}(\lambda;x,t)\thicksim I,T\thicksim T_0$, and $U_{\pm}(\lambda;t)=0$, that is, $\psi_{\pm}(\lambda;x,t)$ also solve the $t$-part \eqref{T3}.

\end{proof}

\subsubsection{Scattering Matrix and Reflection Coefficients}
Since the Jost solutions $\psi_{\pm}(\lambda;x,t)$ solve the both parts of the modified Zakharov-Shabat eigenvalue problem \eqref{T2}-\eqref{T3}. Therefore, there exist a constant scattering matrix $S(\lambda)=(s_{ij}(\lambda))_{2\times2}$ independent of $x$ and $t$ to satisfy the linear relation between $\psi_{+}(\lambda;x,t)$ and $\psi_{-}(\lambda;x,t)$, where $s_{ij}(\lambda)$ are called the scattering coefficients, we have
\begin{align}\label{T12}
\psi_{+}(\lambda;x,t)=\psi_{-}(\lambda;x,t)S(\lambda),\quad\lambda\in\Sigma,
\end{align}
let $$\psi_+(\lambda;x,t)=(\psi_{+1},\psi_{+2})=\bigg(\begin{array}{cc} \psi_{+11} & \psi_{+12} \\ \psi_{+21} & \psi_{+22} \end{array}\bigg),\,\psi_-(\lambda;x,t)=(\psi_{-1},\psi_{-2})=\bigg(\begin{array}{cc} \psi_{-11} & \psi_{-12} \\ \psi_{-21} & \psi_{-22} \end{array}\bigg),$$ substituting into Eq. \eqref{T12} and using Cramer's rule, one can obtain
\begin{align}\label{T13}
\begin{split}
&s_{11}(\lambda)=\mathrm{det}(\psi_{+1}(\lambda;x,t),\psi_{-2}(\lambda;x,t)),\,s_{12}(\lambda)=\mathrm{det}(\psi_{+2}(\lambda;x,t),\psi_{-2}(\lambda;x,t)),\\
&s_{21}(\lambda)=\mathrm{det}(\psi_{-1}(\lambda;x,t),\psi_{+1}(\lambda;x,t)),\,s_{22}(\lambda)=\mathrm{det}(\psi_{-1}(\lambda;x,t),\psi_{+2}(\lambda;x,t)),\\
\end{split}
\end{align}
where $\mathrm{det}(\psi_{-}(\lambda;x,t))=1$.

\begin{prop}
Suppose that $q(x,t)\in L^1(\mathbb{R})$. Then, $s_{11}(\lambda)$ can be analytically extended to $D^+$ and continuously extended to $D^+\cup\Sigma$, while $s_{22}(\lambda)$ can be analytically extended to $D^-$ and continuously extended to $D^-\cup\Sigma$. Moreover, both $s_{12}(\lambda)$ and $s_{21}(\lambda)$ are continuous in $\Sigma$.
\end{prop}
\begin{proof}
From proposition \eqref{P1}, we know that $\psi_{+1}(\lambda;x,t)$ and $\psi_{-2}(\lambda;x,t)$ can be extended analytically to $D^+$ and continuously extended to $D^+\cup\Sigma$. From Eqs. \eqref{T13}, we see that $s_{11}(\lambda)=\mathrm{det}(\psi_{+1}(\lambda;x,t),\psi_{-2}(\lambda;x,t))$, one can obtain $s_{11}(\lambda)$ can be analytically extended to $D^+$ and continuously extended to $D^+\cup\Sigma$. Similarly, we can prove $s_{22}(\lambda)$ can be analytically extended to $D^-$ and continuously extended to $D^-\cup\Sigma$. Due to $\psi_{\pm1}(\lambda;x,t)$ and $\psi_{\pm2}(\lambda;x,t)$ are continuous in $\Sigma$, both $s_{12}(\lambda)$ and $s_{21}(\lambda)$ are continuous in $\Sigma$.
\end{proof}

Note that it cannot be ruled out that $s_{11}(\lambda)$ and $s_{22}(\lambda)$ may exist zeros along $\Sigma$. In order to study the RH problem in the inverse process, we focus on the potential without spectral singularity. In general, the reflection coefficients $\rho(\lambda)$ and $\tilde{\rho}(\lambda)$ are defined by
\begin{align}\label{T14}
\rho(\lambda)=\frac{s_{21}(\lambda)}{s_{11}(\lambda)},\quad \tilde{\rho}(\lambda)=\frac{s_{12}(\lambda)}{s_{22}(\lambda)}.
\end{align}

\subsubsection{Symmetry Properties}
\begin{prop}\label{P4}
$X(\lambda;x,t)$, $T(\lambda;x,t)$, Jost solutions, modified Jost solutions, scattering matrix and reflection coefficients have two kinds of symmetry reductions as follow:\\
$\bullet$ The first symmetry reduction\\
\begin{align}\label{T15}
\begin{split}
X(\lambda;x,t)&=\sigma_2X(\lambda^*;x,t)^*\sigma_2,\quad T(\lambda;x,t)=\sigma_2T(\lambda^*;x,t)^*\sigma_2,\\
\psi_{\pm}(\lambda;x,t)&=\sigma_2\psi_{\pm}(\lambda^*;x,t)^*\sigma_2,\quad \mu_{\pm}(\lambda;x,t)=\sigma_2\mu_{\pm}(\lambda^*;x,t)^*\sigma_2,\\
S(\lambda)&=\sigma_2S(\lambda^*)^*\sigma_2,\quad \rho(\lambda)=-\tilde{\rho}(\lambda^*)^*.
\end{split}
\end{align}
$\bullet$ The second symmetry reduction\\
\begin{align}\label{T16}
\begin{split}
X(\lambda;x,t)&=\sigma_1X(-\lambda^*;x,t)^*\sigma_1,\quad T(\lambda;x,t)=\sigma_1T(-\lambda^*;x,t)^*\sigma_1,\\
\psi_{\pm}(\lambda;x,t)&=\sigma_1\psi_{\pm}(-\lambda^*;x,t)^*\sigma_1,\quad \mu_{\pm}(\lambda;x,t)=\sigma_1\mu_{\pm}(-\lambda^*;x,t)^*\sigma_1,\\
S(\lambda)&=\sigma_1S(-\lambda^*)^*\sigma_1,\quad \rho(\lambda)=\tilde{\rho}(-\lambda^*)^*.
\end{split}
\end{align}

\end{prop}

\begin{proof}
First, let's prove the first symmetry reduction. Since we know the specific forms of the matrices $X$ and $T$, we can obtain the symmetry reduction of $X$ and $T$ by direct calculation. In order to obtain the symmetry reduction for Jost solutions, in fact, we just need to demonstrate that $\sigma_2\psi_{\pm}(\lambda^*;x,t)^*\sigma_2$ is also a solution of Eq. \eqref{T2} and admits the same asymptotic behavior as the Jost solutions $\psi_{\pm}(\lambda;x,t)$.

Considering
\begin{align}\label{T17}
\psi_{\pm,x}(\lambda;x,t)=X(\lambda;x,t)\psi_{\pm}(\lambda;x,t),
\end{align}
replacing $\lambda$ with $\lambda^*$, and taking conjugate both sides simultaneously, we can obtain $$\psi_{\pm,x}(\lambda^*;x,t)^*=X(\lambda^*;x,t)^*\psi_{\pm}(\lambda^*;x,t)^*,$$ substituting $X(\lambda^*;x,t)^*=\sigma_2^{-1}X(\lambda;x,t)\sigma_2^{-1}$ into previous formula, then both sides of the obtained formula multiplying matrix $\sigma_2$ on left and right at the same time, one can obtain
\begin{align}\nonumber
\begin{split}
\sigma_2\psi_{\pm,x}(\lambda^*;x,t)^*\sigma_2=X(\lambda;x,t)\sigma^{-1}_2\sigma^{-1}_2(\sigma_2\psi_{\pm}(\lambda^*;x,t)^*\sigma_2),
\end{split}
\end{align}
where $\sigma^{-1}_2=\bigg(\begin{array}{cc} 0 & \mathrm{i} \\-\mathrm{i} & 0 \end{array}\bigg)$, so $\sigma^{-1}_2\sigma^{-1}_2=I$, we have
\begin{align}\label{T18}
\sigma_2\psi_{\pm,x}(\lambda^*;x,t)^*\sigma_2=X(\lambda;x,t)(\sigma_2\psi_{\pm}(\lambda^*;x,t)^*\sigma_2).
\end{align}

By comparing Eqs. \eqref{T17} and \eqref{T18}, we know that $\sigma_2\psi_{\pm,x}(\lambda^*;x,t)^*\sigma_2$ is also a solution of Eq. \eqref{T17}, and $\psi_{\pm}(\lambda;x,t)=\sigma_2\psi_{\pm}(\lambda^*;x,t)^*\sigma_2$ because they have the same asymptotic behavior
$$\psi_{\pm}(\lambda;x,t),\quad \sigma_2\psi_{\pm}(\lambda^*;x,t)^*\sigma_2\thicksim I, \text{ as } x\rightarrow\pm\infty.$$

According to Eq. \eqref{T10} and $\theta(\lambda^*)^*=\theta(\lambda)$, we have $\psi_{\pm}(\lambda;x,t)=\mu_{\pm}(\lambda;x,t)\mathrm{e}^{\mathrm{i}\theta(\lambda;x,t)\sigma_3}$ and $\psi_{\pm}(\lambda^*;x,t)^*=\mu_{\pm}(\lambda^*;x,t)^*\mathrm{e}^{-\mathrm{i}\theta(\lambda;x,t)\sigma_3}$. Untilizing $\psi_{\pm}(\lambda;x,t)=\sigma_2\psi_{\pm}(\lambda^*;x,t)^*\sigma_2$, one can derive $\mu_{\pm}(\lambda;x,t)=\sigma_2\mu_{\pm}(\lambda^*;x,t)^*\mathrm{e}^{-\mathrm{i}\theta(\lambda;x,t)\sigma_3}\sigma_2\mathrm{e}^{-\mathrm{i}\theta(\lambda;x,t)\sigma_3}=\sigma_2\mu_{\pm}(\lambda^*;x,t)^*\sigma_2.$

Next, we demonstrate the symmetry reduction of the scattering matrix. We substitute $\lambda=\lambda^*$ into Eq. \eqref{T12} and conjugate both sides of the equation. We have $\psi_{+}(\lambda^*;x,t)^*=\psi_{-}(\lambda^*;x,t)^*S(\lambda^*)^*$, From the symmetry reduction of Jost solutions, one can obtain $\psi_{+}(\lambda;x,t)=\psi_{-}(\lambda;x,t)\sigma^{-1}_2\sigma^{-1}_2(\sigma_2S(\lambda^*)^*\sigma_2)$, then
\begin{align}\label{T19}
\psi_{+}(\lambda;x,t)=\psi_{-}(\lambda;x,t)(\sigma_2S(\lambda^*)^*\sigma_2).
\end{align}

By comparing Eqs. \eqref{T12} and \eqref{T19}, we can obtain $S(\lambda)=\sigma_2S(\lambda^*)^*\sigma_2$. According to the symmetry reduction of scattering matrix, we have
\begin{align}\nonumber
\bigg(\begin{array}{cc} s_{11}(\lambda) & s_{12}(\lambda) \\s_{21}(\lambda) & s_{22}(\lambda) \end{array}\bigg)=\bigg(\begin{array}{cc} s_{22}(\lambda^*)^* & -s_{21}(\lambda^*)^* \\-s_{12}(\lambda^*)^* & s_{11}(\lambda^*)^* \end{array}\bigg),
\end{align}
we then obtain these symmetric relations
\begin{align}\label{BT1}
\begin{split}
s_{11}(\lambda)=s_{22}(\lambda^*)^*,\quad s_{12}(\lambda)=-s_{21}(\lambda^*)^*,\quad s_{21}(\lambda)=-s_{12}(\lambda^*)^*,\quad s_{22}(\lambda)=s_{11}(\lambda^*)^*,
\end{split}
\end{align}
Since $\rho(\lambda)=\frac{s_{21}(\lambda)}{s_{11}(\lambda)}=\frac{-s_{12}(\lambda^*)^*}{s_{22}(\lambda^*)^*}=-\tilde{\rho}(\lambda^*)^*$, the symmetry reduction of reflection coefficient is $\rho(\lambda)=-\tilde{\rho}(\lambda^*)^*$.

Similarly, by repeating the above process, we can prove the second symmetry reduction. This completes the proof.

\end{proof}

\subsubsection{Asymptotic Behaviors}
In order to propose and solve the matrix RH problem for the inverse problem, the asymptotic behavior of the modified Jost solution and the scattering matrix must be determined as $\lambda\rightarrow\infty$. The asymptotic behaviors of the modified Jost solution can be derived by employing the usual Wentzel-Kramers-Brillouin (WKB) expansion.
\begin{prop}\label{P2}
The asymptotic behaviors of the modified Jost solutions are as follows:
\begin{align}\label{T20}
\mu_{\pm}(\lambda;x,t)=\mathrm{e}^{\mathrm{i}\nu_{\pm}(x,t)\sigma_3}+O\bigg(\frac{1}{\lambda}\bigg),\text{ as }\lambda\rightarrow\infty,
\end{align}
where functions $\nu_{\pm}(x,t)$ read as
\begin{align}\label{T21}
\nu_{\pm}(x,t)=\frac12\int^x_{\pm\infty}|q(y,t)|^2dy.
\end{align}
\end{prop}

\begin{proof}
By utilizing the WKB expansion, as $\lambda\rightarrow\infty$, we substitute $$\mu_{\pm}(\lambda;x,t)=\sum\limits^n_{i=0}\frac{\mu_{\pm}^{[i]}(x,t)}{\lambda^i}+O\left(\frac{1}{\lambda}\right)$$ into equivalent Lax pair \eqref{BT2}, and compare the coefficients of different powers of $\lambda$, we have
\begin{align}\label{T22}
O(\lambda^2):\Big(\mu^{[0]}_{\pm}(x,t)\Big)^{\mathrm{off}}=0,
\end{align}
where $\left(\mu^{[0]}_{\pm}(x,t)\right)^{\mathrm{off}}$ represent the off-diagonal parts of $\mu^{[0]}_{\pm}(x,t)$. Then
\begin{align}\label{T23}
\begin{split}
O(\lambda):\Big(\mu^{[1]}_{\pm}(x,t)\Big)^{\mathrm{off}}=\frac{\mathrm{i}}{2}\sigma_3Q(x,t)\mu^{[0]}_{\pm}(x,t),
\end{split}
\end{align}
\begin{align}\label{T24}
\begin{split}
O(1):\mu^{[0]}_{\pm}(x,t)=C^{\mathrm{diag}}\mathrm{e}^{\mathrm{i}\nu_{\pm}(x,t)\sigma_3},
\end{split}
\end{align}
where $\nu_{\pm}(x,t)=\frac12\int^x_{\pm\infty}|q(y,t)|^2dy$, $\mu_{\pm}(\lambda;x,t)=C^{\mathrm{diag}}\mathrm{e}^{\mathrm{i}\nu_{\pm}(x,t)\sigma_3}+O\left(\frac{1}{\lambda}\right),\text{ as }\lambda\rightarrow\infty$. When $x\rightarrow\pm\infty$, we have $\nu_{\pm}(x,t)\rightarrow0,\mu_{\pm}(\lambda;x,t)\rightarrow I$, one can obtain $C^{\mathrm{diag}}=I$, and deduce the asymptotic behavior as $\lambda\rightarrow\infty$.

\end{proof}

\begin{prop}\label{PP2}
The asymptotic behavior for the scattering matrix is as follows:
\begin{align}\label{T25}
S(\lambda)=\mathrm{e}^{-\mathrm{i}\nu\sigma_3}+O\bigg(\frac{1}{\lambda}\bigg),\text{ as }\lambda\rightarrow\infty,
\end{align}
where the constant $\nu$ reads as
\begin{align}\label{T26}
\nu=\frac12\int^{+\infty}_{-\infty}|q(y,t)|^2dy.
\end{align}

\end{prop}

\begin{proof}
From the relationship among scattering matrix, Jost solution and modified Jost solution, we have $$\mu_{+}(\lambda;x,t)\mathrm{e}^{\mathrm{i}\theta(\lambda;x,t)\sigma_3}=\mu_{-}(\lambda;x,t)\mathrm{e}^{\mathrm{i}\theta(\lambda;x,t)\sigma_3}S(\lambda).$$ According to proposition \eqref{P2}, when $\lambda\rightarrow\infty$, $\mu_{\pm}(\lambda;x,t)=\mathrm{e}^{i\nu_{\pm}(x,t)\sigma_3}+O\left(\frac{1}{\lambda}\right)$, we have
\begin{align}\nonumber
\bigg[\mathrm{e}^{\mathrm{i}\nu_{+}(x,t)\sigma_3}+O\bigg(\frac{1}{\lambda}\bigg)\bigg]\mathrm{e}^{\mathrm{i}\theta(\lambda;x,t)\sigma_3}=\bigg[\mathrm{e}^{\mathrm{i}\nu_{-}(x,t)\sigma_3}+O\bigg(\frac{1}{\lambda}\bigg)\bigg]\mathrm{e}^{\mathrm{i}\theta(\lambda;x,t)\sigma_3}S(\lambda),
\end{align}
then both sides of the equation pre-multiply $\mathrm{e}^{-\mathrm{i}[\nu_{-}(x,t)+\theta(\lambda;x,t)]\sigma_3}$ at the same time, one can obtain
\begin{align}\nonumber
S(\lambda)=\mathrm{e}^{-\mathrm{i}[\nu_{-}(x,t)-\nu_{+}(x,t)]\sigma_3}+o\bigg(\frac{1}{\lambda}\bigg)=\mathrm{e}^{-\mathrm{i}\nu\sigma_3}+O\bigg(\frac{1}{\lambda}\bigg),
\end{align}
where $\nu=\nu_{-}(x,t)-\nu_{+}(x,t)=\frac12\int^{+\infty}_{-\infty}|q(y,t)|^2dy$.

On the other hand, substituting the WKB expansion $\mu_{\pm}(\lambda;x,t)=\sum\limits^n_{i=0}\frac{\mu_{\pm}^{[i]}(x,t)}{\lambda^i}+O\big(\frac{1}{\lambda}\big)$( as $\lambda\rightarrow\infty)$ into the time part of equivalent lax pairs \eqref{BT3} and matching the $O(\lambda^i)(i =6,5,4,3,2,1,0)$ in order, combining the $O(1):\mu^{[0]}_{\pm,t}(x,t)+4\mathrm{i}\big[\mu^{[6]}_{\pm}(x,t),\sigma_3\big]=(T-T_0)\mu^{[0]}_{\pm}(x,t)$ and $\big(\mu^{[0]}_{\pm}(x,t)\big)^{\mathrm{off}}=0$, one can yield $\mu^{[0]}_{\pm,t}(x,t)=0$, $\nu_{\pm,t}(x,t)=0$ and $\nu_t=0$ as $x\rightarrow\pm\infty$. That is, $\nu$ does not depend on the variable $t$. Furthermore, for the Zakharov-Shabat spectral problem of the NLS equation, $I_1=\int^{\infty}_{-\infty}|u(x,t)|^2dx$ is conservations of mass (also called power), and one has $I_{1,t}=0$ \cite{YangJK(2010)}. Here, one can derive $I_1=2\nu=\int^{+\infty}_{-\infty}|q(x,t)|^2dx$ (here $x$ and $y$ in \eqref{T26} are equivalent) is the conservations of mass for the modified Zakharov-Shabat spectral problem of the TOFKN, and we also obtain $\nu_t=0$ by using the $I_{1,t}=0$. Next we give a simple proof of pure algebra for $I_{1,t}=0$.

According to the Eq.\eqref{T1}, we have
\begin{align}\nonumber
\begin{split}
&q_t=\frac92|q|^4q_x+3q^2|q|^2q^*_x+3\mathrm{i}|q|^2q_{xx}+3\mathrm{i}q^2_xq^*+3\mathrm{i}qq^*_xq_x-q_{xxx},\\
&q^*_t=\frac92|q|^4q^*_x+3(q^*)^2|q|^2q_x-3\mathrm{i}|q|^2q^*_{xx}-3\mathrm{i}(q^*_x)^2q-3\mathrm{i}q^*q_xq^*_x-q^*_{xxx},
\end{split}
\end{align}
from $I_1=\int^{+\infty}_{-\infty}|q(x,t)|^2dx$, we have
\begin{align}\nonumber
\begin{split}
I_{1,t}&=\int^{+\infty}_{-\infty}[q(x,t)q^*(x,t)]_tdx\\
&=\frac{5}{2}q^3(q^*)^3\Big|^{+\infty}_{-\infty}+3\mathrm{i}|q|^2(q_xq^*-qq^*_x)\Big|^{+\infty}_{-\infty}-(q_{xx}q^*+qq^*_{xx})\Big|^{+\infty}_{-\infty}+q_xq^*_x\Big|^{+\infty}_{-\infty},
\end{split}
\end{align}
Since $q(x,t)=0\text{ as }x\rightarrow\pm\infty$, one can derive $I_{1,t}=0$ and $\nu_t=0$. Therefore, $\nu$ is a constant. Proof complete.

\end{proof}

\subsection{Inverse Problem with ZBCs and Double Poles}

In the following parts, we will propose an inverse problem with ZBCs and solve it to obtain $N$-double poles solution for the TOFKN \eqref{T1}.

\subsubsection{Discrete Spectrum with with ZBCs and Double Zeros}
The discrete spectrum of the studied scattering problem is the set of all values $\lambda\in\mathbb{C}\backslash\Sigma$ which satisfy the eigenfunctions exist in $L^2(\mathbb{R})$. As was proposed by Biondini and Kova\v{c}i\v{c} \cite{Biondini(2014)}, there are exactly the values of $\lambda$ in $D^{+}$ such that $s_{11}(\lambda)=0$ and those values in $D^{-}$ such that $s_{22}(\lambda)=0$. Differing from the previous results with simple poles \cite{Zhang(2020),Kaup(1978),XuJ(2020)}, we here suppose that $s_{11}(\lambda)$ has $N$ double zeros in $\Lambda_0 = \{\lambda\in\mathbb{C}: \mathrm{Re}\lambda>0, \mathrm{Im}\lambda>0\}$ denoted by $\lambda_n$, $n =1,2,\cdots,N$, that is, $s_{11}(\lambda_n)=s'_{11}(\lambda_n)=0$, and $s''_{11}\neq0$. It follows from the symmetry reduce of the scattering matrix that
\begin{align}\label{T27}
\bigg\{
\begin{aligned}
s_{11}(\pm\lambda_n)=s_{22}(\pm\lambda_n^*)=0,\\
s'_{11}(\pm\lambda_n)=s'_{22}(\pm\lambda_n^*)=0.
\end{aligned}
\end{align}

Therefore, the corresponding discrete spectrum is defined by the set
\begin{align}\label{T28}
\Lambda=\{\pm\lambda_n,\pm\lambda^*_n\}^N_{n=1},
\end{align}
whose distributions are shown in Fig. \ref{F1}. When a given $\lambda_0\in\Lambda\cap D^{+}$, one can obtain that $\psi_{+1}(\lambda_0;x,t)$ and $\psi_{-2}(\lambda_0;x,t)$ are linearly dependent by combining Eq. \eqref{T13} and $s_{11}(\lambda_0)=0$. Similarly, when a given $\lambda_0\in\Lambda\cap D^{-}$, one can obtain that $\psi_{+2}(\lambda_0;x,t)$ and $\psi_{-1}(\lambda_0;x,t)$ are linearly dependent by combining Eq. \eqref{T13} and $s_{22}(\lambda_0)=0$. For convenience, we introduce the norming constant $b[\lambda_0]$ such that
\begin{align}\label{T29}
\begin{split}
\psi_{+1}(\lambda_0;x,t)=b[\lambda_0]\psi_{-2}(\lambda_0;x,t),\text{ as }\lambda_0\in\Lambda\cap D^{+},\\
\psi_{+2}(\lambda_0;x,t)=b[\lambda_0]\psi_{-1}(\lambda_0;x,t),\text{ as }\lambda_0\in\Lambda\cap D^{-}.
\end{split}
\end{align}

For a given $\lambda_0\in\Lambda\cap D^{+}$, according to $s_{11}(\lambda_0)=\mathrm{det}(\psi_{+1}(\lambda_0;x,t),\psi_{-2}(\lambda_0;x,t))$ in Eq. \eqref{T13}, and taking derivative respect to $\lambda$ $(\text{here }\lambda=\lambda_0)$ on both sides of previous equation, we have
\begin{align}\nonumber
\begin{split}
s'_{11}(\lambda_0)=\mathrm{det}(\psi'_{+1}(\lambda_0;x,t)-b[\lambda_0]\psi'_{-2}(\lambda_0;x,t),\psi_{-2}(\lambda_0;x,t)),
\end{split}
\end{align}
note that $s'_{11}(\lambda_0)=0$, $\psi'_{+1}(\lambda_0;x,t)-b[\lambda_0]\psi'_{-2}(\lambda_0;x,t)$ and $\psi_{-2}(\lambda_0;x,t)$ are linearly dependent. Similarly, for a given $\lambda_0\in\Lambda\cap D^{-}$, one can obtain that $\psi'_{+2}(\lambda_0;x,t)-b[\lambda_0]\psi'_{-1}(\lambda_0;x,t)$ and $\psi_{-1}(\lambda_0;x,t)$ are linearly dependent by combining Eq. \eqref{T13} and $s'_{22}(\lambda_0)=0$. For convenience, we define another norming constant $d[\lambda_0]$ such that
\begin{align}\label{T30}
\begin{split}
&\psi'_{+1}(\lambda_0;x,t)-b[\lambda_0]\psi'_{-2}(\lambda_0;x,t)=d[\lambda_0]\psi_{-2}(\lambda_0;x,t),\text{ as }\lambda_0\in\Lambda\cap D^{+},\\
&\psi'_{+2}(\lambda_0;x,t)-b[\lambda_0]\psi'_{-1}(\lambda_0;x,t)=d[\lambda_0]\psi_{-1}(\lambda_0;x,t),\text{ as }\lambda_0\in\Lambda\cap D^{-}.
\end{split}
\end{align}

On the other hand, we notice that $\psi_{+1}(\lambda;x,t)$ and $s_{11}(\lambda)$ are analytic on $D^{+}$, and suppose $\lambda_0$ is the double zeros of $s_{11}$. Let $\psi_{+1}(\lambda;x,t)$ and $s_{11}(\lambda)$ carry out Taylor expansion at $\lambda=\lambda_0$, we have
\begin{align}\nonumber
\begin{split}
\frac{\psi_{+1}(\lambda;x,t)}{s_{11}(\lambda)}&=\frac{2\psi_{+1}(\lambda_0;x,t)}{s''_{11}(\lambda_0)}(\lambda-\lambda_0)^{-2}+\bigg(\frac{2\psi'_{+1}(\lambda_0;x,t)}{s''_{11}(\lambda_0)}-\frac{2\psi_{+1}(\lambda_0;x,t)s'''_{11}(\lambda_0)}{3s''^2_{11}(\lambda_0)}\bigg)\\
&\quad(\lambda-\lambda_0)^{-1}+\cdots,
\end{split}
\end{align}
Then, one has the compact form
\begin{align}\nonumber
\begin{split}
\mathop{P_{-2}}_{\lambda=\lambda_0}\bigg[\frac{\psi_{+1}(\lambda;x,t)}{s_{11}(\lambda)}\bigg]=\frac{2\psi_{+1}(\lambda_0;x,t)}{s''_{11}(\lambda_0)}=\frac{2b[\lambda_0]\psi_{-2}(\lambda_0;x,t)}{s''_{11}(\lambda_0)},\text{ as }\lambda_0\in\Lambda\cap D^{+},
\end{split}
\end{align}
where $\underset{\lambda=\lambda_0}{P_{-2}}[f(\lambda;x,t)]$ denotes the coefficient of $O((\lambda-\lambda_0)^{-2})$ term in the Laurent series expansion of $f(\lambda;x,t)$ at $\lambda=\lambda_0$. Likewise, we have
\begin{align}\nonumber
\begin{split}
\mathop{\mathrm{Res}}_{\lambda=\lambda_0}\bigg[\frac{\psi_{+1}(\lambda;x,t)}{s_{11}(\lambda)}\bigg]&=\frac{2\psi'_{+1}(\lambda_0;x,t)}{s''_{11}(\lambda_0)}-\frac{2\psi_{+1}(\lambda_0;x,t)s'''_{11}(\lambda_0)}{3s''^2_{11}(\lambda_0)}\\
&=\frac{2b[\lambda_0]\psi'_{-2}(\lambda_0;x,t)}{s''_{11}(\lambda_0)}+\bigg[\frac{2b[\lambda_0]}{s''_{11}(\lambda_0)}\bigg(\frac{d[\lambda_0]}{b[\lambda_0]}-\frac{s'''_{11}(\lambda_0)}{3s''_{11}(\lambda_0)}\bigg)\bigg]\psi_{-2}(\lambda_0;x,t),\\
&\quad\text{ as }\lambda_0\in\Lambda\cap D^{+},
\end{split}
\end{align}
where $\underset{\lambda=\lambda_0}{\mathrm{Res}}[f(\lambda;x,t)]$ denotes the coefficient of $O((\lambda-\lambda_0)^{-1})$ term in the Laurent series expansion of $f(\lambda;x,t)$ at $\lambda=\lambda_0$.

Similarly,  for the case of $\psi_{+2}(\lambda;x,t)$ and $s_{22}(\lambda)$ are analytic on $D^{-}$, we repeat the above process and obtain
\begin{align}\nonumber
\begin{split}
&\mathop{P_{-2}}_{\lambda=\lambda_0}\bigg[\frac{\psi_{+2}(\lambda;x,t)}{s_{22}(\lambda)}\bigg]=\frac{2\psi_{+2}(\lambda_0;x,t)}{s''_{22}(\lambda_0)}=\frac{2b[\lambda_0]\psi_{-1}(\lambda_0;x,t)}{s''_{22}(\lambda_0)},\text{ as }\lambda_0\in\Lambda\cap D^{-},\\
&\mathop{\mathrm{Res}}_{\lambda=\lambda_0}\bigg[\frac{\psi_{+2}(\lambda;x,t)}{s_{22}(\lambda)}\bigg]=\frac{2b[\lambda_0]\psi'_{-1}(\lambda_0;x,t)}{s''_{22}(\lambda_0)}+\bigg[\frac{2b[\lambda_0]}{s''_{22}(\lambda_0)}\bigg(\frac{d[\lambda_0]}{b[\lambda_0]}-\frac{s'''_{22}(\lambda_0)}{3s''_{22}(\lambda_0)}\bigg)\bigg]\psi_{-1}(\lambda_0;x,t),\\
&\quad\text{ as }\lambda_0\in\Lambda\cap D^{-}.
\end{split}
\end{align}

Moreover, let
\begin{align}\label{T31}
A[\lambda_0]=\left\{
\begin{aligned}
\frac{2b[\lambda_0]}{s''_{11}(\lambda_0)},\text{ as }\lambda_0\in\Lambda\cap D^{+},\\
\frac{2b[\lambda_0]}{s''_{22}(\lambda_0)},\text{ as }\lambda_0\in\Lambda\cap D^{-},
\end{aligned}
\right.\quad
B[\lambda_0]=\left\{
\begin{aligned}
\frac{d[\lambda_0]}{b[\lambda_0]}-\frac{s'''_{11}(\lambda_0)}{3s''_{11}(\lambda_0)},\text{ as }\lambda_0\in\Lambda\cap D^{+},\\
\frac{d[\lambda_0]}{b[\lambda_0]}-\frac{s'''_{22}(\lambda_0)}{3s''_{22}(\lambda_0)},\text{ as }\lambda_0\in\Lambda\cap D^{-}.
\end{aligned}
\right.
\end{align}
Thus we obtain the following simple expression
\begin{align}\label{T32}
\begin{split}
&\mathop{P_{-2}}_{\lambda=\lambda_0}\bigg[\frac{\psi_{+1}(\lambda;x,t)}{s_{11}(\lambda)}\bigg]=A[\lambda_0]\psi_{-2}(\lambda_0;x,t),\text{ as }\lambda_0\in\Lambda\cap D^{+},\\
&\mathop{P_{-2}}_{\lambda=\lambda_0}\bigg[\frac{\psi_{+2}(\lambda;x,t)}{s_{22}(\lambda)}\bigg]=A[\lambda_0]\psi_{-1}(\lambda_0;x,t),\text{ as }\lambda_0\in\Lambda\cap D^{-},\\
&\mathop{\mathrm{Res}}_{\lambda=\lambda_0}\bigg[\frac{\psi_{+1}(\lambda;x,t)}{s_{11}(\lambda)}\bigg]=A[\lambda_0][\psi'_{-2}(\lambda_0;x,t)+B[\lambda_0]\psi_{-2}(\lambda_0;x,t)],\text{ as }\lambda_0\in\Lambda\cap D^{+},\\
&\mathop{\mathrm{Res}}_{\lambda=\lambda_0}\bigg[\frac{\psi_{+2}(\lambda;x,t)}{s_{22}(\lambda)}\bigg]=A[\lambda_0][\psi'_{-1}(\lambda_0;x,t)+B[\lambda_0]\psi_{-1}(\lambda_0;x,t)],\text{ as }\lambda_0\in\Lambda\cap D^{-}.
\end{split}
\end{align}

Accordingly, by mean of Eqs. \eqref{T29}-\eqref{T31} as well as proposition \ref{P4}, we can derive the following symmetry relations.
\begin{prop}
For $\lambda_0\in\Lambda$, the two symmetry relations for $A[\lambda_0]$ and $B[\lambda_0]$ can be deduced as follows:\\
$\bullet$ The first symmetry relation $A[\lambda_0]=-A[\lambda^*_0]^*$, $B[\lambda_0]=B[\lambda^*_0]^*$.\\
$\bullet$ The second symmetry relation $A[\lambda_0]=A[-\lambda^*_0]^*$, $B[\lambda_0]=-B[-\lambda^*_0]^*$.
\end{prop}

\subsubsection{The Matrix RH Problem with ZBCs and Double Poles}
In general, according to the relation \eqref{T12} of two Jost solutions $\psi_{\pm}(\lambda;x,t)$, we will investigate the inverse problem with ZBCs by establishing a RH problem. In order to pose and solve the RH problem conveniently, we define
\begin{align}\label{T33}
\zeta_n=\bigg\{
\begin{aligned}
\lambda_n,\quad\quad &n=1,2,\cdots,N,\\
-\lambda_{n-N},\quad &n=N+1,N+2,\cdots,2N.
\end{aligned}
\end{align}

Then, a matrix RH problem is proposed as follows.
\begin{prop}
Define the sectionally meromorphic matrices
\begin{align}\label{T34}
M(\lambda;x,t)=\Bigg\{
\begin{aligned}
M^+(\lambda;x,t)=\bigg(\frac{\mu_{+1}(\lambda;x,t)}{s_{11}(\lambda)},\mu_{-2}(\lambda;x,t)\bigg),\text{ as } \lambda\in D^+,\\
M^-(\lambda;x,t)=\bigg(\mu_{-1}(\lambda;x,t),\frac{\mu_{+2}(\lambda;x,t)}{s_{11}(\lambda)}\bigg),\text{ as } \lambda\in D^-,
\end{aligned}
\Bigg.
\end{align}
where $\lim\limits_{\substack{\lambda'\rightarrow\lambda\\\lambda'\in D^{\pm}}}M(\lambda';x,t)=M^{\pm}(\lambda;x,t)$. Then, the multiplicative matrix RH problem is given below:\\
$\bullet$ Analyticity: $M(\lambda;x,t)$ is analytic in $D^+\cup D^-\backslash\Lambda$ and has the double poles in $\Lambda$, whose principal parts of the Laurent series at each double pole $\zeta_n$ or $\zeta^*_n$, are determined as
\begin{align}\label{T35}
\begin{split}
&\mathop{\mathrm{Res}}_{\lambda=\zeta_n}M^+(\lambda;x,t)=\left(A[\zeta_n]\mathrm{e}^{-2\mathrm{i}\theta(\zeta_n;x,t)}\{\mu'_{-2}(\zeta_n;x,t)+[B[\zeta_n]-2\mathrm{i}\theta'(\zeta_n;x,t)]\mu_{-2}(\zeta_n;x,t)\},0\right),\\
&\mathop{\mathrm{P}_{-2}}_{\lambda=\zeta_n}M^+(\lambda;x,t)=\left(A[\zeta_n]\mathrm{e}^{-2\mathrm{i}\theta(\zeta_n;x,t)}\mu_{-2}(\zeta_n;x,t),0\right),\\
&\mathop{\mathrm{Res}}_{\lambda=\zeta^*_n}M^+(\lambda;x,t)=\left(0,A[\zeta^*_n]\mathrm{e}^{2\mathrm{i}\theta(\zeta^*_n;x,t)}\{\mu'_{-1}(\zeta^*_n;x,t)+[B[\zeta^*_n]+2\mathrm{i}\theta'(\zeta^*_n;x,t)]\mu_{-1}(\zeta^*_n;x,t)\}\right),\\
&\mathop{\mathrm{P}_{-2}}_{\lambda=\zeta^*_n}M^+(\lambda;x,t)=\left(0,A[\zeta^*_n]\mathrm{e}^{2\mathrm{i}\theta(\zeta^*_n;x,t)}\mu_{-1}(\zeta^*_n;x,t)\right).
\end{split}
\end{align}
$\bullet$ Jump condition:
\begin{align}\label{T36}
M^-(\lambda;x,t)=M^+(\lambda;x,t)[I-J(\lambda;x,t)],\text{ as }\lambda\in\Sigma,
\end{align}
where
\begin{align}\label{T37}
J(\lambda;x,t)=\mathrm{e}^{\mathrm{i}\theta(\lambda;x,t)\widehat{\sigma_3}}\bigg(\begin{array}{cc} 0 & -\tilde{\rho}(\lambda) \\ \rho(\lambda) & \rho(\lambda)\tilde{\rho}(\lambda) \end{array}\bigg).
\end{align}
$\bullet$ Asymptotic behavior:
\begin{align}\label{T38}
M(\lambda;x,t)=\mathrm{e}^{\mathrm{i}\nu_{-}(x,t)\sigma_3}+O\bigg(\frac{1}{\lambda}\bigg),\text{ as }\lambda\rightarrow\infty.
\end{align}

\end{prop}
\begin{proof}
For the analyticity of $M(\lambda;x,t)$, It follows from Eqs. \eqref{T10} and \eqref{T32} that for each double poles $\zeta_n\in D^+$ or $\zeta^*_n\in D^-$. Now, we consider $\zeta_n\in D^+$, and obtain
\begin{small}
\begin{align}\nonumber
\begin{split}
&\mathop{\mathrm{Res}}_{\lambda=\zeta_n}\bigg[\frac{\mu_{+1}(\zeta_n;x,t)}{s_{11}(\zeta_n)}\bigg]=A[\zeta_n]\mathrm{e}^{-2\mathrm{i}\theta(\zeta_n;x,t)}\{\mu'_{-2}(\zeta_n;x,t)+[B[\zeta_n]-2\mathrm{i}\theta'(\zeta_n;x,t)]\mu_{-2}(\zeta_n;x,t)\},\Rightarrow\\
&\mathop{\mathrm{Res}}_{\lambda=\zeta_n}M^+(\lambda;x,t)=\bigg(A[\zeta_n]\mathrm{e}^{-2\mathrm{i}\theta(\zeta_n;x,t)}\{\mu'_{-2}(\zeta_n;x,t)+[B[\zeta_n]-2\mathrm{i}\theta'(\zeta_n;x,t)]\mu_{-2}(\zeta_n;x,t)\},0\bigg),
\end{split}
\end{align}
\end{small}
where the $``'"$ denotes the partial derivative with respect to $\lambda$ $(\text{here }\lambda=\zeta_n)$, and
\begin{align}\nonumber
\begin{split}
&\mathop{P_{-2}}_{\lambda=\zeta_n}\bigg[\frac{\mu_{+1}(\zeta_n;x,t)}{s_{11}(\zeta_n)}\bigg]=A[\zeta_n]\mathrm{e}^{-2\mathrm{i}\theta(\zeta_n;x,t)}\mu_{-2}(\zeta_n;x,t),\Rightarrow\\
&\mathop{\mathrm{P}_{-2}}_{\lambda=\zeta_n}M(\lambda;x,t)=\bigg(A[\zeta_n]\mathrm{e}^{-2\mathrm{i}\theta(\zeta_n;x,t)}\mu_{-2}(\zeta_n;x,t),0\bigg).
\end{split}
\end{align}

Similarly, we also can obtain the analyticity for $\zeta^*_n\in D^-$. It follows from Eqs. \eqref{T10} and \eqref{T12} that
\begin{align}\nonumber
\bigg\{
\begin{aligned}
&\mu_{+1}(\lambda;x,t)=\mu_{-1}(\lambda;x,t)s_{11}(\lambda)+\mu_{-2}(\lambda;x,t)\mathrm{e}^{-2\mathrm{i}\theta(\lambda;x,t)}s_{21}(\lambda),\\
&\mu_{+2}(\lambda;x,t)=\mu_{-1}(\lambda;x,t)\mathrm{e}^{2\mathrm{i}\theta(\lambda;x,t)}s_{12}(\lambda)+\mu_{-2}(\lambda;x,t)s_{22}(\lambda),
\end{aligned}
\end{align}
by combining Eqs. \eqref{T14} and \eqref{T34}, one can obtain
\begin{footnotesize}
\begin{align}\nonumber
\begin{split}
&M^+(\lambda;x,t)=\bigg(\frac{\mu_{+1}(\lambda;x,t)}{s_{11}(\lambda)},\mu_{-2}(\lambda;x,t)\bigg)=\bigg(\mu_{-1}(\lambda;x,t)+\mu_{-2}(\lambda;x,t)\mathrm{e}^{-2\mathrm{i}\theta(\lambda;x,t)}\rho(\lambda),\mu_{-2}(\lambda;x,t)\bigg),\\
&M^-(\lambda;x,t)=\bigg(\mu_{-1}(\lambda;x,t),\frac{\mu_{+2}(\lambda;x,t)}{s_{22}(\lambda)}\bigg)=\bigg(\mu_{-1}(\lambda;x,t),\mu_{-1}(\lambda;x,t)\mathrm{e}^{2\mathrm{i}\theta(\lambda;x,t)}\tilde{\rho}(\lambda)+\mu_{-2}(\lambda;x,t)\bigg),
\end{split}
\end{align}
\end{footnotesize}
and
\begin{small}
\begin{align}\nonumber
\bigg(\mu_{-1}(\lambda;x,t),\frac{\mu_{+2}(\lambda;x,t)}{s_{22}(\lambda)}\bigg)=\bigg(\frac{\mu_{+1}(\lambda;x,t)}{s_{11}(\lambda)},\mu_{-2}(\lambda;x,t)\bigg)\bigg(\begin{array}{cc} 1 & \mathrm{e}^{2\mathrm{i}\theta(\lambda;x,t)}\tilde{\rho}(\lambda) \\ -\mathrm{e}^{-2\mathrm{i}\theta(\lambda;x,t)}\rho(\lambda) & 1-\rho(\lambda)\tilde{\rho}(\lambda) \end{array}\bigg),\\
\end{align}
\end{small}
that is $M^-(z;x,t)=M^+(z;x,t)(I-J(z;x,t))$, where $J(\lambda;x,t)$ is given by Eq. \eqref{T37}. The asymptotic behaviors of the modified Jost solutions $\mu_{\pm}(\lambda;x,t)$ and scattering matrix $S(\lambda)$ given in propositions \eqref{P2} and \eqref{PP2} can easily lead to that of $M(\lambda;x,t)$. Specifically,
\begin{align}\nonumber
\begin{split}
&M^{+}(\lambda;x,t)=\Bigg(\begin{array}{cc} \frac{\mathrm{e}^{\mathrm{i}\nu_{+}(x,t)}}{\mathrm{e}^{-\mathrm{i}\nu}} & 0 \\ 0 & \mathrm{e}^{-\mathrm{i}\nu_{-}(x,t)} \end{array}\Bigg)+O\bigg(\frac{1}{\lambda}\bigg)=\mathrm{e}^{\mathrm{i}\nu_{-}(x,t)\sigma_3}+O\bigg(\frac{1}{\lambda}\bigg),\text{ as }\lambda\in D^{+},\\
&M^{-}(\lambda;x,t)=\Bigg(\begin{array}{cc} \mathrm{e}^{\mathrm{i}\nu_{-}(x,t)} & 0 \\ 0 & \frac{\mathrm{e}^{-\mathrm{i}\nu_{+}(x,t)}}{\mathrm{e}^{\mathrm{i}\nu}} \end{array}\Bigg)+O\bigg(\frac{1}{\lambda}\bigg)=\mathrm{e}^{\mathrm{i}\nu_{-}(x,t)\sigma_3}+O\bigg(\frac{1}{\lambda}\bigg),\text{ as }\lambda\in D^{-}.
\end{split}
\end{align}

Therefore, we obtain the asymptotic behavior \eqref{T38} of $M(\lambda;x,t)$ when $\lambda\rightarrow\infty$. This completes the proof.

\end{proof}

By subtracting out the asymptotic values as $\lambda\rightarrow\infty$ and the singularity contributions, one can regularize the RH problem as a normative form. Then, applying the Plemelj's formula, the solutions of the corresponding matrix RH problem can be established by an integral equation.

\begin{prop}\label{P8}
The solution of the above-mentioned matrix Riemann-Hilbert problem can be expressed as
\begin{small}
\begin{align}\label{T39}
\begin{split}
M(\lambda;x,t)=&\mathrm{e}^{\mathrm{i}\nu_{-}(x,t)\sigma_3}+\frac{1}{2\pi\mathrm{i}}\int_{\Sigma}\frac{M^+(\xi;x,t)J(\xi;x,t)}{\xi-\lambda}d\xi+\sum^{2N}_{n=1}\bigg(C_n(\lambda)\bigg[\mu'_{-2}(\zeta_n;x,t)+\bigg(D_n+\\
&\frac{1}{\lambda-\zeta_n}\bigg)\mu_{-2}(\zeta_n;x,t)\bigg],\widetilde{C}_n(\lambda)\bigg[\mu'_{-1}(\zeta^*_n;x,t)+\bigg(\widetilde{D}_n+\frac{1}{\lambda-\zeta^*_n}\bigg)\mu_{-1}(\zeta^*_n;x,t)\bigg]\bigg),
\end{split}
\end{align}
\end{small}
where $\lambda\in\mathbb{C}\backslash\Sigma$, $\int_{\Sigma}$ is an integral along the oriented contour exhibited in Fig. \ref{F1}, and
\begin{align}\label{T40}
\begin{split}
&C_n(\lambda)=\frac{A[\zeta_n]}{\lambda-\zeta_n}\mathrm{e}^{-2\mathrm{i}\theta(\zeta_n;x,t)},\quad \widetilde{C}_n(\lambda)=\frac{A[\zeta^*_n]}{\lambda-\zeta^*_n}\mathrm{e}^{2\mathrm{i}\theta(\zeta^*_n;x,t)},\\
&D_n=B[\zeta_n]-2\mathrm{i}\theta'(\zeta_n;x,t),\quad \widetilde{D}_n=B[\zeta^*_n]+2\mathrm{i}\theta'(\zeta^*_n;x,t),
\end{split}
\end{align}
$\mu_{-k}$ and $\mu'_{-k}$ $(k=1,2)$ satisfy
\begin{small}
\begin{align}\label{T41}
\begin{split}
&\mu_{-1}(\zeta^*_n;x,t)=\mathrm{e}^{\mathrm{i}\nu_{-}(x,t)\sigma_3}\bigg(\begin{array}{cc} 1 \\ 0  \end{array}\bigg)+\sum_{s=1}^{2N}C_s(\zeta^*_n)\bigg[\mu'_{-2}(\zeta_s;x,t)+\bigg(D_s+\frac{1}{\zeta^*_n-\zeta_s}\bigg)\mu_{-2}(\zeta_s;x,t)\bigg]+\\
&\qquad\qquad\qquad\frac{1}{2\pi \mathrm{i}}\int_{\Sigma}\frac{(M^+(\xi;x,t)J(\xi;x,t))_{1}}{\xi-\zeta^*_n}d\xi,\\
&\mu_{-2}(\zeta_s;x,t)=\mathrm{e}^{\mathrm{i}\nu_{-}(x,t)\sigma_3}\bigg(\begin{array}{cc} 0 \\ 1  \end{array}\bigg)+\sum_{j=1}^{2N}\widetilde{C}_j(\zeta_s)\bigg[\mu'_{-1}(\zeta^*_j;x,t)+\bigg(\widetilde{D}_j+\frac{1}{\zeta_s-\zeta^*_j}\bigg)\mu_{-1}(\zeta^*_j;x,t)\bigg]+\\
&\qquad\qquad\qquad\frac{1}{2\pi \mathrm{i}}\int_{\Sigma}\frac{(M^+(\xi;x,t)J(\xi;x,t))_{2}}{\xi-\zeta_s}d\xi,\\
&\mu'_{-1}(\zeta^*_n;x,t)=-\sum_{s=1}^{2N}\frac{C_s(\zeta^*_n)}{\zeta^*_n-\zeta_s}\bigg[\mu'_{-2}(\zeta_s;x,t)+\bigg(D_s+\frac{2}{\zeta^*_n-\zeta_s}\bigg)\mu_{-2}(\zeta_s;x,t)\bigg]+\\
&\qquad\qquad\qquad\frac{1}{2\pi \mathrm{i}}\int_{\Sigma}\frac{(M^+(\xi;x,t)J(\xi;x,t))_{1}}{(\xi-\zeta^*_n)^2}d\xi,\\
&\mu'_{-2}(\zeta_s;x,t)=-\sum_{j=1}^{2N}\frac{\widetilde{C}_j(\zeta_s)}{\zeta_s-\zeta^*_j}\bigg[\mu'_{-1}(\zeta^*_j;x,t)+\bigg(\widetilde{D}_j+\frac{2}{\zeta_s-\zeta^*_j}\bigg)\mu_{-1}(\zeta^*_j;x,t)\bigg]+\\
&\qquad\qquad\qquad\frac{1}{2\pi \mathrm{i}}\int_{\Sigma}\frac{(M^+(\xi;x,t)J(\xi;x,t))_{2}}{(\xi-\zeta_s)^2}d\xi,
\end{split}
\end{align}
\end{small}
where $(M^+(\xi;x,t)J(\xi;x,t))_{j}$ $(j=1,2)$ represent $j$th column of matrix $M^+(\xi;x,t)J(\xi;x,t)$.
\end{prop}

\begin{proof}
In order to regularize the RH problem, one has to subtract out the asymptotic values as $\lambda\rightarrow\infty$ which exhibited in Eq. \eqref{T38} and the singularity contributions. Then, the jump condition \eqref{T36} becomes
\begin{small}
\begin{align}\label{T42}
\begin{split}
&M^-(\lambda;x,t)-\mathrm{e}^{\mathrm{i}\nu_{-}(x,t)\sigma_3}-\sum_{n=1}^{2N}\Bigg[\frac{\mathop{\mathrm{P}_{-2}}\limits_{\lambda=\zeta_n}M(\lambda;x,t)}{(\lambda-\zeta_n)^2}+\frac{\mathop{\mathrm{Res}}\limits_{\lambda=\zeta_n}M(\lambda;x,t)}{\lambda-\zeta_n}+\frac{\mathop{\mathrm{P}_{-2}}\limits_{\lambda=\zeta^*_n}M(\lambda;x,t)}{(\lambda-\zeta^*_n)^2}+\\
&\frac{\mathop{\mathrm{Res}}\limits_{\lambda=\zeta^*_n}M(\lambda;x,t)}{\lambda-\zeta^*_n}\Bigg]=M^+(\lambda;x,t)-\mathrm{e}^{\mathrm{i}\nu_{-}(x,t)\sigma_3}-\sum_{n=1}^{2N}\Bigg[\frac{\mathop{\mathrm{P}_{-2}}\limits_{\lambda=\zeta_n}M(\lambda;x,t)}{(\lambda-\zeta_n)^2}+\frac{\mathop{\mathrm{Res}}\limits_{\lambda=\zeta_n}M(\lambda;x,t)}{\lambda-\zeta_n}+\\
&\frac{\mathop{\mathrm{P}_{-2}}\limits_{\lambda=\zeta^*_n}M(\lambda;x,t)}{(\lambda-\zeta^*_n)^2}+\frac{\mathop{\mathrm{Res}}\limits_{\lambda=\zeta^*_n}M(\lambda;x,t)}{\lambda-\zeta^*_n}\Bigg]-M^+(\lambda;x,t)J(\lambda;x,t),
\end{split}
\end{align}
\end{small}
where $\mathop{\mathrm{P}_{-2}}\limits_{\lambda=\zeta_n}M(\lambda;x,t),\mathop{\mathrm{Res}}\limits_{\lambda=\zeta_n}M(\lambda;x,t),\mathop{\mathrm{P}_{-2}}\limits_{\lambda=\zeta^*_n}M(\lambda;x,t),\mathop{\mathrm{Res}}\limits_{\lambda=\zeta^*_n}M(\lambda;x,t)$ have given in Eq. \eqref{T35}. By using Plemelj's formula, one can obtain the solution \eqref{T39} with formula \eqref{T40} of the matrix RH problem. By combining Eqs. \eqref{T34} and \eqref{T39}, $\mu_{-1}(\zeta^*_n;x,t)$ is the first column element of the solution \eqref{T39} as the double-pole $\lambda=\zeta^*_n\in D^{-}$, $\mu_{-2}(\zeta_s;x,t)$ is the second column element of the solution \eqref{T39} as the double-pole $\lambda=\zeta_s\in D^{+}$. The specific expressions of $\mu_{-1}(\zeta^*_n;x,t),\mu_{-2}(\zeta_s;x,t)$ and their first derivative to $\lambda$ are provided in Eq. \eqref{T41}.

\end{proof}

\subsubsection{Reconstruction Formula of the Potential with ZBCs and Double poles}
From the solution \eqref{T39} of the matrix RH problem, we have
\begin{align}\label{T43}
M(\lambda;x,t)=\mathrm{e}^{\mathrm{i}\nu_{-}(x,t)\sigma_3}+\frac{M^{[1]}(x,t)}{\lambda}+O\bigg(\frac{1}{\lambda^2}\bigg),\text{ as }\lambda\rightarrow\infty,
\end{align}
where
\begin{align}\label{T44}
\begin{split}
M^{[1]}(x,t)=&-\frac{1}{2\pi \mathrm{i}}\int_{\Sigma}M^+(\xi;x,t)J(\xi;x,t)d\xi+\sum^{2N}_{n=1}\big\{A[\zeta_n]\mathrm{e}^{-2\mathrm{i}\theta(\zeta_n;x,t)}\big[\mu'_{-2}(\zeta_n;x,t)+\\
&D_n\mu_{-2}(\zeta_n;x,t)\big],A[\zeta^*_n]\mathrm{e}^{2\mathrm{i}\theta(\zeta^*_n;x,t)}\big[\mu'_{-1}(\zeta^*_n;x,t)+\widetilde{D}_n\mu_{-1}(\zeta^*_n;x,t)\big]\big\}.
\end{split}
\end{align}

Substituting Eq. \eqref{T43} into Eq. \eqref{BT2} and matching $O(\lambda)$ term, we have
\begin{align}\nonumber
O(\lambda):\mathrm{i}\big[M^{[1]}(x,t)\sigma_3-\sigma_3M^{[1]}(x,t)\big]=Q(x,t)\mathrm{e}^{\mathrm{i}\nu_{-}(x,t)\sigma_3},
\end{align}
then, by expanding the above equation, one can find the reconstruction formula of the double-pole solution (potential) for the TOFKN \eqref{T1} with ZBCs as follows
\begin{align}\label{T45}
q(x,t)=-2\mathrm{i}\mathrm{e}^{\mathrm{i}\nu_{-}(x,t)}M^{[1]}_{12}(x,t),
\end{align}
where $M^{[1]}_{12}(x,t)$ represents the first row and second column element of the matrix $M^{[1]}(x,t)$, and
\begin{align}\label{T46}
\begin{split}
M^{[1]}_{12}(x,t)=&-\frac{1}{2\pi \mathrm{i}}\int_{\Sigma}\big(M^{+}(\xi;x,t)J(\xi;x,t)\big)_{12}d\xi+\\
&\sum^{2N}_{n=1}\big\{A[\zeta^*_n]\mathrm{e}^{2\mathrm{i}\theta(\zeta^*_n;x,t)}\big[\mu'_{-11}(\zeta^*_n;x,t)+\widetilde{D}_n\mu_{-11}(\zeta^*_n;x,t)\big]\big\},
\end{split}
\end{align}
where $\mu_{-11}(\zeta^*_n;x,t)$ and $\mu'_{-11}(\zeta^*_n;x,t)$ represents the first row element of the column vector $\mu_{-1}(\zeta^*_n;x,t)$ and $\mu'_{-1}(\zeta^*_n;x,t)$ respectively. When taking row vector $\alpha=\big(\alpha^{(1)},\alpha^{(2)}\big)$ and column vector $\gamma=(\gamma^{(1)},\gamma^{(2)})^{\mathrm{T}}$, where
\begin{align}\label{T47}
\begin{split}
&\alpha^{(1)}=\left(A[\zeta^*_n]\mathrm{e}^{2\mathrm{i}\theta(\zeta^*_n;x,t)}\widetilde{D}_n\right)_{1\times2N},\quad\alpha^{(2)}=\left(A[\zeta^*_n]\mathrm{e}^{2\mathrm{i}\theta(\zeta^*_n;x,t)}\right)_{1\times2N},\\
&\gamma^{(1)}=\big(\mu_{-11}(\zeta^*_n;x,t)\big)_{1\times2N},\quad\gamma^{(2)}=\big(\mu'_{-11}(\zeta^*_n;x,t)\big)_{1\times2N},
\end{split}
\end{align}
we can obtain a more concise reconstruction formulation of the double poles solution (potential) for the TOFKN \eqref{T1} with ZBCs and as follows
\begin{align}\label{T48}
q(x,t)=-2\mathrm{i}\mathrm{e}^{\mathrm{i}\nu_{-}(x,t)}\bigg(\alpha\gamma-\frac{1}{2\pi \mathrm{i}}\int_{\Sigma}\big(M^{+}(\xi;x,t)J(\xi;x,t)\big)_{12}d\xi\bigg).
\end{align}

\subsubsection{Trace Formulae with ZBCs and double poles}
The so-called trace formulae are the scattering coefficients $s_{11}(\lambda)$ and $s_{22}(\lambda)$ are formulated in terms of the discrete spectrum $\Lambda$ and reflection coefficients $\rho(\lambda)$ and $\tilde{\rho}(\lambda)$. We know that $s_{11}(\lambda),s_{22}(\lambda)$ are analytic on $D^{+},D^{-}$, respectively. The discrete spectral points $\zeta_n$'s are the double zeros of $s_{11}(\lambda)$, while $\zeta^*_n$'s are the double zeros of $s_{22}(\lambda)$. Define the functions $\beta^{\pm}(\lambda)$ as follows:
\begin{align}\label{T49}
\beta^{+}(\lambda)=s_{11}(\lambda)\prod^{2N}_{n=1}\bigg(\frac{\lambda-\zeta^*_n}{\lambda-\zeta_n}\bigg)^2\mathrm{e}^{\mathrm{i}\nu},\,\beta^{-}(\lambda)=s_{22}(\lambda)\prod^{2N}_{n=1}\bigg(\frac{\lambda-\zeta_n}{\lambda-\zeta^*_n}\bigg)^2\mathrm{e}^{-\mathrm{i}\nu}.
\end{align}

Then, $\beta^{+}(\lambda)$ and $\beta^{-}(\lambda)$ are analytic and have no zero in $D^{+}$ and $D^{-}$, respectively. Furthermore, we have the relation $\beta^{+}(\lambda)\beta^{-}(\lambda)=s_{11}(\lambda)s_{22}(\lambda)$ and the asymptotic behaviors $\beta^{\pm}(\lambda)\rightarrow1,\text{ as }\lambda\rightarrow\infty$.

According to $\mathrm{det}(S(\lambda))=s_{11}s_{22}-s_{21}s_{12}=1$, we can derive
\begin{align}\nonumber
\frac{1}{s_{11}s_{22}}=1-\frac{s_{21}s_{12}}{s_{11}s_{22}}=1-\rho(\lambda)\tilde{\rho}(\lambda),
\end{align}
by taking logarithm on both sides of the above equation at the same time, that is
\begin{align}\nonumber
-\mathrm{log}(s_{11}s_{22})=\mathrm{log}[1-\rho(\lambda)\tilde{\rho}(\lambda)]\Rightarrow\mathrm{log}[\beta^{+}(\lambda)\beta^{-}(\lambda)]=-\mathrm{log}[1-\rho(\lambda)\tilde{\rho}(\lambda)],
\end{align}
and employing the Plemelj' formula such that we have
\begin{align}\label{T50}
\mathrm{log}\beta^{\pm}(\lambda)=\mp\frac{1}{2\pi\mathrm{i}}\int_{\Sigma}\frac{\mathrm{log}[1-\rho(\lambda)\tilde{\rho}(\lambda)]}{\xi-\lambda}d\xi,\quad\lambda\in D^{\pm}.
\end{align}

Then, substituting Eq. \eqref{T50} into  Eq. \eqref{T49}, we can obtain the trace formulae
\begin{align}\label{T51}
\begin{split}
&s_{11}(\lambda)=\mathrm{exp}\bigg(-\frac{1}{2\pi \mathrm{i}}\int_{\Sigma}\frac{\mathrm{log}[1-\rho(\lambda)\tilde{\rho}(\lambda)]}{\xi-\lambda}d\xi\bigg)\prod^{2N}_{n=1}\bigg(\frac{\lambda-\zeta_n}{\lambda-\zeta^*_n}\bigg)^2\mathrm{e}^{-\mathrm{i}\nu},\\
&s_{22}(\lambda)=\mathrm{exp}\bigg(\frac{1}{2\pi \mathrm{i}}\int_{\Sigma}\frac{\mathrm{log}[1-\rho(\lambda)\tilde{\rho}(\lambda)]}{\xi-\lambda}d\xi\bigg)\prod^{2N}_{n=1}\bigg(\frac{\lambda-\zeta^*_n}{\lambda-\zeta_n}\bigg)^2\mathrm{e}^{\mathrm{i}\nu}.
\end{split}
\end{align}

\subsubsection{Reflectionless Potential with ZBCs: Double-Pole Solitons}
Now, we consider a special kind of reflectionless potential $q(x,t)$ with the reflection coefficients $\rho(\lambda)=\tilde{\rho}(\lambda)=0$. From the Volterra integral equation \eqref{T11}, one obtains $\psi_{\pm}(0;x,t)=\mu_{\pm}(0;x,t)=I$, and $s_{11}(0)=1$. Combining the trace formula, one obtains that there exists an integer $l\in\mathbb{Z}$ such that $\nu=8\sum_{n=1}^N\mathrm{arg}(\zeta_n)+2\pi l$.

Then Eqs. \eqref{T41} and \eqref{T48} with $J(\lambda;x,t)=0_{2\times2}$ become
\begin{small}
\begin{align}\label{T52}
\begin{split}
&\mu_{-11}(\zeta^*_n;x,t)=\mathrm{e}^{\mathrm{i}\nu_{-}(x,t)}+\sum_{s=1}^{2N}C_s(\zeta^*_n)\bigg[\mu'_{-21}(\zeta_s;x,t)+\bigg(D_s+\frac{1}{\zeta^*_n-\zeta_s}\bigg)\mu_{-21}(\zeta_s;x,t)\bigg],\\
&\mu_{-21}(\zeta_s;x,t)=\sum_{j=1}^{2N}\widetilde{C}_j(\zeta_s)\bigg[\mu'_{-11}(\zeta^*_j;x,t)+\bigg(\widetilde{D}_j+\frac{1}{\zeta_s-\widetilde{\zeta_j}}\bigg)\mu_{-11}(\zeta^*_j;x,t)\bigg],\\
&\mu'_{-11}(\zeta^*_n;x,t)=-\sum_{s=1}^{2N}\frac{C_s(\zeta^*_n)}{\zeta^*_n-\zeta_s}\bigg[\mu'_{-21}(\zeta_s;x,t)+\bigg(D_s+\frac{2}{\zeta^*_n-\zeta_s}\bigg)\mu_{-21}(\zeta_s;x,t)\bigg],\\
&\mu'_{-21}(\zeta_s;x,t)=-\sum_{j=1}^{2N}\frac{\widetilde{C}_j(\zeta_s)}{\zeta_s-\zeta^*_j}\bigg[\mu'_{-11}(\zeta^*_j;x,t)+\bigg(\widetilde{D}_j+\frac{2}{\zeta_s-\zeta^*_j}\bigg)\mu_{-11}(\zeta^*_j;x,t)\bigg],
\end{split}
\end{align}
\end{small}
and
\begin{align}\label{T53}
q(x,t)=-2\mathrm{i}\mathrm{e}^{\mathrm{i}\nu_{-}(x,t)}\alpha\gamma.
\end{align}

\begin{thm}
The explicit expression for the double-pole solution of the TOFKN \eqref{T1} with ZBCs is given by determinant formula
\begin{align}\label{T54}
q(x,t)=2\mathrm{i}\frac{\mathrm{det}(I-G)}{\mathrm{det}R}\Bigg(\frac{\mathrm{det}\widetilde{R}}{\mathrm{det}(I-\widetilde{G})}\Bigg)^2,
\end{align}
where
\begin{align}\label{T55}
\begin{split}
&R=\bigg(\begin{array}{cc} 0 & \alpha \\ \tau & I-G \end{array}\bigg),\quad \widetilde{R}=\bigg(\begin{array}{cc} 0 & \alpha \\ \tau & I-\widetilde{G} \end{array}\bigg),\\
&\tau=\bigg(\begin{array}{cc} \tau^{(1)} \\ \tau^{(2)} \end{array}\bigg),\quad \tau^{(1)}=(1)_{2N\times1},\quad \tau^{(2)}=(0)_{2N\times1},
\end{split}
\end{align}
the $4N\times4N$ partitioned matrix $G=\bigg(\begin{array}{cc} G^{(1,1)} & G^{(1,2)} \\ G^{(2,1)} & G^{(2,2)} \end{array}\bigg)$ with $G^{(i,j)}=\Big(g^{(i,j)}_{n,j}\Big)_{2N\times2N}\,(i,j=1,2)$ is given by
\begin{small}
\begin{align}\nonumber
\begin{split}
&g^{(1,1)}_{n,j}=\sum^{2N}_{s=1}C_s(\zeta^*_n)\widetilde{C}_j(\zeta_s)\bigg[-\frac{1}{\zeta_s-\zeta^*_j}\bigg(\widetilde{D}_j+\frac{2}{\zeta_s-\zeta^*_j}\bigg)+\bigg(D_s+\frac{1}{\zeta^*_n-\zeta_s}\bigg)\bigg(\widetilde{D}_j+\frac{1}{\zeta_s-\zeta^*_j}\bigg)\bigg],\\
&g^{(1,2)}_{n,j}=\sum^{2N}_{s=1}C_s(\zeta^*_n)\widetilde{C}_j(\zeta_s)\bigg[-\frac{1}{\zeta_s-\zeta^*_j}+\bigg(D_s+\frac{1}{\zeta^*_n-\zeta_s}\bigg)\bigg],\\
&g^{(2,1)}_{n,j}=\sum^{2N}_{s=1}\frac{C_s(\zeta^*_n)\widetilde{C}_j(\zeta_s)}{\zeta^*_n-\zeta_s}\bigg[\frac{1}{\zeta_s-\zeta^*_j}\bigg(\widetilde{D}_j+\frac{2}{\zeta_s-\zeta^*_j}\bigg)-\bigg(D_s+\frac{2}{\zeta^*_n-\zeta_s}\bigg)\bigg(\widetilde{D}_j+\frac{1}{\zeta_s-\zeta^*_j}\bigg)\bigg],\\
&g^{(2,2)}_{n,j}=\sum^{2N}_{s=1}\frac{C_s(\zeta^*_n)\widetilde{C}_j(\zeta_s)}{\zeta^*_n-\zeta_s}\bigg[\frac{1}{\zeta_s-\zeta^*_j}-\bigg(D_s+\frac{2}{\zeta^*_n-\zeta_s}\bigg)\bigg],
\end{split}
\end{align}
\end{small}
the $4N\times4N$ partitioned matrix $\widetilde{G}=\bigg(\begin{array}{cc} \widetilde{G}^{(1,1)} & \widetilde{G}^{(1,2)} \\ \widetilde{G}^{(2,1)} & \widetilde{G}^{(2,2)} \end{array}\bigg)$ with $\widetilde{G}^{(i,j)}=\Big(\widetilde{g}^{(i,j)}_{n,j}\Big)_{2N\times2N}\,(i,j=1,2)$ is given by
\begin{footnotesize}
\begin{align}\nonumber
\begin{split}
&\widetilde{g}^{(1,1)}_{n,j}=\zeta^*_n\sum^{2N}_{s=1}\frac{C_s(\zeta^*_n)\widetilde{C}_j(\zeta_s)}{\zeta_s}\bigg[-\frac{1}{\zeta_s-\zeta^*_j}\bigg(\widetilde{D}_j+\frac{2}{\zeta_s-\zeta^*_j}\bigg)+\frac{\zeta_s}{\zeta^*_j}\bigg(D_s+\frac{1}{\zeta^*_n-\zeta_s}-\frac{1}{\zeta_s}\bigg)\bigg(\widetilde{D}_j+\frac{1}{\zeta_s-\zeta^*_j}-\frac{1}{\zeta^*_j}\bigg)\bigg],\\
&\widetilde{g}^{(1,2)}_{n,j}=\zeta^*_n\sum^{2N}_{s=1}\frac{C_s(\zeta^*_n)\widetilde{C}_j(\zeta_s)}{\zeta_s}\bigg[-\frac{1}{\zeta_s-\zeta^*_j}+\frac{\zeta_s}{\zeta^*_j}\bigg(D_s+\frac{1}{\zeta^*_n-\zeta_s}-\frac{1}{\zeta_s}\bigg)\bigg],\\
&\widetilde{g}^{(2,1)}_{n,j}=\sum^{2N}_{s=1}\frac{C_s(\zeta^*_n)\widetilde{C}_j(\zeta_s)}{\zeta^*_n-\zeta_s}\bigg[\frac{1}{\zeta_s-\zeta^*_j}\bigg(\widetilde{D}_j+\frac{2}{\zeta_s-\zeta^*_j}\bigg)-\frac{\zeta_s}{\zeta^*_j}\bigg(D_s+\frac{2}{\zeta^*_n-\zeta_s}\bigg)\bigg(\widetilde{D}_j+\frac{1}{\zeta_s-\zeta^*_j}-\frac{1}{\zeta^*_j}\bigg)\bigg],\\
&\widetilde{g}^{(2,2)}_{n,j}=\sum^{2N}_{s=1}\frac{C_s(\zeta^*_n)\widetilde{C}_j(\zeta_s)}{\zeta^*_n-\zeta_s}\bigg[\frac{1}{\zeta_s-\zeta^*_j}-\frac{\zeta_s}{\zeta^*_j}\bigg(D_s+\frac{2}{\zeta^*_n-\zeta_s}\bigg)\bigg].
\end{split}
\end{align}
\end{footnotesize}

\end{thm}
\begin{proof}
From the Eqs. \eqref{T47}, \eqref{T52} and \eqref{T53}, the reflectionless potential is derived by determinant formula:
\begin{align}\label{T56}
q(x,t)=2\mathrm{i}\frac{\mathrm{det}(R)}{\mathrm{det}(I-G)}\mathrm{e}^{2\mathrm{i}\nu_{-}(x,t)}.
\end{align}

However, this formula \eqref{T56} is implicit due to unknown function $\nu_{-}(x,t)$ is included. We need to deduce an exact form for the reflectionless potential. From the trace formulae \eqref{T51} and Volterra integral equation \eqref{T11} as $x\rightarrow-\infty$, we get
\begin{align}\label{T57}
\begin{split}
M(\lambda;x,t)=&I+\lambda\sum^{2N}_{n=1}\Bigg[\frac{\mathop{\mathrm{P}_{-2}}\limits_{\lambda=\zeta_n}\big(M(\lambda;x,t)/\lambda\big)}{(\lambda-\zeta_n)^2}+\frac{\mathop{\mathrm{Res}}\limits_{\lambda=\zeta_n}\big(M(\lambda;x,t)/\lambda\big)}{\lambda-\zeta_n}\\
&+\frac{\mathop{\mathrm{P}_{-2}}\limits_{\lambda=\zeta^*_n}\big(M(\lambda;x,t)/\lambda\big)}{(\lambda-\zeta^*_n)^2}+\frac{\mathop{\mathrm{Res}}\limits_{\lambda=\zeta^*_n}\big(M(\lambda;x,t)/\lambda\big)}{\lambda-\zeta^*_n}\Bigg],
\end{split}
\end{align}
which can yield the $\gamma$ given by Eq. \eqref{T47} explicitly. Then, substituting $\gamma$ into the formula of the potential, one yields
\begin{align}\label{T58}
q(x,t)=2\mathrm{i}\frac{\mathrm{det}(\widetilde{R})}{\mathrm{det}(I-\widetilde{G})}\mathrm{e}^{\mathrm{i}\nu_{-}(x,t)},
\end{align}
then, by combining Eq. \eqref{T57} with Eq. \eqref{T58}, we can obtain the determinant formula \eqref{T54}. Complete the proof.

\end{proof}

For example, we obtain the $N$-double-pole solutions of the TOFKN \eqref{T1} with ZBCs via Eq. \eqref{T54}:

$\bullet$ When taking parameters $N=1,\zeta_1=\frac12+\frac12\mathrm{i},A[\zeta_1]=1,B[\zeta_1]=1$, we can obtain the explicit 1-double-pole solution and give out relevant plots in Fig. \ref{F2}. Figs. \ref{F2} (a) and (b) exhibit the three-dimensional and density diagrams for the exact 1-double-pole soliton solution of the TOFKN with ZBCs, which is equivalent to the elastic collisions of two bright-bright solitons. Fig. \ref{F2} (c) displays the distinct profiles of the exact 1-double-pole soliton solution at $t=\pm6,0$.

Compared with the classical second-order flow KN system which also be called the DNLS equation in Ref. \cite{Zhangg(2020)}, the comparison made between the density diagrams of the 1-double-pole soliton solutions for the TOFKN and the DNLS equation shows that the trajectories of solutions are different obviously, it also means the introduction of third-order dispersion and quintic nonlinear term of KN systems can affect the trajectories of solutions. By comparing with Ref. \cite{Zhangg(2020)}, we find that the wave heights of the 1-double-pole solution at $x=0$ and $t=0$ are consistent, and the wave heights are all $1168/157$. Therefore, we make a guess that the energy at $t=0$ of the central peak for the 1-double-pole solution of the TOFKN is the same with that corresponding to the DNLS equation by selecting the same parameters. In other words, third-order dispersion and quintic nonlinear term of KN systems have little effect on the maximum amplitude of the solutions.

\begin{rem}
The parameter selection of the 1-double-pole soliton solution shown in Fig. \ref{F2} is consistent with that in reference \cite{Zhangg(2020)}.
\end{rem}

\begin{figure}[htbp]
\centering
\subfigure[]{
\begin{minipage}[t]{0.33\textwidth}
\centering
\includegraphics[height=4.5cm,width=4.5cm]{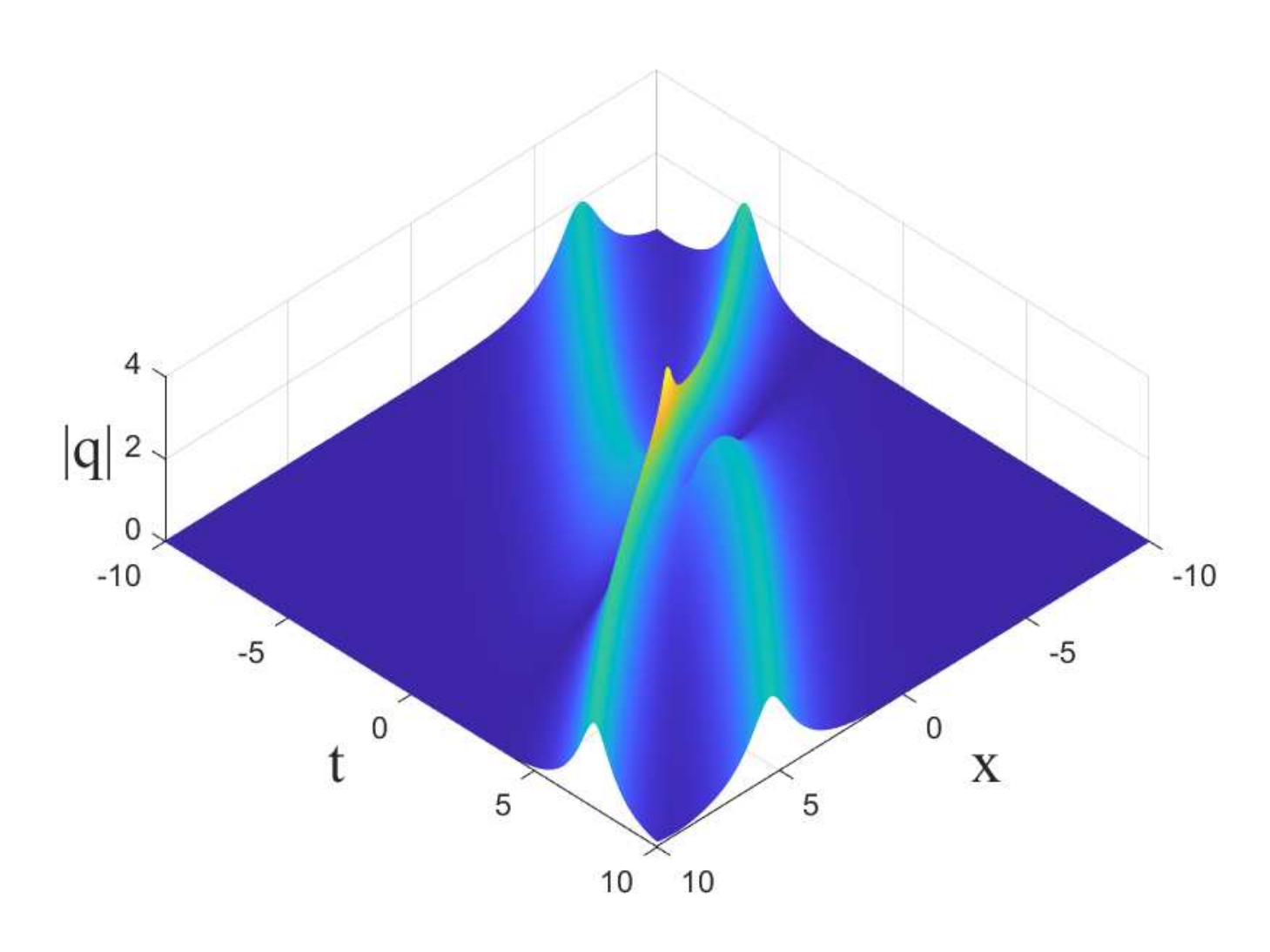}
%\caption{fig1}
\end{minipage}
}%
\subfigure[]{
\begin{minipage}[t]{0.33\textwidth}
\centering
\includegraphics[height=4.5cm,width=4.5cm]{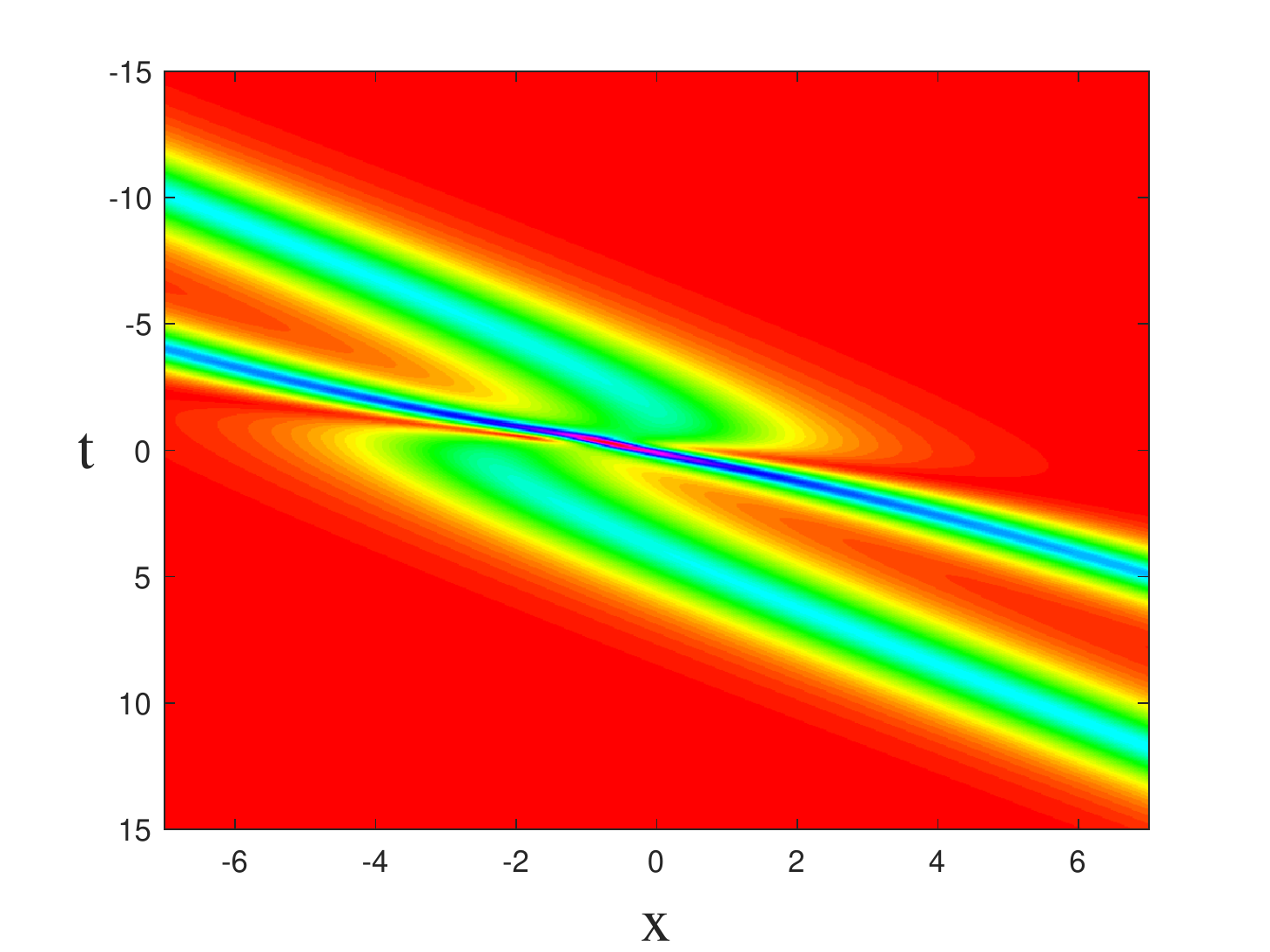}
%\caption{fig2}
\end{minipage}%
}%
\subfigure[]{
\begin{minipage}[t]{0.33\textwidth}
\centering
\includegraphics[height=4.5cm,width=4.5cm]{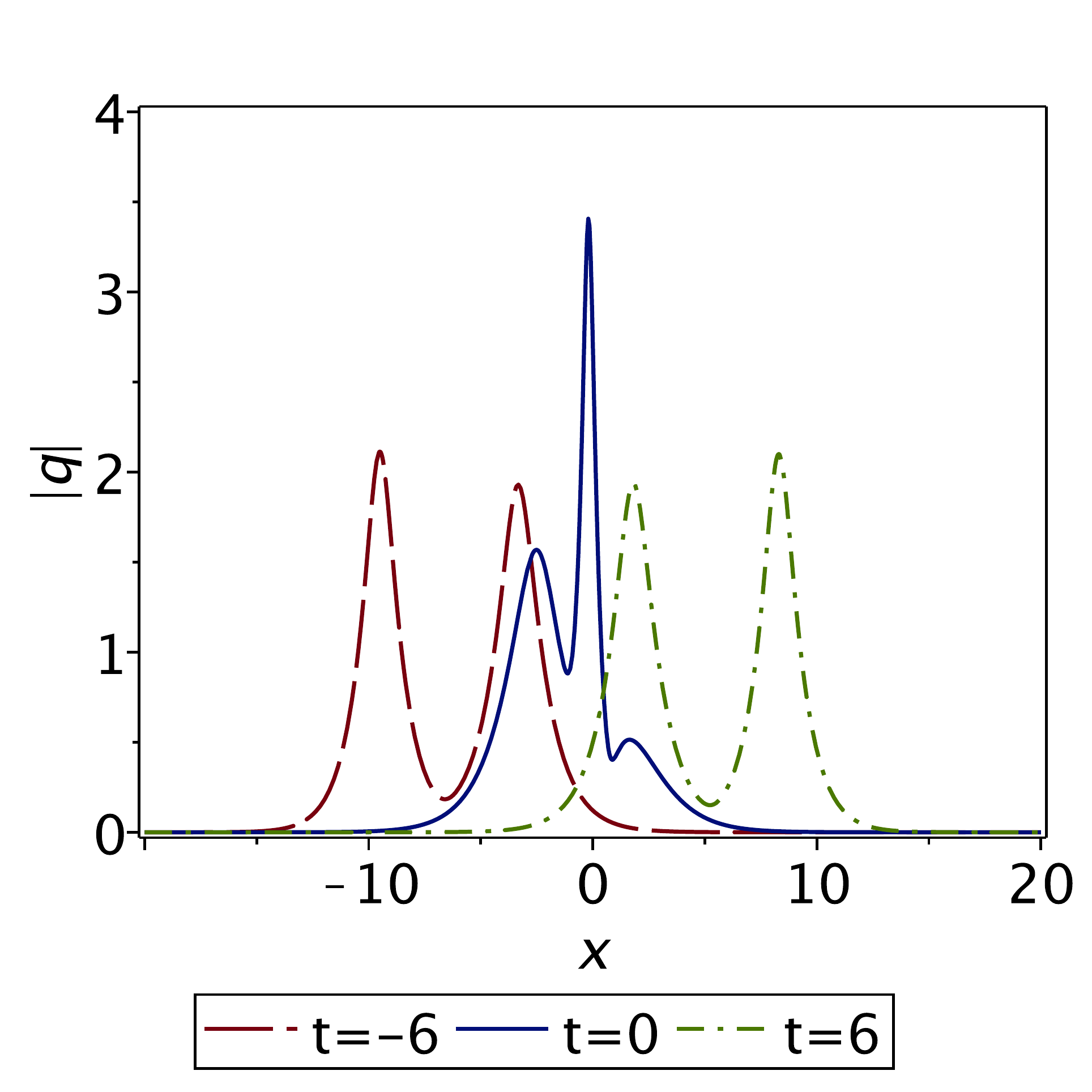}
%\caption{fig1}
\end{minipage}
}%
\centering
\caption{(Color online) The 1-double-pole soliton solution of TOFKN \eqref{T1} with ZBCs and $N=1,\zeta_1=\frac12+\frac12\mathrm{i},A[\zeta_1]=B[\zeta_1]=1$. (a) The three-dimensional plot; (b) The density plot; (c) The sectional drawings at $t=-6$ (dashed line), $t=0$ (solid line), and $t =6$ (dash-dot line).}
\label{F2}
\end{figure}

$\bullet$ When taking parameters $N=2,\zeta_1=\frac12+\frac12\mathrm{i},A[\zeta_1]=B[\zeta_1]=1,\zeta_2=\frac13+\frac12\mathrm{i},A[\zeta_2]=B[\zeta_2]=1$, we can obtain the explicit 2-double-pole solution and give out relevant plots in Fig. \ref{F3}. Figs. \ref{F3} (a) and (b) exhibit the three-dimensional and density diagrams for the exact 2-double-pole soliton solution of the TOFKN with ZBCs, which is equivalent to the interaction of two 1-double-pole soliton solutions. Fig. \ref{F3} (c) displays the distinct profiles of the exact 2-double-pole soliton solution at $t=\pm20,0$.

\begin{figure}[htbp]
\centering
\subfigure[]{
\begin{minipage}[t]{0.33\textwidth}
\centering
\includegraphics[height=4.5cm,width=4.5cm]{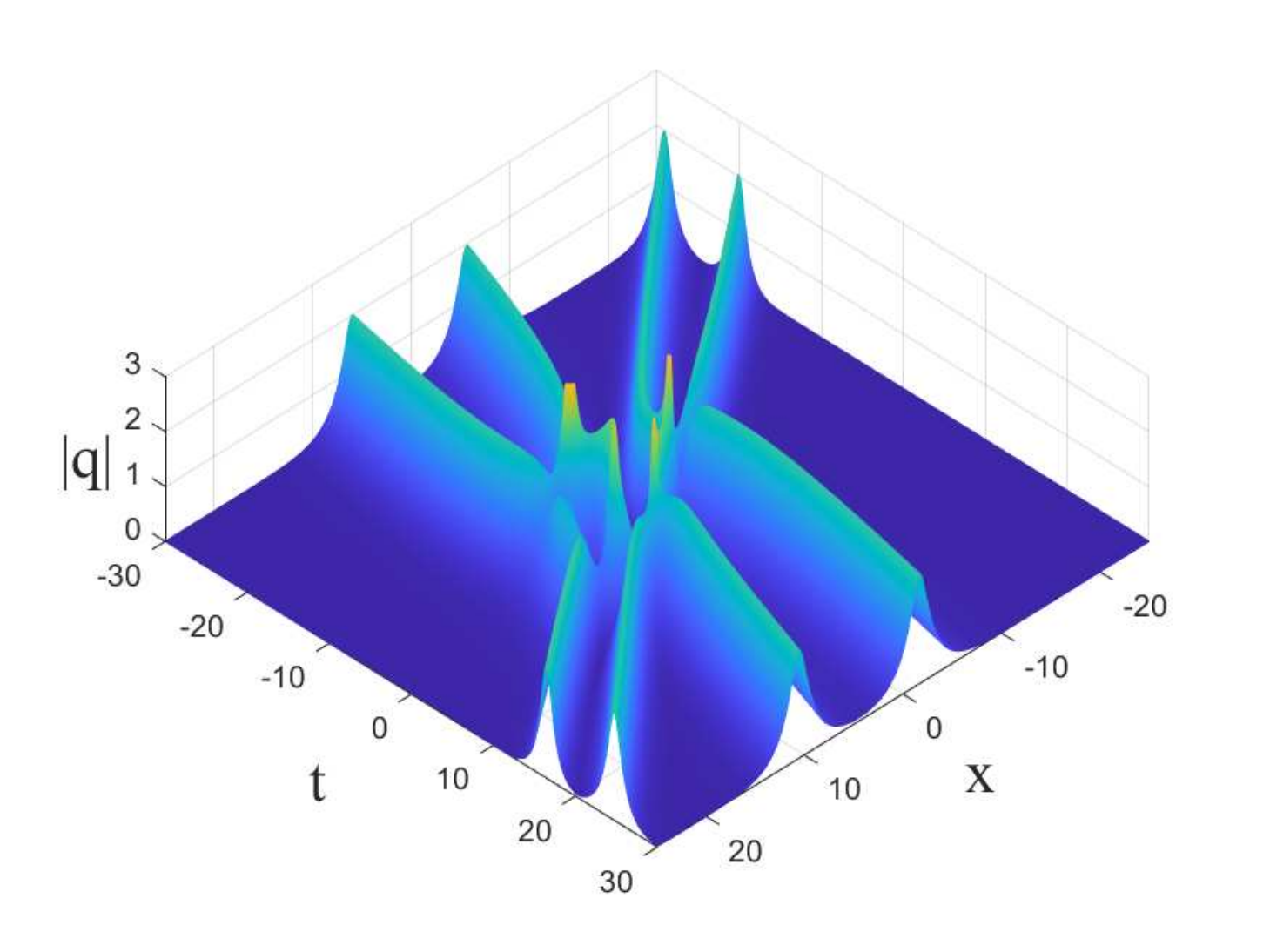}
%\caption{fig1}
\end{minipage}
}%
\subfigure[]{
\begin{minipage}[t]{0.33\textwidth}
\centering
\includegraphics[height=4.5cm,width=4.5cm]{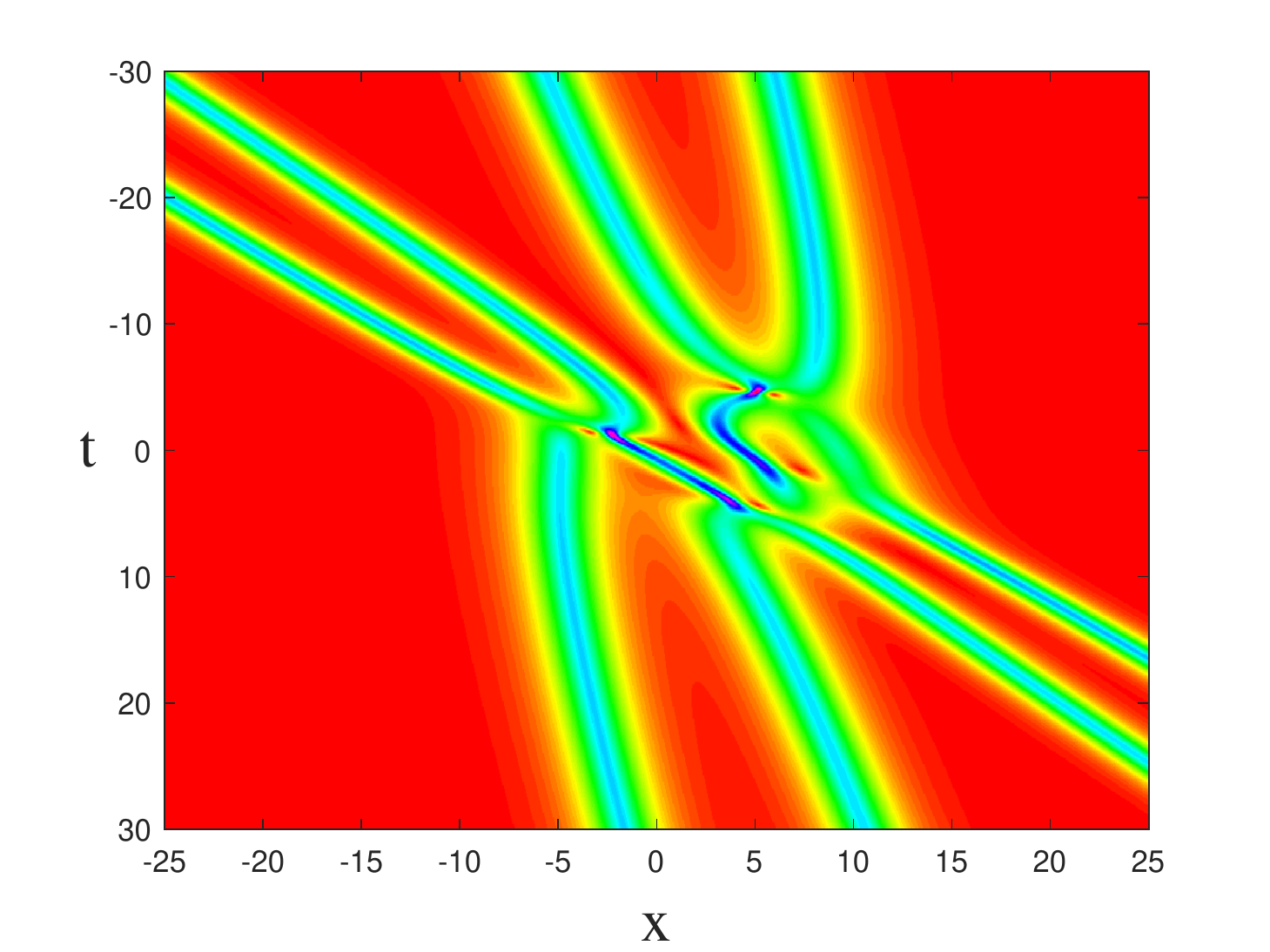}
%\caption{fig2}
\end{minipage}%
}%
\subfigure[]{
\begin{minipage}[t]{0.33\textwidth}
\centering
\includegraphics[height=4.5cm,width=4.5cm]{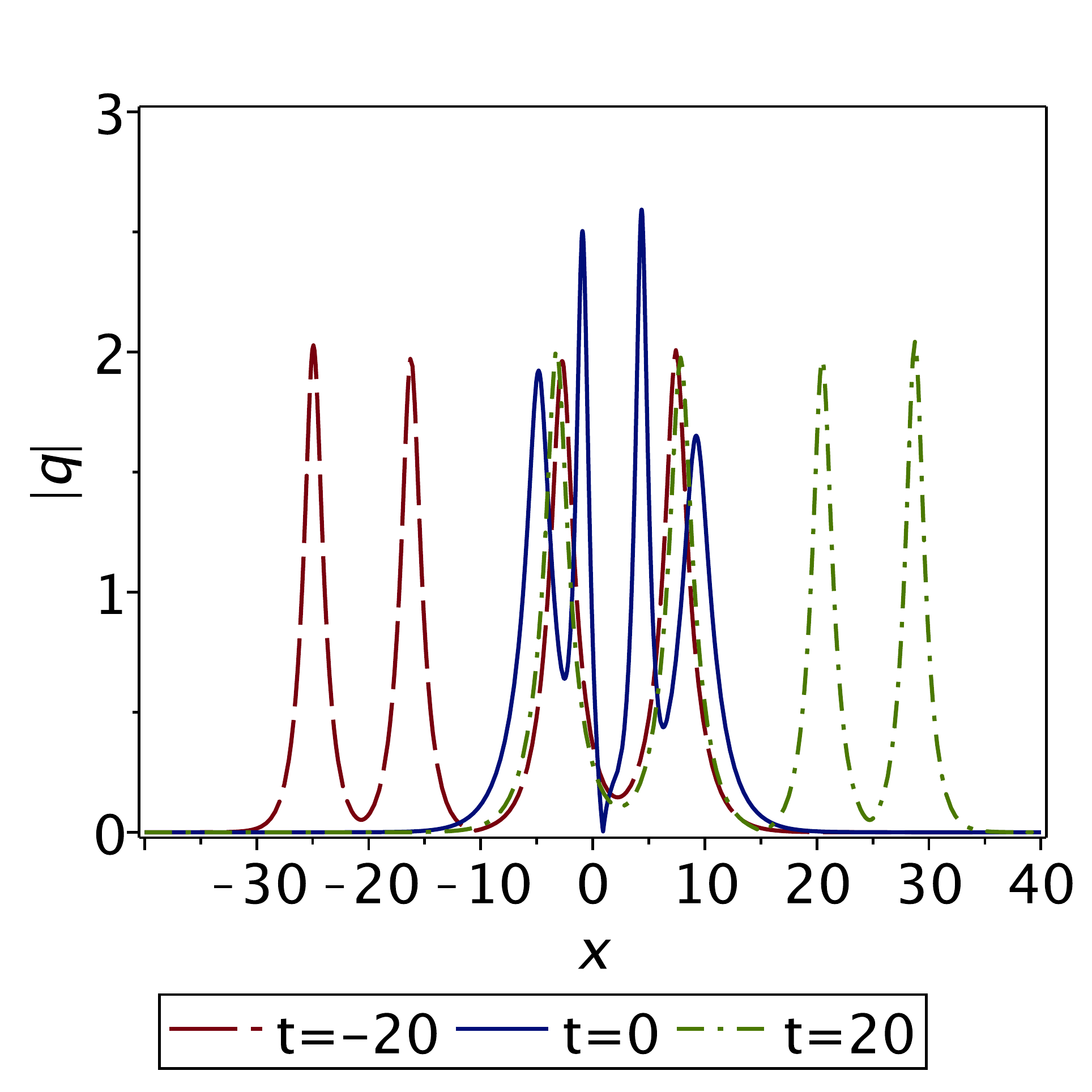}
%\caption{fig1}
\end{minipage}
}%
\centering
\caption{(Color online) The 2-double-pole soliton solution of TOFKN \eqref{T1} with ZBCs and $N=2,\zeta_1=\frac12+\frac12\mathrm{i},A[\zeta_1]=B[\zeta_1]=1,\zeta_2=\frac13+\frac12\mathrm{i},A[\zeta_2]=B[\zeta_2]=1$. (a) The three-dimensional plot; (b) The density plot; (c) The sectional drawings at $t=-20$ (dashed line), $t=0$ (solid line), and $t =20$ (dash-dot line).}
\label{F3}
\end{figure}

Moreover, we analyze the asymptotic states of the 1-double poles soliton solution of TOFKN \eqref{T1} with ZBCs as $t\rightarrow\infty$. Especially, we take $N=1, \zeta1=\frac12+\mathrm{i}, A[\zeta_1]=B[\zeta_1]=1$, and we can derive the asymptotic state of the 1-double poles soliton solution by means of Maple symbol calculations, given by
\begin{align}\label{ZD-AS}
\begin{split}
|q(x,t)|^2\longrightarrow\left\{
\begin{aligned}
\frac{19865600000\mathrm{e}^{4\theta_1}}{390625\mathrm{e}^{8\theta_1}-1862400000\mathrm{e}^{4\theta_1}+6166282240000}, \text{as }\theta_1=x+\frac{11}{4}t-\frac12\mathrm{log}(t)\\
\frac{8136949760\mathrm{e}^{4\theta_1}}{6166282240000\mathrm{e}^{8\theta_1}-762839040\mathrm{e}^{4\theta_1}+65536}, \text{as }\theta_1=x+\frac{11}{4}t+\frac12\mathrm{log}(t).
\end{aligned}
\right.
\end{split}
\end{align}
From the above expression, one can obtain that the 1-double poles soliton solution degrades into the two one-soliton solution as $t\rightarrow\infty$. Of which, the center trajectories are $x+\frac{11}{4}t-\frac12\mathrm{log}(t)$ and $x+\frac{11}{4}t+\frac12\mathrm{log}(t)$, respectively. When $t\rightarrow\infty$, the position shift of the
two standard one soliton solution is $\mathrm{log}(t)$, which depends on $t$.

\subsection{Inverse Problem with ZBCs and Triple Poles}

In the section, we devote to propose an inverse problem with ZBCs and solve it to obtain accurate triple poles solutions for the TOFKN \eqref{T1}.

\subsubsection{Discrete Spectrum with ZBCs and Triple Zeros}
Differing from the previous results with double poles, we here suppose that $s_{11}(\lambda)$ has $N$ triple zeros in $\Lambda_0 = \{\lambda\in\mathbb{C}: \mathrm{Re}\lambda>0, \mathrm{Im}\lambda>0\}$ denoted by $\lambda_n$, $n =1,2,\cdots,N$, that is, $s_{11}(\lambda_n)=s'_{11}(\lambda_n)=s''_{11}(\lambda_n)=0$, and $s'''_{11}(\lambda_n)\neq0$. Obviously, the corresponding discrete spectrum of triple poles is consistent with the discrete spectrum of double poles, as shown in Fig. \ref{F1}.

Once given $\lambda_0\in\Lambda\cap D^{+}$, we obtain that $\psi''_{+1}(\lambda_0;x,t)-b[\lambda_0]\psi''_{-2}(\lambda_0;x,t)-2d[\lambda_0]\psi'_{-2}(\lambda_0;x,t)$ and $\psi_{-2}(\lambda_0;x,t)$ are linearly dependent by combining Eqs. \eqref{T29}-\eqref{T30} and $s''_{11}(\lambda_0)=0$. Similarly, when a given $\lambda_0\in\Lambda\cap D^{-}$, one can obtain that $\psi''_{+2}(\lambda_0;x,t)-b[\lambda_0]\psi''_{-1}(\lambda_0;x,t)-2d[\lambda_0]\psi'_{-1}(\lambda_0;x,t)$ and $\psi_{-1}(\lambda_0;x,t)$ are linearly dependent by combining $s''_{22}(\lambda_0)=0$. For convenience, we introduce the norming constant $h[\lambda_0]$ such that
\begin{align}\label{DS-TZ1}
\begin{split}
\psi''_{+1}(\lambda_0;x,t)-b[\lambda_0]\psi''_{-2}(\lambda_0;x,t)-2d[\lambda_0]\psi'_{-2}(\lambda_0;x,t)=h[\lambda_0]\psi_{-2}(\lambda_0;x,t),\text{ as }\lambda_0\in\Lambda\cap D^{+},\\
\psi''_{+2}(\lambda_0;x,t)-b[\lambda_0]\psi''_{-1}(\lambda_0;x,t)-2d[\lambda_0]\psi'_{-1}(\lambda_0;x,t)=h[\lambda_0]\psi_{-1}(\lambda_0;x,t),\text{ as }\lambda_0\in\Lambda\cap D^{-}.
\end{split}
\end{align}

Then we notice that $\psi_{+1}(\lambda;x,t)$ and $s_{11}(\lambda)$ are analytic on $D^{+}$, and suppose $\lambda_0$ is the triple zeros of $s_{11}$. Let $\psi_{+1}(\lambda;x,t)$ and $s_{11}(\lambda)$ carry out Taylor expansion at $\lambda=\lambda_0$, we have
\begin{align}\nonumber
\begin{split}
&\frac{\psi_{+1}(\lambda;x,t)}{s_{11}(\lambda)}=\frac{6\psi_{+1}(\lambda_0;x,t)}{s'''_{11}(\lambda_0)}(\lambda-\lambda_0)^{-3}+\frac{-3\psi_{+1}(\lambda_0;x,t)s''''_{11}(\lambda_0)+12\psi'_{+1}(\lambda_0;x,t)s'''_{11}(\lambda_0)}{2(s'''_{11}(\lambda_0))^2}\\
&(\lambda-\lambda_0)^{-2}+\frac{3\psi_{+1}(\lambda_0;x,t)(s''''_{11}(\lambda_0))^2-12\psi'_{+1}(\lambda_0;x,t)s'''_{11}(\lambda_0)s''''_{11}(\lambda_0)+24\psi''_{+1}(\lambda_0;x,t)(s'''_{11}(\lambda_0))^2}{8(s'''_{11}(\lambda_0))^3}\\
&(\lambda-\lambda_0)^{-1}+\cdots,
\end{split}
\end{align}
Thus, as $\lambda_0\in\Lambda\cap D^{+}$, by means of Eqs. \eqref{T29}, \eqref{T30} and \eqref{DS-TZ1}, one has the compact form
\begin{align}\nonumber
\begin{split}
\mathop{P_{-3}}_{\lambda=\lambda_0}\bigg[\frac{\psi_{+1}(\lambda;x,t)}{s_{11}(\lambda)}\bigg]=\frac{6b[\lambda_0]\psi_{-2}(\lambda_0;x,t)}{s'''_{11}(\lambda_0)},
\end{split}
\end{align}
where $\underset{\lambda=\lambda_0}{P_{-3}}[f(\lambda;x,t)]$ denotes the coefficient of $O((\lambda-\lambda_0)^{-3})$ term in the Laurent series expansion of $f(\lambda;x,t)$ at $\lambda=\lambda_0$, and
\begin{align}\nonumber
\begin{split}
&\mathop{P_{-2}}_{\lambda=\lambda_0}\bigg[\frac{\psi_{+1}(\lambda;x,t)}{s_{11}(\lambda)}\bigg]=\frac{6b[\lambda_0]}{s'''_{11}(\lambda_0)}\bigg[\psi'_{-2}(\lambda_0;x,t)+\bigg(\frac{d[\lambda_0]}{b[\lambda_0]}-\frac{s''''_{11}(\lambda_0)}{4s'''_{11}(\lambda_0)}\bigg)\psi_{-2}(\lambda_0;x,t)\bigg],\\
&\mathop{\mathrm{Res}}_{\lambda=\lambda_0}\bigg[\frac{\psi_{+1}(\lambda;x,t)}{s_{11}(\lambda)}\bigg]=\frac{6b[\lambda_0]}{s'''_{11}(\lambda_0)}\bigg[\frac12\psi''_{-2}(\lambda_0;x,t)+\bigg(\frac{d[\lambda_0]}{b[\lambda_0]}-\frac{s''''_{11}(\lambda_0)}{4s'''_{11}(\lambda_0)}\bigg)\psi'_{-2}(\lambda_0;x,t)+\\
&\qquad\qquad\qquad\qquad\qquad\bigg(\frac{h[\lambda_0]}{2b[\lambda_0]}-\frac{d[\lambda_0]s''''_{11}(\lambda_0)}{4b[\lambda_0]s'''_{11}(\lambda_0)}+\frac{(s''''_{11}(\lambda_0))^2}{16(s'''_{11}(\lambda_0))^2}\bigg)\psi_{-2}(\lambda_0;x,t)\bigg].
\end{split}
\end{align}

Similarly, for the case of $\psi_{+2}(\lambda;x,t)$ and $s_{22}(\lambda)$ are analytic on $D^{-}$, as $\lambda_0\in\Lambda\cap D^{-}$ we repeat the above process and obtain
\begin{align}\nonumber
\begin{split}
&\mathop{P_{-3}}_{\lambda=\lambda_0}\bigg[\frac{\psi_{+2}(\lambda;x,t)}{s_{22}(\lambda)}\bigg]=\frac{6b[\lambda_0]\psi_{-1}(\lambda_0;x,t)}{s'''_{22}(\lambda_0)},\\
&\mathop{P_{-2}}_{\lambda=\lambda_0}\bigg[\frac{\psi_{+2}(\lambda;x,t)}{s_{22}(\lambda)}\bigg]=\frac{6b[\lambda_0]}{s'''_{22}(\lambda_0)}\bigg[\psi'_{-1}(\lambda_0;x,t)+\bigg(\frac{d[\lambda_0]}{b[\lambda_0]}-\frac{s''''_{22}(\lambda_0)}{4s'''_{22}(\lambda_0)}\bigg)\psi_{-1}(\lambda_0;x,t)\bigg],\\
&\mathop{\mathrm{Res}}_{\lambda=\lambda_0}\bigg[\frac{\psi_{+2}(\lambda;x,t)}{s_{22}(\lambda)}\bigg]=\frac{6b[\lambda_0]}{s'''_{22}(\lambda_0)}\bigg[\frac12\psi''_{-1}(\lambda_0;x,t)+\bigg(\frac{d[\lambda_0]}{b[\lambda_0]}-\frac{s''''_{22}(\lambda_0)}{4s'''_{22}(\lambda_0)}\bigg)\psi'_{-1}(\lambda_0;x,t)+\\
&\qquad\qquad\qquad\qquad\qquad\bigg(\frac{h[\lambda_0]}{2b[\lambda_0]}-\frac{d[\lambda_0]s''''_{22}(\lambda_0)}{4b[\lambda_0]s'''_{22}(\lambda_0)}+\frac{(s''''_{22}(\lambda_0))^2}{16(s'''_{22}(\lambda_0))^2}\bigg)\psi_{-1}(\lambda_0;x,t)\bigg].
\end{split}
\end{align}

Moreover, let
\begin{align}\label{DS-TZ2}
\begin{split}
&\widetilde{A}[\lambda_0]=\left\{
\begin{aligned}
\frac{6b[\lambda_0]}{s'''_{11}(\lambda_0)},\text{ as }\lambda_0\in\Lambda\cap D^{+},\\
\frac{6b[\lambda_0]}{s'''_{22}(\lambda_0)},\text{ as }\lambda_0\in\Lambda\cap D^{-},
\end{aligned}
\right.\quad
\widetilde{B}[\lambda_0]=\left\{
\begin{aligned}
\frac{d[\lambda_0]}{b[\lambda_0]}-\frac{s''''_{11}(\lambda_0)}{4s'''_{11}(\lambda_0)},\text{ as }\lambda_0\in\Lambda\cap D^{+},\\
\frac{d[\lambda_0]}{b[\lambda_0]}-\frac{s''''_{22}(\lambda_0)}{4s'''_{22}(\lambda_0)},\text{ as }\lambda_0\in\Lambda\cap D^{-},
\end{aligned}
\right.\\
&\widetilde{C}[\lambda_0]=\left\{
\begin{aligned}
\frac{h[\lambda_0]}{2b[\lambda_0]}-\frac{d[\lambda_0]s''''_{11}(\lambda_0)}{4b[\lambda_0]s'''_{11}(\lambda_0)}+\frac{(s''''_{11}(\lambda_0))^2}{16(s'''_{11}(\lambda_0))^2},\text{ as }\lambda_0\in\Lambda\cap D^{+},\\
\frac{h[\lambda_0]}{2b[\lambda_0]}-\frac{d[\lambda_0]s''''_{22}(\lambda_0)}{4b[\lambda_0]s'''_{22}(\lambda_0)}+\frac{(s''''_{22}(\lambda_0))^2}{16(s'''_{22}(\lambda_0))^2},\text{ as }\lambda_0\in\Lambda\cap D^{-}.
\end{aligned}
\right.
\end{split}
\end{align}

Then, we have
\begin{align}\label{DS-TZ3}
\begin{split}
&\mathop{P_{-3}}_{\lambda=\lambda_0}\Bigg[\frac{\psi_{+1}(\lambda;x,t)}{s_{11}(\lambda)}\Bigg]=\widetilde{A}[\lambda_0]\psi_{-2}(\lambda_0;x,t),\text{ as }\lambda_0\in\Lambda\cap D^{+},\\
&\mathop{P_{-3}}_{\lambda=\lambda_0}\Bigg[\frac{\psi_{+2}(\lambda;x,t)}{s_{22}(\lambda)}\Bigg]=\widetilde{A}[\lambda_0]\psi_{-1}(\lambda_0;x,t),\text{ as }\lambda_0\in\Lambda\cap D^{-},\\
&\mathop{P_{-2}}_{\lambda=\lambda_0}\Bigg[\frac{\psi_{+1}(\lambda;x,t)}{s_{11}(\lambda)}\Bigg]=\widetilde{A}[\lambda_0][\psi'_{-2}(\lambda_0;x,t)+\widetilde{B}[\lambda_0]\psi_{-2}(\lambda_0;x,t)],\text{ as }\lambda_0\in\Lambda\cap D^{+},\\
&\mathop{P_{-2}}_{\lambda=\lambda_0}\Bigg[\frac{\psi_{+2}(\lambda;x,t)}{s_{22}(\lambda)}\Bigg]=\widetilde{A}[\lambda_0][\psi'_{-1}(\lambda_0;x,t)+\widetilde{B}[\lambda_0]\psi_{-1}(\lambda_0;x,t)],\text{ as }\lambda_0\in\Lambda\cap D^{-},\\
&\mathop{\mathrm{Res}}_{\lambda=\lambda_0}\Bigg[\frac{\psi_{+1}(\lambda;x,t)}{s_{11}(\lambda)}\Bigg]=\widetilde{A}[\lambda_0]\bigg[\frac12\psi''_{-2}(\lambda_0;x,t)+\widetilde{B}[\lambda_0]\psi'_{-2}(\lambda_0;x,t)+\widetilde{C}[\lambda_0]\psi_{-2}(\lambda_0;x,t)\bigg],\\
&\qquad\qquad\qquad\qquad\qquad\text{ as }\lambda_0\in\Lambda\cap D^{+},\\
&\mathop{\mathrm{Res}}_{\lambda=\lambda_0}\Bigg[\frac{\psi_{+2}(\lambda;x,t)}{s_{22}(\lambda)}\Bigg]=\widetilde{A}[\lambda_0]\bigg[\frac12\psi''_{-1}(\lambda_0;x,t)+\widetilde{B}[\lambda_0]\psi'_{-1}(\lambda_0;x,t)+\widetilde{C}[\lambda_0]\psi_{-1}(\lambda_0;x,t)\bigg],\\
&\qquad\qquad\qquad\qquad\qquad\text{ as }\lambda_0\in\Lambda\cap D^{-}.
\end{split}
\end{align}

Accordingly, by mean of Eqs. \eqref{T29}-\eqref{T30}, Eqs. \eqref{DS-TZ1}-\eqref{DS-TZ2} as well as proposition \ref{P4}, we can derive the following symmetry relations.
\begin{prop}\label{DS-TZ-P1}
For $\lambda_0\in\Lambda$, the two symmetry relations for $\widetilde{A}[\lambda_0]$, $\widetilde{B}[\lambda_0]$ and $\widetilde{C}[\lambda_0]$ can be deduced as follows:\\
$\bullet$ The first symmetry relation $\widetilde{A}[\lambda_0]=-\widetilde{A}[\lambda^*_0]^*$, $\widetilde{B}[\lambda_0]=\widetilde{B}[\lambda^*_0]^*$, $\widetilde{C}[\lambda_0]=\widetilde{C}[\lambda^*_0]^*$.\\
$\bullet$ The second symmetry relation $\widetilde{A}[\lambda_0]=-\widetilde{A}[-\lambda^*_0]^*$, $\widetilde{B}[\lambda_0]=-\widetilde{B}[-\lambda^*_0]^*$ , $\widetilde{C}[\lambda_0]=\widetilde{C}[-\lambda^*_0]^*$.
\end{prop}

\subsubsection{The Matrix RH Problem with ZBCs and Triple Poles}

Similarly, a matrix RH problem is built as follows.
\begin{prop}\label{MRHP-P1}
Define the sectionally meromorphic matrices
\begin{small}
\begin{align}\label{MRHP-1}
M(\lambda;x,t)=\left\{
\begin{aligned}
M^+(\lambda;x,t)=\bigg(\frac{\mu_{+1}(\lambda;x,t)}{s_{11}(\lambda)},\mu_{-2}(\lambda;x,t)\bigg),\text{ as } \lambda\in D^+,\\
M^-(\lambda;x,t)=\bigg(\mu_{-1}(\lambda;x,t),\frac{\mu_{+2}(\lambda;x,t)}{s_{22}(\lambda)}\bigg),\text{ as } \lambda\in D^-,
\end{aligned}
\right.
\end{align}
\end{small}
where $\lim\limits_{\substack{\lambda'\rightarrow\lambda\\\lambda'\in D^{\pm}}}M(\lambda';x,t)=M^{\pm}(\lambda;x,t)$. Then, the multiplicative matrix RH problem is given below:\\
$\bullet$ Analyticity: $M(\lambda;x,t)$ is analytic in $D^+\cup D^-\backslash\Lambda$ and has the triple poles in $\Lambda$, whose principal parts of the Laurent series at each triple pole $\zeta_n$ or $\zeta^*_n$, are determined as
\begin{small}
\begin{align}\label{MRHP-2}
\begin{split}
&\mathop{\mathrm{Res}}_{\lambda=\zeta_n}M^+(\lambda;x,t)=\bigg(\widetilde{A}[\zeta_n]\mathrm{e}^{-2\mathrm{i}\theta(\zeta_n;x,t)}\bigg[\frac12\mu''_{-2}(\zeta_n;x,t)+(\widetilde{B}[\zeta_n]-2\mathrm{i}\theta'(\zeta_n;x,t))\mu'_{-2}(\zeta_n;x,t)\\
&\qquad\qquad\qquad\qquad+(\widetilde{C}[\zeta_n]-2(\theta'(\zeta_n))^2-\mathrm{i}\theta''(\zeta_n)-2\mathrm{i}\theta'(\zeta_n)\widetilde{B}[\zeta_n])\mu_{-2}(\zeta_n;x,t)\bigg],0\bigg),\\
&\mathop{\mathrm{P}_{-2}}_{\lambda=\zeta_n}M^+(\lambda;x,t)=\big(\widetilde{A}[\zeta_n]\mathrm{e}^{-2i\theta(\zeta_n;x,t)}[\mu'_{-2}(\zeta_n;x,t)+(\widetilde{B}[\zeta_n]-2\mathrm{i}\theta'(\zeta_n;x,t))\mu_{-2}(\zeta_n;x,t)],0\big),\\
&\mathop{\mathrm{P}_{-3}}_{\lambda=\zeta_n}M^+(\lambda;x,t)=\big(\widetilde{A}[\zeta_n]\mathrm{e}^{-2\mathrm{i}\theta(\zeta_n;x,t)}\mu_{-2}(\zeta_n;x,t),0\big),\\
&\mathop{\mathrm{Res}}_{\lambda=\zeta^*_n}M^+(\lambda;x,t)=\bigg(0,\widetilde{A}[\zeta^*_n]\mathrm{e}^{2\mathrm{i}\theta(\zeta^*_n;x,t)}\bigg[\frac12\mu''_{-1}(\zeta^*_n;x,t)+(\widetilde{B}[\zeta^*_n]+2\mathrm{i}\theta'(\zeta^*_n;x,t))\mu'_{-1}(\zeta^*_n;x,t)\\
&\qquad\qquad\qquad\qquad+(\widetilde{C}[\zeta^*_n]-2(\theta'(\zeta^*_n))^2+\mathrm{i}\theta''(\zeta^*_n)+2\mathrm{i}\theta'(\zeta^*_n)\widetilde{B}[\zeta^*_n])\mu_{-1}(\zeta^*_n;x,t)\bigg]\bigg),\\
&\mathop{\mathrm{P}_{-2}}_{\lambda=\zeta^*_n}M^+(\lambda;x,t)=\big(0,\widetilde{A}[\zeta^*_n]\mathrm{e}^{2i\theta(\zeta^*_n;x,t)}[\mu'_{-1}(\zeta^*_n;x,t)+(\widetilde{B}[\zeta_n]+2\mathrm{i}\theta'(\zeta^*_n;x,t))\mu_{-1}(\zeta^*_n;x,t)]\big),\\
&\mathop{\mathrm{P}_{-3}}_{\lambda=\zeta^*_n}M^+(\lambda;x,t)=\big(0,\widetilde{A}[\zeta^*_n]\mathrm{e}^{2\mathrm{i}\theta(\zeta^*_n;x,t)}\mu_{-1}(\zeta^*_n;x,t)\big).
\end{split}
\end{align}
\end{small}
$\bullet$ Jump condition:
\begin{align}\label{MRHP-3}
M^-(\lambda;x,t)=M^+(\lambda;x,t)[I-J(\lambda;x,t)],\text{ as }\lambda\in\Sigma,
\end{align}
where
\begin{align}\label{MRHP-4}
J(\lambda;x,t)=\mathrm{e}^{\mathrm{i}\theta(\lambda;x,t)\widehat{\sigma_3}}\left(\begin{array}{cc} 0 & -\tilde{\rho}(\lambda) \\ \rho(\lambda) & \rho(\lambda)\tilde{\rho}(\lambda) \end{array}\right).
\end{align}
$\bullet$ Asymptotic behavior:
\begin{align}\label{MRHP-5}
M(\lambda;x,t)=\mathrm{e}^{\mathrm{i}\nu_{-}(x,t)\sigma_3}+O\bigg(\frac{1}{\lambda}\bigg),\text{ as }\lambda\rightarrow\infty.
\end{align}

\end{prop}

Therefore utilizing asymptotic values as $\lambda\rightarrow\infty$ and the singularity contributions, one can regularize the RH problem as a normative form. Then, applying the Plemelj's formula, the solutions of the corresponding matrix RH problem can be solved as follows

\begin{prop}\label{MRHP-P2}
The solution of the above-mentioned matrix Riemann-Hilbert problem can be expressed as
\begin{small}
\begin{align}\label{MRHP-6}
\begin{split}
M(\lambda;x,t)=&\mathrm{e}^{\mathrm{i}\nu_{-}(x,t)\sigma_3}+\frac{1}{2\pi \rm i}\int_{\Sigma}\frac{M^+(\xi;x,t)J(\xi;x,t)}{\xi-\lambda}d\xi+\sum^{2N}_{n=1}\bigg(C_n(\lambda)\bigg[\frac12\mu''_{-2}(\zeta_n;x,t)+\bigg(D_n+\\
&\frac{1}{\lambda-\zeta_n}\bigg)\mu'_{-2}(\zeta_n;x,t)+\bigg(\frac{1}{(\lambda-\zeta_n)^2}+\frac{D_n}{\lambda-\zeta_n}+F_n\bigg)\mu_{-2}(\zeta_n;x,t)\bigg],\widehat{C}_n(\lambda)\bigg[\frac12\mu''_{-1}(\zeta^*_n;x,t)\\
&+\bigg(\widehat{D}_n+\frac{1}{\lambda-\zeta^*_n}\bigg)\mu'_{-1}(\zeta^*_n;x,t)+\bigg(\frac{1}{(\lambda-\zeta^*_n)^2}+\frac{\widehat{D}_n}{\lambda-\zeta^*_n}+\widehat{F}_n\bigg)\mu_{-1}(\zeta^*_n;x,t)\bigg),
\end{split}
\end{align}
\end{small}
where $\lambda\in\mathbb{C}\backslash\Sigma$, $\int_{\Sigma}$ is an integral along the oriented contour exhibited in Fig. \ref{F1}, and
\begin{small}
\begin{align}\label{MRHP-7}
\begin{split}
&C_n(\lambda)=\frac{\widetilde{A}[\zeta_n]}{\lambda-\zeta_n}\mathrm{e}^{-2\mathrm{i}\theta(\zeta_n;x,t)},\quad \widehat{C}_n(\lambda)=\frac{\widetilde{A}[\zeta^*_n]}{\lambda-\zeta^*_n}\mathrm{e}^{2\mathrm{i}\theta(\zeta^*_n;x,t)},\\
&D_n=\widetilde{B}[\zeta_n]-2\mathrm{i}\theta'(\zeta_n;x,t),\quad \widehat{D}_n=\widetilde{B}[\zeta^*_n]+2\mathrm{i}\theta'(\zeta^*_n;x,t),\\
&F_n=\widetilde{C}[\zeta_n]-2(\theta'(\zeta_n))^2-\mathrm{i}\theta''(\zeta_n)-2\mathrm{i}\theta'(\zeta_n)\widetilde{B}[\zeta_n],\\
&\widehat{F}_n=\widetilde{C}[\zeta^*_n]-2(\theta'(\zeta^*_n))^2+\mathrm{i}\theta''(\zeta^*_n)+2\mathrm{i}\theta'(\zeta^*_n)\widetilde{B}[\zeta^*_n],
\end{split}
\end{align}
\end{small}
in which, $\mu_{-k}$, $\mu'_{-k}$ and $\mu''_{-k}$ $(k=1,2)$ satisfy
\begin{footnotesize}
\begin{align}\label{MRHP-8}
\begin{split}
&\mu_{-1}(\zeta^*_n;x,t)=\mathrm{e}^{\mathrm{i}\nu_{-}(x,t)\sigma_3}\bigg(\begin{array}{cc} 1 \\ 0  \end{array}\bigg)+\sum_{s=1}^{2N}C_s(\zeta^*_n)\bigg[\frac12\mu''_{-2}(\zeta_s;x,t)+\bigg(D_s+\frac{1}{\zeta^*_n-\zeta_s}\bigg)\mu'_{-2}(\zeta_s;x,t)+\\
&\qquad\qquad\qquad\bigg(\frac{1}{(\zeta^*_n-\zeta_s)^2}+\frac{D_s}{\zeta^*_n-\zeta_s}+F_s\bigg)\mu_{-2}(\zeta_s;x,t)\bigg]+\frac{1}{2\pi \rm i}\int_{\Sigma}\frac{(M^+(\xi;x,t)J(\xi;x,t))_{1}}{\xi-\zeta^*_n}d\xi,\\
&\mu_{-2}(\zeta_s;x,t)=\mathrm{e}^{\mathrm{i}\nu_{-}(x,t)\sigma_3}\bigg(\begin{array}{cc} 0 \\ 1  \end{array}\bigg)+\sum_{j=1}^{2N}\widehat{C}_j(\zeta_s)\bigg[\frac12\mu''_{-1}(\zeta^*_j;x,t)+\bigg(\widehat{D}_j+\frac{1}{\zeta_s-\zeta^*_j}\bigg)\mu'_{-1}(\zeta^*_j;x,t)+\\
&\qquad\qquad\qquad\bigg(\frac{1}{(\zeta_s-\zeta^*_j)^2}+\frac{\widehat{D}_j}{\zeta_s-\zeta^*_j}+\widehat{F}_j\bigg)\mu_{-1}(\zeta^*_j;x,t)\bigg]+\frac{1}{2\pi \rm i}\int_{\Sigma}\frac{(M^+(\xi;x,t)J(\xi;x,t))_{2}}{\xi-\zeta_s}d\xi,\\
&\mu'_{-1}(\zeta^*_n;x,t)=-\sum_{s=1}^{2N}\frac{C_s(\zeta^*_n)}{\zeta^*_n-\zeta_s}\bigg[\frac12\mu''_{-2}(\zeta_s;x,t)+\bigg(D_s+\frac{2}{\zeta^*_n-\zeta_s}\bigg)\mu'_{-2}(\zeta_s;x,t)+\bigg(\frac{3}{(\zeta^*_n-\zeta_s)^2}+\\
&\qquad\qquad\qquad\frac{2D_s}{\zeta^*_n-\zeta_s}+F_s\bigg)\mu_{-2}(\zeta_s;x,t)\bigg]+\frac{1}{2\pi \rm i}\int_{\Sigma}\frac{(M^+(\xi;x,t)J(\xi;x,t))_{1}}{(\xi-\zeta^*_n)^2}d\xi,\\
&\mu'_{-2}(\zeta_s;x,t)=-\sum_{j=1}^{2N}\frac{\widehat{C}_j(\zeta_s)}{\zeta_s-\zeta^*_j}\bigg[\frac12\mu''_{-1}(\zeta^*_j;x,t)+\bigg(\widehat{D}_j+\frac{2}{\zeta_s-\zeta^*_j}\bigg)\mu'_{-1}(\zeta^*_j;x,t)+\bigg(\frac{3}{(\zeta_s-\zeta^*_j)^2}+\\
&\qquad\qquad\qquad\frac{2\widehat{D}_j}{\zeta_s-\zeta^*_j}+\widehat{F}_j\bigg)\mu_{-1}(\zeta^*_j;x,t)\bigg]+\frac{1}{2\pi \rm i}\int_{\Sigma}\frac{(M^+(\xi;x,t)J(\xi;x,t))_{2}}{(\xi-\zeta_s)^2}d\xi,\\
&\mu''_{-1}(\zeta^*_n;x,t)=\sum_{s=1}^{2N}\frac{2C_s(\zeta^*_n)}{(\zeta^*_n-\zeta_s)^2}\bigg[\frac12\mu''_{-2}(\zeta_s;x,t)+\bigg(D_s+\frac{3}{\zeta^*_n-\zeta_s}\bigg)\mu'_{-2}(\zeta_s;x,t)+\bigg(\frac{6}{(\zeta^*_n-\zeta_s)^2}+\\
&\qquad\qquad\qquad\frac{3D_s}{\zeta^*_n-\zeta_s}+F_s\bigg)\mu_{-2}(\zeta_s;x,t)\bigg]+\frac{1}{2\pi \rm i}\int_{\Sigma}\frac{2(M^+(\xi;x,t)J(\xi;x,t))_{1}}{(\xi-\zeta^*_n)^3}d\xi,\\
&\mu''_{-2}(\zeta_s;x,t)=\sum_{j=1}^{2N}\frac{2\widehat{C}_j(\zeta_s)}{(\zeta_s-\zeta^*_j)^2}\bigg[\frac12\mu''_{-1}(\zeta^*_j;x,t)+\bigg(\widehat{D}_j+\frac{3}{\zeta_s-\zeta^*_j}\bigg)\mu'_{-1}(\zeta^*_j;x,t)+\bigg(\frac{6}{(\zeta_s-\zeta^*_j)^2}+\\
&\qquad\qquad\qquad\frac{3\widehat{D}_j}{\zeta_s-\zeta^*_j}+\widehat{F}_j\bigg)\mu_{-1}(\zeta^*_j;x,t)\bigg]+\frac{1}{2\pi \rm i}\int_{\Sigma}\frac{2(M^+(\xi;x,t)J(\xi;x,t))_{2}}{(\xi-\zeta_s)^3}d\xi,
\end{split}
\end{align}
\end{footnotesize}
where $(M^+(\xi;x,t)J(\xi;x,t))_{j}$ $(j=1,2)$ represent $j$th column of matrix $M^+(\xi;x,t)J(\xi;x,t)$.
\end{prop}

\begin{proof}
In order to regularize the RH problem, one has to subtract out the asymptotic values as $\lambda\rightarrow\infty$ which exhibited in Eq. \eqref{MRHP-5} and the singularity contributions. Then, the jump condition \eqref{MRHP-3} becomes
\begin{align}\label{MRHP-9}
\begin{split}
&M^-(\lambda;x,t)-\mathrm{e}^{\mathrm{i}\nu_{-}(x,t)\sigma_3}-\sum_{n=1}^{2N}\Bigg[\frac{\mathop{\mathrm{P}_{-3}}\limits_{\lambda=\zeta_n}M^+(\lambda;x,t)}{(\lambda-\zeta_n)^3}+\frac{\mathop{\mathrm{P}_{-2}}\limits_{\lambda=\zeta_n}M^+(\lambda;x,t)}{(\lambda-\zeta_n)^2}+\frac{\mathop{\mathrm{Res}}\limits_{\lambda=\zeta_n}M^+(\lambda;x,t)}{\lambda-\zeta_n}+\\
&\frac{\mathop{\mathrm{P}_{-3}}\limits_{\lambda=\zeta^*_n}M^-(\lambda;x,t)}{(\lambda-\zeta^*_n)^3}+\frac{\mathop{\mathrm{P}_{-2}}\limits_{\lambda=\zeta^*_n}M^-(\lambda;x,t)}{(\lambda-\zeta^*_n)^2}+\frac{\mathop{\mathrm{Res}}\limits_{\lambda=\zeta^*_n}M^-(\lambda;x,t)}{\lambda-\zeta^*_n}\Bigg]=M^+(\lambda;x,t)-\mathrm{e}^{\mathrm{i}\nu_{-}(x,t)\sigma_3}-\\
&\sum_{n=1}^{2N}\Bigg[\frac{\mathop{\mathrm{P}_{-3}}\limits_{\lambda=\zeta_n}M^+(\lambda;x,t)}{(\lambda-\zeta_n)^3}+\frac{\mathop{\mathrm{P}_{-2}}\limits_{\lambda=\zeta_n}M^+(\lambda;x,t)}{(\lambda-\zeta_n)^2}+\frac{\mathop{\mathrm{Res}}\limits_{\lambda=\zeta_n}M^+(\lambda;x,t)}{\lambda-\zeta_n}+\frac{\mathop{\mathrm{P}_{-3}}\limits_{\lambda=\zeta^*_n}M^-(\lambda;x,t)}{(\lambda-\zeta^*_n)^3}+\\
&\frac{\mathop{\mathrm{P}_{-2}}\limits_{\lambda=\zeta^*_n}M^-(\lambda;x,t)}{(\lambda-\zeta^*_n)^2}+\frac{\mathop{\mathrm{Res}}\limits_{\lambda=\zeta^*_n}M^-(\lambda;x,t)}{\lambda-\zeta^*_n}\Bigg]-M^+(\lambda;x,t)J(\lambda;x,t),
\end{split}
\end{align}
where $\mathop{\mathrm{P}_{-3}}\limits_{\lambda=\zeta_n}M^+,\mathop{\mathrm{P}_{-2}}\limits_{\lambda=\zeta_n}M^+,\mathop{\mathrm{Res}}\limits_{\lambda=\zeta_n}M^+,\mathop{\mathrm{P}_{-3}}\limits_{\lambda=\zeta^*_n}M^-,\mathop{\mathrm{P}_{-2}}\limits_{\lambda=\zeta^*_n}M^-,\mathop{\mathrm{Res}}\limits_{\lambda=\zeta^*_n}M^-$ have given in Eq. \eqref{MRHP-2}. By using Plemelj's formula, one can obtain the solution \eqref{MRHP-6} with formula \eqref{MRHP-7} of the matrix RH problem. By combining Eqs. \eqref{MRHP-1} and \eqref{MRHP-6}, $\mu_{-1}(\zeta^*_n;x,t)$ is the first column element of the solution \eqref{MRHP-6} as the triple-pole $\lambda=\zeta^*_n\in D^{-}$, $\mu_{-2}(\zeta_s;x,t)$ is the second column element of the solution \eqref{MRHP-6} as the triple-pole $\lambda=\zeta_s\in D^{+}$. The specific expressions of $\mu_{-1}(\zeta^*_n;x,t),\mu_{-2}(\zeta_s;x,t)$ and their first and second derivative to $\lambda$ are provided in Eq. \eqref{MRHP-8}.

\end{proof}

\subsubsection{Reconstruction Formula of the Potential with ZBCs and Triple poles}
Similarly, the reconstruction formula of the triple-pole solution (potential) for the TOFKN \eqref{T1} with ZBCs is consistent with Eq. \eqref{T45}. From the Eq. \eqref{T43} and solution \eqref{MRHP-6} of the matrix RH problem, we have
\begin{align}\label{RFP-TP-1}
M^{[1]}_{12}(x,t)=&-\frac{1}{2\pi\mathrm{i}}\int_{\Sigma}\big(M^{+}(\xi;x,t)J(\xi;x,t)\big)_{12}d\xi+\sum^{2N}_{n=1}\bigg\{A[\zeta^*_n]\mathrm{e}^{2\mathrm{i}\theta(\zeta^*_n;x,t)}\\
&\bigg[\frac12\mu''_{-11}(\zeta^*_n;x,t)+\widehat{D}_n\mu'_{-11}(\zeta^*_n;x,t)+\widehat{F}_n\mu_{-11}(\zeta^*_n;x,t)\bigg]\bigg\},
\end{align}
where $\mu_{-11}(\zeta^*_n;x,t), \mu'_{-11}(\zeta^*_n;x,t)$ and $\mu''_{-11}(\zeta^*_n;x,t)$ represents the first row element of the column vector $\mu_{-1}(\zeta^*_n;x,t), \mu'_{-1}(\zeta^*_n;x,t)$ and $\mu''_{-1}(\zeta^*_n;x,t)$, respectively. Then taking row vector $\alpha=\big(\alpha^{(1)},\alpha^{(2)},\alpha^{(3)}\big)$ and column vector $\gamma=(\gamma^{(1)},\gamma^{(2)},\gamma^{(3)})^{\mathrm{T}}$, where
\begin{align}\label{RFP-TP-2}
\begin{split}
&\alpha^{(1)}=\big(\widetilde{A}[\zeta^*_n]\mathrm{e}^{2\mathrm{i}\theta(\zeta^*_n;x,t)}\widehat{F}_n\big)_{1\times2N},\,\alpha^{(2)}=\big(\widetilde{A}[\zeta^*_n]\mathrm{e}^{2\mathrm{i}\theta(\zeta^*_n;x,t)}\widehat{D}_n\big)_{1\times2N},\\
&\alpha^{(3)}=\bigg(\frac12\widetilde{A}[\zeta^*_n]\mathrm{e}^{2\mathrm{i}\theta(\zeta^*_n;x,t)}\bigg)_{1\times2N},\\
&\gamma^{(1)}=\big(\mu_{-11}(\zeta^*_n;x,t)\big)_{1\times2N},\,\gamma^{(2)}=\big(\mu'_{-11}(\zeta^*_n;x,t)\big)_{1\times2N},\,\gamma^{(2)}=\big(\mu''_{-11}(\zeta^*_n;x,t)\big)_{1\times2N},
\end{split}
\end{align}
we can obtain a more concise reconstruction formulation of the triple poles solution (potential) for the TOFKN \eqref{T1} with ZBCs and as follows
\begin{align}\label{RFP-TP-3}
q(x,t)=-2\mathrm{i}\mathrm{e}^{\mathrm{i}\nu_{-}(x,t)}\bigg(\alpha\gamma-\frac{1}{2\pi\mathrm{i}}\int_{\Sigma}\big(M^{+}(\xi;x,t)J(\xi;x,t)\big)_{12}d\xi\bigg).
\end{align}

\subsubsection{Trace Formulae with ZBCs and triple poles}
The discrete spectral points $\zeta_n$'s are the triple zeros of $s_{11}(\lambda)$, while $\zeta^*_n$'s are the triple zeros of $s_{22}(\lambda)$. Define the functions $\beta^{\pm}(\lambda)$ as follows:
\begin{align}\label{TF-1}
\beta^{+}(\lambda)=s_{11}(\lambda)\prod^{2N}_{n=1}\bigg(\frac{\lambda-\zeta^*_n}{\lambda-\zeta_n}\bigg)^3\mathrm{e}^{\mathrm{i}\nu},\,\beta^{-}(\lambda)=s_{22}(\lambda)\prod^{2N}_{n=1}\bigg(\frac{\lambda-\zeta_n}{\lambda-\zeta^*_n}\bigg)^3\mathrm{e}^{-\mathrm{i}\nu}.
\end{align}

Then, $\beta^{+}(\lambda)$ and $\beta^{-}(\lambda)$ are analytic and have no zero in $D^{+}$ and $D^{-}$, respectively. Furthermore, we have the relation $\beta^{+}(\lambda)\beta^{-}(\lambda)=s_{11}(\lambda)s_{22}(\lambda)$ and the asymptotic behaviors $\beta^{\pm}(\lambda)\rightarrow1,\text{ as }\lambda\rightarrow\infty$.

By means of employing the Cauchy projectors and Plemelj' formula, we have
\begin{align}\label{TF-2}
\mathrm{log}\beta^{\pm}(\lambda)=\mp\frac{1}{2\pi \mathrm{i}}\int_{\Sigma}\frac{\mathrm{log}[1-\rho(\lambda)\tilde{\rho}(\lambda)]}{\xi-\lambda}d\xi,\quad\lambda\in D^{\pm}.
\end{align}

After substituting Eq. \eqref{TF-2} into  Eq. \eqref{TF-1}, we can obtain the trace formulae
\begin{align}\label{TF-3}
\begin{split}
&s_{11}(\lambda)=\mathrm{exp}\bigg(-\frac{1}{2\pi\mathrm{i}}\int_{\Sigma}\frac{\mathrm{log}[1-\rho(\lambda)\tilde{\rho}(\lambda)]}{\xi-\lambda}d\xi\bigg)\prod^{2N}_{n=1}\bigg(\frac{\lambda-\zeta_n}{\lambda-\zeta^*_n}\bigg)^3\mathrm{e}^{-\mathrm{i}\nu},\\
&s_{22}(\lambda)=\mathrm{exp}\bigg(\frac{1}{2\pi\mathrm{i}}\int_{\Sigma}\frac{\mathrm{log}[1-\rho(\lambda)\tilde{\rho}(\lambda)]}{\xi-\lambda}d\xi\bigg)\prod^{2N}_{n=1}\bigg(\frac{\lambda-\zeta^*_n}{\lambda-\zeta_n}\bigg)^3\mathrm{e}^{\mathrm{i}\nu}.
\end{split}
\end{align}

\subsubsection{Reflectionless Potential with ZBCs: Triple-Pole Solitons}
Now, we consider the case of reflectionless potential $q(x,t)$ with the reflection coefficients $\rho(\lambda)=\tilde{\rho}(\lambda)=0$. Then the Eqs. \eqref{MRHP-8} and \eqref{RFP-TP-3} with $J(\lambda;x,t)=0_{2\times2}$ become
\begin{small}
\begin{align}\label{RP-TPS-1}
\begin{split}
&\mu_{-11}(\zeta^*_n;x,t)=\mathrm{e}^{\mathrm{i}\nu_{-}(x,t)}+\sum_{s=1}^{2N}C_s(\zeta^*_n)\bigg[\frac12\mu''_{-21}(\zeta_s;x,t)+\bigg(D_s+\frac{1}{\zeta^*_n-\zeta_s}\bigg)\mu'_{-21}(\zeta_s;x,t)+\\
&\qquad\qquad\qquad\bigg(\frac{1}{(\zeta^*_n-\zeta_s)^2}+\frac{D_s}{\zeta^*_n-\zeta_s}+F_s\bigg)\mu_{-21}(\zeta_s;x,t)\bigg],\\
&\mu_{-21}(\zeta_s;x,t)=\sum_{j=1}^{2N}\widehat{C}_j(\zeta_s)\bigg[\frac12\mu''_{-11}(\zeta^*_j;x,t)+\bigg(\widehat{D}_j+\frac{1}{\zeta_s-\zeta^*_j}\bigg)\mu'_{-11}(\zeta^*_j;x,t)+\bigg(\frac{1}{(\zeta_s-\zeta^*_j)^2}+\\
&\qquad\qquad\qquad\frac{\widehat{D}_j}{\zeta_s-\zeta^*_j}+\widehat{F}_j\bigg)\mu_{-11}(\zeta^*_j;x,t)\bigg],\\
&\mu'_{-11}(\zeta^*_n;x,t)=-\sum_{s=1}^{2N}\frac{C_s(\zeta^*_n)}{\zeta^*_n-\zeta_s}\bigg[\frac12\mu''_{-21}(\zeta_s;x,t)+\bigg(D_s+\frac{2}{\zeta^*_n-\zeta_s}\bigg)\mu'_{-21}(\zeta_s;x,t)+\bigg(\frac{3}{(\zeta^*_n-\zeta_s)^2}+\\
&\qquad\qquad\qquad\frac{2D_s}{\zeta^*_n-\zeta_s}+F_s\bigg)\mu_{-21}(\zeta_s;x,t)\bigg],\\
&\mu'_{-21}(\zeta_s;x,t)=-\sum_{j=1}^{2N}\frac{\widehat{C}_j(\zeta_s)}{\zeta_s-\zeta^*_j}\bigg[\frac12\mu''_{-11}(\zeta^*_j;x,t)+\bigg(\widehat{D}_j+\frac{2}{\zeta_s-\zeta^*_j}\bigg)\mu'_{-11}(\zeta^*_j;x,t)+\bigg(\frac{3}{(\zeta_s-\zeta^*_j)^2}+\\
&\qquad\qquad\qquad\frac{2\widehat{D}_j}{\zeta_s-\zeta^*_j}+\widehat{F}_j\bigg)\mu_{-11}(\zeta^*_j;x,t)\bigg],\\
&\mu''_{-11}(\zeta^*_n;x,t)=\sum_{s=1}^{2N}\frac{2C_s(\zeta^*_n)}{(\zeta^*_n-\zeta_s)^2}\bigg[\frac12\mu''_{-21}(\zeta_s;x,t)+\bigg(D_s+\frac{3}{\zeta^*_n-\zeta_s}\bigg)\mu'_{-21}(\zeta_s;x,t)+\bigg(\frac{6}{(\zeta^*_n-\zeta_s)^2}+\\
&\qquad\qquad\qquad\frac{3D_s}{\zeta^*_n-\zeta_s}+F_s\bigg)\mu_{-21}(\zeta_s;x,t)\bigg],\\
&\mu''_{-21}(\zeta_s;x,t)=\sum_{j=1}^{2N}\frac{2\widehat{C}_j(\zeta_s)}{(\zeta_s-\zeta^*_j)^2}\bigg[\frac12\mu''_{-11}(\zeta^*_j;x,t)+\bigg(\widehat{D}_j+\frac{3}{\zeta_s-\zeta^*_j}\bigg)\mu'_{-11}(\zeta^*_j;x,t)+\bigg(\frac{6}{(\zeta_s-\zeta^*_j)^2}+\\
&\qquad\qquad\qquad\frac{3\widehat{D}_j}{\zeta_s-\zeta^*_j}+\widehat{F}_j\bigg)\mu_{-11}(\zeta^*_j;x,t)\bigg],
\end{split}
\end{align}
\end{small}
and
\begin{align}\label{RP-TPS-2}
q(x,t)=-2\mathrm{i}\mathrm{e}^{\mathrm{i}\nu_{-}(x,t)}\alpha\gamma.
\end{align}

\begin{thm}\label{RP-TPS-thm1}
The general expression of N-triple poles soliton of the TOFKN \eqref{T1} with ZBCs is given by determinant formula
\begin{align}\label{RP-TPS-3}
q(x,t)=2\mathrm{i}\mathrm{e}^{2\mathrm{i}\nu_{-}(x,t)}\frac{\mathrm{det}(R)}{\mathrm{det}(I-G)},
\end{align}
where
\begin{align}\label{RP-TPS-4}
\begin{split}
&R=\bigg(\begin{array}{cc} 0 & \alpha \\ \tau & I-G \end{array}\bigg),\,\tau=\big(\tau^{(1)},\tau^{(2)},\tau^{(3)}\big)^{\mathrm{T}},\\
&\tau^{(1)}=(1)_{1\times2N},\,\tau^{(2)}=(0)_{1\times2N},\tau^{(3)}=(0)_{1\times2N},
\end{split}
\end{align}
the $6N\times6N$ partitioned matrix $G=\Bigg(\begin{array}{ccc} G^{(1,1)} & G^{(1,2)} & G^{(1,3)} \\ G^{(2,1)} & G^{(2,2)} & G^{(2,3)} \\ G^{(3,1)} & G^{(3,2)} & G^{(3,3)} \end{array}\Bigg)$ with the $G^{(i,j)}=\Big(g^{(i,j)}_{n,j}\Big)_{2N\times2N}$ $(i,j=1,2,3)$ is given by
\begin{footnotesize}
\begin{align}\label{RP-TPS-5}
\begin{split}
&g^{(1,1)}_{n,j}=\sum^{2N}_{s=1}C_s(\zeta^*_n)\widehat{C}_j(\zeta_s)\bigg[\frac{1}{(\zeta_s-\zeta^*_j)^2}\bigg(\frac{6}{(\zeta_s-\zeta^*_j)^2}+\frac{3\widehat{D}_j}{\zeta_s-\zeta^*_j}+\widehat{F}_j\bigg)-\frac{1}{\zeta_s-\zeta^*_j}\bigg(D_s+\frac{1}{\zeta^*_n-\zeta_s}\bigg)\\
&\bigg(\frac{3}{(\zeta_s-\zeta^*_j)^2}+\frac{2\widehat{D}_j}{\zeta_s-\zeta^*_j}+\widehat{F}_j\bigg)+\bigg(\frac{1}{(\zeta^*_n-\zeta_s)^2}+\frac{D_s}{\zeta^*_n-\zeta_s}+F_s\bigg)\bigg(\frac{1}{(\zeta_s-\zeta^*_j)^2}+\frac{\widehat{D}_j}{\zeta_s-\zeta^*_j}+\widehat{F}_j\bigg)\bigg],\\
&g^{(1,2)}_{n,j}=\sum^{2N}_{s=1}C_s(\zeta^*_n)\widehat{C}_j(\zeta_s)\bigg[\frac{1}{(\zeta_s-\zeta^*_j)^2}\bigg(\widehat{D}_j+\frac{3}{\zeta_s-\zeta^*_j}\bigg)-\frac{1}{\zeta_s-\zeta^*_j}\bigg(D_s+\frac{1}{\zeta^*_n-\zeta_s}\bigg)\bigg(\widehat{D}_j+\frac{2}{\zeta_s-\zeta^*_j}\bigg)\\
&+\bigg(\frac{1}{(\zeta^*_n-\zeta_s)^2}+\frac{D_s}{\zeta^*_n-\zeta_s}+F_s\bigg)\bigg(\widehat{D}_j+\frac{1}{\zeta_s-\zeta^*_j}\bigg)\bigg],\\
&g^{(1,3)}_{n,j}=\sum^{2N}_{s=1}\frac{C_s(\zeta^*_n)\widehat{C}_j(\zeta_s)}{2}\bigg[\frac{1}{(\zeta_s-\zeta^*_j)^2}-\frac{1}{\zeta_s-\zeta^*_j}\bigg(D_s+\frac{1}{\zeta^*_n-\zeta_s}\bigg)+\bigg(\frac{1}{(\zeta^*_n-\zeta_s)^2}+\frac{D_s}{\zeta^*_n-\zeta_s}+F_s\bigg)\bigg],\\
&g^{(2,1)}_{n,j}=-\sum^{2N}_{s=1}\frac{C_s(\zeta^*_n)\widehat{C}_j(\zeta_s)}{\zeta^*_n-\zeta_s}\bigg[\frac{1}{(\zeta_s-\zeta^*_j)^2}\bigg(\frac{6}{(\zeta_s-\zeta^*_j)^2}+\frac{3\widehat{D}_j}{\zeta_s-\zeta^*_j}+\widehat{F}_j\bigg)-\frac{1}{\zeta_s-\zeta^*_j}\bigg(D_s+\frac{2}{\zeta^*_n-\zeta_s}\bigg)\\
&\bigg(\frac{3}{(\zeta_s-\zeta^*_j)^2}+\frac{2\widehat{D}_j}{\zeta_s-\zeta^*_j}+\widehat{F}_j\bigg)+\bigg(\frac{3}{(\zeta^*_n-\zeta_s)^2}+\frac{2D_s}{\zeta^*_n-\zeta_s}+F_s\bigg)\bigg(\frac{1}{(\zeta_s-\zeta^*_j)^2}+\frac{\widehat{D}_j}{\zeta_s-\zeta^*_j}+\widehat{F}_j\bigg)\bigg],\\
&g^{(2,2)}_{n,j}=-\sum^{2N}_{s=1}\frac{C_s(\zeta^*_n)\widehat{C}_j(\zeta_s)}{\zeta^*_n-\zeta_s}\bigg[\frac{1}{(\zeta_s-\zeta^*_j)^2}\bigg(\widehat{D}_j+\frac{3}{\zeta_s-\zeta^*_j}\bigg)-\frac{1}{\zeta_s-\zeta^*_j}\bigg(D_s+\frac{2}{\zeta^*_n-\zeta_s}\bigg)\bigg(\widehat{D}_j+\frac{2}{\zeta_s-\zeta^*_j}\bigg)\\
&+\bigg(\frac{3}{(\zeta^*_n-\zeta_s)^2}+\frac{2D_s}{\zeta^*_n-\zeta_s}+F_s\bigg)\bigg(\widehat{D}_j+\frac{1}{\zeta_s-\zeta^*_j}\bigg)\bigg],\\
&g^{(2,3)}_{n,j}=-\sum^{2N}_{s=1}\frac{C_s(\zeta^*_n)\widehat{C}_j(\zeta_s)}{2(\zeta^*_n-\zeta_s)}\bigg[\frac{1}{(\zeta_s-\zeta^*_j)^2}-\frac{1}{\zeta_s-\zeta^*_j}\bigg(D_s+\frac{2}{\zeta^*_n-\zeta_s}\bigg)+\bigg(\frac{3}{(\zeta^*_n-\zeta_s)^2}+\frac{2D_s}{\zeta^*_n-\zeta_s}+F_s\bigg)\bigg],\\
&g^{(3,1)}_{n,j}=\sum^{2N}_{s=1}\frac{2C_s(\zeta^*_n)\widehat{C}_j(\zeta_s)}{(\zeta^*_n-\zeta_s)^2}\bigg[\frac{1}{(\zeta_s-\zeta^*_j)^2}\bigg(\frac{6}{(\zeta_s-\zeta^*_j)^2}+\frac{3\widehat{D}_j}{\zeta_s-\zeta^*_j}+\widehat{F}_j\bigg)-\frac{1}{\zeta_s-\zeta^*_j}\bigg(D_s+\frac{3}{\zeta^*_n-\zeta_s}\bigg)\\
&\bigg(\frac{3}{(\zeta_s-\zeta^*_j)^2}+\frac{2\widehat{D}_j}{\zeta_s-\zeta^*_j}+\widehat{F}_j\bigg)+\bigg(\frac{6}{(\zeta^*_n-\zeta_s)^2}+\frac{3D_s}{\zeta^*_n-\zeta_s}+F_s\bigg)\bigg(\frac{1}{(\zeta_s-\zeta^*_j)^2}+\frac{\widehat{D}_j}{\zeta_s-\zeta^*_j}+\widehat{F}_j\bigg)\bigg],\\
&g^{(3,2)}_{n,j}=\sum^{2N}_{s=1}\frac{2C_s(\zeta^*_n)\widehat{C}_j(\zeta_s)}{(\zeta^*_n-\zeta_s)^2}\bigg[\frac{1}{(\zeta_s-\zeta^*_j)^2}\bigg(\widehat{D}_j+\frac{3}{\zeta_s-\zeta^*_j}\bigg)-\frac{1}{\zeta_s-\zeta^*_j}\bigg(D_s+\frac{3}{\zeta^*_n-\zeta_s}\bigg)\bigg(\widehat{D}_j+\frac{2}{\zeta_s-\zeta^*_j}\bigg)\\
&+\bigg(\frac{6}{(\zeta^*_n-\zeta_s)^2}+\frac{3D_s}{\zeta^*_n-\zeta_s}+F_s\bigg)\bigg(\widehat{D}_j+\frac{1}{\zeta_s-\zeta^*_j}\bigg)\bigg],\\
&g^{(3,3)}_{n,j}=\sum^{2N}_{s=1}\frac{C_s(\zeta^*_n)\widehat{C}_j(\zeta_s)}{(\zeta^*_n-\zeta_s)^2}\bigg[\frac{1}{(\zeta_s-\zeta^*_j)^2}-\frac{1}{\zeta_s-\zeta^*_j}\bigg(D_s+\frac{3}{\zeta^*_n-\zeta_s}\bigg)+\bigg(\frac{6}{(\zeta^*_n-\zeta_s)^2}+\frac{3D_s}{\zeta^*_n-\zeta_s}+F_s\bigg)\bigg].
\end{split}
\end{align}
\end{footnotesize}

\end{thm}
\begin{proof}
We rewrite the linear system \eqref{RP-TPS-1} in the matrix form
\begin{align}\nonumber
\gamma-G\gamma=\tau,
\end{align}
where
\begin{align}\nonumber
\begin{split}
&\gamma=(\gamma^{(1)},\gamma^{(2)},\gamma^{(3)})^{\mathrm{T}},\,\gamma^{(1)}=(\mu_{-11}(\zeta^*_n;x,t))_{1\times2N},\\
&\gamma^{(2)}=(\mu'_{-11}(\zeta^*_n;x,t))_{1\times2N},\,\gamma^{(1)}=(\mu''_{-11}(\zeta^*_n;x,t))_{1\times2N},\\
&\tau=(\tau^{(1)},\tau^{(2)},\tau^{(3)})^{\mathrm{T}},\,\tau^{(1)}=(\mathrm{e}^{\mathrm{i}\nu_-})_{1\times2N},\,\tau^{(2)}=(0)_{1\times2N},\,\tau^{(3)}=(0)_{1\times2N}.
\end{split}
\end{align}

The $6N\times6N$ matrix $G$ is shown in Eq. \eqref{RP-TPS-5}. Combining Eqs. \eqref{RFP-TP-2} and \eqref{RP-TPS-2} with the case of reflectionless potential, the triple poles soliton solution \eqref{RP-TPS-3} can be given out.

\end{proof}

For example, we obtain the $N$-triple-pole solutions of the TOFKN \eqref{T1} with ZBCs via Theorem \ref{RP-TPS-thm1}.

$\bullet$ When taking parameters $N=1,\zeta_1=\frac12+\frac12\mathrm{i},A[\zeta_1]=B[\zeta_1]=C[\zeta_1]=1$, we can obtain the 1-triple-pole solution and give out relevant plots in Fig. \ref{ZB-TP-F1}. Figs. \ref{ZB-TP-F1} (a) and (b) exhibit the three-dimensional and density diagrams for the 1-triple-pole soliton solution of the TOFKN with ZBCs, which is equivalent to the elastic collisions of three bright-bright-bright solitons. Fig. \ref{ZB-TP-F1} (c) displays the distinct profiles of the 1-triple-pole soliton solution at $t=\pm6,0$. Combined with Figs. \ref{F2} and \ref{F3}, as $N = 2$, we can obtain a 2-triple-pole solution with six single soliton branches, due to the form of solution is too complex, here we omit the relevant plots.

\begin{figure}[htbp]
\centering
\subfigure[]{
\begin{minipage}[t]{0.33\textwidth}
\centering
\includegraphics[height=4.5cm,width=4.5cm]{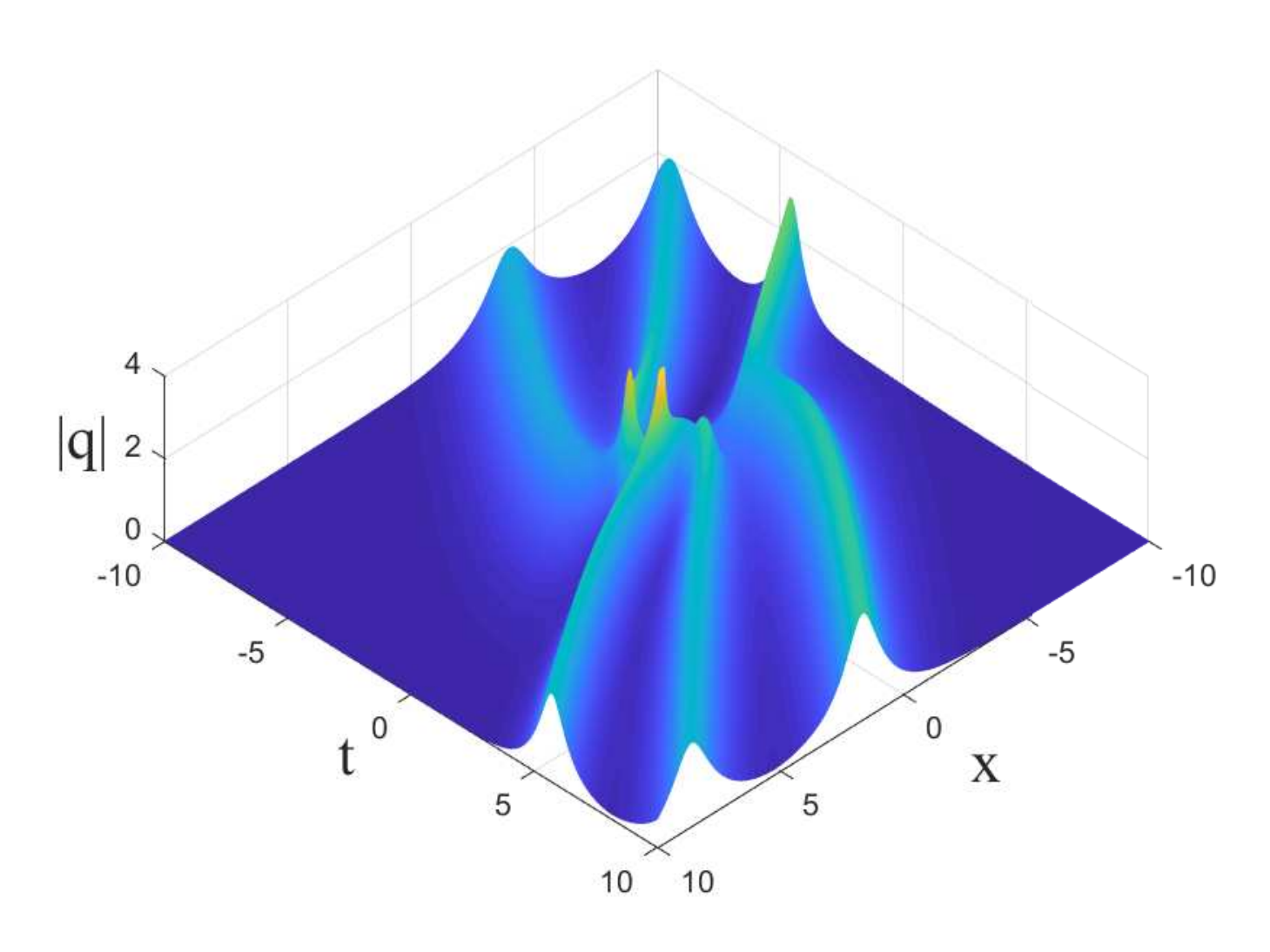}
%\caption{fig1}
\end{minipage}
}%
\subfigure[]{
\begin{minipage}[t]{0.33\textwidth}
\centering
\includegraphics[height=4.5cm,width=4.5cm]{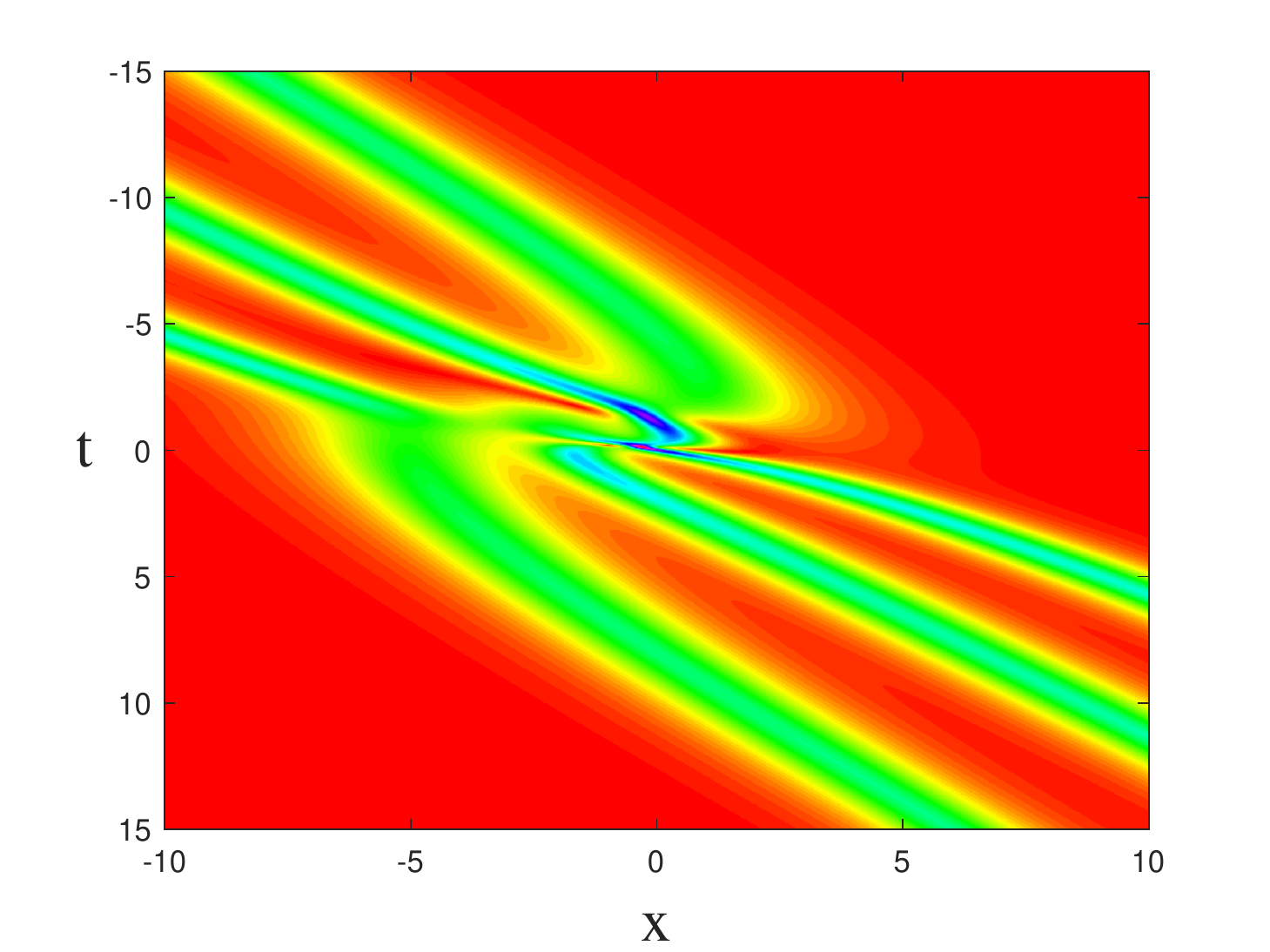}
%\caption{fig2}
\end{minipage}%
}%
\subfigure[]{
\begin{minipage}[t]{0.33\textwidth}
\centering
\includegraphics[height=4.5cm,width=4.5cm]{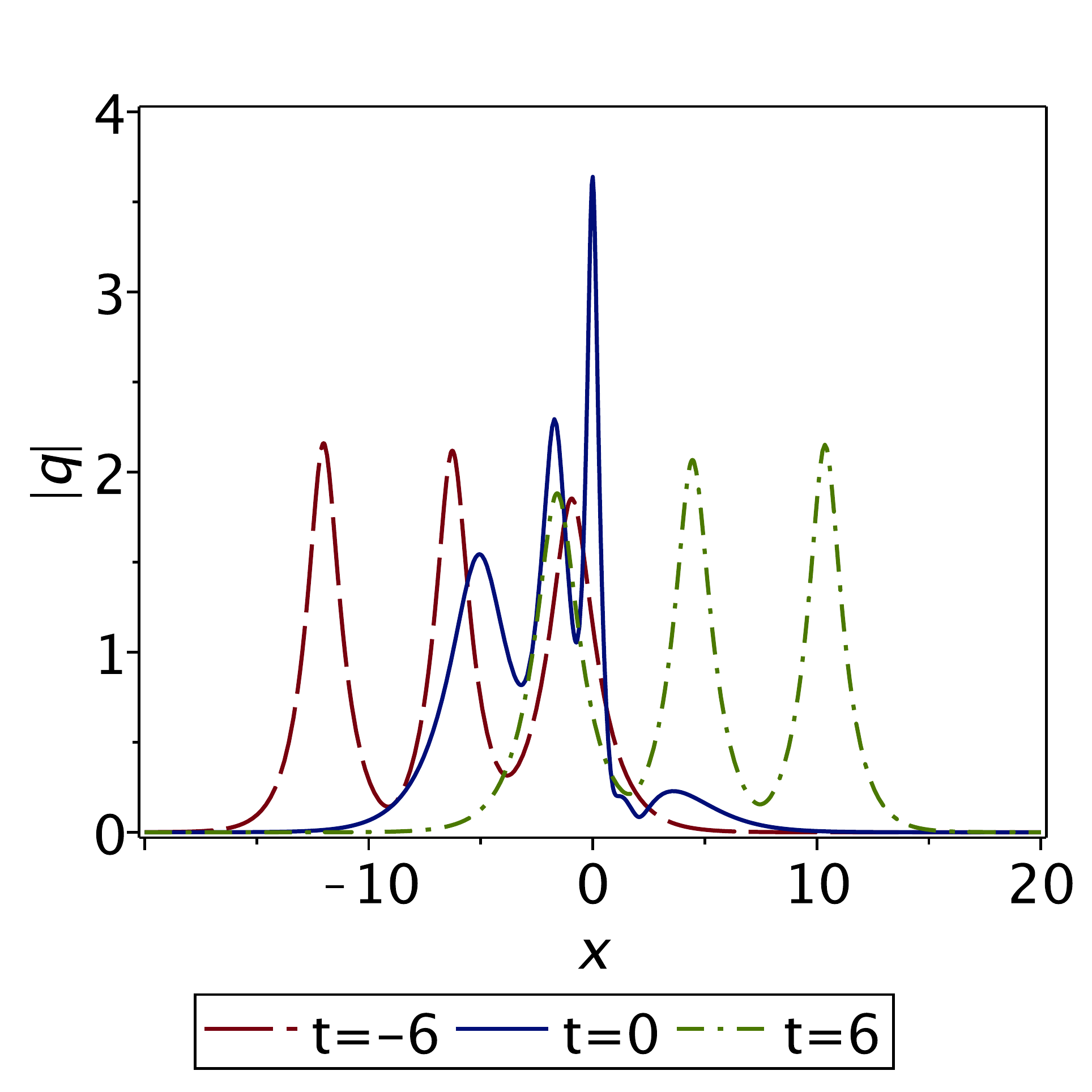}
%\caption{fig1}
\end{minipage}
}%
\centering
\caption{(Color online) The 1-triple-pole soliton solution of TOFKN \eqref{T1} with ZBCs and $N=1,\zeta_1=\frac12+\frac12\mathrm{i},A[\zeta_1]=B[\zeta_1]=C[\zeta_1]=1$. (a) The three-dimensional plot; (b) The density plot; (c) The sectional drawings at $t=-6$ (dashed line), $t=0$ (solid line), and $t =6$ (dash-dot line).}
\label{ZB-TP-F1}
\end{figure}

Moreover, let $N=1,\zeta_1=\frac12+\mathrm{i},A[\zeta_1]=B[\zeta_1]=C[\zeta_1]=1$, with the aid of Maple symbol calculations, one can derive the asymptotic state of the 1-triple poles soliton solution, as following
\begin{tiny}
\begin{align}\label{RP-TPS-6}
\begin{split}
|q(x,t)|^2\longrightarrow\left\{
\begin{aligned}
&\frac{240870400000000000000\mathrm{e}^{4\theta_2}}{3814697265625\mathrm{e}^{8\theta_2}-22581600000000000000\mathrm{e}^{4\theta_2}+92829679353856000000000000}, \text{as }\theta_2=x+\frac{11}{4}t-\mathrm{log}(t)\\
&\frac{16552487849191014400\mathrm{e}^{4\theta_2}}{24334743464537227264000000\mathrm{e}^{8\theta_2}-1551795735861657600\mathrm{e}^{4\theta_2}+68719476736}, \text{as }\theta_2=x+\frac{11}{4}t\\
&\frac{608368586613430681600000000\mathrm{e}^{4\theta_2}}{92829679353856000000000000\mathrm{e}^{8\theta_2}-57034554995009126400000000\mathrm{e}^{4\theta_2}+24334743464537227264000000},\\
&\text{as }\theta_2=x+\frac{11}{4}t+\mathrm{log}(t).
\end{aligned}
\right.
\end{split}
\end{align}
\end{tiny}

From the above expression, one can obtain that the 1-triple poles soliton solution degrades into the three one-soliton solution as $t\rightarrow\infty$. Of which, the center trajectories are $x+\frac{11}{4}t-\mathrm{log}(t)$, $x+\frac{11}{4}t$ and $x+\frac{11}{4}t+\mathrm{log}(t)$, respectively. When $t\rightarrow\infty$, the position shift of the
three standard one soliton solution is $\mathrm{log}(t)$, which depends on $t$.

\section{The Construction and Solve of RH problem with NZBCs}\label{Sec3}
In this section, we will find the RH problem for the TOFKN \eqref{T1} with the NZBCs
\begin{align}\label{T59}
q(x,t)\sim q_{\pm},\text{ as }x\rightarrow\pm\infty,
\end{align}
where $|q_{\pm}|=q_0>0$, and $q_{\pm}$ are independent of $x,t$.

\subsection{Direct Scattering with NZBCs}
In this subsection, we will investigate direct scattering with NZBCs for the modified Zakharov-Shabat eigenvalue problem of TOFKN \eqref{T1} in detail.

\subsubsection{Jost Solutions, Analyticity, and Continuity}
As $x\rightarrow\pm\infty$, we consider the asymptotic scattering problem of the modified Zakharov-Shabat eigenvalue problem \eqref{T2}-\eqref{T3}:
\begin{align}\label{T60}
\begin{split}
&\Psi_x=X_{\pm}\Psi,\quad\Psi_t=T_{\pm}\Psi,\quad X_{\pm}=\mathrm{i}\lambda^2\sigma_3+\lambda Q_{\pm},\\
&T_{\pm}=\bigg(4\lambda^4+\frac32q^4_0-2\lambda^2q^2_0\bigg)X_{\pm},\quad Q_{\pm}=\Bigg(\begin{array}{cc} 0 & q_{\pm} \\-q^*_{\pm} & 0 \end{array}\Bigg).
\end{split}
\end{align}

\begin{prop}
the fundamental matrix solution of Eq. \eqref{T60} is presented as follows:
\begin{align}\label{T61}
\Psi^{\mathrm{bg}}_{\pm}(\lambda;x,t)=\Bigg\{
\begin{aligned}
&Y_{\pm}(\lambda)\mathrm{e}^{\mathrm{i}\theta(\lambda;x,t)\sigma_3},\quad\quad\quad\quad\lambda\neq\pm \mathrm{i}q_0,\text{ and }\lambda+k(\lambda)\neq0,\\
&I+\bigg(x+\frac{15}{2}q^4_0t\bigg)X_{\pm}(\lambda),\quad\lambda=\pm \mathrm{i}q_0,
\end{aligned}
\end{align}
where
\begin{align}\label{T62}
\begin{split}
&Y_{\pm}(\lambda)=\Bigg(\begin{array}{cc} 1 & \frac{\mathrm{i}q_{\pm}}{z} \\ \frac{\mathrm{i}q^*_{\pm}}{z} & 1 \end{array}\Bigg)=I+\frac{\mathrm{i}}{z}\sigma_3Q_{\pm},\quad z=\lambda+k(\lambda),\\
&\theta(\lambda;x,t)=\lambda k(\lambda)\bigg[x+\bigg(4\lambda^4+\frac32q^4_0-2\lambda^2q^2_0\bigg)t\bigg],\quad k^2(\lambda)=\lambda^2+q^2_0.
\end{split}
\end{align}

\end{prop}

\begin{proof}
For the case of $\lambda=\pm \mathrm{i}q_0$, we can directly calculate and obtain $\Psi^{\mathrm{bg}}_{\pm}(\lambda;x,t)=I+\Big(x+\frac{15}{2}q^4_0t\Big)X_{\pm}(\lambda)$. However, when $\lambda\neq\pm \mathrm{i}q_0$, based on eigenvalues $\pm \mathrm{i}\lambda k(\lambda)$, one can derive $k^2(\lambda)=\lambda^2+q^2_0$, and the eigenvector matrix of $X_{\pm}$ and $T_{\pm}$ can be written
\begin{align}\label{T63}
\begin{split}
&Y_{\pm}(\lambda)=\Bigg(\begin{array}{cc} 1 & \frac{\mathrm{i}q_{\pm}}{\lambda+k(\lambda)} \\ \frac{\mathrm{i}q^*_{\pm}}{\lambda+k(\lambda)} & 1 \end{array}\Bigg)=\Bigg(\begin{array}{cc} 1 & \frac{\mathrm{i}q_{\pm}}{z} \\ \frac{\mathrm{i}q^*_{\pm}}{z} & 1 \end{array}\Bigg)=I+\frac{\mathrm{i}}{z}\sigma_3Q_{\pm},\quad z=\lambda+k(\lambda)\neq0.
\end{split}
\end{align}
Thus, $X_{\pm}$ and $T_{\pm}$ can are diagonalized by the eigenvector matrix \eqref{T63} as
\begin{align}\label{T64}
X_{\pm}=Y_{\pm}[\mathrm{i}\lambda k(\lambda)\sigma_3]Y^{-1}_{\pm},\quad T_{\pm}=Y_{\pm}\bigg[\bigg(4\lambda^2+\frac32q^4_0-2\lambda^2q^2_0\bigg)\mathrm{i}\lambda k(\lambda)\sigma_3\bigg]Y^{-1}_{\pm},
\end{align}
and direct calculation shows that
\begin{align}\label{T65}
\mathrm{det}Y_{\pm}=1+\frac{q^2_0}{z^2}\overset{\triangle}{=}\eta,
\end{align}
and
\begin{align}\label{T66}
Y^{-1}_{\pm}=\frac{1}{\eta}\Bigg(\begin{array}{cc} 1 & -\frac{\mathrm{i}q_{\pm}}{z} \\ -\frac{\mathrm{i}q^*_{\pm}}{z} & 1 \end{array}\Bigg)=\frac{1}{\eta}\bigg(I-\frac{\mathrm{i}}{z}\sigma_3Q_{\pm}\bigg).
\end{align}
Then substituting \eqref{T64} into \eqref{T60}, we have
\begin{align}\label{T67}
\big(Y^{-1}_{\pm}\Psi\big)_x=\mathrm{i}\lambda k(\lambda)\sigma_3\big(Y^{-1}_{\pm}\Psi\big),\quad \big(Y^{-1}_{\pm}\Psi\big)_t=\bigg(4\lambda^2+\frac32q^4_0-2\lambda^2q^2_0\bigg)\mathrm{i}\lambda k(\lambda)\sigma_3\big(Y^{-1}_{\pm}\Psi\big),
\end{align}
from which we can derive the solution of the asymptotic spectral problem \eqref{T60}
\begin{align}\label{T68}
\Psi=\Psi^{\mathrm{bg}}_{\pm}(\lambda;x,t)=Y_{\pm}(\lambda)\mathrm{e}^{\mathrm{i}\theta(\lambda;x,t)\sigma_3},
\end{align}
where $\theta(\lambda;x,t)$ is given by \eqref{T62}. This completes the proof.

\end{proof}

Since $k(\lambda)$ stands for a two-sheeted Riemann surface. Therefore, in order to avoid multi-valued case of eigenvalue $k(\lambda)$, we define a new uniformization variable: $z=\lambda+k(\lambda)$, which was first introduced by Faddeev and Takhtajan in 1987 \cite{Faddeev(1987)}. Next, we will illustrate the scattering problem on a standard $z$-plane instead of the two-sheeted Riemann surface by utilizing two single-valued inverse mappings:
\begin{align}\label{T69}
k(\lambda)=\frac12\bigg(z+\frac{q^2_0}{z}\bigg),\quad\lambda=\frac12\bigg(z-\frac{q^2_0}{z}\bigg),
\end{align}
where the inverse mapping  evidently satisfy the relation $k^2(\lambda)=\lambda^2+q^2_0$. It means that the Jost solutions $\psi_{\pm}(z;x,t)$ of the Lax pairs \eqref{T2}-\eqref{T3} possess the following asymptotics:
\begin{align}\label{T70}
\psi_{\pm}(z;x,t)\sim Y_{\pm}(z)\mathrm{e}^{\mathrm{i}\theta(z;x,t)\sigma_3},\text{ as }x\rightarrow\pm\infty,
\end{align}
and the modified Jost solutions $\mu_{\pm}(z;x,t)$ are derived by making the transform
\begin{align}\label{T71}
\mu_{\pm}(z;x,t)=\psi_{\pm}(z;x,t)\mathrm{e}^{-\mathrm{i}\theta(z;x,t)\sigma_3},
\end{align}
such that
\begin{align}\label{T72}
\mu_{\pm}(z;x,t)\sim Y_{\pm}(z),\text{ as }x\rightarrow\pm\infty.
\end{align}

On the other hand, the modified Jost solutions $\mu_{\pm}(z;x,t)$ satisfy the following equivalent Lax pair:
\begin{align}\label{T73}
\Big[Y^{-1}_{\pm}(z)\mu_{\pm}(z;x,t)\Big]_x+\mathrm{i}\lambda k(\lambda)\Big[Y^{-1}_{\pm}(z)\mu_{\pm}(z;x,t),\sigma_3\Big]=Y^{-1}_{\pm}(z)\Delta X_{\pm}\mu_{\pm}(z;x,t),
\end{align}
\begin{align}\label{T74}
\begin{split}
&\Big[Y^{-1}_{\pm}(z)\mu_{\pm}(z;x,t)\Big]_t+\mathrm{i}\lambda k(\lambda)\bigg(4\lambda^4+\frac32q^4_0-2\lambda^2q^2_0\bigg)\Big[Y^{-1}_{\pm}(z)\mu_{\pm}(z;x,t),\sigma_3\Big]\\
&=Y^{-1}_{\pm}(z)\Delta T_{\pm}\mu_{\pm}(z;x,t),
\end{split}
\end{align}
where $\Delta X_{\pm}=X\big(\lambda(z);x,t\big)-X_{\pm}\big(\lambda(z);x,t\big)=\lambda(z)[Q(x,t)-Q_{\pm}(x,t)]:=\lambda(z)\Delta Q_{\pm}(x,t)$ and $\Delta T_{\pm}=T\big(\lambda(z);x,t\big)-T_{\pm}\big(\lambda(z);x,t\big)$. Equations \eqref{T73} and \eqref{T74} can be written in full derivative form
\begin{align}\label{T75}
d\big(\mathrm{e}^{-\mathrm{i}\theta(z;x,t)\widehat{\sigma}_3}Y^{-1}_{\pm}(z)\mu_{\pm}(z;x,t)\big)=\mathrm{e}^{-\mathrm{i}\theta(z;x,t)\widehat{\sigma}_3}\big[Y^{-1}_{\pm}(z)\big(\Delta X_{\pm}dx+\Delta T_{\pm}dt\big)\mu_{\pm}(z;x,t)\big],
\end{align}
which lead to Volterra integral equations
\begin{small}
\begin{align}\label{T76}
\mu_{\pm}(z;x,t)=Y_{\pm}(z)+\left\{
\begin{aligned}
&\lambda(z)\int^x_{\pm\infty}Y_{\pm}(z)\mathrm{e}^{\mathrm{i}k(z)\lambda(z)(x-y)\widehat{\sigma}_3}\big[Y^{-1}_{\pm}(z)\Delta Q_{\pm}(y,t)\mu_{\pm}(z;y,t)\big]dy,\, z\neq0,\pm \mathrm{i}q_0,\\
&\lambda(z)\int^x_{\pm\infty}\big[I+(x-y)X_{\pm}(z)\big]\Delta Q_{\pm}(y,t)\mu_{\pm}(z;y,t)\big]dy,\, z=\pm \mathrm{i}q_0.
\end{aligned}
\right.
\end{align}
\end{small}

Since the exponential function contains $k(z)\lambda(z)$, we define the $z$ plane by the positive and negative values of $\mathrm{Im}\big(k(z)\lambda(z)\big)$, that is
\begin{small}
\begin{align}\nonumber
\begin{split}
\mathrm{Im} \big[k(z)\lambda(z)\big]=\frac{1}{4|z|^4}\big[(2|z|^4-2q^4_0)\mathrm{Re}z\mathrm{Im}z+4q^4_0\mathrm{Re}z\mathrm{Im}z\big]=\frac{1}{4|z|^4}\big[(2|z|^4+2q^4_0)\mathrm{Re}z\mathrm{Im}z\big],
\end{split}
\end{align}
\end{small}
and
\begin{align}\nonumber
\begin{split}
&\mathrm{Im} \big[k(z)\lambda(z)\big]>0,\text{ that is }\mathrm{Re}z\mathrm{Im}z>0,\\
&\mathrm{Im} \big[k(z)\lambda(z)\big]<0,\text{ that is }\mathrm{Re}z\mathrm{Im}z<0,\\
&\mathrm{Im} \big[k(z)\lambda(z)\big]=0,\text{ that is }z^4=q^4_0\Rightarrow\bigg\{
\begin{aligned}
&z=\pm \mathrm{i}q_0,\\
&z=\pm q_0.
\end{aligned}
\end{split}
\end{align}

Therefore, from the above mapping relation between the $\lambda$-plane and $z$-plane, we define $\Sigma$ and $D^{\pm}$ on the $z$-plane as $\Sigma:=\mathbb{R}\cup \mathrm{i}\mathbb{R}\backslash\{0\},D^{\pm}:=\{z\in\mathbb{C}|\pm(\mathrm{Re}z)(\mathrm{Im}z)>0\}$, which are shown in Fig. \ref{F4}.

\begin{figure}[htbp]
\centerline{\begin{tikzpicture}[scale=1.8]
\path [fill=pink] (2.5,0) -- (0.5,0) to
(0.5,-2) -- (2.5,-2);
\path [fill=pink] (4.5,0) -- (2.5,0) to
(2.5,2) -- (4.5,2);
\draw[-][thick](0.5,0)--(0.75,0);
\draw[<-][thick](0.75,0)--(1,0);
\draw[-][thick](1,0)--(2,0);
\draw[<-][thick](2,0)--(2.5,0);
\draw[fill] (2.5,0) circle [radius=0.03];
\draw[->][thick](2.5,0)--(3,0);
\draw[->][thick](3,0)--(4,0);
\draw[-][thick](4,0)--(4.5,0)node[right]{$\mbox{Re}z$};
\draw[-][thick](2.5,2)node[above]{$\mbox{Im}z$}--(2.5,0);
\draw[-][thick](2.5,0)--(2.5,-2);
\draw[->][thick](2.5,-1.5)--(2.5,-0.5);
\draw[->][thick](2.5,-2)--(2.5,-1.5);
\draw[->][thick](2.5,1.5)--(2.5,0.5);
\draw[->][thick](2.5,2)--(2.5,1.5);
\draw[fill] (2.5,-0.3) node[right]{$0$};
\draw[fill] (2.5,1) circle [radius=0.03];
\draw[fill] (2.5,-1) circle [radius=0.03];
\draw[fill](3.8,1.5) circle [radius=0.03] node[right]{$z_{n}$};
\draw[fill] (3.8,-1.5) circle [radius=0.03] node[right]{$z^{*}_{n}$};
\draw[fill] (1.2,1.5) circle [radius=0.03] node[left]{$-z^{*}_{n}$};
\draw[fill] (1.2,-1.5) circle [radius=0.03] node[left]{$-z_{n}$};
\draw[fill] (2,0.5) circle [radius=0.03] node[right]{$-\frac{q^{2}_{0}}{z_{n}}$};
\draw[fill] (2,-0.5) circle [radius=0.03] node[right]{$-\frac{q^{2}_{0}}{z^{*}_{n}}$};
\draw[fill] (3,0.5) circle [radius=0.03] node[left]{$\frac{q^{2}_{0}}{z^{*}_{n}}$};
\draw[fill] (3,-0.5) circle [radius=0.03] node[left]{$\frac{q^{2}_{0}}{z_{n}}$};
\draw[-][thick](3.5,0) arc(0:360:1);
\draw[-][thick](3.5,0) arc(0:30:1);
\draw[-][thick](3.5,0) arc(0:150:1);
\draw[-][thick](3.5,0) arc(0:210:1);
\draw[-][thick](3.5,0) arc(0:330:1);
\draw[fill] (1.8,0.715) circle [radius=0.03];
\draw[fill] (1.7,0.715) circle [radius=0.0] node[above]{$-\omega^{*}_{m}$};
\draw[fill] (1.8,-0.715) circle [radius=0.03];
\draw[fill] (1.7,-0.715) circle [radius=0.0] node[below]{$-\omega_{m}$};
\draw[fill] (3.2,0.715) circle [radius=0.03];
\draw[fill] (3.3,0.715) circle [radius=0.0] node[above]{$\omega_{m}$};
\draw[fill] (3.2,-0.715) circle [radius=0.03];
\draw[fill] (3.3,-0.715) circle [radius=0.0] node[below]{$\omega^{*}_{m}$};
\end{tikzpicture}}
\caption{(Color online) Distribution of the discrete spectrum and jumping curves for the RH problem on complex $z$-plane, Region $D^{+}=\left\{k\in \mathbb{C}\big|\mathrm{Re}z\mathrm{Im}z> 0\right\}$ (pink region), region $D^{-}=\left\{z\in \mathbb{C}\big|\mathrm{Re}z\mathrm{Im}z< 0\right\}$ (white region).}
\label{F4}
\end{figure}
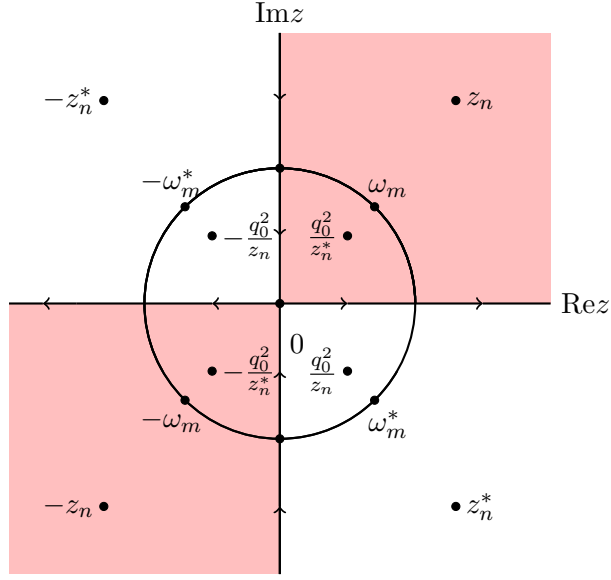

Then, similar to proposition \ref{P1} in Sec. \ref{Sec2} with ZBCs, as well as the Volterra integral equations \eqref{T76} for the case of $z\neq0,\pm \mathrm{i}q_0$, the following analyticity of the (modified) Jost solution can be derived as follows:
\begin{prop}\label{P3}
Suppose $(1+|x|)\big(q(x,t)-q_{\pm}\big)\in L^{1}(\mathbb{R}^{\pm})$. Then, Jost solution $\psi(z;x,t)$ $(\text{modified Jost solution } \mu(z;x,t))$ have the following properties:\\
$\bullet$ Eq. \eqref{T2} has the unique solution $\psi(z;x,t)$ $(\mu(z;x,t))$ satisfying Eq. \eqref{T70} $($Eq. \eqref{T72}$)$ on $\Sigma$.\\
$\bullet$ The column vectors $\psi_{+1}(z;x,t)$ $(\mu_{+1}(z;x,t))$ and $\psi_{-2}(z;x,t)$ $(\mu_{-2}(z;x,t))$ can be analytically extended to $D^+$ and continuously extended to $D^+\cup\Sigma$.\\
$\bullet$ The column vectors $\psi_{+2}(z;x,t)$ $(\mu_{+2}(z;x,t))$ and $\psi_{-1}(z;x,t)$ $(\mu_{-1}(z;x,t))$ can be analytically extended to $D^-$ and continuously extended to $D^-\cup\Sigma$.
\end{prop}

Similar to proposition \ref{P10} in Sec. \ref{Sec2} with ZBCs, one can also confirm that the Jost solutions $\psi(z;x,t)$ are the simultaneous solutions for both parts of the modified Zakharov-Shabat eigenvalue problem Eqs. \eqref{T2}-\eqref{T3} with NZBCs.

\subsubsection{Scattering Matrix, Analyticity, and Continuity}
Since $\psi_{\pm}(z;x,t)$ are two fundamental matrix solutions of the modified Zakharov-Shabat eigenvalue problem Eqs. \eqref{T2}-\eqref{T3} in $z\in\Sigma\backslash\{\pm \mathrm{i}q_0\}$, there exists a linear relation between $\psi_{+}(z;x,t)$ and $\psi_{-}(z;x,t)$, so one can define the constant scattering matrix $S(z)=\big(s_{ij}(z)\big)_{2\times2}$ such that
\begin{align}\label{T77}
\psi_{+}(z;x,t)=\psi_{-}(z;x,t)S(z),\quad z\in\Sigma\backslash\{\pm \mathrm{i}q_0\}.
\end{align}

According to the \eqref{T13} of the derivation process, and combining with the formula $$\mathrm{det}\big(\psi_{\pm}(z;x,t)\big)=\mathrm{det}(Y_{\pm})=\eta,$$ the scattering coefficients shown as:
\begin{align}\label{T78}
\begin{split}
&s_{11}(z)=\frac{\mathrm{det}\big(\psi_{+1}(z;x,t),\psi_{-2}(z;x,t)\big)}{\eta},\quad s _{12}(z)=\frac{\mathrm{det}\big(\psi_{+2}(z;x,t),\psi_{-2}(z;x,t)\big)}{\eta},\\
&s_{21}(z)=\frac{\mathrm{det}\big(\psi_{-1}(z;x,t),\psi_{+1}(z;x,t)\big)}{\eta},\quad s_{22}(z)=\frac{\mathrm{det}\big(\psi_{-1}(z;x,t),\psi_{+2}(z;x,t)\big)}{\eta}.
\end{split}
\end{align}

From these determinant representations, one can extend the analytical regions of $s_{11}(z)$ and $s_{22}(z)$.
\begin{prop}
Suppose that $q(x,t)-q_{\pm}\in L^1(\mathbb{R}^{\pm})$. Then, $s_{11}(z)$ can be analytically extended to $D^+$ and continuously extended to $D^+\cup\big(\Sigma\backslash\{\pm \mathrm{i}q_0\}\big)$, while $s_{22}(z)$ can be analytically extended to $D^-$ and continuously extended to $D^-\cup\big(\Sigma\backslash\{\pm \mathrm{i}q_0\}\big)$. Moreover, both $s_{12}(z)$ and $s_{21}(z)$ are continuous in $\Sigma\backslash{\pm \mathrm{i}q_0}$.
\end{prop}
\begin{proof}
The proposition can be verified by using proposition \eqref{P3} and Eq. \eqref{T78}.
\end{proof}

\begin{prop}
Suppose that $(1+|x|)(q(x,t)-q_{\pm})\in L^1(\mathbb{R}^{\pm})$. Then, $k(z)s_{11}(z)$ can be analytically extended to $D^+$ and continuously extended to $D^+\cup\Sigma$, while $k(z)s_{22}(z)$ can be analytically extended to $D^-$ and continuously extended to $D^-\cup\Sigma$. Moreover, both $k(z)s_{12}(z)$ and $k(z)s_{21}(z)$ are continuous in $\Sigma$.
\end{prop}
\begin{proof}
The proposition can be verified by using proposition \eqref{P3} and Eqs. \eqref{T65}, \eqref{T69} and \eqref{T78}.
\end{proof}

In order to study the RH problem in the inverse process, we focus on the potential without spectral singularity, and suppose $s_{ij}(z)(i,j=1,2)$ are continuous in the branch points $\{\pm \mathrm{i}q_0\}$. The reflection coefficients $\rho(\lambda)$ and $\tilde{\rho}(\lambda)$ are defined as follows:
\begin{align}\label{T79}
\rho(z)=\frac{s_{21}(z)}{s_{11}(z)},\quad \tilde{\rho}(z)=\frac{s_{12}(z)}{s_{22}(z)},\quad z\in\Sigma,
\end{align}
which will be utilized in the inverse scattering problem.

\subsubsection{Symmetry Structures}
Compared with the case of ZBCs, The symmetries of $X(z;x,t),T(z;x,t)$, Jost solutions, modified Jost solutions, scattering matrix and reflection coefficients of the case of NZBCs will be more complicated.
\begin{prop}\label{P5}
For the case of NZBCs, the Jost solutions, modified Jost solutions, scattering matrix and reflection coefficients admit following three kinds of reduction conditions on the $z$-plane:\\
$\bullet$ The first symmetry reduction\\
\begin{align}\label{T80}
\begin{split}
X(z;x,t)&=\sigma_2X(z^*;x,t)^*\sigma_2,\quad T(z;x,t)=\sigma_2T(z^*;x,t)^*\sigma_2,\\
\psi_{\pm}(z;x,t)&=\sigma_2\psi_{\pm}(z^*;x,t)^*\sigma_2,\quad \mu_{\pm}(z;x,t)=\sigma_2\mu_{\pm}(z^*;x,t)^*\sigma_2,\\
S(z)&=\sigma_2S(z^*)^*\sigma_2,\quad \rho(z)=-\tilde{\rho}(z^*)^*.
\end{split}
\end{align}
$\bullet$ The second symmetry reduction\\
\begin{align}\label{T81}
\begin{split}
X(z;x,t)&=\sigma_1X(-z^*;x,t)^*\sigma_1,\quad T(z;x,t)=\sigma_1T(-z^*;x,t)^*\sigma_1,\\
\psi_{\pm}(z;x,t)&=\sigma_1\psi_{\pm}(-z^*;x,t)^*\sigma_1,\quad \mu_{\pm}(z;x,t)=\sigma_1\mu_{\pm}(-z^*;x,t)^*\sigma_1,\\
S(z)&=\sigma_1S(-z^*)^*\sigma_1,\quad \rho(z)=\tilde{\rho}(-z^*)^*.
\end{split}
\end{align}
$\bullet$ The third symmetry reduction\\
\begin{align}\label{T82}
\begin{split}
&X(z;x,t)=X\bigg(-\frac{q^2_0}{z};x,t\bigg),\quad T(z;x,t)=T\bigg(-\frac{q^2_0}{z};x,t\bigg),\\
&\psi_{\pm}(z;x,t)=\frac{\mathrm{i}}{z}\psi_{\pm}\bigg(-\frac{q^2_0}{z};x,t\bigg)\sigma_3Q_{\pm},\quad \mu_{\pm}(z;x,t)=\frac{\mathrm{i}}{z}\mu_{\pm}\bigg(-\frac{q^2_0}{z};x,t\bigg)\sigma_3Q_{\pm},\\
&S(z)=(\sigma_3Q_{-})^{-1}S\bigg(-\frac{q^2_0}{z}\bigg)\sigma_3Q_{+},\quad \rho(z)=\frac{q^*_{-}}{q_{-}}\tilde{\rho}\bigg(-\frac{q^2_0}{z}\bigg).
\end{split}
\end{align}

\end{prop}

\begin{proof}
The proof of the first symmetry reduction and the second symmetry reduction are similar to proposition \eqref{P4}. We only need to replace $\lambda$ in proposition \eqref{P4} with $z$. Next, we mainly prove the third symmetry reduction.

Since $\lambda(z)=\frac12\big(z-\frac{q^2_0}{z}\big)$, one can derive $\lambda\big(-\frac{q^2_0}{z}\big)=\frac12\big(-\frac{q^2_0}{z}+z\big)=\lambda(z)$, we have
\begin{align}\nonumber
X(\lambda(z);x,t)=X\bigg(\lambda\bigg(-\frac{q^2_0}{z}\bigg);x,t\bigg)\Rightarrow X(z;x,t)=X\bigg(-\frac{q^2_0}{z};x,t\bigg),\\
T(\lambda(z);x,t)=T\bigg(\lambda\bigg(-\frac{q^2_0}{z}\bigg);x,t\bigg)\Rightarrow T(z;x,t)=T\bigg(-\frac{q^2_0}{z};x,t\bigg).
\end{align}
Then $\psi_{\pm}\bigg(-\frac{q^2_0}{z};x,t\bigg)C$ is also the solution of Eqs. \eqref{T2}-\eqref{T3}, for any $2\times2$ matrix $C$ independent of $x$ and $t$. Hence, according to \eqref{T70}, it is apparent that
\begin{align}\label{T83}
\psi_{\pm}\bigg(-\frac{q^2_0}{z};x,t\bigg)C\sim Y_{\pm}\bigg(-\frac{q^2_0}{z}\bigg)\mathrm{e}^{\mathrm{i}\theta\big(-\frac{q^2_0}{z};x,t\big)\sigma_3}C=Y_{\pm}\bigg(-\frac{q^2_0}{z}\bigg)\mathrm{e}^{-\mathrm{i}\theta(z;x,t)\sigma_3}C,\text{ as }x\rightarrow\pm\infty,
\end{align}
where $\theta\big(-\frac{q^2_0}{z};x,t\big)=-\theta(z;x,t)$, and noting that
\begin{align}\nonumber
\frac{\mathrm{i}}{z}Y_{\pm}\bigg(-\frac{q^2_0}{z}\bigg)\mathrm{e}^{-\mathrm{i}\theta(z;x,t)\sigma_3}\sigma_3Q_{\pm}=Y_{\pm}(z)\mathrm{e}^{\mathrm{i}\theta(z;x,t)\sigma_3},
\end{align}
which together with \eqref{T83} implies that $C=\frac{\mathrm{i}}{z}\sigma_3Q_{\pm}$, we then obtain the symmetric relation
\begin{align}\label{T84}
\psi_{\pm}(z;x,t)=\frac{\mathrm{i}}{z}\psi_{\pm}\bigg(-\frac{q^2_0}{z};x,t\bigg)\sigma_3Q_{\pm},
\end{align}
and by using the transformation \eqref{T71}, we have
\begin{align}\nonumber
\begin{split}
\mu_{\pm}(z;x,t)&=\frac{\mathrm{i}}{z}\mu_{\pm}\bigg(-\frac{q^2_0}{z};x,t\bigg)\mathrm{e}^{\mathrm{i}\theta\big(-\frac{q^2_0}{z};x,t\big)\sigma_3}\sigma_3Q_{\pm}\mathrm{e}^{-\mathrm{i}\theta(z;x,t)\sigma_3}=\frac{\mathrm{i}}{z}\mu_{\pm}\bigg(-\frac{q^2_0}{z};x,t\bigg)\sigma_3Q_{\pm}.
\end{split}
\end{align}

In addition, By using Eq. \eqref{T77}, we have
\begin{align}\label{T85}
\psi_{+}\bigg(-\frac{q^2_0}{z};x,t\bigg)=\psi_{-}\bigg(-\frac{q^2_0}{z};x,t\bigg)S\bigg(-\frac{q^2_0}{z}\bigg).
\end{align}
and according the symmetry of Jost solutions \eqref{T84}, one has following formulas:
\begin{align}\label{T86}
\psi_{+}\bigg(-\frac{q^2_0}{z};x,t\bigg)=\frac{z}{\mathrm{i}}\psi_{+}(z;x,t)(\sigma_3Q_{+})^{-1},\quad \psi_{-}\bigg(-\frac{q^2_0}{z};x,t\bigg)=\frac{z}{\mathrm{i}}\psi_{-}(z;x,t)(\sigma_3Q_{-})^{-1},
\end{align}
then substituting \eqref{T86} into Eq. \eqref{T85}, one can obtain
\begin{align}\label{T87}
\begin{split}
\psi_{+}(z;x,t)=\psi_{-}(z;x,t)(\sigma_3Q_{-})^{-1}S\bigg(-\frac{q^2_0}{z}\bigg)\sigma_3Q_{+},
\end{split}
\end{align}
combining \eqref{T87} with \eqref{T77}, and considering the same asymptotic behavior,
\begin{align}\nonumber
S(z),\quad S\bigg(-\frac{q^2_0}{z}\bigg)\thicksim I, \text{ as } x\rightarrow\pm\infty,
\end{align}
one can lead to
\begin{align}\label{T88}
S(z)=(\sigma_3Q_{-})^{-1}S\bigg(-\frac{q^2_0}{z}\bigg)\sigma_3Q_{+}.
\end{align}

According to the symmetry reduction \eqref{T88} of scattering matrix, we have
\begin{align}\nonumber
\begin{split}
\bigg(\begin{array}{cc} s_{11}(z) & s_{12}(z) \\s_{21}(z) & s_{22}(z) \end{array}\bigg)=\Bigg(\begin{array}{cc} \frac{q^*_{+}}{q^*_{-}}s_{22}\big(-\frac{q^2_0}{z}\big) & \frac{q_{+}}{q^*_{-}}s_{21}\big(-\frac{q^2_0}{z}\big) \\ \frac{q^*_{+}}{q_{-}}s_{12}\big(-\frac{q^2_0}{z}\big) & \frac{q_{+}}{q_{-}}s_{11}\big(-\frac{q^2_0}{z}\big) \end{array}\Bigg),
\end{split}
\end{align}
we then obtain these symmetric relations
\begin{align}\label{T89}
\begin{split}
&s_{11}(z)=\frac{q^*_{+}}{q^*_{-}}s_{22}\bigg(-\frac{q^2_0}{z}\bigg),\quad s_{12}(z)=\frac{q_{+}}{q^*_{-}}s_{21}\bigg(-\frac{q^2_0}{z}\bigg),\\ &s_{21}(z)=\frac{q^*_{+}}{q_{-}}s_{12}\bigg(-\frac{q^2_0}{z}\bigg),\quad s_{22}(z)=\frac{q_{+}}{q_{-}}s_{11}\bigg(-\frac{q^2_0}{z}\bigg).
\end{split}
\end{align}

Since $\rho(z)=\frac{s_{21}(z)}{s_{11}(z)}=\frac{\frac{q^*_{+}}{q_{-}}s_{12}\big(-\frac{q^2_0}{z}\big)}{\frac{q^*_{+}}{q^*_{-}}s_{22}\big(-\frac{q^2_0}{z}\big)}=\frac{q^*_{-}}{q_{-}}\tilde{\rho}\Big(-\frac{q^2_0}{z}\Big)$, the symmetry reduction of reflection coefficient is $\rho(z)=\frac{q^*_{-}}{q_{-}}\tilde{\rho}\Big(-\frac{q^2_0}{z}\Big)$. This completes the proof.

\end{proof}

\subsubsection{Asymptotic Behaviors}
In order to propose and solve the matrix RH problem for the inverse problem, it is necessary to discuss the asymptotic behaviors of the modified Jost solutions and scattering matrix as $z\rightarrow0$ and $z\rightarrow\infty$, which differ from the case of ZBCs. The asymptotic behaviors of the modified Jost solution are deduce with the aid of the general WKB expansions.

\begin{prop}\label{P7}
The asymptotic behaviors of the modified Jost solutions are exhibited as
\begin{align}\label{T90}
\mu_{\pm}(z;x,t)=\left\{
\begin{aligned}
&\frac{\mathrm{i}}{z}\mathrm{e}^{\mathrm{i}\nu_{\pm}(x,t)\sigma_3}\sigma_3Q_{\pm}+O(1),\text{ as } z\rightarrow0,\\
&\mathrm{e}^{\mathrm{i}\nu_{\pm}(x,t)\sigma_3}+O\Big(\frac{1}{z}\Big),\quad\quad\,\text{ as } z\rightarrow\infty,
\end{aligned}
\right.
\end{align}
where
\begin{align}\label{T91}
\nu_{\pm}(x,t)=\frac12\int^x_{\pm\infty}\big(|q(y,t)|^2-q^2_0\big)dy.
\end{align}

\end{prop}

\begin{proof}
First, we consider the following asymptotic WKB expansion of the modified Jost solutions $\mu_{\pm}(z;x,t)$ as $z\rightarrow0$
\begin{align}\label{T92}
\mu_{\pm}(z;x,t)=\sum\limits^n_{i=-1}\mu_{\pm}^{[i]}(x,t)z^{i}+O\big(z^{n+1}\big)=\frac{\mu^{[-1]}_{\pm}(x,t)}{z}+\mu^{[0]}_{\pm}(x,t)+\mu^{[1]}_{\pm}(x,t)z+\cdots.
\end{align}

Substituting \eqref{T92} into the equivalent Lax pair \eqref{T73}, and comparing the coefficients of different powers of $z$, we have
\begin{align}\nonumber
O(z^{-4}):\Big(\mu^{[-1]}_{\pm}(x,t)\Big)^{\mathrm{diag}}=0,
\end{align}
where $\Big(\mu^{[-1]}_{\pm}(x,t)\Big)^{\mathrm{diag}}$ represent the diagonal parts of $\mu^{[-1]}_{\pm}(x,t)$.
\begin{align}\nonumber
\begin{split}
O(z^{-3}):\Big(\mu^{[0]}_{\pm}(x,t)\Big)^{\mathrm{diag}}=\mathrm{diag}\bigg(\frac{-\mathrm{i}q}{q^2_0}\mu^{[-1]}_{\pm21}(x,t),\frac{-\mathrm{i}q^*}{q^2_0}\mu^{[-1]}_{\pm12}(x,t)\bigg),
\end{split}
\end{align}
where $\mu^{[-1]}_{\pm21}(x,t)$ and $\mu^{[-1]}_{\pm12}(x,t)$ represent the second row and first column element and the first row and second column element of matrix $\mu^{[-1]}_{\pm}(x,t)$, respectively. Moreover, $$\mathrm{diag}\bigg(\frac{-\mathrm{i}q}{q^2_0}\mu^{[-1]}_{\pm21}(x,t),\frac{-\mathrm{i}q^*}{q^2_0}\mu^{[-1]}_{\pm12}(x,t)\bigg)=\Bigg(\begin{array}{cc} \frac{-\mathrm{i}q}{q^2_0}\mu^{[-1]}_{\pm21}(x,t) & 0 \\ 0 & \frac{-\mathrm{i}q^*}{q^2_0}\mu^{[-1]}_{\pm12}(x,t) \end{array}\Bigg).$$
\begin{align}\nonumber
\begin{split}
O(z^{-2}):\Big(\mu^{[-1]}_{\pm}(x,t)\Big)_x=\frac12\mathrm{i}\big(|q(x,t)|^2-q^2_0\big)\sigma_3\mu^{[-1]}_{\pm}(x,t).
\end{split}
\end{align}
Therefore, we obtain that $\mu^{[-1]}_{\pm}(x,t)=\mathrm{e}^{\int^x_{\pm\infty}\frac{\mathrm{i}}{2}\big(|q(y,t)|^2-q^2_0\big)dy\sigma_3}C$, where $C$ is off-diagonal constant matrix. Let $\nu_{\pm}(x,t)$ satisfy Eq. \eqref{T91}, one can obtain $\mu^{[-1]}_{\pm}(x,t)=\mathrm{e}^{\mathrm{i}\nu_{\pm}(x,t)\sigma_3}C$ and $\lim\limits_{x\rightarrow\pm\infty}\mu^{[-1]}_{\pm}(x,t)=C$. According to the expansion \eqref{T92}, we have
\begin{align}\nonumber
\lim\limits_{x\rightarrow\pm\infty}z\mu_{\pm}(z;x,t)=zY_{\pm}(z)=z\Big(I+\frac{\mathrm{i}}{z}\sigma_3Q_{\pm}\Big)=\lim\limits_{x\rightarrow\pm\infty}\Big(\mu^{[-1]}_{\pm}(x,t)+z\mu^{[0]}_{\pm}(x,t)+\cdots\Big),
\end{align}
so we have $C=\mathrm{i}\sigma_3Q_{\pm}$ and $\mu^{[-1]}_{\pm}(x,t)=\mathrm{i}\mathrm{e}^{\mathrm{i}\nu_{\pm}(x,t)\sigma_3}\sigma_3Q_{\pm}$. Finally, we obtain the asymptotic behaviors
\begin{align}\nonumber
\mu_{\pm}(z;x,t)=\frac{\mathrm{i}}{z}\mathrm{e}^{\mathrm{i}\nu_{\pm}(x,t)\sigma_3}\sigma_3Q_{\pm}+O(1),\text{ as } z\rightarrow0.
\end{align}

Next, we consider the asymptotic expansions of the Jost solutions $\mu_{\pm}(z;x,t)$ as $z\rightarrow\infty$
\begin{align}\label{T93}
\mu_{\pm}(z;x,t)=\sum\limits^n_{i=0}\frac{\mu_{\pm}^{[i]}(x,t)}{z^i}+O\bigg(\frac{1}{z^{n+1}}\bigg)=\mu^{[0]}_{\pm}(x,t)+\frac{\mu^{[1]}_{\pm}(x,t)}{z}+\frac{\mu^{[2]}_{\pm}(x,t)}{z^2}+\cdots.
\end{align}

By utilizing same way, substituting \eqref{T93} into the equivalent Lax pair \eqref{T73}, and comparing the coefficients of different powers of $z$, we have
\begin{align}\nonumber
O(z^{2}):\Big(\mu^{[0]}_{\pm}(x,t)\Big)^{\mathrm{off}}=0,
\end{align}
where $\Big(\mu^{[0]}_{\pm}(x,t)\Big)^{\mathrm{off}}$ represent the off-diagonal parts of $\mu^{[0]}_{\pm}(x,t)$.

\begin{align}\nonumber
\begin{split}
O(z):\Big(\mu^{[1]}_{\pm}(x,t)\Big)^{\mathrm{off}}=\mathrm{off}\Big(\mathrm{i}q\mu^{[0]}_{\pm22}(x,t),\mathrm{i}q^*\mu^{[0]}_{\pm11}(x,t)\Big),
\end{split}
\end{align}
where $\mathrm{off}\Big(\mathrm{i}q\mu^{[0]}_{\pm22}(x,t),\mathrm{i}q^*\mu^{[0]}_{\pm11}(x,t)\Big)=\Bigg(\begin{array}{cc} 0 & \mathrm{i}q\mu^{[0]}_{\pm22}(x,t) \\ \mathrm{i}q^*\mu^{[0]}_{\pm11}(x,t) & 0 \end{array}\Bigg)$.

\begin{align}\nonumber
\begin{split}
O(1):\Big(\mu^{[0]}_{\pm}(x,t)\Big)_x=\frac{\mathrm{i}}{2}\big(|q(x,t)|^2-q^2_0\big)\sigma_3.
\end{split}
\end{align}

Therefore, we obtain that $\mu^{[0]}_{\pm}(x,t)=\widetilde{C}\mathrm{e}^{\int^x_{\pm\infty}\frac{\mathrm{i}}{2}\big(|q(y,t)|^2-q^2_0\big)dy\sigma_3}$, where $\widetilde{C}$ is diagonal constant matrix. Let $\nu_{\pm}(x,t)$ satisfy Eq. \eqref{T91}, one can obtain $\mu^{[0]}_{\pm}(x,t)=\widetilde{C}\mathrm{e}^{\mathrm{i}\nu_{\pm}(x,t)\sigma_3}$ and $\lim\limits_{x\rightarrow\pm\infty}\mu^{[0]}_{\pm}(x,t)=\widetilde{C}$. According to the expansion \eqref{T93}, we have
\begin{align}\nonumber
\lim\limits_{\substack{x\rightarrow\pm\infty\\ z\rightarrow\infty}}\mu_{\pm}(z;x,t)=\lim\limits_{z\rightarrow\infty}Y_{\pm}(z)=\lim\limits_{z\rightarrow\infty}\bigg(I+\frac{\mathrm{i}}{z}\sigma_3Q_{\pm}\bigg)=I,
\end{align}
so we have $\widetilde{C}=I$ and $\mu^{[0]}_{\pm}(x,t)=\mathrm{e}^{\mathrm{i}\nu_{\pm}(x,t)\sigma_3}$. Finally, we get the asymptotic behaviors
\begin{align}\nonumber
\mu_{\pm}(z;x,t)=\mathrm{e}^{\mathrm{i}\nu_{\pm}(x,t)\sigma_3}+O\bigg(\frac{1}{z}\bigg),\text{ as } z\rightarrow\infty.
\end{align}

This completes the proof.
\end{proof}

The asymptotic behaviors of the scattering matrix can be directly yielded by exploiting the determinant representations of the scattering matrix and the asymptotic behaviors of the modified Jost solutions.
\begin{prop}\label{PP7}
The asymptotic behaviors of the scattering matrix are as follows
\begin{align}\label{T94}
S(z)=\left\{
\begin{aligned}
&\mathrm{diag}\bigg(\frac{q_{-}}{q_{+}},\frac{q_{+}}{q_{-}}\bigg)\mathrm{e}^{\mathrm{i}\nu(t)\sigma_3}+O(z),\text{ as } z\rightarrow0,\\
&\mathrm{e}^{-\mathrm{i}\nu(t)\sigma_3}+O\bigg(\frac{1}{z}\bigg),\quad\quad\quad\quad\quad\text{ as } z\rightarrow\infty,
\end{aligned}
\right.
\end{align}
where
\begin{align}\label{T95}
\nu(t)=\frac12\int^{+\infty}_{-\infty}\big(|q(y,t)|^2-q^2_0\big)dy.
\end{align}

\end{prop}
\begin{proof}
From the relationship \eqref{T71} between Jost solutions and modified Jost solutions, the determinant representations \eqref{T78} of the scattering matrix and the asymptotic behaviors \eqref{T90}, for $z\rightarrow0$, we have
\begin{align}\nonumber
\begin{split}
s_{11}(z)=&\frac{\mathrm{det}\big(\psi_{+1}(z;x,t),\psi_{-2}(z;x,t)\big)}{\eta}=\frac{\mathrm{det}\big(\mu_{+1}(z;x,t),\mu_{-2}(z;x,t)\big)}{1+\frac{q^2_0}{z^2}}=\\
&\mathrm{det}\Bigg(\begin{array}{cc} O(1) & \frac{\mathrm{i}}{z}q_{-}\mathrm{e}^{-\mathrm{i}\nu_{-}(x,t)}+O(1) \\ \frac{\mathrm{i}}{z}q^*_{+}\mathrm{e}^{\mathrm{i}\nu_{+}(x,t)}+O(1) & O(1) \end{array}\Bigg) \bigg(\frac{z^2}{q^2_0}-\frac{z^4}{q^4_0}+\cdots\bigg)\\
=&\bigg\{\frac{q_{-}q^*_{+}}{z^2}\mathrm{e}^{-\mathrm{i}[\nu_{-}(x,t)-\nu_{+}(x,t)]}+O(z)\bigg\}\bigg(\frac{z^2}{q^2_0}-\frac{z^4}{q^4_0}+\cdots\bigg)=\frac{q_{-}}{q_{+}}\mathrm{e}^{-\mathrm{i}\nu(t)}+O(z),
\end{split}
\end{align}
where $\frac{1}{1+\frac{q^2_0}{z^2}}=\frac{z^2}{z^2+q^2_0}=\frac{z^2}{q^2_0}\cdot\frac{1}{1-\big(-\frac{z^2}{q^2_0}\big)}=\frac{z^2}{q^2_0}-\frac{z^4}{q^4_0}+\cdots$, and $\nu(t)=\nu_{-}(x,t)-\nu_{+}(x,t)=\frac12\int^{+\infty}_{-\infty}\big(|q(y,t)|^2-q^2_0\big)dy$. Similarly, one can obtain
\begin{align}\nonumber
s_{12}=O(z),\quad s_{21}=O(z),\quad s_{22}=\frac{q_{+}}{q_{-}}\mathrm{e}^{\mathrm{i}\nu(t)}+O(z),
\end{align}
so we have the asymptotic behavior \eqref{T94} as $z\rightarrow0$.

When $z\rightarrow\infty$, by using the same way, we have

\begin{align}\nonumber
\begin{split}
s_{11}(z)&=\frac{\mathrm{det}\big(\psi_{+1}(z;x,t),\psi_{-2}(z;x,t)\big)}{\eta}=\frac{\mathrm{det}\big(\mu_{+1}(z;x,t),\mu_{-2}(z;x,t)\big)}{1+\frac{q^2_0}{z^2}}\\
&=\mathrm{det}\bigg(\begin{array}{cc} \mathrm{e}^{\mathrm{i}\nu_{+}(x,t)}+O(z^{-1}) & O(z^{-1}) \\ O(z^{-1}) & \mathrm{e}^{-\mathrm{i}\nu_{-}(x,t)}+O(z^{-1}) \end{array}\bigg)\bigg(1-\frac{q^2_0}{z^2}+\frac{q^4_0}{z^4}+\cdots\bigg)\\
&=\big\{\mathrm{e}^{-\mathrm{i}[\nu_{-}(x,t)-\nu_{+}(x,t)]}+O(z^{-1})\big\}\bigg(1-\frac{q^2_0}{z^2}+\frac{q^4_0}{z^4}+\cdots\bigg)=\mathrm{e}^{-\mathrm{i}\nu(t)}+O(z^{-1}),
\end{split}
\end{align}

where $\frac{1}{1+\frac{q^2_0}{z^2}}=\frac{1}{1-\big(-\frac{q^2_0}{z^2}\big)}=1-\frac{q^2_0}{z^2}+\frac{q^4_0}{z^4}+\cdots$, and
\begin{align}\nonumber
s_{12}=O(z^{-1}),\quad s_{21}=O(z^{-1}),\quad s_{22}=\mathrm{e}^{\mathrm{i}\nu(t)}+O(z^{-1}),
\end{align}
Hence, we also obtain the asymptotic behavior \eqref{T94} as $z\rightarrow\infty$.

Moreover, on the one hand, substituting the WKB expansion $$\mu_{\pm}(\lambda;x,t)=\sum\limits^n_{i=-1}\mu_{\pm}^{[i]}(x,t)z^{i}+O(z^{n+1})(\text{ as }z\rightarrow0)$$ into the time part of equivalent Lax pairs \eqref{T74} and matching the $O(\lambda^i)(i=-6,\cdots,0)$ in order, combining the
\begin{align}\nonumber
\begin{split}
O(1):&-\frac{13}{16}\mathrm{i}q_0^{10}\big[\mu^{[2]}_{\pm}(x,t),\sigma_3\big]+\frac{39}{16}q^8_0\sigma_3Q_{\pm}\big[\mu^{[1]}_{\pm}(x,t),\sigma_3\big]+\frac{39}{16}\mathrm{i}q^8_0\big[\mu^{[0]}_{\pm}(x,t),\sigma_3\big]-\\
&\frac{27}{16}q^6_0\sigma_3Q_{\pm}\big[\mu^{[-1]}_{\pm}(x,t),\sigma_3\big]+\mathrm{i}\sigma_3Q_{\pm}\mu^{[-1]}_{\pm,t}(x,t)=\mathrm{i}\sigma_3Q_{\pm}(T-T_{\pm})\mu^{[-1]}_{\pm}(x,t)
\end{split}
\end{align}
and $\Big(\mu^{[-1]}_{\pm}(x,t)\Big)^{\mathrm{diag}}=0$, one can yield $\mu^{[-1]}_{\pm,t}(x,t)\neq0$, $\nu_{\pm,t}(x,t)\neq0$ and $\nu_t\neq0$ as $x\rightarrow\pm\infty$. That is, $\nu$ is dependent on the variable $t$.

On the other hand, substituting the WKB expansion $$\mu_{\pm}(\lambda;x,t)=\sum\limits^n_{i=0}\frac{\mu_{\pm}^{[i]}(x,t)}{z^i}+O\Big(\frac{1}{z^{n+1}}\Big)(\text{ as }z\rightarrow\infty)$$ into the time part of equivalent Lax pairs \eqref{T74} and matching the $O(\lambda^i)(i=6,\cdots,0)$ in order, combining the
\begin{align}\nonumber
\begin{split}
O(1):&\frac{1}{16}\mathrm{i}\big[\mu^{[6]}_{\pm}(x,t),\sigma_3\big]+\frac{1}{16}\sigma_3Q_{\pm}\big[\mu^{[5]}_{\pm}(x,t),\sigma_3\big]-\frac{3}{16}\mathrm{i}q^2_0\big[\mu^{[4]}_{\pm}(x,t),\sigma_3\big]-\\
&\frac{13}{16}q^2_0\sigma_3Q_{\pm}\big[\mu^{[3]}_{\pm}(x,t),\sigma_3\big]+\frac52\mathrm{i}q^4_0\big[\mu^{[2]}_{\pm}(x,t),\sigma_3\big]+\frac52q^4_0\sigma_3Q_{\pm}\big[\mu^{[1]}_{\pm}(x,t),\sigma_3\big]-\\
&\frac{39}{16}\mathrm{i}q^6_0\big[\mu^{[0]}_{\pm}(x,t),\sigma_3\big]+\mu^{[0]}_{\pm,t}(x,t)=(T-T_{\pm})\mu^{[0]}_{\pm}(x,t)
\end{split}
\end{align}
and $\Big(\mu^{[0]}_{\pm}(x,t)\Big)^{\mathrm{off}}=0$, one can yield $\mu^{[0]}_{\pm,t}(x,t)\neq0$, $\nu_{\pm,t}(x,t)\neq0$ and $\nu_t\neq0$ as $x\rightarrow\pm\infty$. That is, $\nu$ also is dependent on the variable $t$. In summary, $\nu$ is a function about time t. Proof complete.

\end{proof}

\subsection{Inverse Problem with NZBCs and Double Poles}

\subsubsection{Discrete Spectrum with NZBCs and Double Zeros}
In this section, we consider the case of $s_{11}(z)$ with double zeros and suppose that $s_{11}(z)$ has $N_1$ double zeros in $Z_0=\{z\in\mathbb{C}|\mathrm{Re}z>0,\mathrm{Im}z>0,|z|>q_0\}$ denoted by $z_n$, and $N_2$ double zeros on the circle $W_0=\{z\in\mathbb{C}|z=q_0\mathrm{e}^{\mathrm{i}\phi},0<\phi<\frac{\pi}{2}\}$ denoted by $\omega_m$, that is, $s_{11}(z_0)=s'_{11}(z_0)=0$ and $s''_{11}(z_0)\neq0$ if $z_0$ is a double zero of $s_{11}(z)$. From the symmetries of the scattering matrix presented in proposition \eqref{P5}, the $4(2N_1+N_2)$ discrete spectrum can be given by the set
\begin{align}\label{T96}
Z=\Bigg\{\pm z_n,\pm z^*_n, \pm \frac{q_0^2}{z_n},\pm \frac{q^2_0}{z^*_n}\Bigg\}^{N_1}_{n=1}\bigcup\big\{\pm\omega_m,\pm\omega^*_m\big\}^{N_2}_{m=1},
\end{align}
the distributions are shown in Figure \ref{F4}. When a given $z_0\in Z\cap D^{+}$, one can obtain that $\psi_{+1}(z_0;x,t)$ and $\psi_{-2}(z_0;x,t)$ are linearly dependent by combining Eq. \eqref{T78} and $s_{11}(z_0)=0$. Similarly, when a given $z_0\in Z\cap D^{-}$, one can obtain that $\psi_{+2}(z_0;x,t)$ and $\psi_{-1}(z_0;x,t)$ are linearly dependent by combining Eq. \eqref{T78} and $s_{22}(z_0)=0$. For convenience, we introduce the norming constant $b[z_0]$ such that
\begin{align}\label{T97}
\begin{split}
\psi_{+1}(z_0;x,t)=b[z_0]\psi_{-2}(z_0;x,t),\text{ as }z_0\in Z\cap D^{+},\\
\psi_{+2}(z_0;x,t)=b[z_0]\psi_{-1}(z_0;x,t),\text{ as }z_0\in Z\cap D^{-}.
\end{split}
\end{align}

For a given $z_0\in Z\cap D^{+}$, according to $s_{11}(z_0)=\frac{\mathrm{det}\big(\psi_{+1}(z;x,t),\psi_{-2}(z;x,t)\big)}{\eta}$ in Eq. \eqref{T78}, and taking derivative respect to $z$ $(\text{here }z=z_0)$ on both sides of this equation, combining $s_{11}(z_0)=0,\, s'_{11}(z_0)=0$ and Eq. \eqref{T97}, we have
\begin{align}\nonumber
\begin{split}
\mathrm{det}(\psi'_{+1}(z_0;x,t)-b[z_0]\psi'_{-2}(z_0;x,t),\psi_{-2}(z_0;x,t))=0,
\end{split}
\end{align}
it is evidence that $\psi'_{+1}(z_0;x,t)-b[z_0]\psi'_{-2}(z_0;x,t)$ and $\psi_{-2}(z_0;x,t)$ are linearly dependent. Similarly, when a given $z_0\in Z\cap D^{-}$, one can obtain that $\psi'_{+2}(z_0;x,t)-b[z_0]\psi'_{-1}(z_0;x,t)$ and $\psi_{-1}(z_0;x,t)$ are linearly dependent by combining Eq. \eqref{T78} and $s'_{22}(z_0)=0$. For convenience, we define another norming constant $d[z_0]$ such that
\begin{align}\label{T98}
\begin{split}
&\psi'_{+1}(z_0;x,t)-b[z_0]\psi'_{-2}(z_0;x,t)=d[z_0]\psi_{-2}(z_0;x,t),\text{ as }z_0\in Z\cap D^{+},\\
&\psi'_{+2}(z_0;x,t)-b[z_0]\psi'_{-1}(z_0;x,t)=d[z_0]\psi_{-1}(z_0;x,t),\text{ as }z_0\in Z\cap D^{-}.
\end{split}
\end{align}

On the other hand, we notice that $\psi_{+1}(z;x,t)$ and $s_{11}(z)$ are analytic on $D^{+}$. Suppose $z_0$ is the double zeros of $s_{11}(z)$, let $\psi_{+1}(z;x,t)$ and $s_{11}(z)$ carry out Taylor expansion at $z=z_0$, we have
\begin{footnotesize}
\begin{align}\nonumber
\begin{split}
\frac{\psi_{+1}(z;x,t)}{s_{11}(z)}&=\frac{2\psi_{+1}(z_0;x,t)}{s''_{11}(z_0)}(z-z_0)^{-2}+\Bigg(\frac{2\psi'_{+1}(z_0;x,t)}{s''_{11}(z_0)}-\frac{2\psi_{+1}(z_0;x,t)s'''_{11}(z_0)}{3s''^2_{11}(z_0)}\Bigg)(z-z_0)^{-1}+\cdots,
\end{split}
\end{align}
\end{footnotesize}
Then, one has the compact form
\begin{align}\nonumber
\begin{split}
\mathop{P_{-2}}_{z=z_0}\bigg[\frac{\psi_{+1}(z;x,t)}{s_{11}(z)}\bigg]=\frac{2\psi_{+1}(z_0;x,t)}{s''_{11}(z_0)}=\frac{2b[z_0]\psi_{-2}(z_0;x,t)}{s''_{11}(z_0)},\text{ as }z_0\in Z\cap D^{+},
\end{split}
\end{align}
where $\underset{z=z_0}{P_{-2}}[f(z;x,t)]$ denotes the coefficient of $O((z-z_0)^{-2})$ term in the Laurent series expansion of $f(z;x,t)$ at $z=z_0$. and
\begin{footnotesize}
\begin{align}\nonumber
\begin{split}
\mathop{\mathrm{Res}}_{z=z_0}\bigg[\frac{\psi_{+1}(z;x,t)}{s_{11}(z)}\bigg]=\frac{2b[z_0]\psi'_{-2}(z_0;x,t)}{s''_{11}(z_0)}+\bigg[\frac{2b[z_0]}{s''_{11}(z_0)}\bigg(\frac{d[z_0]}{b[z_0]}-\frac{s'''_{11}(z_0)}{3s''_{11}(z_0)}\bigg)\bigg]\psi_{-2}(z_0;x,t),\text{ as }z_0\in Z\cap D^{+},
\end{split}
\end{align}
\end{footnotesize}
where $\underset{z=z_0}{\mathrm{Res}}[f(z;x,t)]$ denotes the coefficient of $O((z-z_0)^{-1})$ term in the Laurent series expansion of $f(z;x,t)$ at $z=z_0$.

Similarly,  for the case of $\psi_{+2}(z;x,t)$ and $s_{22}(z)$ are analytic on $D^{-}$, we repeat the above process and obtain
\begin{align}\nonumber
\begin{split}
&\mathop{P_{-2}}_{z=z_0}\bigg[\frac{\psi_{+2}(z;x,t)}{s_{22}(z)}\bigg]=\frac{2\psi_{+2}(z_0;x,t)}{s''_{22}(z_0)}=\frac{2b[z_0]\psi_{-1}(z_0;x,t)}{s''_{22}(z_0)},\text{ as }z_0\in Z\cap D^{-},\\
&\mathop{\mathrm{Res}}_{z=z_0}\bigg[\frac{\psi_{+2}(z;x,t)}{s_{22}(z)}\bigg]=\frac{2b[z_0]\psi'_{-1}(z_0;x,t)}{s''_{22}(z_0)}+\bigg[\frac{2b[z_0]}{s''_{22}(z_0)}\bigg(\frac{d[z_0]}{b[z_0]}-\frac{s'''_{22}(z_0)}{3s''_{22}(z_0)}\bigg)\bigg]\psi_{-1}(z_0;x,t),\\
&\qquad\qquad\qquad\qquad\quad\text{ as }z_0\in Z\cap D^{-}.
\end{split}
\end{align}

Moreover, let
\begin{align}\label{T99}
A[z_0]=\left\{
\begin{aligned}
\frac{2b[z_0]}{s''_{11}(z_0)},\text{ as }z_0\in Z\cap D^{+},\\
\frac{2b[z_0]}{s''_{22}(z_0)},\text{ as }z_0\in Z\cap D^{-},
\end{aligned}
\right.\quad
B[z_0]=\left\{
\begin{aligned}
\frac{d[z_0]}{b[z_0]}-\frac{s'''_{11}(z_0)}{3s''_{11}(z_0)},\text{ as }z_0\in Z\cap D^{+},\\
\frac{d[z_0]}{b[z_0]}-\frac{s'''_{22}(z_0)}{3s''_{22}(z_0)},\text{ as }z_0\in Z\cap D^{-}.
\end{aligned}
\right.
\end{align}

Then, we have
\begin{align}\label{T100}
\begin{split}
&\mathop{P_{-2}}_{z=z_0}\bigg[\frac{\psi_{+1}(z;x,t)}{s_{11}(z)}\bigg]=A[z_0]\psi_{-2}(z_0;x,t),\text{ as }z_0\in Z\cap D^{+},\\
&\mathop{P_{-2}}_{z=z_0}\bigg[\frac{\psi_{+2}(z;x,t)}{s_{22}(z)}\bigg]=A[z_0]\psi_{-1}(z_0;x,t),\text{ as }z_0\in Z\cap D^{-},\\
&\mathop{\mathrm{Res}}_{z=z_0}\bigg[\frac{\psi_{+1}(z;x,t)}{s_{11}(z)}\bigg]=A[z_0][\psi'_{-2}(z_0;x,t)+B[z_0]\psi_{-2}(z_0;x,t)],\text{ as }z_0\in Z\cap D^{+},\\
&\mathop{\mathrm{Res}}_{z=z_0}\bigg[\frac{\psi_{+2}(z;x,t)}{s_{22}(z)}\bigg]=A[z_0][\psi'_{-1}(z_0;x,t)+B[z_0]\psi_{-1}(z_0;x,t)],\text{ as }z_0\in Z\cap D^{-}.
\end{split}
\end{align}

Accordingly, by mean of Eqs. \eqref{T97}-\eqref{T99} as well as proposition \ref{P5}, we can derive the following symmetry relations.
\begin{prop}\label{P6}
For a $z_0\in Z$, the three symmetry relations for $A[z_0]$ and $B[z_0]$ are deduced as follow\\
$\bullet$ The first symmetry relation $A[z_0]=-A[z^*_0]^*,\quad B[z_0]=B[z^*_0]^*$.\\
$\bullet$ The second symmetry relation $A[z_0]=A[-z^*_0]^*,\quad B[z_0]=-B[-z^*_0]^*$.\\
$\bullet$ The third symmetry relation $A[z_0]=\frac{z^4_0q^*_{-}}{q^4_0q_{-}}A\Big[-\frac{q^2_0}{z_0}\Big],\quad B[z_0]=\frac{q^2_0}{z^2_0}B\Big[-\frac{q^2_0}{z_0}\Big]+\frac{2}{z_0}$.

\end{prop}

In the following subsections, we will propose an inverse problem with NZBCs and solve it to obtain explicit double solutions for the TOFKN \eqref{T1}.
\subsubsection{The Matrix RH Problem with NZBCs and Double Poles}
Similar to the case of ZBCs in Section 2, the matrix RH problem for the NZBCs can also be established. In order to pose and solve the RH problem conveniently, we define $\widetilde{\zeta}_n=-\frac{q^2_0}{\zeta_n}$ with
\begin{align}\label{T103}
\zeta_n=\left\{
\begin{aligned}
&z_n,\quad\quad\quad\quad\quad n=1,2,\cdots,N_1,\\
&-z_{n-N},\quad\quad\quad n=N_1+1,N_1+2,\cdots,2N_1,\\
&\frac{q^2_0}{z^*_{n-2N_1}},\quad\quad\quad n=2N_1+1,2N_1+2,\cdots,3N_1,\\
&-\frac{q^2_0}{z^*_{n-3N_1}},\quad\quad n=3N_1+1,3N_1+2,\cdots,4N_1,\\
&\omega_{n-4N_1},\quad\quad\quad n=4N_1+1,4N_1+2,\cdots,4N_1+N_2,\\
&-\omega_{n-4N_1-N_2},\quad n=4N_1+N_2+1,4N_1+N_2+2,\cdots,4N_1+2N_2.
\end{aligned}
\right.
\end{align}

Then, we can proposed a matrix RH problem as follows.
\begin{prop}\label{P9}
Define the sectionally meromorphic matrices
\begin{align}\label{T104}
M(z;x,t)=\left\{
\begin{aligned}
M^+(z;x,t)=\bigg(\frac{\mu_{+1}(z;x,t)}{s_{11}(z)},\mu_{-2}(z;x,t)\bigg),\text{ as }z\in D^+,\\
M^-(z;x,t)=\bigg(\mu_{-1}(z;x,t),\frac{\mu_{+2}(z;x,t)}{s_{22}(z)}\bigg),\text{ as }z\in D^-,
\end{aligned}
\right.
\end{align}
where $\lim\limits_{\substack{z'\rightarrow z\\ z'\in D^{\pm}}}M(z';x,t)=M^{\pm}(z;x,t)$. Then, the multiplicative matrix RH problem is given bellow:\\
$\bullet$ Analyticity: $M(z;x,t)$ is analytic in $D^+\cup D^-\backslash Z$ and has the double poles in $Z$, whose principal parts of the Laurent series at each double pole $\zeta_n$ or $\widetilde{\zeta}_n$, are determined as
\begin{align}\label{T105}
\begin{split}
&\mathop{\mathrm{Res}}_{z=\zeta_n}M(z;x,t)=\big(A[\zeta_n]\mathrm{e}^{-2\mathrm{i}\theta(\zeta_n;x,t)}\{\mu'_{-2}(\zeta_n;x,t)+[B[\zeta_n]-2\mathrm{i}\theta'(\zeta_n;x,t)]\mu_{-2}(\zeta_n;x,t)\},0\big),\\
&\mathop{\mathrm{P}_{-2}}_{z=\zeta_n}M(z;x,t)=\big(A[\zeta_n]\mathrm{e}^{-2\mathrm{i}\theta(\zeta_n;x,t)}\mu_{-2}(\zeta_n;x,t),0\big),\\
&\mathop{\mathrm{Res}}_{z=\widetilde{\zeta}_n}M(z;x,t)=\big(0,A[\widetilde{\zeta}_n]\mathrm{e}^{2\mathrm{i}\theta(\widetilde{\zeta}_n;x,t)}\{\mu'_{-1}(\widetilde{\zeta}_n;x,t)+[B[\widetilde{\zeta}_n]+2\mathrm{i}\theta'(\widetilde{\zeta}_n;x,t)]\mu_{-1}(\widetilde{\zeta}_n;x,t)\}\big),\\
&\mathop{\mathrm{P}_{-2}}_{z=\widetilde{\zeta}_n}M(z;x,t)=\big(0,A[\widetilde{\zeta}_n]\mathrm{e}^{2\mathrm{i}\theta(\widetilde{\zeta}_n;x,t)}\mu_{-1}(\widetilde{\zeta}_n;x,t)\big).
\end{split}
\end{align}
$\bullet$ Jump condition:
\begin{align}\label{T106}
M^-(z;x,t)=M^+(z;x,t)[I-J(z;x,t)],\text{ as }z\in\Sigma,
\end{align}
where
\begin{align}\label{T107}
J(z;x,t)=\mathrm{e}^{\mathrm{i}\theta(z;x,t)\widehat{\sigma_3}}\bigg(\begin{array}{cc} 0 & -\tilde{\rho}(z) \\ \rho(z) & \rho(z)\tilde{\rho}(z) \end{array}\bigg).
\end{align}
$\bullet$ Asymptotic behavior:
\begin{align}\label{T108}
M(z;x,t)=\left\{
\begin{aligned}
&\frac{\mathrm{i}}{z}\mathrm{e}^{\mathrm{i}\nu_{-}(x,t)\sigma_3}\sigma_3Q_{-}+O(1),\text{ as } z\rightarrow0,\\
&\mathrm{e}^{\mathrm{i}\nu_{-}(x,t)\sigma_3}+\frac{1}{z},\quad\quad\quad\text{ as } z\rightarrow\infty.
\end{aligned}
\right.
\end{align}

\end{prop}

\begin{proof}
For the analyticity of $M(z;x,t)$, It follows from Eqs. \eqref{T71} and \eqref{T100} that for each double poles $\zeta_n\in D^+$ or $\widetilde{\zeta}_n\in D^-$. Now, we consider $\zeta_n\in D^+$, and obtain
\begin{small}
\begin{align}\nonumber
\begin{split}
&\mathop{\mathrm{Res}}_{z=\zeta_n}\bigg[\frac{\mu_{+1}(\zeta_n;x,t)}{s_{11}(\zeta_n)}\bigg]=A[\zeta_n]\mathrm{e}^{-2\mathrm{i}\theta(\zeta_n;x,t)}\{\mu'_{-2}(\zeta_n;x,t)+[B[\zeta_n]-2\mathrm{i}\theta'(\zeta_n;x,t)]\mu_{-2}(\zeta_n;x,t)\},\Rightarrow\\
&\mathop{\mathrm{Res}}_{z=\zeta_n}M^+(z;x,t)=\Big(A[\zeta_n]\mathrm{e}^{-2\mathrm{i}\theta(\zeta_n;x,t)}\{\mu'_{-2}(\zeta_n;x,t)+[B[\zeta_n]-2\mathrm{i}\theta'(\zeta_n;x,t)]\mu_{-2}(\zeta_n;x,t)\},0\Big),
\end{split}
\end{align}
\end{small}
where the $``'"$ denotes the partial derivative with respect to $z$ $(\text{here }z=\zeta_n)$, and
\begin{align}\nonumber
\begin{split}
&\mathop{P_{-2}}_{z=\zeta_n}\bigg[\frac{\mu_{+1}(\zeta_n;x,t)}{s_{11}(\zeta_n)}\bigg]=A[\zeta_n]\mathrm{e}^{-2\mathrm{i}\theta(\zeta_n;x,t)}\mu_{-2}(\zeta_n;x,t),\Rightarrow\\
&\mathop{\mathrm{P}_{-2}}_{\lambda=\zeta_n}M(\lambda;x,t)=\Big(A[\zeta_n]\mathrm{e}^{-2\mathrm{i}\theta(\zeta_n;x,t)}\mu_{-2}(\zeta_n;x,t),0\Big).
\end{split}
\end{align}

Similarly, we also can obtain the analyticity for $\widetilde{\zeta}_n\in D^-$. It follows from Eqs. \eqref{T71} and \eqref{T77} that
\begin{align}\nonumber
\bigg\{
\begin{aligned}
&\mu_{+1}(z;x,t)=\mu_{-1}(z;x,t)s_{11}(z)+\mu_{-2}(z;x,t)\mathrm{e}^{-2\mathrm{i}\theta(z;x,t)}s_{21}(z),\\
&\mu_{+2}(z;x,t)=\mu_{-1}(z;x,t)\mathrm{e}^{2\mathrm{i}\theta(z;x,t)}s_{12}(z)+\mu_{-2}(z;x,t)s_{22}(z),
\end{aligned}
\end{align}
by combining formula \eqref{T79}, one can obtain
\begin{align}\nonumber
\begin{split}
&M^+(z;x,t)=\Big(\mu_{-1}(z;x,t)+\mu_{-2}(z;x,t)\mathrm{e}^{-2\mathrm{i}\theta(z;x,t)}\rho(z),\mu_{-2}(z;x,t)\Big),\\
&M^-(z;x,t)=\Big(\mu_{-1}(z;x,t),\mu_{-1}(z;x,t)\mathrm{e}^{2\mathrm{i}\theta(z;x,t)}\tilde{\rho}(z)+\mu_{-2}(z;x,t)\Big),
\end{split}
\end{align}
and
\begin{small}
\begin{align}\nonumber
&\bigg(\mu_{-1}(z;x,t),\frac{\mu_{+2}(z;x,t)}{s_{22}(z)}\bigg)=\bigg(\frac{\mu_{+1}(z;x,t)}{s_{11}(z)},\mu_{-2}(z;x,t)\bigg)\bigg(\begin{array}{cc} 1 & \mathrm{e}^{2\mathrm{i}\theta(z;x,t)}\tilde{\rho}(z) \\ -\mathrm{e}^{-2\mathrm{i}\theta(z;x,t)}\rho(z) & 1-\rho(z)\tilde{\rho}(z) \end{array}\bigg),\\
\end{align}
\end{small}
that is
\begin{align}\nonumber
M^-(z;x,t)=M^+(z;x,t)(I-J(z;x,t)),
\end{align}
where $J(z;x,t)$ is given by Eq. \eqref{T107}. The asymptotic behaviors of the modified Jost solutions $\mu_{\pm}(z;x,t)$ and scattering matrix $S(z)$ given in propositions \eqref{P7} and \eqref{PP7} can easily lead to the asymptotic behavior of $M(z;x,t)$. Specifically, when $z\rightarrow0$, we have
\begin{align}\nonumber
\begin{split}
&M^{+}(z;x,t)=\bigg(\frac{\mu_{+1}(z;x,t)}{s_{11}(z)},\mu_{-2}(z;x,t)\bigg)=\Bigg(\begin{array}{cc} 0 & \frac{\mathrm{i}}{z}q_{-}\mathrm{e}^{\mathrm{i}\nu_{-}(x,t)} \\ \frac{\mathrm{i}q^*_{+}\mathrm{e}^{-\mathrm{i}\nu_{+}(x,t)}}{z\frac{q_{-}}{q_{+}}\mathrm{e}^{\mathrm{i}\nu(t)}} & 0 \end{array}\Bigg)+O(1)\\
&\qquad\qquad=\frac{\mathrm{i}}{z}\bigg(\begin{array}{cc} 0 & q_{-}\mathrm{e}^{\mathrm{i}\nu_{-}(x,t)} \\ q^*_{-}\mathrm{e}^{-\mathrm{i}\nu_{-}(x,t)} & 0 \end{array}\bigg)+O(1)=\frac{\mathrm{i}}{z}\mathrm{e}^{\mathrm{i}\nu_{-}(x,t)\sigma_3}\sigma_3Q_{-}+O(1),\\
&M^{-}(z;x,t)=\bigg(\mu_{-1}(z;x,t),\frac{\mu_{+2}(z;x,t)}{s_{22}(z)}\bigg)=\Bigg(\begin{array}{cc} 0 & \frac{\frac{\mathrm{i}}{z}q_{+}\mathrm{e}^{\mathrm{i}\nu_{+}(x,t)}}{\frac{q_{+}}{q_{-}}\mathrm{e}^{-\mathrm{i}\nu(t)}} \\ \frac{\mathrm{i}}{z}q^*_{-}\mathrm{e}^{-\mathrm{i}\nu_{-}(x,t)} & 0 \end{array}\Bigg)+O(1)\\
&\qquad\qquad=\frac{\mathrm{i}}{z}\bigg(\begin{array}{cc} 0 & q_{-}\mathrm{e}^{\mathrm{i}\nu_{-}(x,t)} \\ q^*_{-}\mathrm{e}^{-\mathrm{i}\nu_{-}(x,t)} & 0 \end{array}\bigg)+O(1)=\frac{\mathrm{i}}{z}\mathrm{e}^{\mathrm{i}\nu_{-}(x,t)\sigma_3}\sigma_3Q_{-}+O(1),
\end{split}
\end{align}
it can obtain that $M(z;x,t)=\frac{\mathrm{i}}{z}\mathrm{e}^{\mathrm{i}\nu_{-}(x,t)\sigma_3}\sigma_3Q_{-}+O(1)\text{ as }z\rightarrow0$. Moreover, when $z\rightarrow\infty$, we have
\begin{align}\nonumber
\begin{split}
&M^{+}(z;x,t)=\bigg(\frac{\mu_{+1}(z;x,t)}{s_{11}(z)},\mu_{-2}(z;x,t)\bigg)=\Bigg(\begin{array}{cc} \frac{\mathrm{e}^{\mathrm{i}\nu_{+}(x,t)}}{\mathrm{e}^{-\mathrm{i}\nu(t)}} & 0 \\ 0 & \mathrm{e}^{-\mathrm{i}\nu_{-}(x,t)} \end{array}\Bigg)+O\Big(\frac{1}{z}\Big)\\
&\qquad\qquad=\bigg(\begin{array}{cc} \mathrm{e}^{\mathrm{i}\nu_{-}(x,t)} & 0 \\ 0 & \mathrm{e}^{-\mathrm{i}\nu_{-}(x,t)} \end{array}\bigg)+O\Big(\frac{1}{z}\Big)=\mathrm{e}^{\mathrm{i}\nu_{-}(x,t)\sigma_3}+O\Big(\frac{1}{z}\Big),\\
&M^{-}(z;x,t)=\bigg(\mu_{-1}(z;x,t),\frac{\mu_{+2}(z;x,t)}{s_{22}(z)}\bigg)=\Bigg(\begin{array}{cc} \mathrm{e}^{\mathrm{i}\nu_{-}(x,t)} & 0 \\ 0 & \frac{\mathrm{e}^{-\mathrm{i}\nu_{+}(x,t)}}{\mathrm{e}^{\mathrm{i}\nu(t)}} \end{array}\Bigg)+O\Big(\frac{1}{z}\Big)\\
&\qquad\qquad=\bigg(\begin{array}{cc} \mathrm{e}^{\mathrm{i}\nu_{-}(x,t)} & 0 \\ 0 & \mathrm{e}^{-\mathrm{i}\nu_{-}(x,t)} \end{array}\bigg)+O\Big(\frac{1}{z}\Big)=\mathrm{e}^{\mathrm{i}\nu_{-}(x,t)\sigma_3}+O\Big(\frac{1}{z}\Big),
\end{split}
\end{align}
it can obtain that $M(z;x,t)=\mathrm{e}^{\mathrm{i}\nu_{-}(x,t)\sigma_3}+O\big(\frac{1}{z}\big)$. Therefore, we obtain the asymptotic behavior \eqref{T108} of $M(z;x,t)$ when $z\rightarrow0$ and $z\rightarrow\infty$. This completes the proof.

\end{proof}

\begin{prop}
The solution of the matrix Riemann-Hilbert problem with double poles can be expressed as
\begin{align}\label{T109}
\begin{split}
M(z;x,t)=&\mathrm{e}^{\mathrm{i}\nu_{-}(x,t)\sigma_3}\Big(I+\frac{\mathrm{\mathrm{i}}}{z}\sigma_3Q_{-}\Big)+\frac{1}{2\pi \mathrm{i}}\int_{\Sigma}\frac{M^+(\xi;x,t)J(\xi;x,t)}{\xi-z}d\xi+\\
&\sum^{4N_1+2N_2}_{n=1}\bigg(C_n(z)\bigg[\mu'_{-2}(\zeta_n;x,t)+\bigg(D_n+\frac{1}{z-\zeta_n}\bigg)\mu_{-2}(\zeta_n;x,t)\bigg],\\
&\widetilde{C}_n(z)\bigg[\mu'_{-1}(\widetilde{\zeta}_n;x,t)+\bigg(\widetilde{D}_n+\frac{1}{z-\widetilde{\zeta}_n}\bigg)\mu_{-1}(\widetilde{\zeta}_n;x,t)\bigg]\bigg),
\end{split}
\end{align}
where $\int_{\Sigma}$ stands for an integral along the oriented contour displayed in Fig. \ref{F4},
\begin{align}\label{T110}
\begin{split}
&C_n(z)=\frac{A[\zeta_n]}{z-\zeta_n}\mathrm{e}^{-2\mathrm{i}\theta(\zeta_n;x,t)},\quad \widetilde{C}_n(z)=\frac{A[\widetilde{\zeta}_n]}{z-\widetilde{\zeta}_n}\mathrm{e}^{2\mathrm{i}\theta(\widetilde{\zeta}_n;x,t)},\\
&D_n=B[\zeta_n]-2\mathrm{i}\theta'(\zeta_n;x,t),\quad \widetilde{D}_n=B[\widetilde{\zeta}_n]+2\mathrm{i}\theta'(\widetilde{\zeta}_n;x,t),
\end{split}
\end{align}
$\mu_{-2}(\zeta_n;x,t)$ and $\mu'_{-2}(\zeta_n;x,t)$ are determined via $\mu_{-1}(\widetilde{\zeta}_n;x,t)$ and $\mu'_{-1}(\widetilde{\zeta}_n;x,t)$ as
\begin{align}\label{T111}
\mu_{-2}(\zeta_n;x,t)=\frac{\mathrm{i}q_{-}}{\zeta_n}\mu_{-1}(\widetilde{\zeta}_n;x,t),\, \mu'_{-2}(\zeta_n;x,t)=-\frac{\mathrm{i}q_{-}}{\zeta^2_n}\mu_{-1}(\widetilde{\zeta}_n;x,t)+\frac{\mathrm{i}q_{-}q^2_0}{\zeta^3_n}\mu'_{-1}(\widetilde{\zeta}_n;x,t),
\end{align}
and $\mu_{-1}(\widetilde{\zeta}_n;x,t)$ and $\mu'_{-1}(\widetilde{\zeta}_n;x,t)$ satisfy the linear system of $8N_1+4N_2$ as bellow:
\begin{align}\label{T112}
\begin{split}
&\sum^{4N_1+2N_2}_{n=1}\bigg\{\widetilde{C}_n(\zeta_s)\mu'_{-1}(\widetilde{\zeta}_n;x,t)+\bigg[\widetilde{C}_n(\zeta_s)\bigg(\widetilde{D}_n+\frac{1}{\zeta_s-\widetilde{\zeta}_n}\bigg)-\frac{\mathrm{i}q_{-}}{\zeta_s}\delta_{s,n}\bigg]\mu_{-1}(\widetilde{\zeta}_n;x,t)\bigg\}=\\
&-\mathrm{e}^{\mathrm{i}\nu_{-}(x,t)\sigma_3}\bigg(\begin{array}{cc} \frac{\mathrm{i}q_{-}}{\zeta_s} \\ 1  \end{array}\bigg)-\frac{1}{2\pi \mathrm{i}}\int_{\Sigma}\frac{(M^+(\xi;x,t)J(\xi;x,t))_{2}}{\xi-\zeta_s}d\xi,\\
&\sum^{4N_1+2N_2}_{n=1}\bigg\{\bigg(\frac{\widetilde{C}_n(\zeta_s)}{\zeta_s-\widetilde{\zeta}_n}+\frac{\mathrm{i}q_{-}q^2_0}{\zeta^3_s}\delta_{s,n}\bigg)\mu'_{-1}(\widetilde{\zeta}_n;x,t)+\bigg[\frac{\widetilde{C}_n(\zeta_s)}{\zeta_s-\widetilde{\zeta}_n}\bigg(\widetilde{D}_n+\frac{2}{\zeta_s-\widetilde{\zeta}_n}\bigg)-\frac{\mathrm{i}q_{-}}{\zeta^2_s}\delta_{s,n}\bigg]\\
&\mu_{-1}(\widetilde{\zeta}_n;x,t)\bigg\}=-\mathrm{e}^{\mathrm{i}\nu_{-}(x,t)\sigma_3}\bigg(\begin{array}{cc} \frac{\mathrm{i}q_{-}}{\zeta^2_s} \\ 0  \end{array}\bigg)+\frac{1}{2\pi \mathrm{i}}\int_{\Sigma}\frac{(M^+(\xi;x,t)J(\xi;x,t))_{2}}{(\xi-\zeta_s)^2}d\xi,
\end{split}
\end{align}
where $s=1,2,\cdots,4N_1+2N_2$ and $\delta_{s,n}$ are the Kronecker $\delta$-symbol.

\end{prop}

\begin{proof}
Similar to the proof of proposition \eqref{P8} for the case of ZBCs. In order to regularize the RH problem established in proposition \eqref{P9} for the case of NZBCs, one has to subtract out the asymptotic values as $z\rightarrow0$ and $z\rightarrow\infty$ given by Eq. \eqref{T108} and the singular contributions. Then, the jump condition \eqref{T106} becomes
\begin{align}\label{T113}
\begin{split}
&M^-(z;x,t)-\frac{\mathrm{i}}{z}\mathrm{e}^{\mathrm{i}\nu_{-}(x,t)\sigma_3}\sigma_3Q_{-}-\mathrm{e}^{\mathrm{i}\nu_{-}(x,t)\sigma_3}-\\
&\sum_{n=1}^{4N_1+2N_2}\left[\frac{\mathop{\mathrm{P}_{-2}}\limits_{z=\zeta_n}M(z;x,t)}{(z-\zeta_n)^2}+\frac{\mathop{\mathrm{Res}}\limits_{z=\zeta_n}M(z;x,t)}{z-\zeta_n}+\frac{\mathop{\mathrm{P}_{-2}}\limits_{z=\widetilde{\zeta}_n}M(z;x,t)}{(z-\widetilde{\zeta}_n)^2}+\frac{\mathop{\mathrm{Res}}\limits_{z=\widetilde{\zeta}_n}M(z;x,t)}{z-\widetilde{\zeta}_n}\right]\\
&=M^+(z;x,t)-\frac{\mathrm{i}}{z}\mathrm{e}^{\mathrm{i}\nu_{-}(x,t)\sigma_3}\sigma_3Q_{-}-\mathrm{e}^{\mathrm{i}\nu_{-}(x,t)\sigma_3}-\\
&\sum_{n=1}^{4N_1+2N_2}\left[\frac{\mathop{\mathrm{P}_{-2}}\limits_{z=\zeta_n}M(z;x,t)}{(z-\zeta_n)^2}+\frac{\mathop{\mathrm{Res}}\limits_{z=\zeta_n}M(z;x,t)}{z-\zeta_n}+\frac{\mathop{\mathrm{P}_{-2}}\limits_{z=\widetilde{\zeta}_n}M(z;x,t)}{(z-\widetilde{\zeta}_n)^2}+\frac{\mathop{\mathrm{Res}}\limits_{z=\widetilde{\zeta}_n}M(z;x,t)}{z-\widetilde{\zeta}_n}\right]\\
&-M^+(z;x,t)J(z;x,t),
\end{split}
\end{align}
where $\mathop{\mathrm{P}_{-2}}\limits_{z=\zeta_n}M(z;x,t),\mathop{\mathrm{Res}}\limits_{z=\zeta_n}M(z;x,t),\mathop{\mathrm{P}_{-2}}\limits_{z=\widetilde{\zeta}_n}M(z;x,t),\mathop{\mathrm{Res}}\limits_{z=\widetilde{\zeta}_n}M(z;x,t)$ have given in Eq. \eqref{T105}. By using Plemelj's formula, one can obtain the solution \eqref{T109} with formula \eqref{T110} of the matrix RH problem. According to the symmetry reduction $\mu_{\pm}(z;x,t)=\frac{\mathrm{i}}{z}\mu_{\pm}\big(-\frac{q^2_0}{z};x,t\big)\sigma_3Q_{\pm}$, we can obtain
\begin{align}\nonumber
&\mu_{-2}(z;x,t)=\frac{\mathrm{i}}{z}q_{-}\mu_{-1}\Big(-\frac{q^2_0}{z};x,t\Big),
\end{align}
and take derivative respect to $z$ on both sides of above equation at the same time. Then, let $z=\zeta_n$, one have following formula
\begin{align}\nonumber
\mu'_{-2}(\zeta_n;x,t)=-\frac{\mathrm{i}q_{-}}{\zeta^2_n}\mu_{-1}(\widetilde{\zeta}_n;x,t)+\frac{\mathrm{i}q_{-}q^2_0}{\zeta^3_n}\mu'_{-1}(\widetilde{\zeta}_n;x,t).
\end{align}

From Eq. \eqref{T104}, when we take $z=\zeta_s$, $\mu_{-2}(\zeta_s;x,t)$ is the second column element of matrix $M(\zeta_s;x,t)$, that is
\begin{align}\label{T114}
\begin{split}
\mu_{-2}(\zeta_s;x,t)=&\bigg(\begin{array}{cc} \frac{\mathrm{i}}{\zeta_s}\mathrm{e}^{\mathrm{i}\nu_{-}(x,t)}q_{-} \\ \mathrm{e}^{-\mathrm{i}\nu_{-}(x,t)} \end{array}\bigg)+\frac{1}{2\pi \mathrm{i}}\int_{\Sigma}\frac{(M^+(\xi;x,t)J(\xi;x,t))_{2}}{\xi-\zeta_s}d\xi+\\
&\sum^{4N_1+2N_2}_{n=1}\bigg\{\widetilde{C}_n(\zeta_s)\bigg[\mu'_{-1}(\widetilde{\zeta}_n;x,t)+\bigg(\widetilde{D}_n+\frac{1}{\zeta_s-\widetilde{\zeta}_n}\bigg)\mu_{-1}(\widetilde{\zeta}_n;x,t)\bigg]\bigg\},
\end{split}
\end{align}
where $s=1,2,\cdots,4N_1+2N_2$ and $\zeta_s$ is equivalent to $\zeta_n$. Then, taking the derivative respect to $z$ $(\text{here }z=\zeta_s)$ on both sides of Eq. \eqref{T114} at the same time, we have
\begin{align}\label{T115}
\begin{split}
\mu'_{-2}(\zeta_s;x,t)=&\bigg(\begin{array}{cc} -\frac{\mathrm{i}}{\zeta^2_s}q_{-}\mathrm{e}^{\mathrm{i}\nu_{-}(x,t)} \\ 0 \end{array}\bigg)+\frac{1}{2\pi \mathrm{i}}\int_{\Sigma}\frac{(M^+(\xi;x,t)J(\xi;x,t))_{2}}{(\xi-\zeta_s)^2}d\xi-\sum^{4N_1+2N_2}_{n=1}\widetilde{C}_n(\zeta_s)\\
&\bigg\{\frac{1}{\zeta_s-\widetilde{\zeta}_n}\bigg[\mu'_{-1}(\widetilde{\zeta}_n;x,t)+\bigg(\widetilde{D}_n+\frac{2}{\zeta_s-\widetilde{\zeta}_n}\bigg)\mu_{-1}(\widetilde{\zeta}_n;x,t)\bigg]\bigg\}.
\end{split}
\end{align}

After that, substituting Eqs. \eqref{T114}-\eqref{T115} into Eq. \eqref{T111} to eliminate $\mu_{-2}(\zeta_s;x,t)$ and $\mu'_{-2}(\zeta_s;x,t)$, by merging and simplifying, one can derived Eq. \eqref{T112}. Completing the proof.

\end{proof}

\subsubsection{Reconstruction Formula of the Potential with NZBCs and Double poles}
From the solution \eqref{T109} of the matrix RH problem, we have
\begin{align}\label{T116}
M(z;x,t)=\mathrm{e}^{\mathrm{i}\nu_{-}(x,t)\sigma_3}+\frac{M^{[1]}(x,t)}{z}+O\Big(\frac{1}{z^2}\Big),\text{ as }z\rightarrow\infty,
\end{align}
where
\begin{small}
\begin{align}\label{T117}
\begin{split}
M^{[1]}(x,t)=&\mathrm{i}\mathrm{e}^{\mathrm{i}\nu_{-}(x,t)\sigma_3}\sigma_3Q_{-}-\frac{1}{2\pi \mathrm{i}}\int_{\Sigma}M^+(\xi;x,t)J(\xi;x,t)d\xi+\sum^{4N_1+2N_2}_{n=1}\big\{A[\zeta_n]\mathrm{e}^{-2\mathrm{i}\theta(\zeta_n;x,t)}\\
&\big[\mu'_{-2}(\zeta_n;x,t)+D_n\mu_{-2}(\zeta_n;x,t)\big],A[\widetilde{\zeta}_n]\mathrm{e}^{2\mathrm{i}\theta(\widetilde{\zeta}_n;x,t)}\big[\mu'_{-1}(\widetilde{\zeta}_n;x,t)+\widetilde{D}_n\mu_{-1}(\widetilde{\zeta}_n;x,t)\big]\big\}.
\end{split}
\end{align}
\end{small}

Substituting Eq. \eqref{T116} into Eq. \eqref{T73} and matching $O(z)$ term, we have
\begin{align}\nonumber
O(z):\frac{\mathrm{i}}{2}\big[M^{[1]}(x,t),\sigma_3\big]+\frac{1}{2}\big[\sigma_3Q_{\pm}\mathrm{e}^{\mathrm{i}\nu_{-}(x,t)\sigma_3},\sigma_3\big]=(Q-Q_{\pm})\mathrm{e}^{\mathrm{i}\nu_{-}(x,t)\sigma_3},
\end{align}
then, by expanding the above equation, one can find the reconstruction formula of the double poles solution (potential) for the TOFKN with NZBCs as follows
\begin{align}\label{T118}
q(x,t)=-\mathrm{i}\mathrm{e}^{\mathrm{i}\nu_{-}(x,t)}M^{[1]}_{12}(x,t),
\end{align}
where $M^{[1]}_{12}(x,t)$ represents the first row and second column element of the matrix $M^{[1]}(x,t)$, and
\begin{small}
\begin{align}\label{T119}
\begin{split}
M^{[1]}_{12}(x,t)=&\mathrm{i}\mathrm{e}^{\mathrm{i}\nu_{-}(x,t)}q_{-}-\frac{1}{2\pi \mathrm{i}}\int_{\Sigma}\big[M^{+}(\xi;x,t)J(\xi;x,t)\big]_{12}d\xi+\\
&\sum^{4N_1+2N_2}_{n=1}\big\{A[\widetilde{\zeta}_n]\mathrm{e}^{2\mathrm{i}\theta(\widetilde{\zeta}_n;x,t)}\big[\mu'_{-11}(\widetilde{\zeta}_n;x,t)+\widetilde{D}_n\mu_{-11}(\widetilde{\zeta}_n;x,t)\big]\big\},
\end{split}
\end{align}
\end{small}
when taking row vector $\alpha=\big(\alpha^{(1)},\alpha^{(2)}\big)$ and column vector $\gamma=(\gamma^{(1)},\gamma^{(2)})^{\mathrm{T}}$, where
\begin{align}\label{T120}
\begin{split}
&\alpha^{(1)}=\Big(A[\widetilde{\zeta}_n]\mathrm{e}^{2\mathrm{i}\theta(\widetilde{\zeta}_n;x,t)}\Big)_{1\times(4N_1+2N_2)},\quad\alpha^{(2)}=\Big(A[\widetilde{\zeta}_n]\mathrm{e}^{2\mathrm{i}\theta(\widetilde{\zeta}_n;x,t)}\widetilde{D}_n\Big)_{1\times(4N_1+2N_2)},\\
&\gamma^{(1)}=\Big(\mu'_{-11}(\widetilde{\zeta}_n;x,t)\Big)_{1\times(4N_1+2N_2)},\qquad\gamma^{(2)}=\Big(\mu_{-11}(\widetilde{\zeta}_n;x,t)\Big)_{1\times(4N_1+2N_2)},
\end{split}
\end{align}
we can obtain a more concise reconstruction formulation of the double-pole solution (potential) for the TOFKN with NZBCs as follows
\begin{align}\label{T121}
q(x,t)=\mathrm{e}^{\mathrm{i}\nu_{-}(x,t)}\bigg[q_{-}\mathrm{e}^{\mathrm{i}\nu_{-}(x,t)}-\mathrm{i}\alpha\gamma+\frac{1}{2\pi}\int_{\Sigma}(M^+(\xi;x,t)J(\xi;x,t))_{12}d\xi\bigg].
\end{align}

\subsubsection{Trace Formulae and Theta Condition with NZBCs and Double poles}
The so-called trace formulae are that the scattering coefficients $s_{11}(z)$ and $s_{22}(z)$ are formulated in terms of the discrete spectrum $Z$ and reflection coefficients $\rho(z)$ and $\tilde{\rho}(z)$. We know that $s_{11}(z),s_{22}(z)$ are analytic on $D^{+},D^{-}$, respectively. The discrete spectral points $\zeta_n$'s are the double zeros of $s_{11}(z)$, while $\widetilde{\zeta}_n$'s are the double zeros of $s_{22}(z)$. Define the functions $\beta^{\pm}(z)$ as follows:
\begin{align}\label{T122}
\begin{split}
&\beta^{+}(z)=s_{11}(z)\prod^{4N_1+2N_2}_{n=1}\Bigg(\frac{z-\widetilde{\zeta}_n}{z-\zeta_n}\Bigg)^2\mathrm{e}^{\mathrm{i}\nu},\\
&\beta^{-}(z)=s_{22}(z)\prod^{4N_1+2N_2}_{n=1}\Bigg(\frac{z-\zeta_n}{z-\widetilde{\zeta}_n}\Bigg)^2\mathrm{e}^{-\mathrm{i}\nu}.
\end{split}
\end{align}

Then, $\beta^{+}(z)$ and $\beta^{-}(z)$ are analytic and have no zero in $D^{+}$ and $D^{-}$, respectively. Furthermore, we have the relation $\beta^{+}(z)\beta^{-}(z)=s_{11}(z)s_{22}(z)$ and the asymptotic behaviors $\beta^{\pm}(z)\rightarrow1,\text{ as }z\rightarrow\infty$.

According to $\mathrm{det}S(z)=s_{11}s_{22}-s_{21}s_{12}=1$, we can derive
\begin{align}\nonumber
\frac{1}{s_{11}(z)s_{22}(z)}=1-\frac{s_{21}(z)s_{12}(z)}{s_{11}(z)s_{22}(z)}=1-\rho(z)\tilde{\rho}(z),
\end{align}
by taking logarithm on both sides of the above equation at the same time, we have
\begin{align}\nonumber
-\mathrm{log}(s_{11}(z)s_{22}(z))=\mathrm{log}[1-\rho(z)\tilde{\rho}(z)]\Rightarrow\mathrm{log}[\beta^{+}(z)\beta^{-}(z)]=-\mathrm{log}[1-\rho(z)\tilde{\rho}(z)],
\end{align}
then employing the Plemelj' formula such that one has
\begin{align}\label{T123}
\mathrm{log}\beta^{\pm}(z)=\mp\frac{1}{2\pi \mathrm{i}}\int_{\Sigma}\frac{\mathrm{log}[1-\rho(z)\tilde{\rho}(z)]}{\xi-z}d\xi,\quad z\in D^{\pm}.
\end{align}

Substituting Eq. \eqref{T123} into  Eq. \eqref{T122}, we can obtain the trace formulae
\begin{align}\label{T124}
\begin{split}
&s_{11}(z)=\mathrm{exp}\Bigg(-\frac{1}{2\pi \mathrm{i}}\int_{\Sigma}\frac{\mathrm{log}[1-\rho(z)\tilde{\rho}(z)]}{\xi-z}d\xi\Bigg)\prod^{4N_1+2N_2}_{n=1}\Bigg(\frac{z-\zeta_n}{z-\widetilde{\zeta}_n}\Bigg)^2\mathrm{e}^{-\mathrm{i}\nu},\\
&s_{22}(z)=\mathrm{exp}\Bigg(\frac{1}{2\pi \mathrm{i}}\int_{\Sigma}\frac{\mathrm{log}[1-\rho(z)\tilde{\rho}(z)]}{\xi-z}d\xi\Bigg)\prod^{4N_1+2N_2}_{n=1}\Bigg(\frac{z-\widetilde{\zeta}_n}{z-\zeta_n}\Bigg)^2\mathrm{e}^{\mathrm{i}\nu}.
\end{split}
\end{align}

As $z\rightarrow0$, we consider formulas \eqref{T124} and the asymptotic behavior of the scattering matrix \eqref{T94}, then the so-called theta condition is obtained. That is to say, there exists $l\in\mathbb{Z}$ such that
\begin{align}\label{T125}
\mathrm{arg}\left(\frac{q_{-}}{q_{+}}\right)+2\nu=16\sum^{N_1}_{n=1}\mathrm{arg}(z_n)+8\sum^{N_2}_{m=1}\mathrm{arg}(\omega_m)+2\pi l+\frac{1}{2\pi}\int_{\Sigma}\frac{\mathrm{log}[1-\rho(\xi)\tilde{\rho}(\xi)]}{\xi}d\xi.
\end{align}

\subsubsection{Reflectionless Potential with NZBCs: Double-Pole Solutions}
For the case of the reflectionless potential: $\rho(z)=\tilde{\rho}(z)=0$, the part jump matrix $J(z;x,t)$ in Eq. \eqref{T107} can be simplified as $J(z;x,t)=0_{2\times2}$. From the Volterra integral equation (72), one can derive $\psi_{\pm}(q_0;x,t)=Y_{\pm}q_0$. Combining with the definition of scattering matrix, one has $S(q_0)=I$ and $q_{+}=q_{-}$. From the theta condition, there exists $i\in\mathbb{Z}$ lead to
\begin{align}\label{T126}
\nu=8\sum^{N_1}_{n=1}\mathrm{arg}(z_n)+4\sum^{N_2}_{m=1}\mathrm{arg}(\omega_m)+\pi i.
\end{align}

Then Eqs. \eqref{T112} and \eqref{T121} with $J(z;x,t)=0_{2\times2}$ become
\begin{align}\label{T127}
\begin{split}
&\sum^{4N_1+2N_2}_{n=1}\bigg\{\widetilde{C}_n(\zeta_s)\mu'_{-11}(\widetilde{\zeta}_n;x,t)+\bigg[\widetilde{C}_n(\zeta_s)\bigg(\widetilde{D}_n+\frac{1}{\zeta_s-\widetilde{\zeta}_n}\bigg)-\frac{iq_{-}}{\zeta_s}\delta_{s,n}\bigg]\mu_{-11}(\widetilde{\zeta}_n;x,t)\bigg\}=\\
&-\mathrm{e}^{\mathrm{i}\nu_{-}(x,t)}\frac{\mathrm{i}q_{-}}{\zeta_s},\\
&\sum^{4N_1+2N_2}_{n=1}\bigg\{\bigg(\frac{\widetilde{C}_n(\zeta_s)}{\zeta_s-\widetilde{\zeta}_n}+\frac{iq_{-}q^2_0}{\zeta^3_s}\delta_{s,n}\bigg)\mu'_{-11}(\widetilde{\zeta}_n;x,t)+\bigg[\frac{\widetilde{C}_n(\zeta_s)}{\zeta_s-\widetilde{\zeta}_n}\bigg(\widetilde{D}_n+\frac{2}{\zeta_s-\widetilde{\zeta}_n}\bigg)-\frac{\mathrm{i}q_{-}}{\zeta^2_s}\delta_{s,n}\bigg]\\
&\mu_{-11}(\widetilde{\zeta}_n;x,t)\bigg\}=-\mathrm{e}^{\mathrm{i}\nu_{-}(x,t)}\frac{\mathrm{i}q_{-}}{\zeta^2_s},
\end{split}
\end{align}
and
\begin{align}\label{T128}
q(x,t)=\mathrm{e}^{\mathrm{i}\nu_{-}(x,t)}\left(q_{-}\mathrm{e}^{\mathrm{i}\nu_{-}(x,t)}-\mathrm{i}\alpha\gamma\right).
\end{align}

\begin{thm}
The explicit expression for the double-pole solution of the TOFKN \eqref{T1} with NZBCs is given by determinant formula
\begin{align}\label{T129}
q(x,t)=q_{-}\bigg(\frac{\mathrm{det}(\widetilde{R})}{\mathrm{det}(\widetilde{G})}\bigg)^2\frac{\mathrm{det}(G)}{\mathrm{det}(R)}\bigg(1+\frac{\mathrm{det}(G)}{\mathrm{det}(R)}\bigg),
\end{align}
where
\begin{align}\label{T130}
\begin{split}
&R=\bigg(\begin{array}{cc} 0 & \alpha \\ \tau & G \end{array}\bigg),\quad \widetilde{R}=\bigg(\begin{array}{cc} 0 & \alpha \\ \tau & \widetilde{G} \end{array}\bigg),\quad\tau=\bigg(\begin{array}{cc} \tau^{(1)} \\ \tau^{(2)} \end{array}\bigg),\\
&\tau^{(1)}=\bigg(\frac{1}{\zeta_s}\bigg)_{(4N_1+2N_2)\times1},\quad \tau^{(2)}=\bigg(\frac{1}{\zeta^2_s}\bigg)_{(4N_1+2N_2)\times1},
\end{split}
\end{align}
and the $(8N_1+4N_2)\times(8N_1+4N_2)$ partitioned matrix $G=\bigg(\begin{array}{cc} G^{(1,1)} & G^{(1,2)} \\ G^{(2,1)} & G^{(2,2)} \end{array}\bigg)$ with $G^{(i,j)}=\big(g^{(i,j)}_{s,n}\big)_{(4N_1+2N_2)\times(4N_1+2N_2)}\,(i,j=1,2)$ given by
\begin{align}\nonumber
\begin{split}
&g^{(1,1)}_{s,n}=\widetilde{C}_n(\zeta_s),\quad g^{(1,2)}_{s,n}=\widetilde{C}_n(\zeta_s)\bigg(\widetilde{D}_n+\frac{1}{\zeta_s-\widetilde{\zeta}_n}\bigg)-\frac{\mathrm{i}q_{-}}{\zeta_s}\delta_{s,n},\\
&g^{(2,1)}_{s,n}=\frac{\widetilde{C}_n(\zeta_s)}{\zeta_s-\widetilde{\zeta}_n}+\frac{\mathrm{i}q_{-}q^2_0}{\zeta^3_s}\delta_{s,n},\quad g^{(2,2)}_{s,n}=\frac{\widetilde{C}_n(\zeta_s)}{\zeta_s-\widetilde{\zeta}_n}\bigg(\widetilde{D}_n+\frac{2}{\zeta_s-\widetilde{\zeta}_n}\bigg)-\frac{\mathrm{i}q_{-}}{\zeta^2_s}\delta_{s,n},
\end{split}
\end{align}
and the $(8N_1+4N_2)\times(8N_1+4N_2)$ partitioned matrix $\widetilde{G}=\bigg(\begin{array}{cc} \widetilde{G}^{(1,1)} & \widetilde{G}^{(1,2)} \\ \widetilde{G}^{(2,1)} & \widetilde{G}^{(2,2)} \end{array}\bigg)$ with $\widetilde{G}^{(i,j)}=\big(\widetilde{g}^{(i,j)}_{s,n}\big)_{(4N_1+2N_2)\times(4N_1+2N_2)}\,(i,j=1,2)$ given by
\begin{align}\nonumber
\begin{split}
&\widetilde{g}^{(1,1)}_{s,n}=\frac{\lambda(\zeta_s)}{\lambda(\widetilde{\zeta}_n)}\widetilde{C}_n(\zeta_s),\quad \widetilde{g}^{(1,2)}_{s,n}=\frac{\lambda(\zeta_s)}{\lambda(\widetilde{\zeta}_n)}\widetilde{C}_n(\zeta_s)\bigg(\widetilde{D}_n+\frac{1}{\zeta_s-\widetilde{\zeta}_n}-\frac{\lambda'(\widetilde{\zeta}_n)}{\lambda(\widetilde{\zeta}_n)}\bigg)-\frac{\mathrm{i}q_{-}}{\zeta_s}\delta_{s,n},\\
&\widetilde{g}^{(2,1)}_{s,n}=\frac{\widetilde{C}_n(\zeta_s)}{\lambda(\widetilde{\zeta}_n)}\bigg(\frac{\lambda(\zeta_s)}{\zeta_s-\widetilde{\zeta}_n}-\lambda'(\zeta_s)\bigg)+\frac{\mathrm{i}q_{-}q^2_0}{\zeta^3_s}\delta_{s,n},\\
&\widetilde{g}^{(2,2)}_{s,n}=\frac{}{}\bigg[\frac{\lambda(\zeta_s)}{\zeta_s-\widetilde{\zeta}_n}\bigg(\widetilde{D}_n+\frac{2}{\zeta_s-\widetilde{\zeta}_n}-\frac{\lambda'(\widetilde{\zeta}_n)}{\lambda(\widetilde{\zeta}_n)}\bigg)-\lambda'(\zeta_s)\bigg(\widetilde{D}_n+\frac{1}{\zeta_s-\widetilde{\zeta}_n}-\frac{\lambda'(\widetilde{\zeta}_n)}{\lambda(\widetilde{\zeta}_n)}\bigg)\bigg]-\frac{\mathrm{i}q_{-}}{\zeta^2_s}\delta_{s,n}.
\end{split}
\end{align}

\end{thm}

\begin{proof}
From the Eqs. \eqref{T120}, \eqref{T127} and \eqref{T128}, the reflectionless potential is deduced by determinants:
\begin{align}\label{T131}
q(x,t)=q_{-}\mathrm{e}^{2\mathrm{i}\nu_{-}(x,t)}\bigg(1+\frac{\mathrm{det}(R)}{\mathrm{det}(G)}\bigg).
\end{align}

However, this formula \eqref{T131} is implicit since $\nu_{-}(x,t)$ is included. One needs to derive an explicit form for the reflectionless potential. From the trace formulae \eqref{T124} and Volterra integral equation \eqref{T76} as $x\rightarrow-\infty$, one derives that
\begin{align}\label{T132}
\begin{split}
M(z;x,t)=&Y_{-}(z)+\lambda(z)\sum^{4N_1+2N_2}_{n=1}\left[\frac{\mathop{\mathrm{P}_{-2}}\limits_{z=\zeta_n}\big(M(\lambda(z);x,t)/\lambda(z)\big)}{(z-\zeta_n)^2}+\frac{\mathop{\mathrm{Res}}\limits_{z=\zeta_n}\big(M(\lambda(z);x,t)/\lambda(z)\big)}{z-\zeta_n}\right.\\
&\left.+\frac{\mathop{\mathrm{P}_{-2}}\limits_{z=\widetilde{\zeta}_n}\big(M(\lambda(z);x,t)/\lambda(z)\big)}{(z-\widetilde{\zeta}_n)^2}+\frac{\mathop{\mathrm{Res}}\limits_{z=\widetilde{\zeta}_n}\big(M(\lambda(z);x,t)/\lambda(z)\big)}{z-\widetilde{\zeta}_n}\right],
\end{split}
\end{align}
which can yield the $\gamma$ given by Eq. \eqref{T120} exactly. Then, substituting $\gamma$ into the formula of the potential, one yields
\begin{align}\label{T133}
q(x,t)=q_{-}\mathrm{e}^{\mathrm{i}\nu_{-}(x,t)}\bigg(\mathrm{e}^{\mathrm{i}\nu_{-}(x,t)}+\frac{\mathrm{det}(\widetilde{R})}{\mathrm{det}(\widetilde{G})}\bigg),
\end{align}
then, by combining Eq. \eqref{T131} with Eq. \eqref{T133}, we can obtain Eq. \eqref{T129}, and complete the proof.

\end{proof}

For example, we exhibit the double-pole solutions which contain three kind of types for the TOFKN with NZBCs:\\

$\bullet$ When taking parameters $N_1=0,N_2=1,q_{\pm}=1,\omega_1=\mathrm{e}^{\frac{\pi}{4}\mathrm{i}},A[\omega_1]=\mathrm{i},B[\omega_1]=1+(1-\sqrt{2})\mathrm{i}$, we can obtain the explicit double-pole dark-bright soliton solution $q(x,t)=\frac{E_1}{E_2}$ with
\begin{small}
\begin{align}\nonumber
\begin{split}
E_1=&-\mathrm{i}\big\{(3\sqrt{2}\mathrm{i}t+3\sqrt{2}t+\sqrt{2}-1-\mathrm{i})^2\big(\mathrm{e}^{2x+7t}\big)^2+\big[2(12\sqrt{2}\mathrm{i}t-36\mathrm{i}t^2+2\sqrt{2}\mathrm{i}-12\mathrm{i}t-4\mathrm{i})(3\sqrt{2}\mathrm{i}t\\
&+3\sqrt{2}t+\sqrt{2}-1-\mathrm{i})\mathrm{e}^{x+\frac{7}{2}t}+2(-3\sqrt{2}\mathrm{i}t-\sqrt{2}\mathrm{i}+3\sqrt{2}t-1+\mathrm{i})(3\sqrt{2}\mathrm{i}t+3\sqrt{2}t+\sqrt{2}-1-\mathrm{i})\big]\\
&\mathrm{e}^{2x+7t}+(12\sqrt{2}\mathrm{i}t-36\mathrm{i}t^2+2\sqrt{2}\mathrm{i}-12\mathrm{i}t-4\mathrm{i})^2\big(\mathrm{e}^{x+\frac72t}\big)^2+2(-3\sqrt{2}\mathrm{i}t-\sqrt{2}\mathrm{i}+3\sqrt{2}t-1+\mathrm{i})\\
&(12\sqrt{2}\mathrm{i}t-36\mathrm{i}t^2+2\sqrt{2}\mathrm{i}-12\mathrm{i}t-4\mathrm{i})\mathrm{e}^{x+\frac72t}+(-3\sqrt{2}\mathrm{i}t-\sqrt{2}\mathrm{i}+3\sqrt{2}t-1+\mathrm{i})^2\big\}\big[(12\sqrt{2}\mathrm{i}t-\\
&36\mathrm{i}t^2+2\sqrt{2}\mathrm{i}-12\mathrm{i}t+4\sqrt{2}-24t-4)\mathrm{e}^{2x+7t}+(2\sqrt{2}\mathrm{i}-6\sqrt{2}t-2\sqrt{2}+2)\mathrm{e}^{x+\frac72t}+\\
&(2\sqrt{2}\mathrm{i}-6\sqrt{2}t+2)\mathrm{e}^{3x+\frac{21}{2}t}+\mathrm{i}\mathrm{e}^{4x+14t}+\mathrm{i}\big],
\end{split}
\end{align}
\end{small}
\begin{small}
\begin{align}\nonumber
\begin{split}
E_2=&\big[(12\sqrt{2}t-36t^2+2\sqrt{2}-12t-4)\mathrm{e}^{2x+7t}+(-6\sqrt{2}t-2\sqrt{2}+2)\mathrm{e}^{x+\frac72t}+(6\sqrt{2}t-2)\mathrm{e}^{3x+\frac{21}{2}t}-\\
&\mathrm{e}^{4x+14t}-1\big]\big[(3\sqrt{2}\mathrm{i}t+\sqrt{2}\mathrm{i}-3\sqrt{2}t+1-\mathrm{i})\mathrm{e}^{2x+7t}+(12\sqrt{2}\mathrm{i}t-36\mathrm{i}t^2+2\sqrt{2}\mathrm{i}-12\mathrm{i}t+\\
&2\sqrt{2}-2-2\mathrm{i}-12t)\mathrm{e}^{x+\frac72t}-3\sqrt{2}\mathrm{i}t-3\sqrt{2}t-\sqrt{2}+1+\mathrm{i}\big]^2,
\end{split}
\end{align}
\end{small}
and give out relevant plots in Fig. \ref{F5}. Fig. \ref{F5} (a) and (b) exhibit the three-dimensional and density diagrams for the exact double-pole dark-bright soliton solution of the TOFKN with NZBCs. Fig. \ref{F5} (c) displays the distinct profiles of the exact double-pole soliton for $t=\pm3,0$. It is a semi rational soliton, which is different from the simple pole solution usually expressed by exponential function, even if the double-pole dark-bright soliton solution shows the interaction of dark soliton and bright soliton.

Compared with the classical second-order flow Kaup-Newell system which also be called the DNLS equation in Ref. \cite{Zhangg(2020)}, according to the density diagrams of the double-pole dark-bright soliton solutions for the TOFKN and the DNLS equation in Ref. \cite{Zhangg(2020)} shows that the trajectories of solutions are different obviously. Moreover, by observing the form of exact double-pole dark-bright soliton solution, the average wave velocity of the double-pole dark-bright soliton solution of the DNLS equation from Ref. \cite{Zhangg(2020)} is about $-2$, while the average wave velocity of the double-pole dark-bright soliton solution in this paper is about $-7/2$. It further validates the introduction of third-order dispersion and quintic nonlinear term of Kaup-Newell systems can affect the trajectories and the speed of solutions.

\begin{rem}
The parameter selection of the double-pole dark-bright soliton solution shown in Fig. \ref{F5}. is consistent with that in reference \cite{Zhangg(2020)}.
\end{rem}

\begin{figure}[htbp]
\centering
\subfigure[]{
\begin{minipage}[t]{0.33\textwidth}
\centering
\includegraphics[height=4.5cm,width=4.5cm]{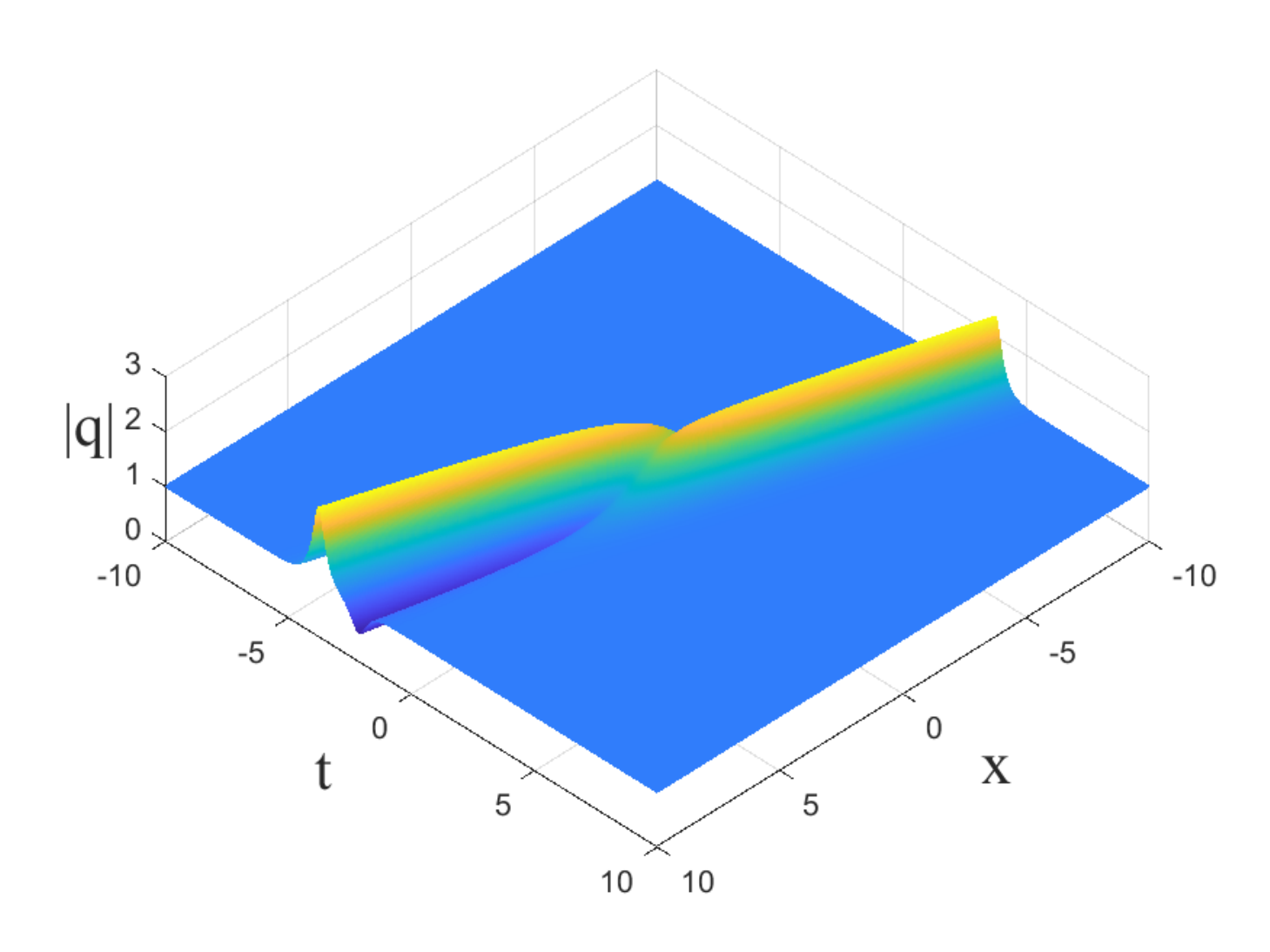}
%\caption{fig1}
\end{minipage}
}%
\subfigure[]{
\begin{minipage}[t]{0.33\textwidth}
\centering
\includegraphics[height=4.5cm,width=4.5cm]{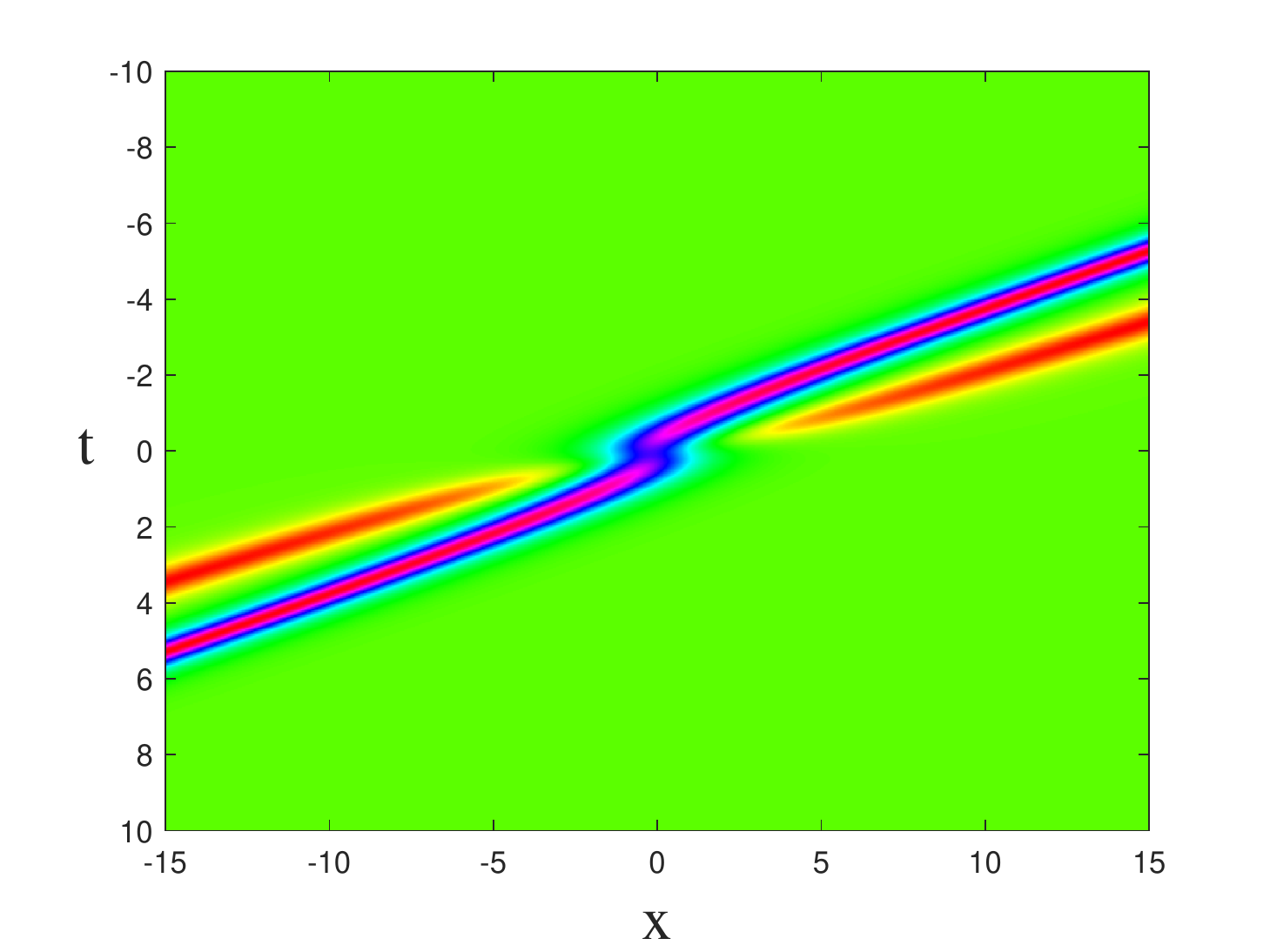}
%\caption{fig2}
\end{minipage}%
}%
\subfigure[]{
\begin{minipage}[t]{0.33\textwidth}
\centering
\includegraphics[height=4.5cm,width=4.5cm]{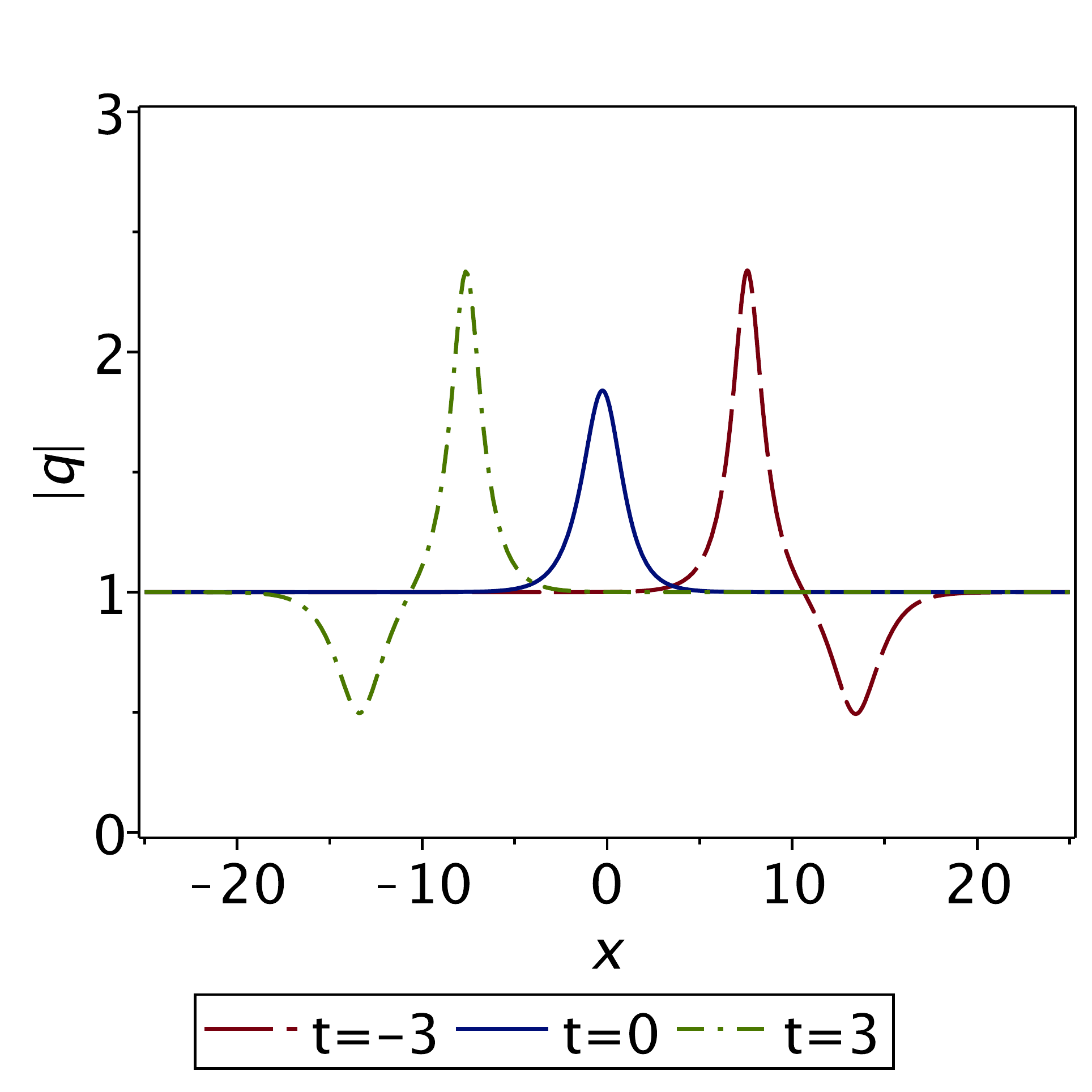}
%\caption{fig1}
\end{minipage}
}%
\centering
\caption{(Color online) The double-pole dark-bright soliton solution of TOFKN \eqref{T1} with NZBCs and $N_1=0,N_2=1,q_{\pm}=1,\omega_1=\mathrm{e}^{\frac{\pi}{4}\mathrm{i}},A[\omega_1]=\mathrm{i},B[\omega_1]=1+(1-\sqrt{2})\mathrm{i}$. (a) The three-dimensional plot; (b) The density plot; (c) The sectional drawings at $t=-3$ (dashed line), $t=0$ (solid line), and $t =3$ (dash-dot line).}
\label{F5}
\end{figure}

$\bullet$ When taking parameters $N_1=1,N_2=0,q_{\pm}=1,\zeta_1=2\mathrm{e}^{\frac{\pi}{4}\mathrm{i}},A[\zeta_1]=B[\zeta_1]=\mathrm{i}$, we can obtain the explicit double-pole breather-breather solution and give out relevant plots in Fig. \ref{F6}. Figs. \ref{F6} (a) and (b) exhibit the three-dimensional and density diagrams for the exact double-pole breather-breather solution of the TOFKN with NZBCs. Fig. \ref{F6} (c) displays the distinct profiles of the exact double-pole breather-breather solution for $t=\pm5,0$. Moreover, from the density plot Fig. \ref{F6} (b), we can find that the propagation of the double-pole breather-breather solution is localized in space $x$ and periodic in time $t$. Obviously, from the Fig. \ref{F6} (a) and (b), we can find that the double-pole breather-breather solution breathes about 10 times at $t=[-3, 3]$ and collides once at near $t=0$.

\begin{figure}[htbp]
\centering
\subfigure[]{
\begin{minipage}[t]{0.33\textwidth}
\centering
\includegraphics[height=4.5cm,width=4.5cm]{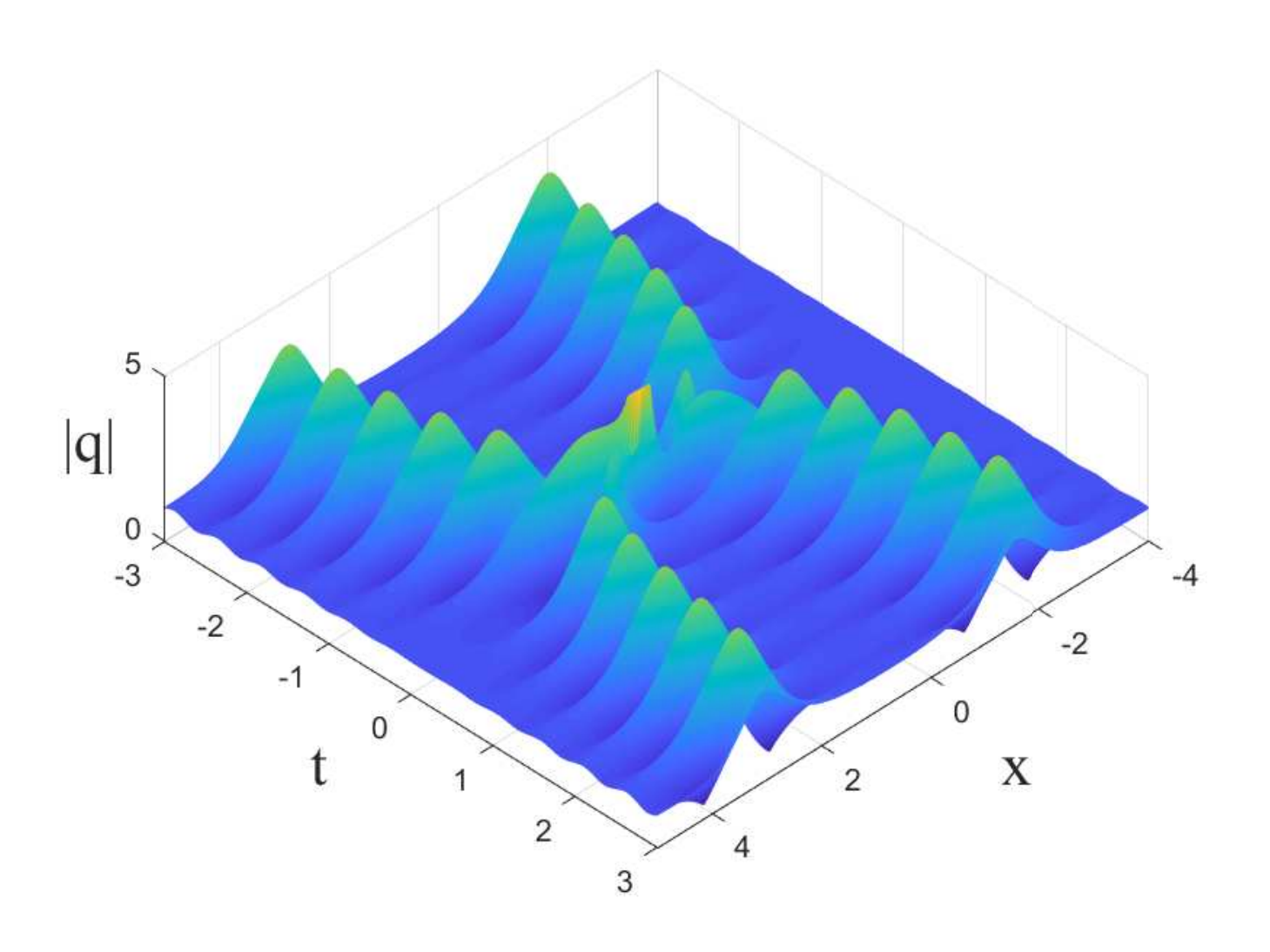}
%\caption{fig1}
\end{minipage}
}%
\subfigure[]{
\begin{minipage}[t]{0.33\textwidth}
\centering
\includegraphics[height=4.5cm,width=4.5cm]{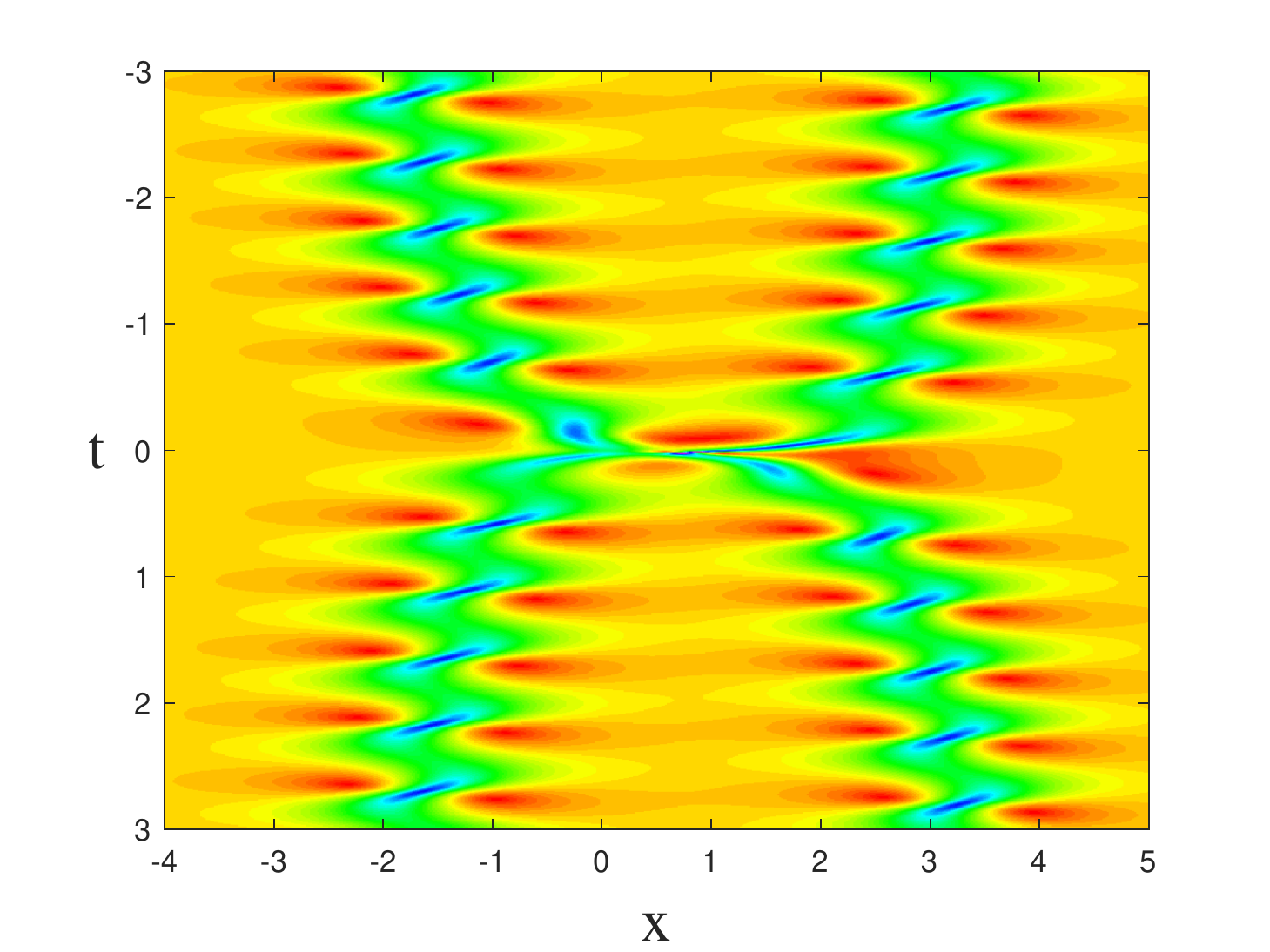}
%\caption{fig2}
\end{minipage}%
}%
\subfigure[]{
\begin{minipage}[t]{0.33\textwidth}
\centering
\includegraphics[height=4.5cm,width=4.5cm]{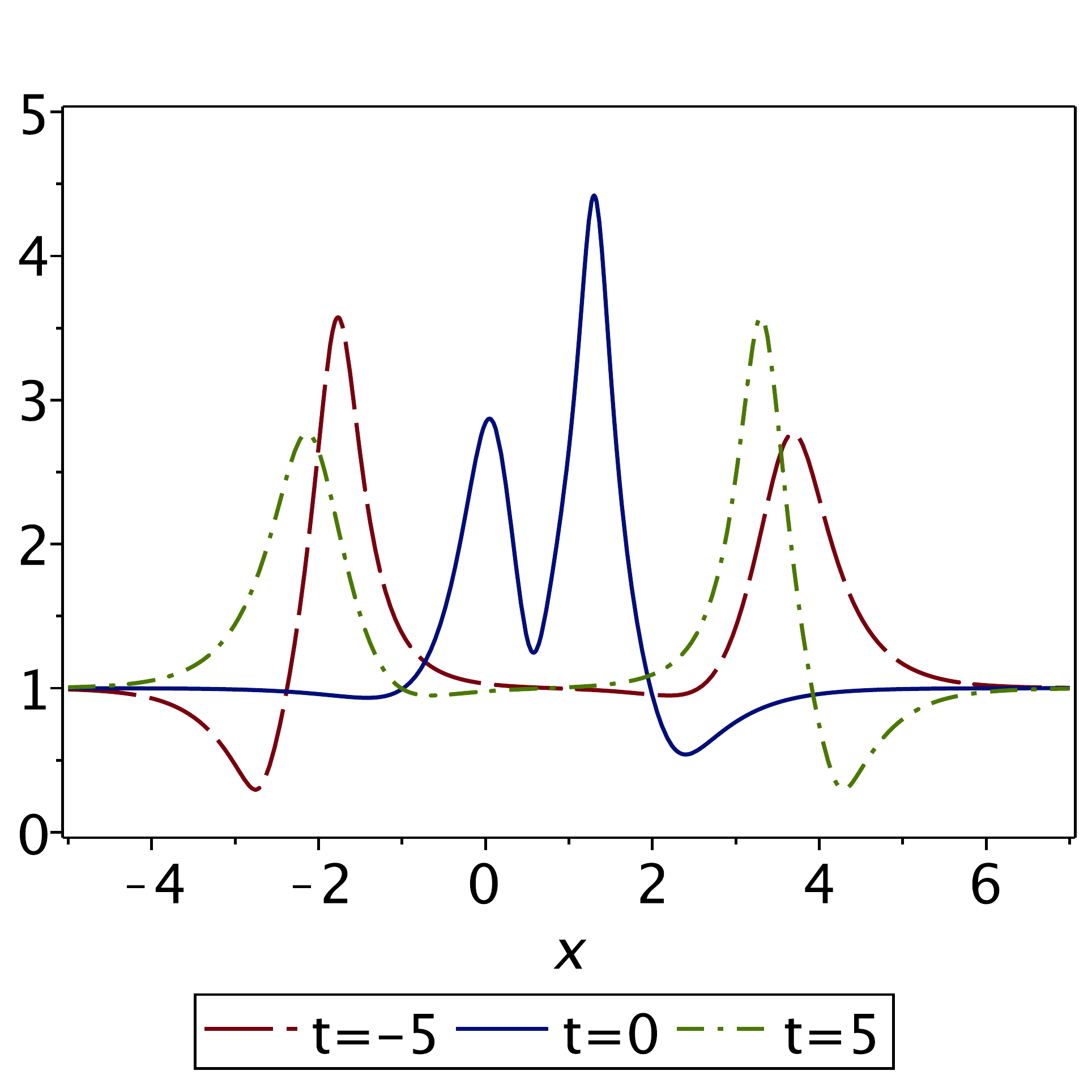}
%\caption{fig1}
\end{minipage}
}%
\centering
\caption{(Color online) The double-pole breather-breather solution of TOFKN \eqref{T1} with NZBCs and $N_1=1,N_2=0,q_{\pm}=1,\zeta_1=2\mathrm{e}^{\frac{\pi}{4}\mathrm{i}},A[\zeta_1]=B[\zeta_1]=\mathrm{i}$. (a) The three-dimensional plot; (b) The density plot; (c) The sectional drawings at $t=-5$ (dashed line), $t=0$ (solid line), and $t=5$ (dash-dot line).}
\label{F6}
\end{figure}

$\bullet$ When taking parameters $N_1=1,N_2=1,q_{\pm}=1,\zeta_1=2\mathrm{e}^{\frac{\pi}{4}\mathrm{i}},\omega_1=\mathrm{e}^{\frac{\pi}{4}\mathrm{i}},A[\zeta_1]=B[\zeta_1]=\mathrm{i},A[\omega_1]=B[\omega_1]=\mathrm{i}$, we can obtain the explicit double-pole breather-breather-dark-bright solution and give out relevant plots in Fig. \ref{F7}. Figs. \ref{F7} (a) and (b) exhibit the three-dimensional and density diagrams for the exact double-pole breather-breather-dark-bright solution of the TOFKN with NZBCs, which is equivalent to a combination of breather-breather solution and dark-bright soliton solutions. Fig. \ref{F7} (c) displays the distinct profiles of the exact double-pole breather-breather-dark-bright solution for $t=\pm5,0$. It is obvious that the double-pole breather-breather-dark-bright solution is an interaction between the double-pole breather-breather solution and the double-pole dark-bright soliton solution. Moreover, We can find that the amplitude at $t=0$ in Fig. \ref{F7} (c) is between that the amplitude at $t=0$ in Fig. \ref{F5} (c) and the amplitude at $t=0$ in Fig. \ref{F6} (c), which also in accordance with the energy conservation law of the interaction between two nonlinear waves.

\begin{figure}[htbp]
\centering
\subfigure[]{
\begin{minipage}[t]{0.33\textwidth}
\centering
\includegraphics[height=4.5cm,width=4.5cm]{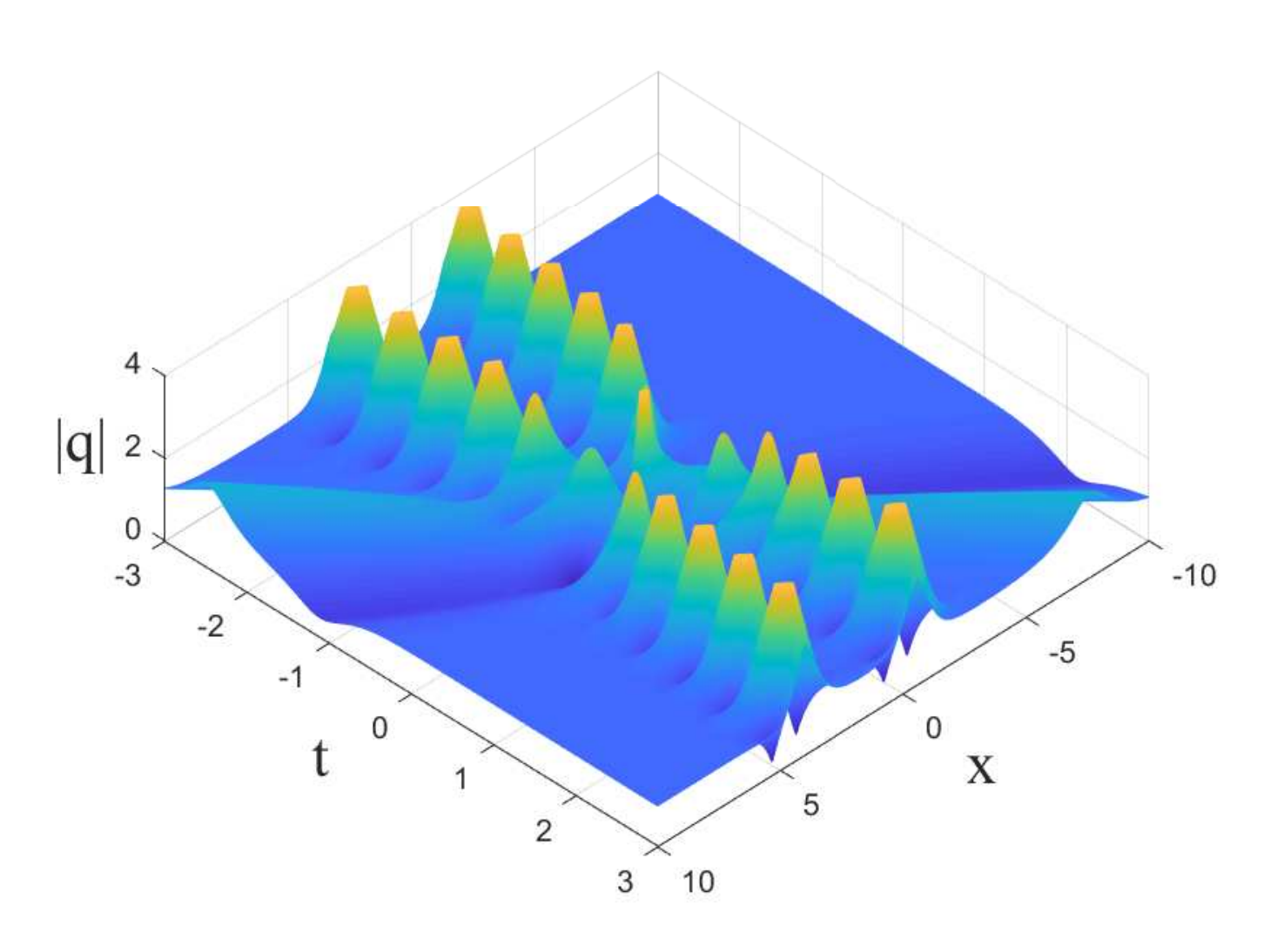}
%\caption{fig1}
\end{minipage}
}%
\subfigure[]{
\begin{minipage}[t]{0.33\textwidth}
\centering
\includegraphics[height=4.5cm,width=4.5cm]{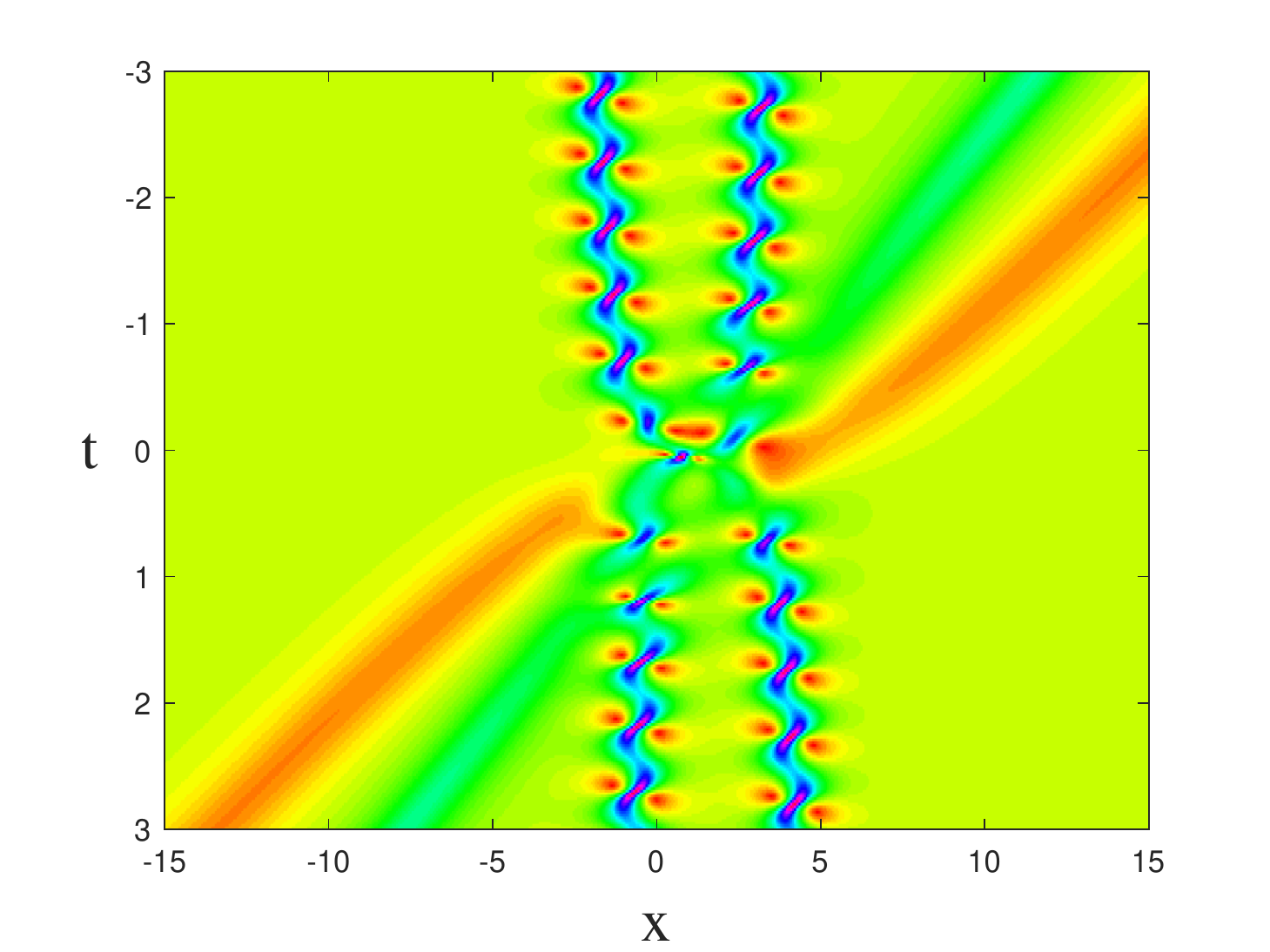}
%\caption{fig2}
\end{minipage}%
}%
\subfigure[]{
\begin{minipage}[t]{0.33\textwidth}
\centering
\includegraphics[height=4.5cm,width=4.5cm]{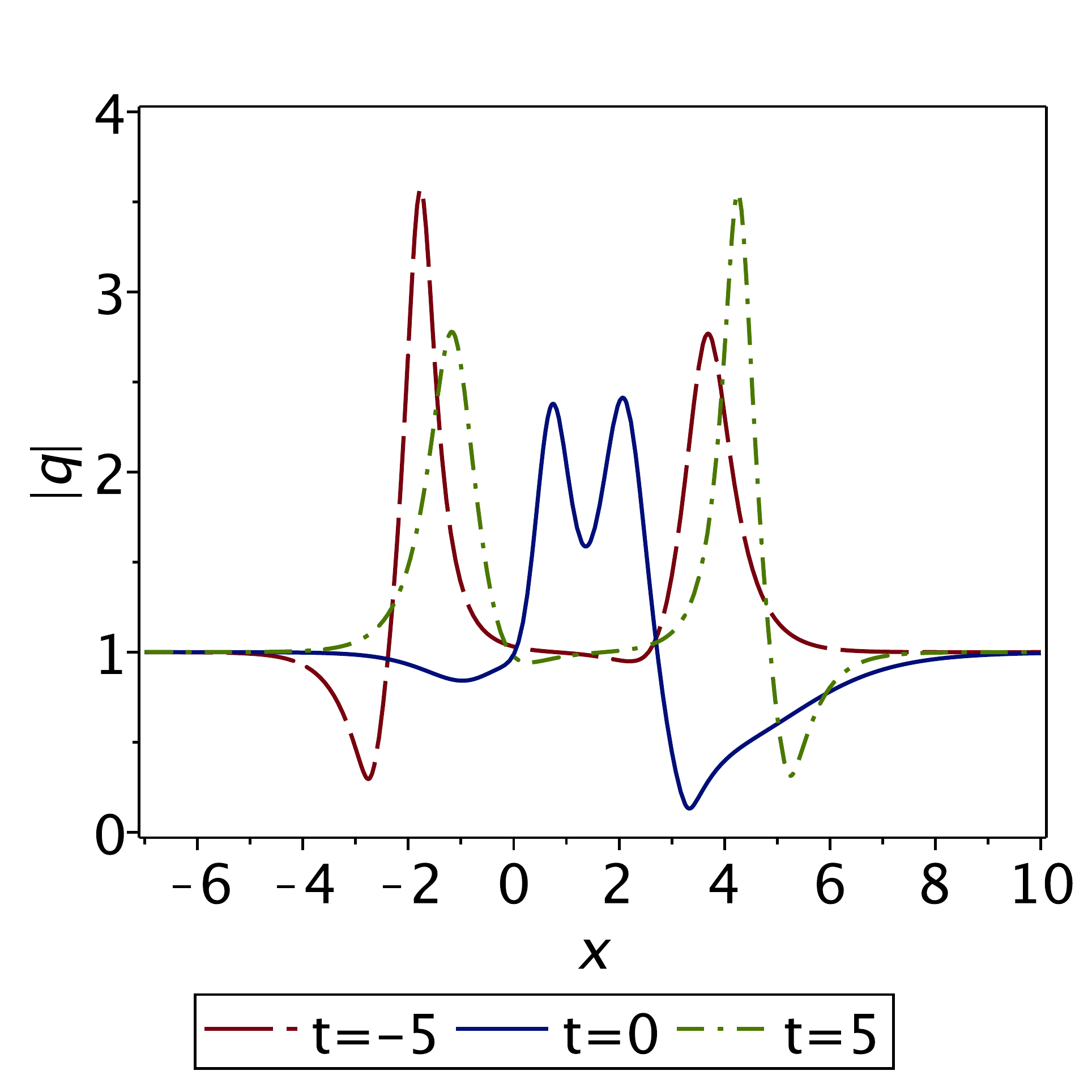}
%\caption{fig1}
\end{minipage}
}%
\centering
\caption{(Color online) The double-pole breather-breather-dark-bright solution of TOFKN \eqref{T1} with NZBCs and $N_1=1,N_2=1,q_{\pm}=1,\zeta_1=2\mathrm{e}^{\frac{\pi}{4}\mathrm{i}},\omega_1=\mathrm{e}^{\frac{\pi}{4}\mathrm{i}},A[\zeta_1]=B[\zeta_1]=\mathrm{i},A[\omega_1]=B[\omega_1]=\mathrm{i}$. (a) The three-dimensional plot; (b) The density plot; (c) The sectional drawings at $t=-5$ (dashed line), $t=0$ (solid line), and $t =5$ (dash-dot line).}
\label{F7}
\end{figure}

\subsection{Inverse Problem with NZBCs and Triple Poles}

In the section, we devote to propose an inverse problem with NZBCs and solve it to obtain accurate triple poles solutions for the TOFKN \eqref{T1}.

\subsubsection{Discrete Spectrum with NZBCs and Triple Zeros}
Differing from the previous results with ZBCs and double poles, we here suppose that $s_{11}(z)$ has $N$ triple zeros in $z_0 = \{z\in\mathbb{C}: \mathrm{Re}z>0, \mathrm{Im}z>0\}$ denoted by $z_n$, $n =1,2,\cdots,N$, that is, $s_{11}(z_n)=s'_{11}(z_n)=s''_{11}(z_n)=0$, and $s'''_{11}(z_n)\neq0$. Obviously, the corresponding discrete spectrum of triple poles is consistent with the discrete spectrum of double poles, as shown in Fig. \ref{F4}.

Once given $z_0\in Z\cap D^{+}$, one can obtain that $\psi''_{+1}(z_0;x,t)-b[z_0]\psi''_{-2}(z_0;x,t)-2d[z_0]\psi'_{-2}(z_0;x,t)$ and $\psi_{-2}(z_0;x,t)$ are linearly dependent by combining Eqs. \eqref{T97}-\eqref{T98} and $s''_{11}(z_0)=0$. Similarly, when a given $z_0\in Z\cap D^{-}$, one can obtain that $\psi''_{+2}(z_0;x,t)-b[z_0]\psi''_{-1}(z_0;x,t)-2d[z_0]\psi'_{-1}(z_0;x,t)$ and $\psi_{-1}(z_0;x,t)$ are linearly dependent by combining $s''_{22}(z_0)=0$. For convenience, we introduce the norming constant $h[z_0]$ such that
\begin{small}
\begin{align}\label{DS-NTZ1}
\begin{split}
\psi''_{+1}(z_0;x,t)-b[z_0]\psi''_{-2}(z_0;x,t)-2d[z_0]\psi'_{-2}(z_0;x,t)=h[z_0]\psi_{-2}(z_0;x,t),\text{ as }z_0\in Z\cap D^{+},\\
\psi''_{+2}(z_0;x,t)-b[z_0]\psi''_{-1}(z_0;x,t)-2d[z_0]\psi'_{-1}(z_0;x,t)=h[z_0]\psi_{-1}(z_0;x,t),\text{ as }z_0\in Z\cap D^{-}.
\end{split}
\end{align}
\end{small}

Then we notice that $\psi_{+1}(z;x,t)$ and $s_{11}(z)$ are analytic on $D^{+}$, Suppose $z_0$ is the triple zeros of $s_{11}$. Let $\psi_{+1}(z;x,t)$ and $s_{11}(z)$ carry out Taylor expansion at $z=z_0$, we have
\begin{small}
\begin{align}\nonumber
\begin{split}
&\frac{\psi_{+1}(z;x,t)}{s_{11}(z)}=\frac{6\psi_{+1}(z_0;x,t)}{s'''_{11}(z_0)}(z-z_0)^{-3}+\frac{-3\psi_{+1}(z_0;x,t)s''''_{11}(z_0)+12\psi'_{+1}(z_0;x,t)s'''_{11}(z_0)}{2(s'''_{11}(z_0))^2}\\
&(z-z_0)^{-2}+\frac{3\psi_{+1}(z_0;x,t)(s''''_{11}(z_0))^2-12\psi'_{+1}(z_0;x,t)s'''_{11}(z_0)s''''_{11}(z_0)+24\psi''_{+1}(z_0;x,t)(s'''_{11}(z_0))^2}{8(s'''_{11}(z_0))^3}\\
&(z-z_0)^{-1}+\cdots,
\end{split}
\end{align}
\end{small}
Then, as $z_0\in Z\cap D^{+}$, one has the compact form
\begin{align}\nonumber
\begin{split}
&\mathop{P_{-3}}_{z=z_0}\bigg[\frac{\psi_{+1}(z;x,t)}{s_{11}(z)}\bigg]=\frac{6b[z_0]\psi_{-2}(z_0;x,t)}{s'''_{11}(z_0)},\\
&\mathop{P_{-2}}_{z=z_0}\bigg[\frac{\psi_{+1}(z;x,t)}{s_{11}(z)}\bigg]=\frac{6b[z_0]}{s'''_{11}(z_0)}\bigg[\psi'_{-2}(z_0;x,t)+\bigg(\frac{d[z_0]}{b[z_0]}-\frac{s''''_{11}(z_0)}{4s'''_{11}(z_0)}\bigg)\psi_{-2}(z_0;x,t)\bigg],\\
&\mathop{\mathrm{Res}}_{z=z_0}\bigg[\frac{\psi_{+1}(z;x,t)}{s_{11}(z)}\bigg]=\frac{6b[z_0]}{s'''_{11}(z_0)}\bigg[\frac12\psi''_{-2}(z_0;x,t)+\bigg(\frac{d[z_0]}{b[z_0]}-\frac{s''''_{11}(z_0)}{4s'''_{11}(z_0)}\bigg)\psi'_{-2}(z_0;x,t)+\\
&\qquad\qquad\qquad\qquad\qquad\bigg(\frac{h[z_0]}{2b[z_0]}-\frac{d[z_0]s''''_{11}(z_0)}{4b[z_0]s'''_{11}(z_0)}+\frac{(s''''_{11}(z_0))^2}{16(s'''_{11}(z_0))^2}\bigg)\psi_{-2}(z_0;x,t)\bigg].
\end{split}
\end{align}

Similarly, for the case of $\psi_{+2}(z;x,t)$ and $s_{22}(z)$ are analytic on $D^{-}$, as $z_0\in Z\cap D^{-}$ we repeat the above process and obtain
\begin{align}\nonumber
\begin{split}
&\mathop{P_{-3}}_{z=z_0}\bigg[\frac{\psi_{+2}(z;x,t)}{s_{22}(z)}\bigg]=\frac{6b[z_0]\psi_{-1}(z_0;x,t)}{s'''_{22}(z_0)},\\
&\mathop{P_{-2}}_{z=z_0}\bigg[\frac{\psi_{+2}(z;x,t)}{s_{22}(z)}\bigg]=\frac{6b[z_0]}{s'''_{22}(z_0)}\bigg[\psi'_{-1}(z_0;x,t)+\bigg(\frac{d[z_0]}{b[z_0]}-\frac{s''''_{22}(z_0)}{4s'''_{22}(z_0)}\bigg)\psi_{-1}(z_0;x,t)\bigg],\\
&\mathop{\mathrm{Res}}_{z=z_0}\bigg[\frac{\psi_{+2}(z;x,t)}{s_{22}(z)}\bigg]=\frac{6b[z_0]}{s'''_{22}(z_0)}\bigg[\frac12\psi''_{-1}(z_0;x,t)+\bigg(\frac{d[z_0]}{b[z_0]}-\frac{s''''_{22}(z_0)}{4s'''_{22}(z_0)}\bigg)\psi'_{-1}(z_0;x,t)+\\
&\qquad\qquad\qquad\qquad\qquad\bigg(\frac{h[z_0]}{2b[z_0]}-\frac{d[z_0]s''''_{22}(z_0)}{4b[z_0]s'''_{22}(z_0)}+\frac{(s''''_{22}(z_0))^2}{16(s'''_{22}(z_0))^2}\bigg)\psi_{-1}(z_0;x,t)\bigg].
\end{split}
\end{align}

Moreover, let
\begin{align}\label{DS-NTZ2}
\begin{split}
&\widetilde{A}[z_0]=\left\{
\begin{aligned}
\frac{6b[z_0]}{s'''_{11}(z_0)},\text{ as }z_0\in Z\cap D^{+},\\
\frac{6b[z_0]}{s'''_{22}(z_0)},\text{ as }z_0\in Z\cap D^{-},
\end{aligned}
\right.\quad
\widetilde{B}[z_0]=\left\{
\begin{aligned}
\frac{d[z_0]}{b[z_0]}-\frac{s''''_{11}(z_0)}{4s'''_{11}(z_0)},\text{ as }z_0\in Z\cap D^{+},\\
\frac{d[z_0]}{b[z_0]}-\frac{s''''_{22}(z_0)}{4s'''_{22}(z_0)},\text{ as }z_0\in Z\cap D^{-},
\end{aligned}
\right.\\
&\widetilde{C}[z_0]=\left\{
\begin{aligned}
\frac{h[z_0]}{2b[z_0]}-\frac{d[z_0]s''''_{11}(z_0)}{4b[z_0]s'''_{11}(z_0)}+\frac{(s''''_{11}(z_0))^2}{16(s'''_{11}(z_0))^2},\text{ as }z_0\in Z\cap D^{+},\\
\frac{h[z_0]}{2b[z_0]}-\frac{d[z_0]s''''_{22}(z_0)}{4b[z_0]s'''_{22}(z_0)}+\frac{(s''''_{22}(z_0))^2}{16(s'''_{22}(z_0))^2},\text{ as }z_0\in Z\cap D^{-}.
\end{aligned}
\right.
\end{split}
\end{align}

Then, we have
\begin{align}\label{DS-NTZ3}
\begin{split}
&\mathop{P_{-3}}_{z=z_0}\bigg[\frac{\psi_{+1}(z;x,t)}{s_{11}(z)}\bigg]=\widetilde{A}[z_0]\psi_{-2}(z_0;x,t),\text{ as }z_0\in Z\cap D^{+},\\
&\mathop{P_{-3}}_{z=z_0}\bigg[\frac{\psi_{+2}(z;x,t)}{s_{22}(z)}\bigg]=\widetilde{A}[z_0]\psi_{-1}(z_0;x,t),\text{ as }z_0\in Z\cap D^{-},\\
&\mathop{P_{-2}}_{z=z_0}\bigg[\frac{\psi_{+1}(z;x,t)}{s_{11}(z)}\bigg]=\widetilde{A}[z_0][\psi'_{-2}(z_0;x,t)+\widetilde{B}[z_0]\psi_{-2}(z_0;x,t)],\text{ as }z_0\in Z\cap D^{+},\\
&\mathop{P_{-2}}_{z=z_0}\bigg[\frac{\psi_{+2}(z;x,t)}{s_{22}(z)}\bigg]=\widetilde{A}[z_0][\psi'_{-1}(z_0;x,t)+\widetilde{B}[z_0]\psi_{-1}(z_0;x,t)],\text{ as }z_0\in Z\cap D^{-},\\
&\mathop{\mathrm{Res}}_{z=z_0}\bigg[\frac{\psi_{+1}(z;x,t)}{s_{11}(z)}\bigg]=\widetilde{A}[z_0]\bigg[\frac12\psi''_{-2}(z_0;x,t)+\widetilde{B}[z_0]\psi'_{-2}(z_0;x,t)+\widetilde{C}[z_0]\psi_{-2}(z_0;x,t)\bigg],\\
&\qquad\qquad\qquad\qquad\qquad\text{ as }z_0\in Z\cap D^{+},\\
&\mathop{\mathrm{Res}}_{z=z_0}\bigg[\frac{\psi_{+2}(z;x,t)}{s_{22}(z)}\bigg]=\widetilde{A}[z_0]\bigg[\frac12\psi''_{-1}(z_0;x,t)+\widetilde{B}[z_0]\psi'_{-1}(z_0;x,t)+\widetilde{C}[z_0]\psi_{-1}(z_0;x,t)\bigg],\\
&\qquad\qquad\qquad\qquad\qquad\text{ as }z_0\in Z\cap D^{-}.
\end{split}
\end{align}

Accordingly, by mean of Eqs. \eqref{T97}-\eqref{T98}, Eqs. \eqref{DS-NTZ1}-\eqref{DS-NTZ2} as well as proposition \ref{P5}, we can derive the following symmetry relations.
\begin{prop}\label{DS-NTZ-P1}
For $z_0\in Z$, the three symmetry relations for $\widetilde{A}[z_0]$, $\widetilde{B}[z_0]$ and $\widetilde{C}[z_0]$ can be deduced as follows:\\
$\bullet$ The first symmetry relation $\widetilde{A}[z_0]=-\widetilde{A}[z^*_0]^*$, $\widetilde{B}[z_0]=\widetilde{B}[z^*_0]^*$, $\widetilde{C}[z_0]=\widetilde{C}[z^*_0]^*$.\\
$\bullet$ The second symmetry relation $\widetilde{A}[z_0]=-\widetilde{A}[-z^*_0]^*$, $\widetilde{B}[z_0]=-\widetilde{B}[-z^*_0]^*$ , $\widetilde{C}[z_0]=\widetilde{C}[-z^*_0]^*$.\\
$\bullet$ The third symmetry relation $\widetilde{A}[z_0]=\frac{z_0^6q_{-}^*}{q_0^6q_{-}}\widetilde{A}[-\frac{q_0^2}{z_0}]$, $\widetilde{B}[z_0]=\frac{q_0^2}{z_0^2}\widetilde{B}[-\frac{q_0^2}{z_0}]+\frac{3}{z_0}$ , $\widetilde{C}[z_0]=\frac{q_0^4}{z_0^4}\widetilde{C}[-\frac{q_0^2}{z_0}]+\frac{2q_0^2}{z_0^3}\widetilde{B}[-\frac{q_0^2}{z_0}]+\frac{3}{z_0^2}$.
\end{prop}

\subsubsection{The Matrix RH Problem with NZBCs and Triple Poles}

Similarly, a matrix RH problem is built as follows.
\begin{prop}\label{MRHP-NTP-P1}
Define the sectionally meromorphic matrices
\begin{align}\label{MRHP-NTP-1}
M(z;x,t)=\left\{
\begin{aligned}
M^+(z;x,t)=\bigg(\frac{\mu_{+1}(z;x,t)}{s_{11}(z)},\mu_{-2}(z;x,t)\bigg),\text{ as } z\in D^+,\\
M^-(z;x,t)=\bigg(\mu_{-1}(z;x,t),\frac{\mu_{+2}(z;x,t)}{s_{22}(z)}\bigg),\text{ as } z\in D^-,
\end{aligned}
\right.
\end{align}
where $\lim\limits_{\substack{z'\rightarrow z\\z'\in D^{\pm}}}M(z';x,t)=M^{\pm}(z;x,t)$. Then, the multiplicative matrix RH problem is given below:\\
$\bullet$ Analyticity: $M(z;x,t)$ is analytic in $D^+\cup D^-\backslash Z$ and has the triple poles in $Z$, whose principal parts of the Laurent series at each triple pole $\zeta_n$ or $\widetilde{\zeta}_n$, are determined as
\begin{align}\label{MRHP-NTP-2}
\begin{split}
&\mathop{\mathrm{Res}}_{z=\zeta_n}M^+(z;x,t)=\bigg(\widetilde{A}[\zeta_n]\mathrm{e}^{-2\mathrm{i}\theta(\zeta_n;x,t)}\bigg[\frac12\mu''_{-2}(\zeta_n;x,t)+(\widetilde{B}[\zeta_n]-2\mathrm{i}\theta'(\zeta_n;x,t))\mu'_{-2}(\zeta_n;x,t)\\
&\qquad\qquad\qquad\qquad+(\widetilde{C}[\zeta_n]-2(\theta'(\zeta_n))^2-\mathrm{i}\theta''(\zeta_n)-2\mathrm{i}\theta'(\zeta_n)\widetilde{B}[\zeta_n])\mu_{-2}(\zeta_n;x,t)\bigg],0\bigg),\\
&\mathop{\mathrm{P}_{-2}}_{z=\zeta_n}M^+(z;x,t)=\big(\widetilde{A}[\zeta_n]\mathrm{e}^{-2\mathrm{i}\theta(\zeta_n;x,t)}[\mu'_{-2}(\zeta_n;x,t)+(\widetilde{B}[\zeta_n]-2\mathrm{i}\theta'(\zeta_n;x,t))\mu_{-2}(\zeta_n;x,t)],0\big),\\
&\mathop{\mathrm{P}_{-3}}_{z=\zeta_n}M^+(z;x,t)=\big(\widetilde{A}[\zeta_n]\mathrm{e}^{-2\mathrm{i}\theta(\zeta_n;x,t)}\mu_{-2}(\zeta_n;x,t),0\big),\\
&\mathop{\mathrm{Res}}_{z=\widetilde{\zeta}_n}M^+(z;x,t)=\bigg(0,\widetilde{A}[\widetilde{\zeta}_n]\mathrm{e}^{2\mathrm{i}\theta(\widetilde{\zeta}_n;x,t)}\bigg[\frac12\mu''_{-1}(\widetilde{\zeta}_n;x,t)+(\widetilde{B}[\widetilde{\zeta}_n]+2\mathrm{i}\theta'(\widetilde{\zeta}_n;x,t))\mu'_{-1}(\widetilde{\zeta}_n;x,t)\\
&\qquad\qquad\qquad\qquad+(\widetilde{C}[\widetilde{\zeta}_n]-2(\theta'(\widetilde{\zeta}_n))^2+\mathrm{i}\theta''(\widetilde{\zeta}_n)+2\mathrm{i}\theta'(\widetilde{\zeta}_n)\widetilde{B}[\widetilde{\zeta}_n])\mu_{-1}(\widetilde{\zeta}_n;x,t)\bigg]\bigg),\\
&\mathop{\mathrm{P}_{-2}}_{z=\widetilde{\zeta}_n}M^+(z;x,t)=\Big(0,\widetilde{A}[\widetilde{\zeta}_n]\mathrm{e}^{2\mathrm{i}\theta(\widetilde{\zeta}_n;x,t)}[\mu'_{-1}(\widetilde{\zeta}_n;x,t)+(\widetilde{B}[\zeta_n]+2\mathrm{i}\theta'(\widetilde{\zeta}_n;x,t))\mu_{-1}(\widetilde{\zeta}_n;x,t)]\Big),\\
&\mathop{\mathrm{P}_{-3}}_{z=\widetilde{\zeta}_n}M^+(z;x,t)=\Big(0,\widetilde{A}[\widetilde{\zeta}_n]\mathrm{e}^{2\mathrm{i}\theta(\widetilde{\zeta}_n;x,t)}\mu_{-1}(\widetilde{\zeta}_n;x,t)\Big).
\end{split}
\end{align}
$\bullet$ Jump condition:
\begin{align}\label{MRHP-NTP-3}
M^-(z;x,t)=M^+(z;x,t)[I-J(z;x,t)],\text{ as }z\in Z,
\end{align}
where
\begin{align}\label{MRHP-NTP-4}
J(z;x,t)=\mathrm{e}^{\mathrm{i}\theta(z;x,t)\widehat{\sigma_3}}\bigg(\begin{array}{cc} 0 & -\tilde{\rho}(z) \\ \rho(z) & \rho(z)\tilde{\rho}(z) \end{array}\bigg).
\end{align}
$\bullet$ Asymptotic behavior:
\begin{align}\label{MRHP-NTP-5}
M(z;x,t)=\left\{
\begin{aligned}
&\frac{\mathrm{i}}{z}\mathrm{e}^{\mathrm{i}\nu_{-}(x,t)\sigma_3}\sigma_3Q_{-}+O(1),\text{ as } z\rightarrow0,\\
&\mathrm{e}^{\mathrm{i}\nu_{-}(x,t)\sigma_3}+\left(\frac{1}{z}\right),\quad\quad\quad\text{ as } z\rightarrow\infty.
\end{aligned}
\right.
\end{align}

\end{prop}

Therefore utilizing asymptotic values as $z\rightarrow\infty$, $z\rightarrow0$ and the singularity contributions, one can regularize the RH problem as a normative form. Then, applying the Plemelj's formula, the solutions of the corresponding matrix RH problem can be solved as follows

\begin{prop}\label{MRHP-NTP-P2}
The solution of the above-mentioned matrix Riemann-Hilbert problem can be expressed as
\begin{small}
\begin{align}\label{MRHP-NTP-6}
\begin{split}
M(z;x,t)=&\mathrm{e}^{\mathrm{i}\nu_{-}(x,t)\sigma_3}\Big(I+\frac{\mathrm{i}}{z}\sigma_3Q_{-}\Big)+\frac{1}{2\pi \mathrm{i}}\int_{\Sigma}\frac{M^+(\xi;x,t)J(\xi;x,t)}{\xi-z}d\xi+\sum^{4N_1+2N_2}_{n=1}\bigg(C_n(z)\bigg[\\
&\frac12\mu''_{-2}(\zeta_n;x,t)+\bigg(D_n+\frac{1}{z-\zeta_n}\bigg)\mu'_{-2}(\zeta_n;x,t)+\bigg(\frac{1}{(z-\zeta_n)^2}+\frac{D_n}{z-\zeta_n}+F_n\bigg)\\
&\mu_{-2}(\zeta_n;x,t)\bigg],\widehat{C}_n(z)\bigg[\frac12\mu''_{-1}(\widetilde{\zeta}_n;x,t)+\bigg(\widehat{D}_n+\frac{1}{z-\widetilde{\zeta}_n}\bigg)\mu'_{-1}(\widetilde{\zeta}_n;x,t)+\bigg(\frac{1}{(z-\widetilde{\zeta}_n)^2}\\
&+\frac{\widehat{D}_n}{z-\widetilde{\zeta}_n}+\widehat{F}_n\bigg)\mu_{-1}(\widetilde{\zeta}_n;x,t)\bigg),
\end{split}
\end{align}
\end{small}
where $z\in\mathbb{C}\backslash\Sigma$, $\int_{\Sigma}$ is an integral along the oriented contour exhibited in Fig. \ref{F4}, and
\begin{align}\label{MRHP-NTP-7}
\begin{split}
&C_n(z)=\frac{\widetilde{A}[\zeta_n]}{z-\zeta_n}\mathrm{e}^{-2\mathrm{i}\theta(\zeta_n;x,t)},\quad \widehat{C}_n(z)=\frac{\widetilde{A}[\widetilde{\zeta}_n]}{z-\widetilde{\zeta}_n}\mathrm{e}^{2\mathrm{i}\theta(\widetilde{\zeta}_n;x,t)},\\
&D_n=\widetilde{B}[\zeta_n]-2\mathrm{i}\theta'(\zeta_n;x,t),\quad \widehat{D}_n=\widetilde{B}[\widetilde{\zeta}_n]+2\mathrm{i}\theta'(\widetilde{\zeta}_n;x,t),\\
&F_n=\widetilde{C}[\zeta_n]-2(\theta'(\zeta_n))^2-\mathrm{i}\theta''(\zeta_n)-2\mathrm{i}\theta'(\zeta_n)\widetilde{B}[\zeta_n],\\
&\widehat{F}_n=\widetilde{C}[\widetilde{\zeta}_n]-2(\theta'(\widetilde{\zeta}_n))^2+\mathrm{i}\theta''(\widetilde{\zeta}_n)+2\mathrm{i}\theta'(\widetilde{\zeta}_n)\widetilde{B}[\widetilde{\zeta}_n],
\end{split}
\end{align}
in which, $\mu_{-2}(\zeta_n;x,t)$, $\mu'_{-2}(\zeta_n;x,t)$ and $\mu''_{-2}(\zeta_n;x,t)$ are determined via $\mu_{-1}(\widetilde{\zeta}_n;x,t)$, $\mu'_{-1}(\widetilde{\zeta}_n;x,t)$ and $\mu''_{-1}(\widetilde{\zeta}_n;x,t)$ as
\begin{align}\label{MRHP-NTP-8}
\begin{split}
&\mu_{-2}(\zeta_n;x,t)=\frac{\mathrm{i}q_{-}}{\zeta_n}\mu_{-1}(\widetilde{\zeta}_n;x,t),\quad \mu'_{-2}(\zeta_n;x,t)=-\frac{\mathrm{i}q_{-}}{\zeta^2_n}\mu_{-1}(\widetilde{\zeta}_n;x,t)+\frac{\mathrm{i}q_{-}q^2_0}{\zeta^3_n}\mu'_{-1}(\widetilde{\zeta}_n;x,t),\\
&\mu''_{-2}(\zeta_n;x,t)=\frac{2\mathrm{i}q_{-}}{\zeta^3_n}\mu_{-1}(\widetilde{\zeta}_n;x,t)-\frac{4\mathrm{i}q_{-}q_0^2}{\zeta^4_n}\mu'_{-1}(\widetilde{\zeta}_n;x,t)+\frac{\mathrm{i}q_{-}q_0^4}{\zeta^5_n}\mu''_{-1}(\widetilde{\zeta}_n;x,t),
\end{split}
\end{align}
and $\mu_{-1}(\widetilde{\zeta}_n;x,t)$, $\mu'_{-1}(\widetilde{\zeta}_n;x,t)$ and $\mu''_{-1}(\widetilde{\zeta}_n;x,t)$ satisfy the linear system of $12N_1+6N_2$ as bellow:
\begin{footnotesize}
\begin{align}\label{MRHP-NTP-9}
\begin{split}
&\sum^{4N_1+2N_2}_{n=1}\bigg\{\frac12\widehat{C}_n(\zeta_s)\mu''_{-1}(\widetilde{\zeta}_n;x,t)+\bigg[\widehat{C}_n(\zeta_s)\bigg(\widehat{D}_n+\frac{1}{\zeta_s-\widetilde{\zeta}_n}\bigg)\bigg]\mu'_{-1}(\widetilde{\zeta}_n;x,t)+\bigg[\widehat{C}_n(\zeta_s)\bigg(\frac{1}{(\zeta_s-\widetilde{\zeta}_n)^2}+\\
&\frac{\widehat{D}_n}{\zeta_s-\widetilde{\zeta}_n}+\widehat{F}_n\bigg)-\frac{\mathrm{i}q_{-}}{\zeta_s}\delta_{s,n}\bigg]\mu_{-1}(\widetilde{\zeta}_n;x,t)\bigg\}=-\mathrm{e}^{\mathrm{i}\nu_{-}(x,t)\sigma_3}\bigg(\begin{array}{cc} \frac{\mathrm{i}q_{-}}{\zeta_s} \\ 1  \end{array}\bigg)-\frac{1}{2\pi \mathrm{i}}\int_{\Sigma}\frac{(M^+(\xi;x,t)J(\xi;x,t))_{2}}{\xi-\zeta_s}d\xi,\\
&\sum^{4N_1+2N_2}_{n=1}\bigg\{\frac{\widehat{C}_n(\zeta_s)}{2(\zeta_s-\widetilde{\zeta}_n)}\mu''_{-1}(\widetilde{\zeta}_n;x,t)+\bigg[\frac{\widehat{C}_n(\zeta_s)}{\zeta_s-\widetilde{\zeta}_n}\bigg(\widehat{D}_n+\frac{2}{\zeta_s-\widetilde{\zeta}_n}\bigg)+\frac{\mathrm{i}q_{-}q^2_0}{\zeta^3_s}\delta_{s,n}\bigg]\mu'_{-1}(\widetilde{\zeta}_n;x,t)+\bigg[\frac{\widehat{C}_n(\zeta_s)}{\zeta_s-\widetilde{\zeta}_n}\\
&\bigg(\frac{3}{(\zeta_s-\widetilde{\zeta}_n)^2}+\frac{2\widehat{D}_n}{\zeta_s-\widetilde{\zeta}_n}+\widehat{F}_n\bigg)-\frac{\mathrm{i}q_{-}}{\zeta^2_s}\delta_{s,n}\bigg]\mu_{-1}(\widetilde{\zeta}_n;x,t)\bigg\}=-\mathrm{e}^{\mathrm{i}\nu_{-}(x,t)\sigma_3}\bigg(\begin{array}{cc} \frac{\mathrm{i}q_{-}}{\zeta^2_s} \\ 0  \end{array}\bigg)+\frac{1}{2\pi \mathrm{i}}\\
&\int_{\Sigma}\frac{(M^+(\xi;x,t)J(\xi;x,t))_{2}}{(\xi-\zeta_s)^2}d\xi,\\
&\sum^{4N_1+2N_2}_{n=1}\bigg\{\bigg[\frac{\widehat{C}_n(\zeta_s)}{(\zeta_s-\widetilde{\zeta}_n)^2}-\frac{\mathrm{i}q_{-}q^4_0}{\zeta^5_s}\delta_{s,n}\bigg]\mu''_{-1}(\widetilde{\zeta}_n;x,t)+\bigg[\frac{2\widehat{C}_n(\zeta_s)}{(\zeta_s-\widetilde{\zeta}_n)^2}\bigg(\widehat{D}_n+\frac{3}{\zeta_s-\widetilde{\zeta}_n}\bigg)+\frac{4\mathrm{i}q_{-}q^2_0}{\zeta^4_s}\delta_{s,n}\bigg]\mu'_{-1}(\widetilde{\zeta}_n;x,t)\\
&+\bigg[\frac{2\widehat{C}_n(\zeta_s)}{(\zeta_s-\widetilde{\zeta}_n)^2}\bigg(\frac{6}{(\zeta_s-\widetilde{\zeta}_n)^2}+\frac{3\widehat{D}_n}{\zeta_s-\widetilde{\zeta}_n}+\widehat{F}_n\bigg)-\frac{2\mathrm{i}q_{-}}{\zeta^3_s}\delta_{s,n}\bigg]\mu_{-1}(\widetilde{\zeta}_n;x,t)\bigg\}=-\mathrm{e}^{\mathrm{i}\nu_{-}(x,t)\sigma_3}\bigg(\begin{array}{cc} \frac{2\mathrm{i}q_{-}}{\zeta^3_s} \\ 0  \end{array}\bigg)\\
&-\frac{1}{2\pi \mathrm{i}}\int_{\Sigma}\frac{2(M^+(\xi;x,t)J(\xi;x,t))_{2}}{(\xi-\zeta_s)^3}d\xi,
\end{split}
\end{align}
\end{footnotesize}
where $s=1,2,\cdots,4N_1+2N_2$ and $\delta_{s,n}$ are the Kronecker $\delta$-symbol.

\end{prop}

\begin{proof}
In order to regularize the RH problem, one has to subtract out the asymptotic values as $z\rightarrow\infty$ and $z\rightarrow0$ which exhibited in Eq. \eqref{MRHP-NTP-5} and the singularity contributions. Then, the jump condition \eqref{MRHP-NTP-3} becomes
\begin{small}
\begin{align}\label{MRHP-NTP-10}
\begin{split}
&M^-(z;x,t)-\mathrm{e}^{\mathrm{i}\nu_{-}(x,t)\sigma_3}-\frac{\mathrm{i}}{z}\mathrm{e}^{\mathrm{i}\nu_{-}(x,t)\sigma_3}\sigma_3Q_{-}-\sum_{n=1}^{4N_1+2N_2}\Bigg[\frac{\mathop{\mathrm{P}_{-3}}\limits_{z=\zeta_n}M^+(z;x,t)}{(z-\zeta_n)^3}+\frac{\mathop{\mathrm{P}_{-2}}\limits_{z=\zeta_n}M^+(z;x,t)}{(z-\zeta_n)^2}+\\
&\frac{\mathop{\mathrm{Res}}\limits_{z=\zeta_n}M^+(z;x,t)}{z-\zeta_n}+\frac{\mathop{\mathrm{P}_{-3}}\limits_{z=\widetilde{\zeta}_n}M^-(z;x,t)}{(z-\widetilde{\zeta}_n)^3}+\frac{\mathop{\mathrm{P}_{-2}}\limits_{z=\widetilde{\zeta}_n}M^-(z;x,t)}{(z-\widetilde{\zeta}_n)^2}+\frac{\mathop{\mathrm{Res}}\limits_{z=\widetilde{\zeta}_n}M^-(z;x,t)}{z-\widetilde{\zeta}_n}\Bigg]=M^+(z;x,t)-\\
&\mathrm{e}^{\mathrm{i}\nu_{-}(x,t)\sigma_3}-\frac{\mathrm{i}}{z}\mathrm{e}^{\mathrm{i}\nu_{-}(x,t)\sigma_3}\sigma_3Q_{-}-\sum_{n=1}^{4N_1+2N_2}\Bigg[\frac{\mathop{\mathrm{P}_{-3}}\limits_{z=\zeta_n}M^+(z;x,t)}{(z-\zeta_n)^3}+\frac{\mathop{\mathrm{P}_{-2}}\limits_{z=\zeta_n}M^+(z;x,t)}{(z-\zeta_n)^2}+\frac{\mathop{\mathrm{Res}}\limits_{z=\zeta_n}M^+(z;x,t)}{z-\zeta_n}\\
&+\frac{\mathop{\mathrm{P}_{-3}}\limits_{z=\widetilde{\zeta}_n}M^-(z;x,t)}{(z-\widetilde{\zeta}_n)^3}+\frac{\mathop{\mathrm{P}_{-2}}\limits_{z=\widetilde{\zeta}_n}M^-(z;x,t)}{(z-\widetilde{\zeta}_n)^2}+\frac{\mathop{\mathrm{Res}}\limits_{z=\widetilde{\zeta}_n}M^-(z;x,t)}{z-\widetilde{\zeta}_n}\Bigg]-M^+(z;x,t)J(z;x,t),
\end{split}
\end{align}
\end{small}
where $\mathop{\mathrm{P}_{-3}}\limits_{z=\zeta_n}M^+,\mathop{\mathrm{P}_{-2}}\limits_{z=\zeta_n}M^+,\mathop{\mathrm{Res}}\limits_{z=\zeta_n}M^+,\mathop{\mathrm{P}_{-3}}\limits_{z=\widetilde{\zeta}_n}M^-,\mathop{\mathrm{P}_{-2}}\limits_{z=\widetilde{\zeta}_n}M^-,\mathop{\mathrm{Res}}\limits_{z=\widetilde{\zeta}_n}M^-$ have given in Eq. \eqref{MRHP-NTP-2}. By using Plemelj's formula, one can obtain the solution \eqref{MRHP-NTP-6} with formula \eqref{MRHP-NTP-7} of the matrix RH problem. According to the symmetry reduction $\mu_{\pm}(z;x,t)=\frac{\mathrm{i}}{z}\mu_{\pm}\big(-\frac{q^2_0}{z};x,t\big)\sigma_3Q_{\pm}$, we can obtain \eqref{MRHP-NTP-8}. By combining Eqs. \eqref{MRHP-NTP-1} and \eqref{MRHP-NTP-6}, $\mu_{-1}(\widetilde{\zeta}_n;x,t)$ is the first column element of the solution \eqref{MRHP-NTP-6} as the triple-pole $z=\widetilde{\zeta}_n\in D^{-}$, $\mu_{-2}(\zeta_s;x,t)$ is the second column element of the solution \eqref{MRHP-NTP-6} as the triple-pole $z=\zeta_s\in D^{+}$. Then, we can obtain the linear system \eqref{MRHP-NTP-9} by utilizing symmetry relation \eqref{MRHP-NTP-8}. Completing the proof.

\end{proof}

\subsubsection{Reconstruction Formula of the Potential with NZBCs and Triple Poles}
Similarly, the reconstruction formula of the triple-pole solution (potential) for the TOFKN \eqref{T1} with NZBCs is consistent with Eq. \eqref{T118}. From the Eq. \eqref{T116} and solution \eqref{MRHP-NTP-6} of the matrix RH problem, we have
\begin{align}\label{RFP-NTP-1}
M^{[1]}_{12}(x,t)=&\mathrm{i}q_{-}\mathrm{e}^{\mathrm{i}\nu_{-}(x,t)}-\frac{1}{2\pi\mathrm{i}}\int_{\Sigma}\big(M^{+}(\xi;x,t)J(\xi;x,t)\big)_{12}d\xi+\sum^{4N_1+2N_2}_{n=1}\bigg\{\widehat{A}[\widetilde{\zeta}_n]\mathrm{e}^{2\mathrm{i}\theta(\widetilde{\zeta}_n;x,t)}\\
&\bigg[\frac12\mu''_{-11}(\widetilde{\zeta}_n;x,t)+\widehat{D}_n\mu'_{-11}(\widetilde{\zeta}_n;x,t)+\widehat{F}_n\mu_{-11}(\widetilde{\zeta}_n;x,t)\bigg]\bigg\},
\end{align}
where $\mu_{-11}(\widetilde{\zeta}_n;x,t), \mu'_{-11}(\widetilde{\zeta}_n;x,t)$ and $\mu''_{-11}(\widetilde{\zeta}_n;x,t)$ represents the first row element of the column vector $\mu_{-1}(\widetilde{\zeta}_n;x,t), \mu'_{-1}(\widetilde{\zeta}_n;x,t)$ and $\mu''_{-1}(\widetilde{\zeta}_n;x,t)$, respectively. Then taking row vector $\alpha=\big(\alpha^{(1)},\alpha^{(2)},\alpha^{(3)}\big)$ and column vector $\gamma=(\gamma^{(1)},\gamma^{(2)},\gamma^{(3)})^{\mathrm{T}}$, where
\begin{align}\label{RFP-NTP-2}
\begin{split}
&\alpha^{(1)}=\bigg(\frac12\widetilde{A}[\widetilde{\zeta}_n]\mathrm{e}^{2\mathrm{i}\theta(\widehat{\zeta}_n;x,t)}\bigg)_{1\times(4N_1+2N_2)},\,\alpha^{(2)}=\big(\widetilde{A}[\widetilde{\zeta}_n]\mathrm{e}^{2\mathrm{i}\theta(\widetilde{\zeta}_n;x,t)}\widehat{D}_n\big)_{1\times(4N_1+2N_2)},\\
&\alpha^{(3)}=\big(\widetilde{A}[\widetilde{\zeta}_n]\mathrm{e}^{2\mathrm{i}\theta(\widetilde{\zeta}_n;x,t)}\widetilde{F}_n\big)_{1\times(4N_1+2N_2)},\\
&\gamma^{(1)}=\big(\mu''_{-11}(\widetilde{\zeta}_n;x,t)\big)_{1\times(4N_1+2N_2)},\,\gamma^{(2)}=\big(\mu'_{-11}(\widetilde{\zeta}_n;x,t)\big)_{1\times(4N_1+2N_2)},\\
&\gamma^{(3)}=\big(\mu_{-11}(\widetilde{\zeta}_n;x,t)\big)_{1\times(4N_1+2N_2)},
\end{split}
\end{align}
we can obtain a more concise reconstruction formulation of the triple poles solution (potential) for the TOFKN \eqref{T1} with NZBCs and as follows
\begin{align}\label{RFP-NTP-3}
q(x,t)=q_{-}\mathrm{e}^{2\mathrm{i}\nu_{-}(x,t)}-\mathrm{i}\mathrm{e}^{\mathrm{i}\nu_{-}(x,t)}\alpha\gamma+\frac{1}{2\pi}\int_{\Sigma}\big(M^{+}(\xi;x,t)J(\xi;x,t)\big)_{12}d\xi.
\end{align}

\subsubsection{Trace Formulae and Theta Condition with NZBCs and Triple Poles}
The discrete spectral points $\zeta_n$'s are the triple zeros of $s_{11}(\lambda)$, while $\widetilde{\zeta}_n$'s are the triple zeros of $s_{22}(\lambda)$. Define the functions $\beta^{\pm}(z)$ as follows:
\begin{align}\label{TFTC-NTP-1}
\beta^{+}(z)=s_{11}(z)\prod^{4N_1+2N_2}_{n=1}\bigg(\frac{z-\widetilde{\zeta}_n}{z-\zeta_n}\bigg)^3\mathrm{e}^{\mathrm{i}\nu},\,\beta^{-}(z)=s_{22}(z)\prod^{4N_1+2N_2}_{n=1}\bigg(\frac{z-\zeta_n}{z-\widetilde{\zeta}_n}\bigg)^3\mathrm{e}^{-\mathrm{i}\nu}.
\end{align}

Then, $\beta^{+}(z)$ and $\beta^{-}(z)$ are analytic and have no zero in $D^{+}$ and $D^{-}$, respectively. Furthermore, we have the relation $\beta^{+}(z)\beta^{-}(z)=s_{11}(z)s_{22}(z)$ and the asymptotic behaviors $\beta^{\pm}(z)\rightarrow1,\text{ as }z\rightarrow\infty$.

By means of employing the Cauchy projectors and Plemelj' formula, we have
\begin{align}\label{TFTC-NTP-2}
\mathrm{log}\beta^{\pm}(z)=\mp\frac{1}{2\pi \mathrm{i}}\int_{\Sigma}\frac{\mathrm{log}[1-\rho(z)\tilde{\rho}(z)]}{\xi-z}d\xi,\, z\in D^{\pm}.
\end{align}

After substituting Eq. \eqref{TFTC-NTP-2} into  Eq. \eqref{TFTC-NTP-1}, we can obtain the trace formulae
\begin{align}\label{TFTC-NTP-3}
\begin{split}
&s_{11}(z)=\mathrm{exp}\Bigg(\frac{\mathrm{i}}{2\pi}\int_{\Sigma}\frac{\mathrm{log}[1-\rho(z)\tilde{\rho}(z)]}{\xi-z}d\xi\Bigg)\prod^{4N_1+2N_2}_{n=1}\bigg(\frac{z-\zeta_n}{z-\zeta^*_n}\bigg)^3\mathrm{e}^{-\mathrm{i}\nu},\\
&s_{22}(z)=\mathrm{exp}\Bigg(-\frac{\mathrm{i}}{2\pi}\int_{\Sigma}\frac{\mathrm{log}[1-\rho(z)\tilde{\rho}(z)]}{\xi-z}d\xi\Bigg)\prod^{4N_1+2N_2}_{n=1}\bigg(\frac{z-\zeta^*_n}{z-\zeta_n}\bigg)^3\mathrm{e}^{\mathrm{i}\nu}.
\end{split}
\end{align}

As $z\rightarrow0$, we can derive the so-called theta condition is obtained. That is, there exists $j\in\mathbb{Z}$ such that
\begin{align}\label{TFTC-NTP-4}
\mathrm{arg}\bigg(\frac{q_{-}}{q_{+}}\bigg)+2\nu=24\sum^{N_1}_{n=1}\mathrm{arg}(z_n)+12\sum^{N_2}_{m=1}\mathrm{arg}(\omega_m)+2\pi j+\frac{1}{2\pi}\int_{\Sigma}\frac{\mathrm{log}[1-\rho(\xi)\tilde{\rho}(\xi)]}{\xi}d\xi.
\end{align}

\subsubsection{Reflectionless Potential with NZBCs: Triple-Pole Solutions}
Now, we consider the case of reflectionless potential $q(x,t)$ with the reflection coefficients $\rho(\lambda)=\tilde{\rho}(\lambda)=0$. Combining with the definition of scattering matrix, one has $S(q_0)=I$ and $q_{+}=q_{-}$. From the theta condition, there exists $j\in\mathbb{Z}$ lead to
\begin{align}\label{RP-NTPS-1}
\nu=12\sum^{N_1}_{n=1}\mathrm{arg}(z_n)+6\sum^{N_2}_{m=1}\mathrm{arg}(\omega_m)+\pi j.
\end{align}
Then the Eqs. \eqref{MRHP-NTP-9} and \eqref{RFP-NTP-3} with $J(\lambda;x,t)=0_{2\times2}$ become
\begin{footnotesize}
\begin{align}\label{RP-NTPS-2}
\begin{split}
&\sum^{4N_1+2N_2}_{n=1}\bigg\{\frac12\widehat{C}_n(\zeta_s)\mu''_{-11}(\widetilde{\zeta}_n;x,t)+\bigg[\widehat{C}_n(\zeta_s)\bigg(\widehat{D}_n+\frac{1}{\zeta_s-\widetilde{\zeta}_n}\bigg)\bigg]\mu'_{-11}(\widetilde{\zeta}_n;x,t)+\bigg[\widehat{C}_n(\zeta_s)\bigg(\frac{1}{(\zeta_s-\widetilde{\zeta}_n)^2}+\\
&\frac{\widehat{D}_n}{\zeta_s-\widetilde{\zeta}_n}+\widehat{F}_n\bigg)-\frac{\mathrm{i}q_{-}}{\zeta_s}\delta_{s,n}\bigg]\mu_{-11}(\widetilde{\zeta}_n;x,t)\bigg\}=-\mathrm{e}^{\mathrm{i}\nu_{-}(x,t)}\frac{\mathrm{i}q_{-}}{\zeta_s},\\
&\sum^{4N_1+2N_2}_{n=1}\bigg\{\frac{\widehat{C}_n(\zeta_s)}{2(\zeta_s-\widetilde{\zeta}_n)}\mu''_{-11}(\widetilde{\zeta}_n;x,t)+\bigg[\frac{\widehat{C}_n(\zeta_s)}{\zeta_s-\widetilde{\zeta}_n}\bigg(\widehat{D}_n+\frac{2}{\zeta_s-\widetilde{\zeta}_n}\bigg)+\frac{\mathrm{i}q_{-}q^2_0}{\zeta^3_s}\delta_{s,n}\bigg]\mu'_{-11}(\widetilde{\zeta}_n;x,t)+\bigg[\frac{\widehat{C}_n(\zeta_s)}{\zeta_s-\widetilde{\zeta}_n}\\
&\bigg(\frac{3}{(\zeta_s-\widetilde{\zeta}_n)^2}+\frac{2\widehat{D}_n}{\zeta_s-\widetilde{\zeta}_n}+\widehat{F}_n\bigg)-\frac{\mathrm{i}q_{-}}{\zeta^2_s}\delta_{s,n}\bigg]\mu_{-11}(\widetilde{\zeta}_n;x,t)\bigg\}=-\mathrm{e}^{\mathrm{i}\nu_{-}(x,t)}\frac{\mathrm{i}q_{-}}{\zeta^2_s},\\
&\sum^{4N_1+2N_2}_{n=1}\bigg\{\bigg[\frac{\widehat{C}_n(\zeta_s)}{(\zeta_s-\widetilde{\zeta}_n)^2}-\frac{\mathrm{i}q_{-}q^4_0}{\zeta^5_s}\bigg]\mu''_{-11}(\widetilde{\zeta}_n;x,t)+\bigg[\frac{2\widehat{C}_n(\zeta_s)}{(\zeta_s-\widetilde{\zeta}_n)^2}\bigg(\widehat{D}_n+\frac{3}{\zeta_s-\widetilde{\zeta}_n}\bigg)+\frac{4\mathrm{i}q_{-}q^2_0}{\zeta^4_s}\delta_{s,n}\bigg]\mu'_{-11}(\widetilde{\zeta}_n;x,t)\\
&+\bigg[\frac{2\widehat{C}_n(\zeta_s)}{(\zeta_s-\widetilde{\zeta}_n)^2}\bigg(\frac{6}{(\zeta_s-\widetilde{\zeta}_n)^2}+\frac{3\widehat{D}_n}{\zeta_s-\widetilde{\zeta}_n}+\widehat{F}_n\bigg)-\frac{2\mathrm{i}q_{-}}{\zeta^3_s}\delta_{s,n}\bigg]\mu_{-11}(\widetilde{\zeta}_n;x,t)\bigg\}=-\mathrm{e}^{\mathrm{i}\nu_{-}(x,t)}\frac{2\mathrm{i}q_{-}}{\zeta^3_s},
\end{split}
\end{align}
\end{footnotesize}
and
\begin{align}\label{RP-NTPS-3}
q(x,t)=q_{-}\mathrm{e}^{2\mathrm{i}\nu_{-}(x,t)}-\mathrm{i}\mathrm{e}^{\mathrm{i}\nu_{-}(x,t)}\alpha\gamma.
\end{align}

\begin{thm}\label{RP-NTPS-thm1}
The general expression of triple poles solution of the TOFKN \eqref{T1} with NZBCs is given by determinant formula
\begin{align}\label{RP-NTPS-4}
q(x,t)=q_{-}\mathrm{e}^{2\mathrm{i}\nu_{-}(x,t)}\bigg(1+\frac{\mathrm{det}(R)}{\mathrm{det}(G)}\bigg),
\end{align}
where
\begin{align}\label{RP-NTPS-5}
\begin{split}
&R=\bigg(\begin{array}{cc} 0 & \alpha \\ \beta & G \end{array}\bigg),\,\beta=\big(\beta^{(1)},\beta^{(2)},\beta^{(3)}\big)^{\mathrm{T}},\,\alpha=\big(\alpha^{(1)},\alpha^{(2)},\alpha^{(3)}\big),\\
&\beta^{(1)}=\bigg(\frac{1}{\zeta_s}\bigg)_{1\times(4N_1+2N_2)},\,\beta^{(2)}=\bigg(\frac{1}{\zeta^2_s}\bigg)_{1\times(4N_1+2N_2)},\beta^{(3)}=\bigg(\frac{2}{\zeta^3_s}\bigg)_{1\times(4N_1+2N_2)},\\
&\alpha^{(1)}=\bigg(\frac12\widetilde{A}[\widetilde{\zeta}_n]\mathrm{e}^{2\mathrm{i}\theta(\widehat{\zeta}_n;x,t)}\bigg)_{1\times(4N_1+2N_2)},\,\alpha^{(2)}=\big(\widetilde{A}[\widetilde{\zeta}_n]\mathrm{e}^{2\mathrm{i}\theta(\widetilde{\zeta}_n;x,t)}\widehat{D}_n\big)_{1\times(4N_1+2N_2)},\\
&\alpha^{(3)}=\big(\widetilde{A}[\widetilde{\zeta}_n]\mathrm{e}^{2\mathrm{i}\theta(\widetilde{\zeta}_n;x,t)}\widetilde{F}_n\big)_{1\times(4N_1+2N_2)},
\end{split}
\end{align}
and the $(12N_1+6N_2)\times(12N_1+6N_2)$ partitioned matrix $G=\Bigg(\begin{array}{ccc} G^{(1,1)} & G^{(1,2)} & G^{(1,3)} \\ G^{(2,1)} & G^{(2,2)} & G^{(2,3)} \\ G^{(3,1)} & G^{(3,2)} & G^{(3,3)} \end{array}\Bigg)$ with $G^{(i,j)}=\Big(g^{(i,j)}_{s,n}\Big)_{(4N_1+2N_2)\times(4N_1+2N_2)}\,(i,j=1,2,3)$ is given by
\begin{footnotesize}
\begin{align}\label{RP-NTPS-6}
\begin{split}
&g^{(1,1)}_{s,n}=\frac12\widehat{C}_n(\zeta_s),\,g^{(1,2)}_{s,n}=\widehat{C}_n(\zeta_s)\bigg(\widehat{D}_n+\frac{1}{\zeta_s-\widetilde{\zeta}_n}\bigg),\,g^{(1,3)}_{s,n}=\widehat{C}_n(\zeta_s)\bigg(\frac{1}{(\zeta_s-\widetilde{\zeta}_n)^2}+\frac{\widehat{D}_n}{\zeta_s-\widetilde{\zeta}_n}+\widehat{F}_n\bigg)-\\
&\frac{\mathrm{i}q_{-}}{\zeta_s}\delta_{s,n},\,g^{(2,1)}_{s,n}=\frac{\widehat{C}_n(\zeta_s)}{2(\zeta_s-\widetilde{\zeta}_n)},\,g^{(2,2)}_{s,n}=\frac{\widehat{C}_n(\zeta_s)}{\zeta_s-\widetilde{\zeta}_n}\bigg(\widehat{D}_n+\frac{2}{\zeta_s-\widetilde{\zeta}_n}\bigg)+\frac{\mathrm{i}q_{-}q^2_0}{\zeta^3_s}\delta_{s,n},\\
&g^{(2,3)}_{s,n}=\frac{\widehat{C}_n(\zeta_s)}{\zeta_s-\widetilde{\zeta}_n}\bigg(\frac{3}{(\zeta_s-\widetilde{\zeta}_n)^2}+\frac{2\widehat{D}_n}{\zeta_s-\widetilde{\zeta}_n}+\widehat{F}_n\bigg)-\frac{\mathrm{i}q_{-}}{\zeta^2_s}\delta_{s,n},\\
&g^{(3,1)}_{s,n}=\frac{\widehat{C}_n(\zeta_s)}{(\zeta_s-\widetilde{\zeta}_n)^2}-\frac{\mathrm{i}q_{-}q^4_0}{\zeta^5_s}\delta_{s,n},\,g^{(3,2)}_{s,n}=\frac{2\widehat{C}_n(\zeta_s)}{(\zeta_s-\widetilde{\zeta}_n)^2}\bigg(\widehat{D}_n+\frac{3}{\zeta_s-\widetilde{\zeta}_n}\bigg)+\frac{4\mathrm{i}q_{-}q^4_0}{\zeta^4_s}\delta_{s,n},\\
&g^{(3,3)}_{s,n}=\frac{2\widehat{C}_n(\zeta_s)}{(\zeta_s-\widetilde{\zeta}_n)^2}\bigg(\frac{6}{(\zeta_s-\widetilde{\zeta}_n)^2}+\frac{3\widehat{D}_n}{\zeta_s-\widetilde{\zeta}_n}+\widehat{F}_n\bigg)-\frac{2\mathrm{i}q_{-}}{\zeta^3_s}\delta_{s,n}.
\end{split}
\end{align}
\end{footnotesize}

\end{thm}
\begin{proof}
We rewrite the linear system \eqref{RP-NTPS-2} in the matrix form
\begin{align}\nonumber
G\gamma=\beta,
\end{align}
where
\begin{align}\nonumber
\begin{split}
&\gamma=(\gamma^{(1)},\gamma^{(2)},\gamma^{(3)})^{\mathrm{T}},\,\gamma^{(1)}=\big(\mu''_{-11}(\widetilde{\zeta}_n;x,t)\big)_{1\times(4N_1+2N_2)},\,\gamma^{(2)}=\big(\mu'_{-11}(\widetilde{\zeta}_n;x,t)\big)_{1\times(4N_1+2N_2)},\\
&\gamma^{(3)}=\big(\mu_{-11}(\widetilde{\zeta}_n;x,t)\big)_{1\times(4N_1+2N_2)},\,\beta=\big(\beta^{(1)},\beta^{(2)},\beta^{(3)}\big)^{\mathrm{T}},\\
&\beta^{(1)}=\bigg(\frac{1}{\zeta_s}\bigg)_{1\times(4N_1+2N_2)},\,\beta^{(2)}=\bigg(\frac{1}{\zeta^2_s}\bigg)_{1\times(4N_1+2N_2)},\beta^{(3)}=\bigg(\frac{2}{\zeta^3_s}\bigg)_{1\times(4N_1+2N_2)}.
\end{split}
\end{align}
Then combining Eq. \eqref{RP-NTPS-3} with the case of reflectionless potential and matrix operation, the triple poles solution \eqref{RP-NTPS-4} can be given out.

\end{proof}

For example, we obtain the triple-pole solutions of the TOFKN \eqref{T1} with NZBCs via Theorem \ref{RP-NTPS-thm1}.

$\bullet$ When taking parameters $N_1=0,N_2=1,q_{\pm}=1,\omega_1=2\mathrm{e}^{\frac{\pi}{6}\mathrm{i}},A[\omega_1]=\mathrm{i},B[\omega_1]=1+(1-\sqrt{2})\mathrm{i},C[\omega_1]=1$, we can obtain the triple-pole bright-dark-bright soliton solution, and give out relevant plots in Fig. \ref{RP-NTPS-F1}. Fig. \ref{RP-NTPS-F1} (a) and (b) exhibit the three-dimensional and density diagrams for the triple-pole bright-dark-bright soliton solution of the TOFKN with NZBCs. Fig. \ref{RP-NTPS-F1} (c) displays the distinct profiles of the triple-pole bright-dark-bright soliton solution for $t=\pm5,0$. It is a semi rational soliton, the triple-pole bright-dark-bright soliton solution shows the interaction of bright soliton, dark soliton and bright soliton. From Fig. \ref{F5}, as taking $N_1=0,N_2=1$, once double-pole solution with NZBCs change to triple-pole solution with NZBCs, the triple-pole solution will arise the second bright soliton branch.

\begin{figure}[htbp]
\centering
\subfigure[]{
\begin{minipage}[t]{0.33\textwidth}
\centering
\includegraphics[height=4.5cm,width=4.5cm]{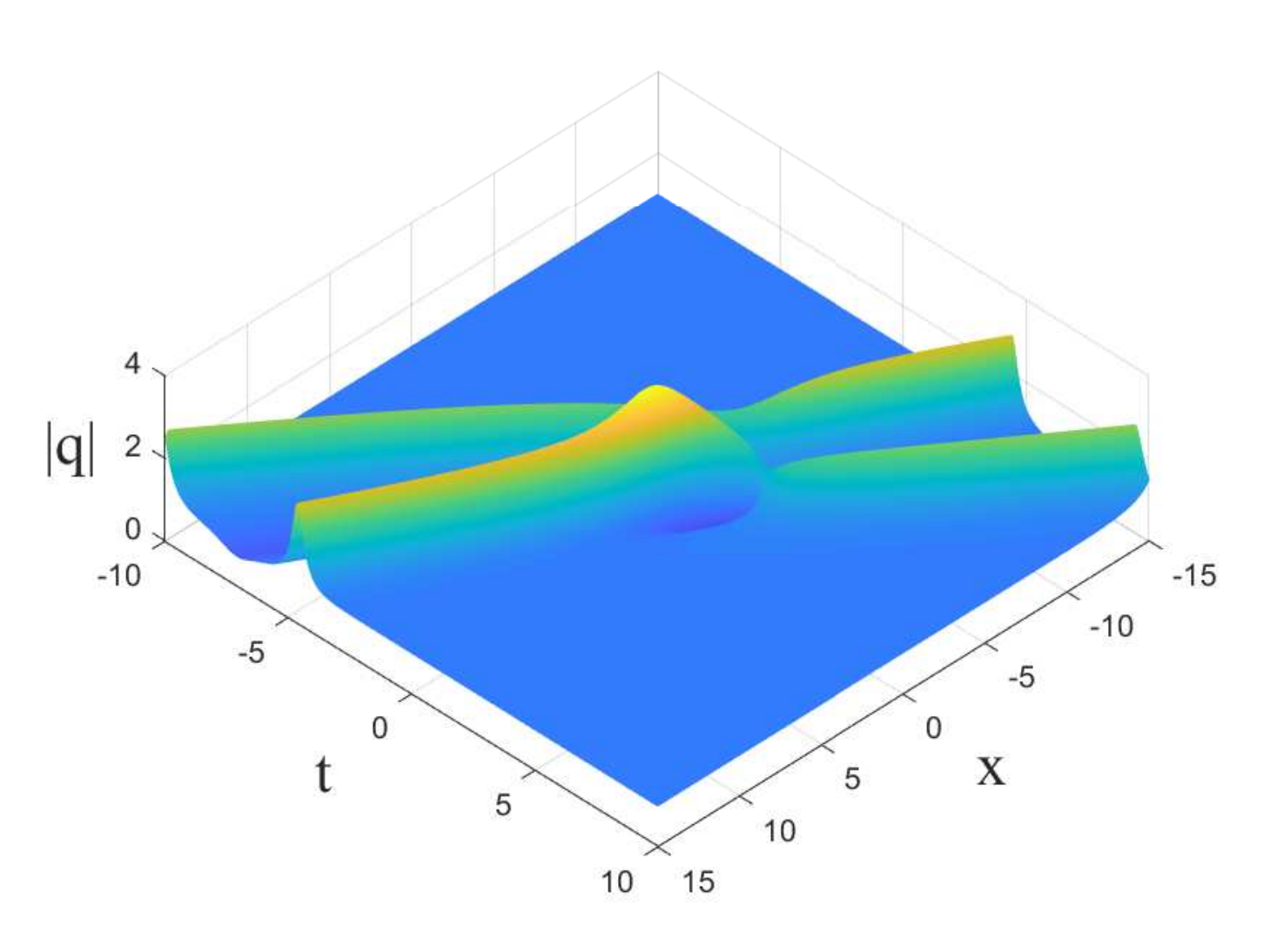}
%\caption{fig1}
\end{minipage}
}%
\subfigure[]{
\begin{minipage}[t]{0.33\textwidth}
\centering
\includegraphics[height=4.5cm,width=4.5cm]{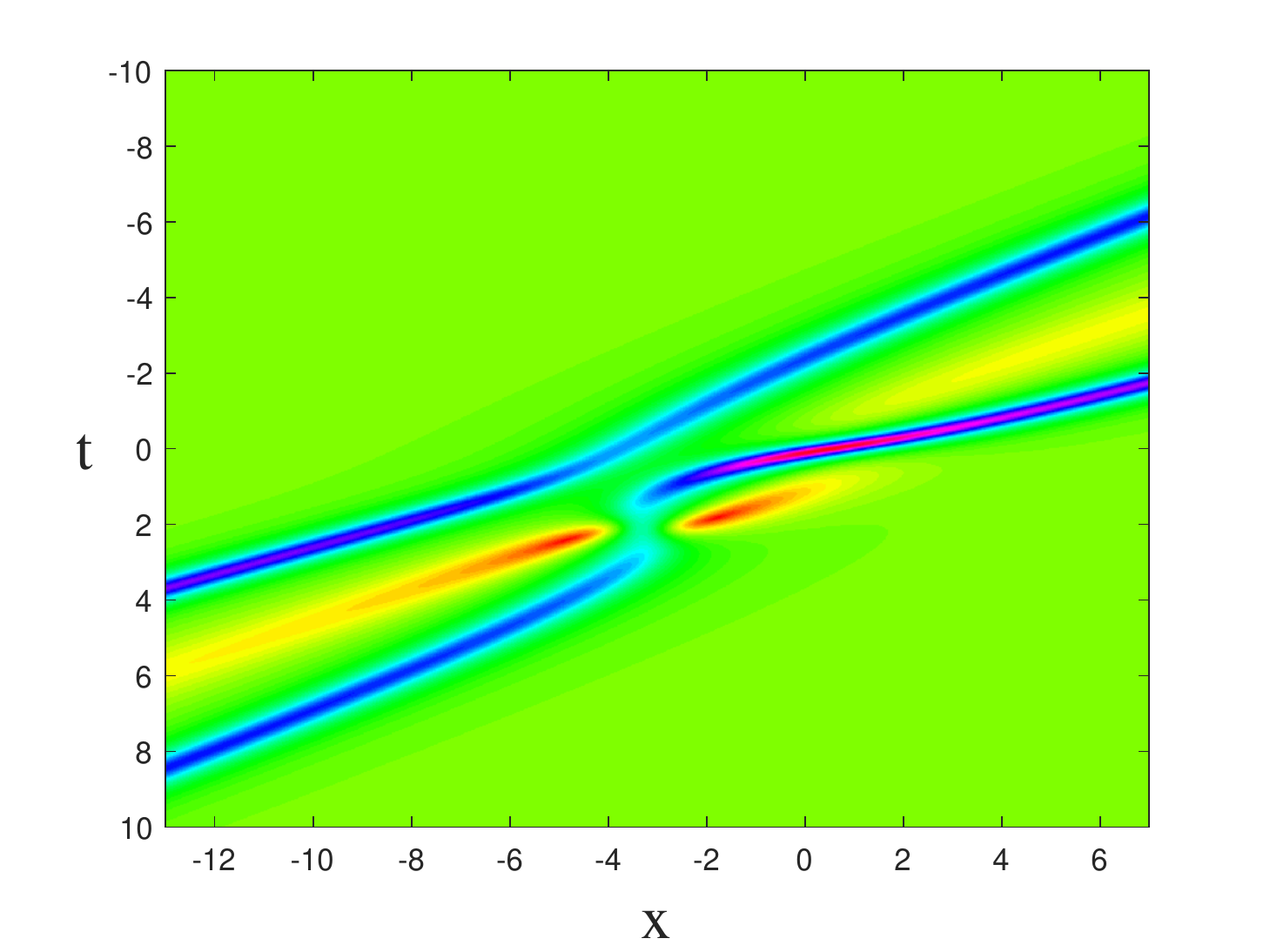}
%\caption{fig2}
\end{minipage}%
}%
\subfigure[]{
\begin{minipage}[t]{0.33\textwidth}
\centering
\includegraphics[height=4.5cm,width=4.5cm]{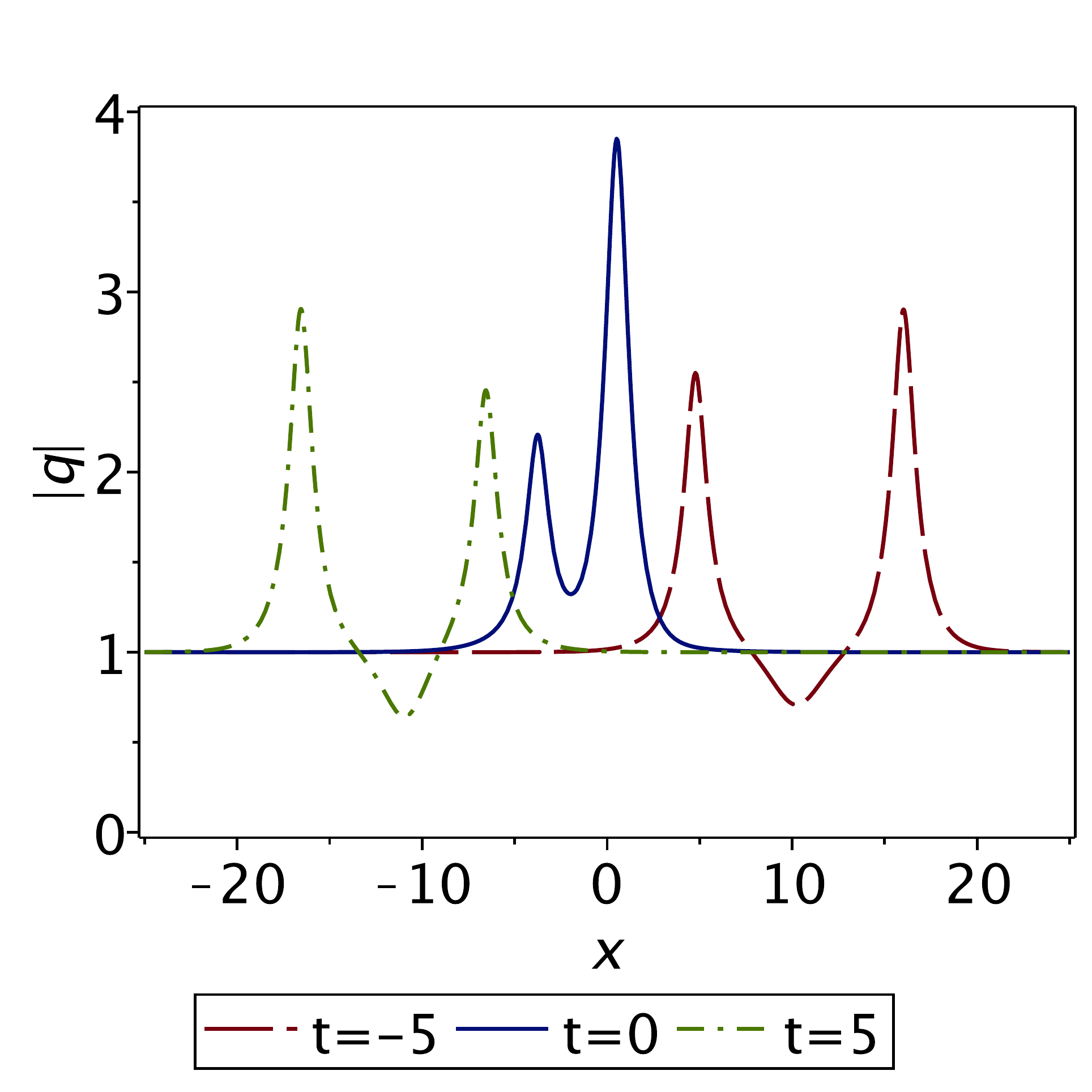}
%\caption{fig1}
\end{minipage}
}%
\centering
\caption{(Color online) The triple-pole bright-dark-bright soliton solution of TOFKN \eqref{T1} with NZBCs and $N_1=0,N_2=1,q_{\pm}=1,\omega_1=2\mathrm{e}^{\frac{\pi}{6}\mathrm{i}},A[\omega_1]=\mathrm{i},B[\omega_1]=1+(1-\sqrt{2})\mathrm{i},C[\omega_1]=1$. (a) The three-dimensional plot; (b) The density plot; (c) The sectional drawings at $t=-5$ (dashed line), $t=0$ (solid line), and $t=5$ (dash-dot line).}
\label{RP-NTPS-F1}
\end{figure}

$\bullet$ When taking parameters $N_1=1,N_2=0,q_{\pm}=1,\zeta_1=\sqrt{2}+\mathrm{i},A[\zeta_1]=B[\zeta_1]=C[\zeta_1]=1$, we can obtain the triple-pole breather-breather-breather solution and give out relevant plots in Fig. \ref{RP-NTPS-F2}. Figs. \ref{RP-NTPS-F2} (a) and (b) exhibit the three-dimensional and density diagrams for the triple-pole breather-breather-breather solution of the TOFKN with NZBCs, respectively. Fig. \ref{RP-NTPS-F2} (c) displays the distinct profiles of the triple-pole breather-breather-breather solution for $t=\pm8,0$. Moreover, from the density plot Fig. \ref{RP-NTPS-F2} (b), we can find that the propagation of the triple-pole breather-breather-breather solution is localized in space $x$ and periodic in time $t$. Obviously, from the Fig. \ref{RP-NTPS-F2} (a) and (b), we can find that the triple-pole breather-breather-breather solution breathes about 10 times at $t=[-8, 8]$ and collides once at near $t=0$. From Fig. \ref{F6}, as taking $N_1=1,N_2=0$, once double-pole solution with NZBCs change to triple-pole solution with NZBCs, the triple-pole solution will arise the third breather branch.

\begin{figure}[htbp]
\centering
\subfigure[]{
\begin{minipage}[t]{0.33\textwidth}
\centering
\includegraphics[height=4.5cm,width=4.5cm]{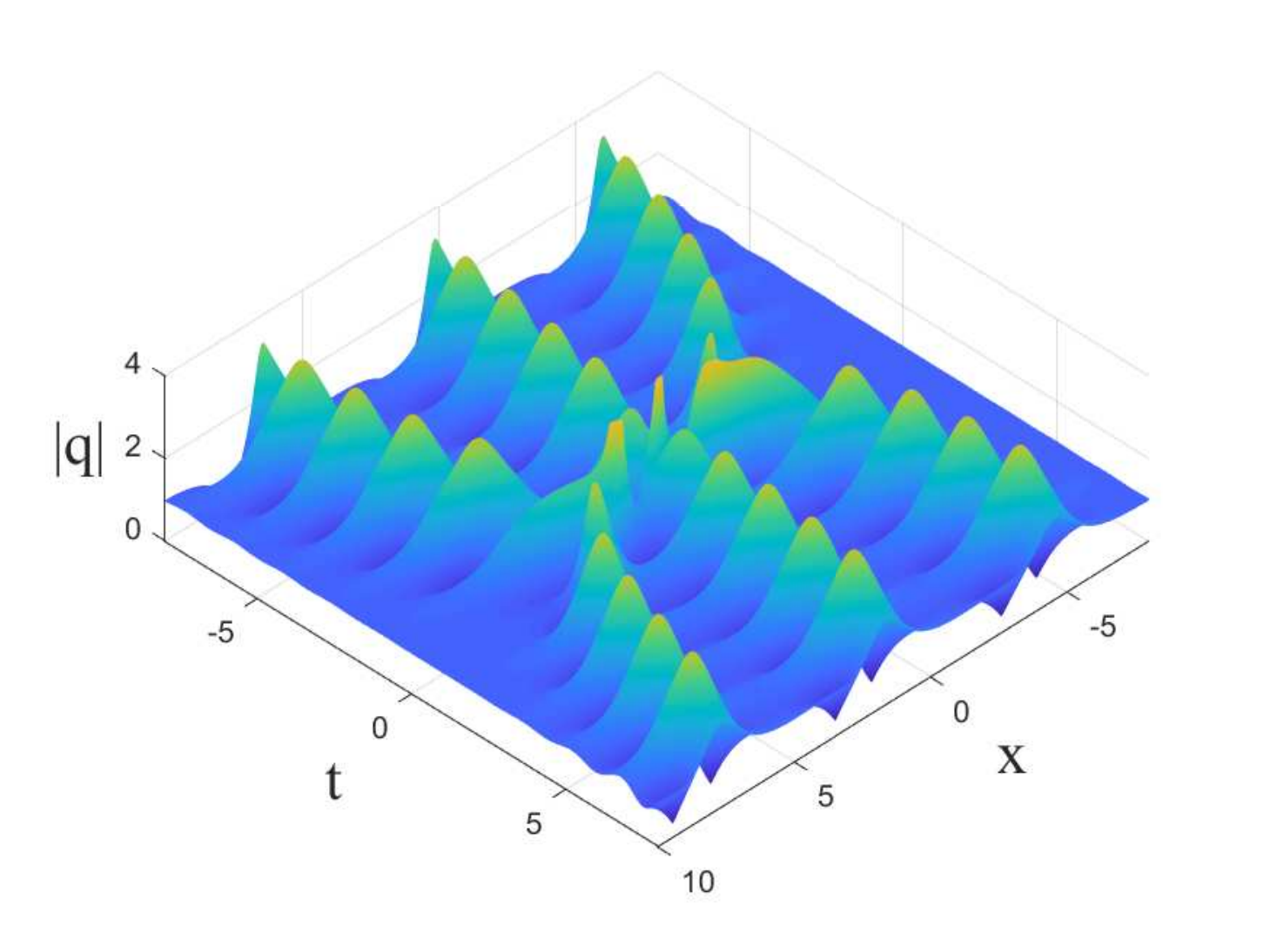}
%\caption{fig1}
\end{minipage}
}%
\subfigure[]{
\begin{minipage}[t]{0.33\textwidth}
\centering
\includegraphics[height=4.5cm,width=4.5cm]{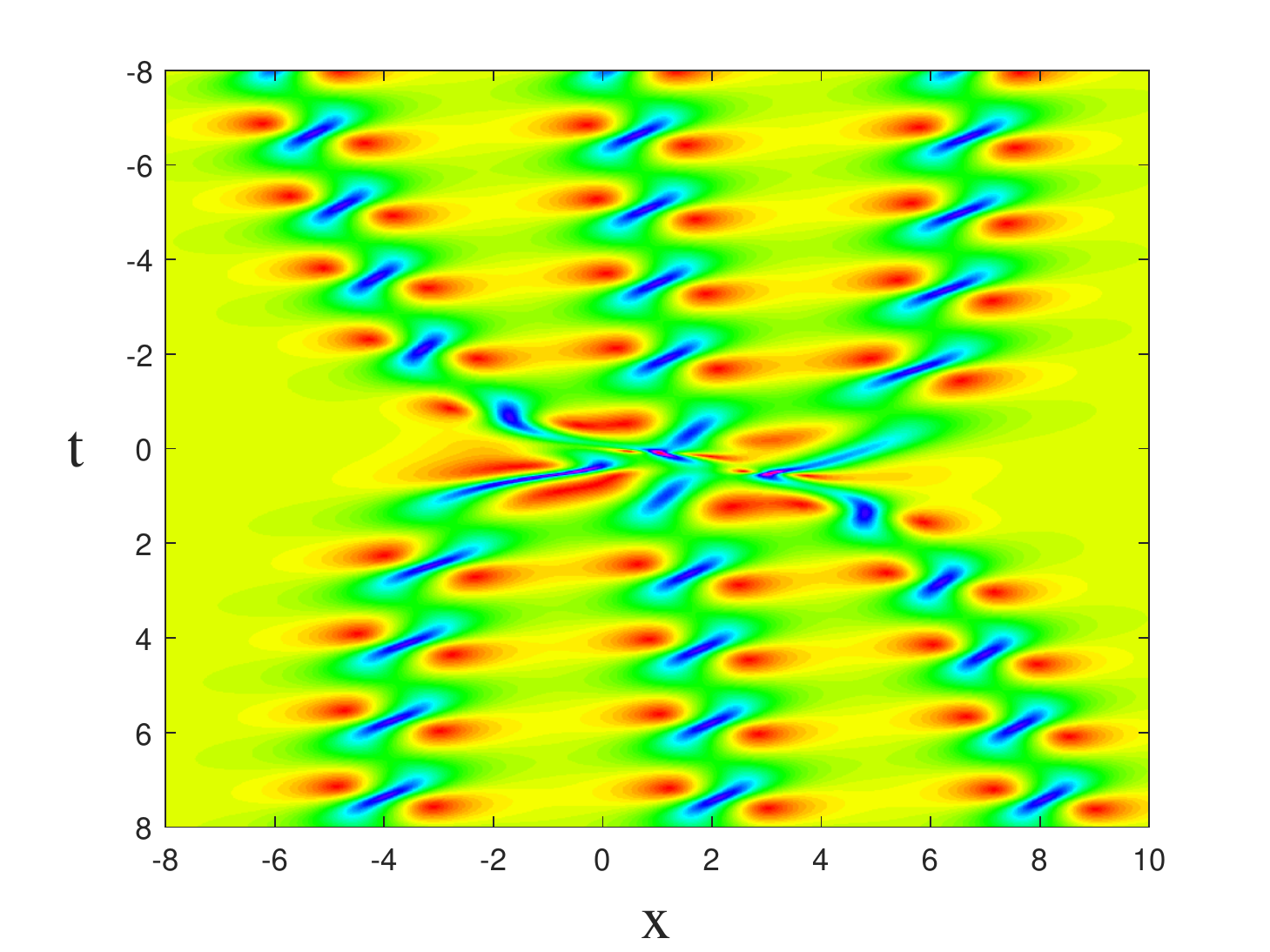}
%\caption{fig2}
\end{minipage}%
}%
\subfigure[]{
\begin{minipage}[t]{0.33\textwidth}
\centering
\includegraphics[height=4.5cm,width=4.5cm]{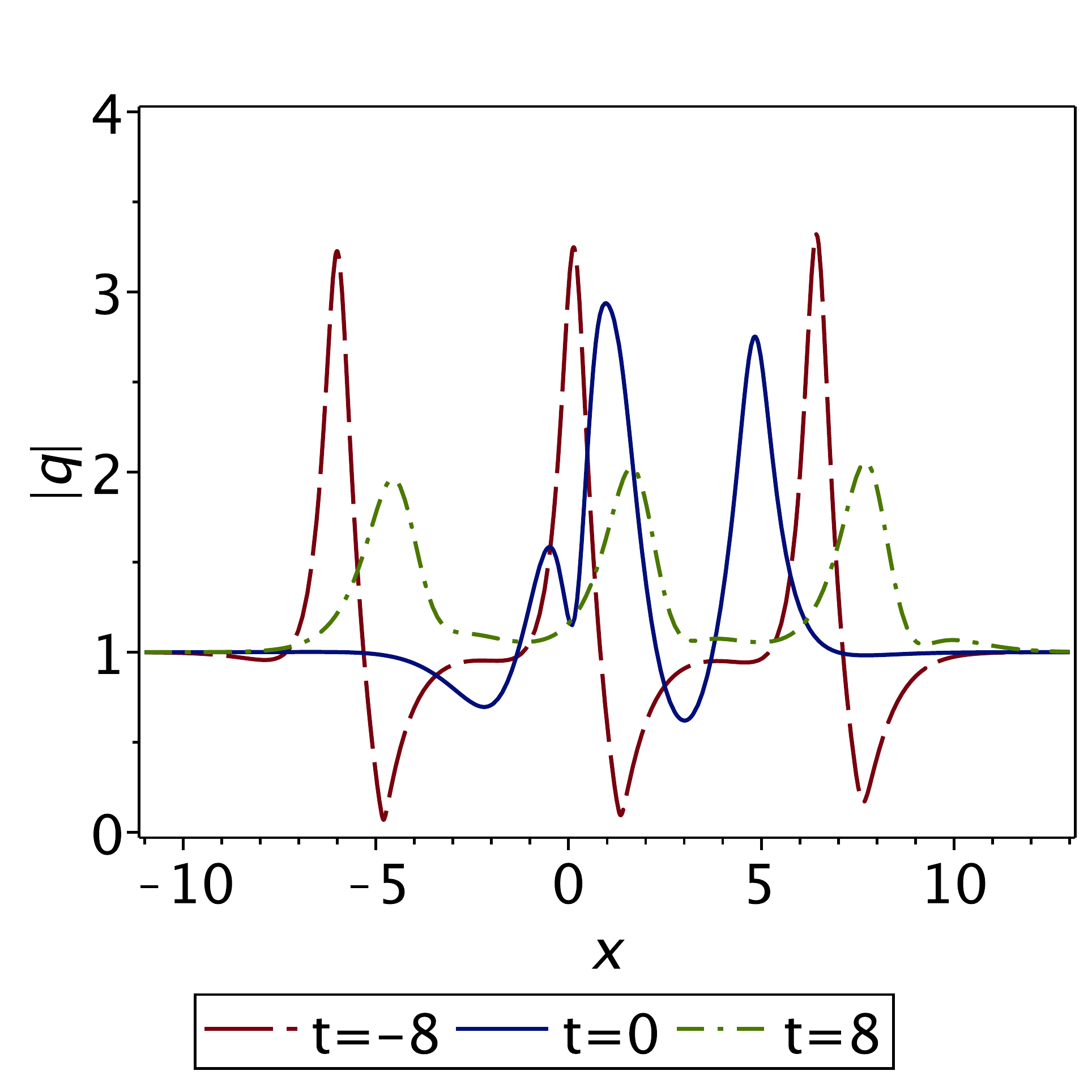}
%\caption{fig1}
\end{minipage}
}%
\centering
\caption{(Color online) The triple-pole breather-breather-breather solution of TOFKN \eqref{T1} with NZBCs and $N_1=1,N_2=0,q_{\pm}=1,\zeta_1=\sqrt{2}+\mathrm{i},A[\zeta_1]=B[\zeta_1]=C[\zeta_1]=1$. (a) The three-dimensional plot; (b) The density plot; (c) The sectional drawings at $t=-8$ (dashed line), $t=0$ (solid line), and $t =8$ (dash-dot line).}
\label{RP-NTPS-F2}
\end{figure}

\section{Conclusions and Discussions}\label{Sec4}
In this paper, we study the IST, double-pole solutions and triple-pole solutions for the TOFKN with ZBCs and NZBCs by means of the RH approach. Starting from spectral problem of the TOFKN, we not only construct the direct scattering problem which illustrates the analyticity, symmetries and asymptotic behaviors, but also establish and solve the inverse problem which can derive the discrete spectrum, trace formula and reflectionless potential with the aid of a matrix RH problem. Specifically, under the condition of reflectionless potential, the general formulas of $N$-double and $N$-triple poles solutions with ZBCs and NZBCs are derived systematically by means of determinants. Vivid illustration, dynamic behavior and asymptotic behavior for some representative double-pole and triple-pole solutions have been given out in details. In the case of ZBCs, we exhibit the dynamic behaviors and associated plots of the $N$-double-pole soliton solutions and $N$-triple-pole soliton solutions when taking fixed $N$, respectively. In the case of NZBCs, when taking different $N_1$ and $N_2$, we give out the dynamic behaviors of abundant double-pole and triple-pole solutions, including the double-pole dark-bright soliton solution as $N_1=0,N_2=1$, the double-pole breather-breather solution as $N_1=1,N_2=0$, the double-pole breather-breather-dark-bright solution as $N_1=1,N_2=1$, the triple-pole bright-dark-bright soliton solution as $N_1=0,N_2=1$ and the triple-pole breather-breather-breather solution as $N_1=1,N_2=0$.

In this work, we have solved many difficult problems and made many novelties. In section \ref{Sec1}, we reveal the corresponding parameter reduction from the coupled TOFKN to the general form of the TOFKN \eqref{T1} specifically, it also provide significant guidance for deriving other higher-order equations of KN systems. For the case of ZBCs in section \ref{Sec2}, we take the lead in obtaining the double-pole and triple-pole soliton solutions with ZBCs of high-order KN system by RH method, and we make a lot of complex derivation and matrix operation, especially in the case of triple-poles. From the asymptotic behavior for the scattering matrix in the direct scattering problem, we connect the constant $\nu$ with the conservation of mass (also called power) $I_1$ of the modified Zakharov-Shabat spectral problem and exhibit essential proof in details. Furthermore, we discover the 1-double-pole soliton solution is equivalent to the elastic collisions of two bright solitons, and the 1-triple-pole soliton solution is equivalent to the elastic collisions of three bright solitons. In order to verify aforementioned conclusion and reveal the long-time asymptotics of the double-pole and triple-pole soliton solution, we derive and analyze the asymptotic states of the 1-double poles soliton solution and the 1-triple poles soliton solution when $t$ tends to infinity. For the case of NZBCs in section \ref{Sec3}, due to the multi-valued case of eigenvalue $k(\lambda)$, we illustrate the scattering problem on a standard $z$-plane by utilizing two single-valued inverse mappings. Since the introduction of $z$-plane, the IST with NZBCs is more complex than the case of ZBCs. Then we take the lead in obtaining the double-pole and triple-pole solutions with NZBCs of high-order KN system by RH approach, we not only make a lot of complex derivation and matrix operation, but also consider more complex symmetries when drawing relevant plots, in which the matrix operation of $12\times12$ is involved in the case of the double-pole breather-breather-dark-bright solution as $N_1=1,N_2=1$ and the triple-pole breather-breather-breather solution as $N_1=1,N_2=0$. Based on the RH method, we further improve the rigorous theory of IST of high-order KN system, and also provides a valuable reference for the IST of the high-order flow equations of KN systems ever other nonlinear integrable systems.

In addition, by comparing the solutions of the TOFKN and the DNLS equation under the ZBCs and NZBCs, it is found that the third-order dispersion and quintic nonlinear term of the KN system can affect both the trajectory and the speed of the solutions. However, the third-order dispersion and quintic nonlinear term of KN system have little influence on the maximum amplitude of the solution at a certain moment. These conclusions are consistent with those in Ref. \cite{Lin(2020)}. These analytical results obtained in this work might have vital reference value for the study of the high-order KN systems even other high-order nonlinear integrable systems, and provide a theoretical basis for possible experimental research and applications. Since the case of the triple-pole solutions with ZBCs and NZBCs involves more complex derivation and matrix calculation, we have not considered the explicit formula of the triple-pole solutions in this paper, which will be further studied in the future. Recently, there are more and more researches on the triple poles solutions even the $N$-pole solutions, and the long-time asymptotics have attracted more and more attention by means of Dbar method and Deift-Zhou steepest descent method. In the future, we will study the multipole solutions and long-time asymptotics for the higher-order flow equations of KN systems even other nonlinear integrable systems.

\section*{Declaration of competing interest}
The authors declare that they have no known competing financial interests or personal relationships that could have appeared to influence the work reported in this paper.

\section*{Acknowledgements}
\hspace{0.3cm}
The authors gratefully acknowledge the support of the National Natural Science Foundation of China (No. 12175069) and Science and Technology Commission of Shanghai Municipality (No. 21JC1402500 and No. 18dz2271000).

\end{document}